\documentclass[aps,prx,letterpaper,reprint,twocolumn,floatfix,superscriptaddress,notitlepage,dvipsnames,svgnames,x11names,table,footinbib,longbibliography]{revtex4-2}
\pdfoutput=1
\usepackage[utf8]{inputenc}
\usepackage[T1]{fontenc}
\usepackage{graphicx}
\usepackage{grffile}
\usepackage{amsfonts}
\usepackage{times}
\usepackage{amsmath,amssymb}
\usepackage{xcolor}
\usepackage{tikz}
\usetikzlibrary{decorations.pathreplacing,arrows.meta,calc}
\usepackage{bm}
\usepackage{booktabs}
\usepackage{dsfont,pifont}
\usepackage{mathtools}
\usepackage{bbm}
\usepackage{colortbl}
\usepackage{arydshln}
\usepackage{booktabs}
\usepackage{mathrsfs}
\usepackage{setspace}
\usepackage{apptools}
\usepackage{appendix}
\usepackage[normalem]{ulem}

\usepackage[braket, qm]{qcircuit}

\usepackage{amsthm}
\newtheorem{theorem}{Theorem}
\newtheorem{proposition}{Proposition}
\newtheorem{definition}{Definition}

\definecolor{gaussred}{RGB}{190,50,50}
\definecolor{isingblue}{RGB}{40,80,180}
\definecolor{medium-blue}{rgb}{0,0,0.8}
\definecolor{forest-green}{cmyk}{0.76,0,0.76,0.25}
\definecolor{hpe-green}{RGB}{2,169,130}
\definecolor{cerulean}{cmyk}{0.98,0.57,0,0.25}

\usepackage[compatibility=false,font=small,labelfont=bf,justification=justified,format=plain]{caption}
\usepackage{ragged2e}
\usepackage{subcaption}

\usepackage[english]{babel}
\usepackage[autostyle,english=american]{csquotes}
\MakeOuterQuote{"}

\usepackage[pdfusetitle]{hyperref}
\hypersetup{pdflang={English},colorlinks=true,linkcolor={hpe-green},citecolor=RoyalBlue,urlcolor=RoyalBlue,breaklinks=true}

\usepackage[nameinlink,capitalize,capitalise]{cleveref}

\graphicspath{{figures/}}

\newcommand{\PRLrefsep}{%
  \par\vspace{0.75\baselineskip}%
  \noindent\makebox[\linewidth][c]{%
    \rule[0.5ex]{0.08\linewidth}{0.25pt}%
    \rule[0.5ex]{0.09\linewidth}{0.45pt}%
    \rule[0.5ex]{0.10\linewidth}{0.70pt}%
    \rule[0.5ex]{0.18\linewidth}{1.00pt}%
    \rule[0.5ex]{0.10\linewidth}{0.70pt}%
    \rule[0.5ex]{0.09\linewidth}{0.45pt}%
    \rule[0.5ex]{0.08\linewidth}{0.25pt}%
  }%
  \par\vspace{0.75\baselineskip}%
}

\newcommand{\Tr}{\operatorname{Tr}}

\begin{document}

\title{Distribution Complexity of Electronic Structure Simulations on Quantum Supercomputers}

\author{Jason Necaise}
\affiliation{HPE Quantum, Emergent Machine Intelligence, Hewlett Packard Labs, CA, USA}
\affiliation{Department of Physics and Astronomy, Dartmouth College, Hanover, NH 03755, USA}

\author{Namit Anand}\thanks{namit.anand@hpe.com}
\affiliation{HPE Quantum, Emergent Machine Intelligence, Hewlett Packard Labs, CA, USA}

\author{Gaurav Gyawali}
\affiliation{HPE Quantum, Emergent Machine Intelligence, Hewlett Packard Labs, CA, USA}

\author{K.~Grace Johnson}
\affiliation{HPE Quantum, Emergent Machine Intelligence, Hewlett Packard Labs, CA, USA}

\author{James D.~Whitfield}
\affiliation{Department of Physics and Astronomy, Dartmouth College, Hanover, NH 03755, USA}

\author{Masoud Mohseni}\thanks{masoud.mohseni@hpe.com}
\affiliation{HPE Quantum, Emergent Machine Intelligence, Hewlett Packard Labs, CA, USA}

\date{\today}

\begin{abstract}

Efficient simulation of strongly-interacting fermionic systems on quantum processing units (QPUs) is a challenging task due to nonlocal mode entanglement generation. However, it is not yet well understood how the structure of entanglement governs the hardness of large-scale quantum chemistry simulations or the scaling of distributing such workloads. Here, we introduce an algorithm for estimating the \textit{distribution complexity} of hybrid quantum-classical simulation for electronic structure Hamiltonians over heterogeneous high-performance architectures. Our algorithm relies on efficient analytical evaluation of the low entanglement boundaries for the orbital rotations and dephasing-induced localization within tensor fragments, in a double-factorized representation. Our entanglement estimation scales as $O(N^3)$ for each fragment, where \(N\) is the number of orbitals. When QPUs are communicating via a quantum network, the cost of distribution per fragment is reduced quadratically from \(O(N^2)\) to \(O(N)\). Similarly, for hybrid quantum-classical approaches, with access to only conventional HPC interconnects, the worst-case cost is reduced from \(O(\exp(N^2))\) to \(O(\exp(N))\). We show that emergent entanglement patterns are induced by the interplay between coherent Gaussian orbital rotations and disordered Coulomb interactions. We discuss the underlying physical mechanisms that govern distribution complexity and introduce model systems that are tunable based on the localizability of fragments and the overlap of interfragment rotations. We characterize three different regimes of hardness for distribution complexity and classical simulability. The framework introduced here enables novel and more efficient quantum-classical application workflows towards utility-scale quantum computing.

\end{abstract}

\maketitle

\let\oldaddcontentsline\addcontentsline%
\renewcommand{\addcontentsline}[3]{}%

%=============================================================================
\section{Introduction}
\label{sec:intro}
%=============================================================================

\begin{figure*}[!thbp]
  \centering
  \includegraphics[width=2\columnwidth]{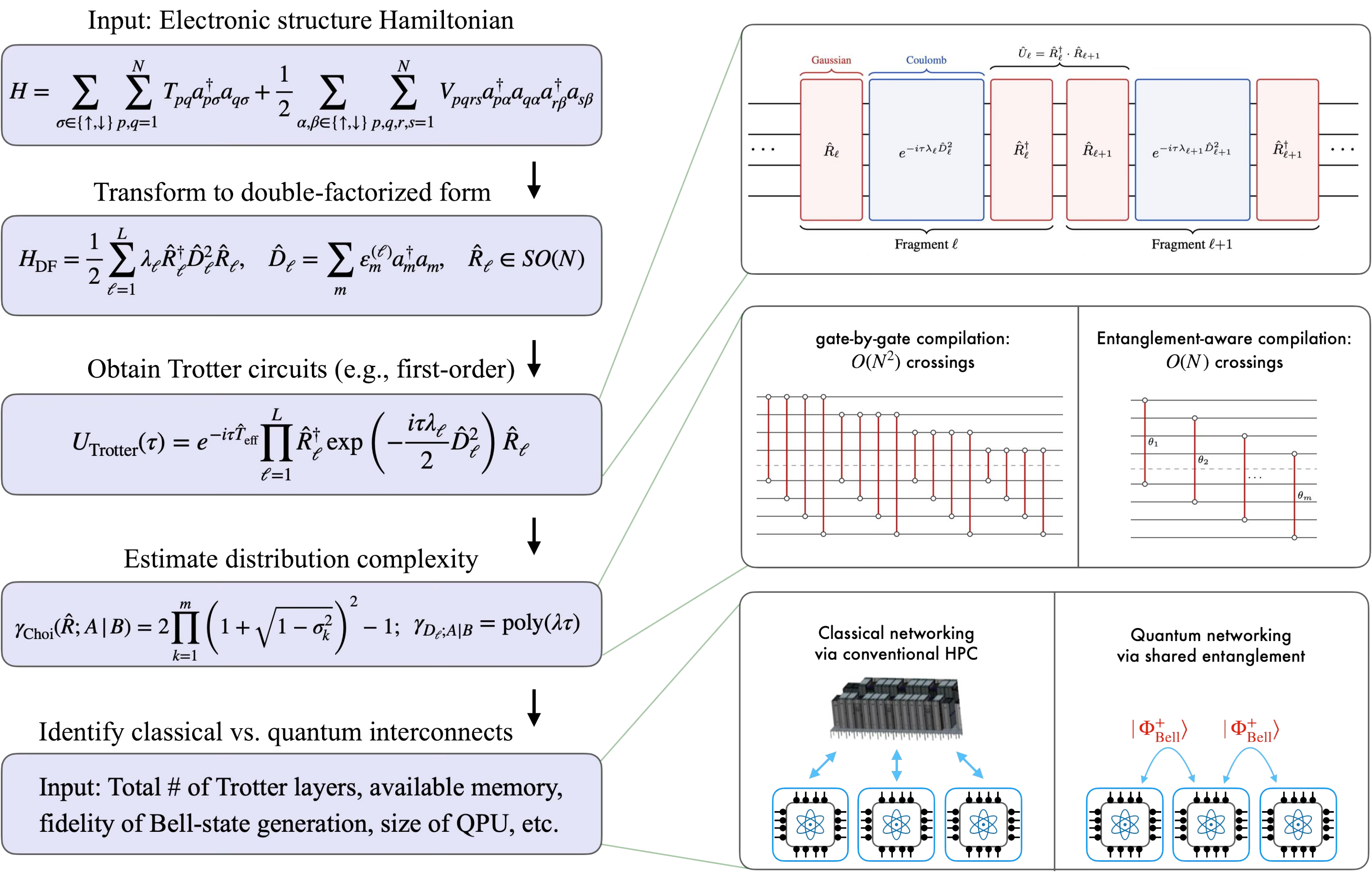}
  \caption{\textbf{Distributed quantum simulation workflow for double-factorized electronic Hamiltonians.} Our work characterizes the distribution cost of Trotter circuits via their operator entanglement, enabling the distribution of quantum simulation workloads among QPUs via conventional HPC interconnects or coherent quantum interconnects. The workflow takes as input an electronic structure Hamiltonian and estimates the overall distribution cost. As an example, for Gaussian orbital rotations, we show how to achieve an entanglement-aware compilation that reduces the inter-QPU communication. If this distribution cost is smaller than some threshold value \(\gamma_\star\) then we simulate it via LOCC-assisted distributed quantum simulation; otherwise we simulate it via entanglement-assisted distributed quantum simulation. The threshold value \(\gamma_\star\) can be a function of the available classical compute, the number of distributed QPUs, quantum hardware characteristics such as gate and measurement times, and the fidelity of Bell pair generation. Our workflow is focused on simulating electronic structure Hamiltonians but can be generalized to other classes of problems.}
  \label{fig:entanglement-hypervisor}
\end{figure*}

Building tightly integrated quantum and HPC systems, or \textit{quantum supercomputers}, capable of tackling large-scale applications will be one of the key technological challenges of the next decade. Over the years, there has been immense progress in identifying hardware, software, and system integration bottlenecks, including the size and characteristics of quantum processors required \cite{mohseni2026buildquantumsupercomputerscaling,chung2026partiallyfaulttolerantquantumcomputation}. Despite significant effort, state-of-the-art quantum algorithms still likely require millions of qubits to enable useful quantum computation at industrial scale. While there is excitement that certain cryptographic applications might show advantage with tens to hundreds of thousands of physical qubits under optimistic assumptions of hardware parameters and quantum error correction capabilities \cite{babbush2026securingellipticcurvecryptocurrencies,webster2026pinnaclearchitecturereducingcost,cain2026shorsalgorithmpossible10000}, the scale of most practical quantum algorithms is realistically beyond what can fit on a single QPU. As a result, utility-scale quantum computing will most likely be achieved via a tightly integrated quantum network of QPUs \cite{Barral2025DistributedQC}. Given the known constraints on quantum interconnects, it will be highly advantageous to identify optimal partitioning techniques for one concrete utility-scale application such as quantum simulation of chemistry and material science. Moreover, it is equally important to offload classically tractable simulations to heterogeneous high-performance classical computing over CPUs, GPUs, FPGAs, and custom-designed ASICs. A systematic understanding of these partitioning schemes and the relevant metrics is currently lacking.

In this work, we provide a rigorous framework for distributing quantum workloads by identifying operator space entanglement as the key metric in circuit partitioning. We focus on the hardness of distributing quantum simulation of electronic structure Hamiltonians in the double-factorized form, which governs the behavior of interacting electrons in molecules and materials. Efficient quantum simulation of such Hamiltonians is anticipated to be one of the cornerstones of quantum computing advantage~\cite{Cao2019QChemQC,motta2021lowrank, lee2021thc,feng2024distributedquantumsimulation}. Our central goal is to characterize the scaling of distributing quantum simulation via quantum interconnects (shared Bell pairs) and/or conventional classical HPC interconnects. In general, the nontrivial cost of Bell pair generation and their low fidelities across QPUs necessitates a careful cost analysis of classical versus quantum interconnects, both in the near-term and the fully fault-tolerant era.

In principle, circuit knitting provides a framework for classically distributing quantum circuits: to compute expectation values of observables, nonlocal gates are replaced by a quasiprobability decompositions (QPD) over local
operations. This replaces entanglement with local unitaries, at the cost of a classical sampling overhead of $\gamma^2$
per cut, where $\gamma$ captures the total strength of terms in quasiprobability decomposition~\cite{peng2020simulating,Tang2021,Bravyi2022,Lowe2023fastquantumcircuit,CarreraVazquez2024,Piveteau2024,piveteau2024knitting,aram_optimal_2025}.  
When multiple gates cross the partition, the distribution overhead is multiplicative,
$\gamma_\mathrm{total} = \prod_k \gamma_k$, so for generic circuits
the total cost grows exponentially in the number of
cuts~\cite{jing2024entanglement}. This serves as a generic upper bound for distributing an arbitrary compiled quantum circuit gate-by-gate. Lower bounds are generally much harder to compute, and even harder to achieve by finding an explicit QPD to match the operator Schmidt decomposition of the total unitary \cite{zanardi_entanglement_2001,styliaris2020informationscrambling,aram_optimal_2025}. Our work extends the class of unitaries for which lower-bounds are either efficient to compute, efficient to achieve (via an explicit compilation scheme), or both. 

Previous work has considered the cost of distributing circuits for electronic structure. Jones and Jacobsen~\cite{jones2025distributed} applied circuit knitting
to fermionic ans\"atze,
compiling the orbital rotations into individual two-qubit Givens gates
and cutting each cross-partition gate independently.
For H$_2$ with Unitary Coupled Cluster Singles and Doubles (UCCSD) in the Jordan--Wigner encoding, $18$ gate cuts produce an overhead of $9^{18} \approx 1.5 \times 10^{17}$. For chains larger than H$_{10}$, most of the extrapolated circuit distribution overheads exceeded $10^{275}$. Of the five ans\"atze they studied, only the Local Unitary Cluster Jastrow (LUCJ) ans\"atz~\cite{Motta2023LUCJ} was remotely feasible, primarily due to the fact that the controlled-phase angles across the spin-sector partition happen to be small. As an application of our circuit partitioning framework, we will show that reported costs of distributing such quantum simulations have been overestimated by several orders of magnitude. 

Such astronomical overheads arise not from the fundamental scaling of the circuit partitioning overhead with (operator) entanglement, but primarily from ignoring the structure of the Trotter circuits. The orbital rotation layers in a double-factorized Trotter step
are not arbitrary unitaries, but rather \emph{Gaussian fermionic
unitaries}. These free-fermion transformations can be fully specified by
$O(N^2)$ parameters despite acting on a $2^N$-dimensional Fock space, where \(N\) is the number of orbitals. The entanglement structure of a Gaussian unitary factorizes over orbital modes,
and the circuit knitting overhead can be computed exactly from the
$N \times N$ rotation matrix without constructing the exponentially
large Fock-space unitary. Similarly, understanding the structure of the diagonalized Coulomb interactions allows us to bound the circuit partitioning overhead. We show that the operator entanglement of the Coulomb layer grows logarithmically in time due to a dephasing mechanism, similar to many-body localization. Therefore, the Coulomb layers can be efficiently distributed across QPUs. Together, these techniques enable significant reduction in the cost of distributing Trotter circuits for double-factorized electronic structure Hamiltonians. Our distributed quantum simulation workflow is summarized in \cref{fig:entanglement-hypervisor}.

Beyond quantifying and reducing distribution complexity, our aim is to carefully characterize the underlying mechanisms of nonlocal mode entanglement generation in electronic structure within a double-factorized representation. These mechanisms include the compatibility between adjacent tensor fragment eigenbases, the eigenvalue spectrum of each fragment, and the number of retained fragments after truncation, among others. This structure is invisible at the level of individual compiled gates. 
To understand the structure of electronic Hamiltonians that makes their classical simulation and distribution nontrivial, we introduce several families of synthetic Hamiltonians that satisfy electronic structure symmetries but have a tunable degree of hardness. These \textit{model molecules} allow us to explore two different mechanisms of distribution hardness: the compatibility of inter-fragment rotations and the block-diagonal structure of the fragment eigenbases. 

It is worth emphasizing that our hybrid quantum-classical algorithms are focused on distributed simulation of electronic structure \textit{across} QPUs, unlike sample-based quantum diagonalization schemes whose focus is to "split" an algorithm across quantum and classical hardware~\cite{Shirakawa2025IBM,gunther2026biomolecularfreeenergies,Kirby2026IBM}. The latter primarily uses classical HPC resources for post-processing of quantum data. In contrast, we employ classical HPC pre-processing to minimize the communication/sampling overhead for either quantum or classical interconnects for distributed computing.

The paper is organized as follows. \cref{sec:distribution_framework} introduces the framework for quantifying the hardness of circuit partitioning via classical vs. quantum interconnects. \cref{df_intro} focuses the discussion to double-factorized electronic Hamiltonians and identifies the Gaussian and Coulomb layers whose costs must be estimated. \cref{sec:structure} explains how molecular structure, truncation, gauge freedom, and fragment ordering affect the resulting distribution complexity, and \cref{sec:synthetic} introduces synthetic Hamiltonian families to illustrate how these mechanisms control the complexity landscape. We conclude with a discussion in \cref{sec:conclusion}.

\section{Quantum and Classical Distribution Complexity}
\label{sec:distribution_framework}

To mitigate the practical upper limits on physical QPU sizes and avoid the substantial cost of building a monolithic quantum computing architecture, one must recast utility-scale quantum simulation as a distributed computing problem. To partition a large (logical) quantum circuit onto (logical) qubits across distinct QPUs, every operation crossing that partition must be replaced by some interprocessor resource. A quantum interconnect can supply shared entanglement that can be used to implement nonlocal operation coherently~\cite{Barral2025DistributedQC,Yue_2019,Wu2023DistributedEntanglement}. Alternatively, circuit knitting has been introduced as a distributed quantum computing method that uses only local quantum operations and classical postprocessing, replacing nonlocal channels by quasiprobability decompositions (QPDs) over implementable local operations~\cite{peng2020simulating,mitarai2021overhead,piveteau2024knitting,johnson2026distributedquantumcomputingadaptive}. Here, we use the term \textit{distribution complexity} to refer to the scaling of the physical and computational resources for pre-processing partitioning, inter-QPU quantum/classical communications, and post-processing required by either strategy. 

% For electronic structure Hamiltonians studied in this work, the distribution scaling is characterized primarily as a function of the number of orbitals.

For a fixed bipartition $A|B$, a QPD of a nonlocal channel $\mathcal U$ has the form
\begin{equation}
    \mathcal U
    =
    \sum_i a_i \ \mathcal F_i ,
    \qquad
    a_i\in\mathbb R ,
    \label{eq:qpd_def_short}
\end{equation}
where each $\mathcal F_i$ is implementable by local operations and classical communication (LOCC) between $A$ and $B$. Any LOCC quantum operation can be written as \(\mathcal F_i = \sum_{j} p_j \mathcal{A}_{j} \otimes \mathcal{B}_{j}\), where \(\mathcal{A}_{j}, \mathcal{B}_{j}\) are local quantum channels and \(\{p_j\}\) is a probability distribution. The ``quasiprobability extent'' \cite{piveteau2025side} of a specific QPD protocol $\Pi$ (also called ``product extent'' when the operators $\mathcal{A}_{j}, \mathcal{B}_{j}$ are unitaries \cite{aram_optimal_2025}), with coefficients $a_i^{(\Pi)}$, is given by the one-norm
\begin{equation}
    \gamma_{\Pi}(\mathcal U)
    =
    \sum_i |a_i^{(\Pi)}| .
    \label{eq:qpd_gamma_protocol_def}
\end{equation}

QPD sampling protocols provide an unbiased estimator for expectation values using only product unitaries on systems $A$ and $B$. For a protocol with one-norm $\gamma_\Pi$, the number of circuit executions required to estimate a bounded observable to precision $\varepsilon$ scales as $O(\gamma_\Pi^2/\varepsilon^2)$. In this work, we discuss the parameter $\gamma_\Pi$ as the fundamental parameter, noting that the explicit circuit sampling overhead of the estimator is set by $\gamma_\Pi^2$~\cite{peng2020simulating,mitarai2021overhead,piveteau2024knitting}. This gives us the following definition.

\begin{definition}
The LOCC-assisted \emph{distribution complexity} of a unitary $\mathcal U$ is the minimum of the quasiprobability extent over all local decompositions:
\begin{equation}
    \gamma_{\mathrm{LOCC}}(\mathcal U)
    =
    \min_{\mathcal U=\sum_i a_i\mathcal F_i}\ 
    \sum_i |a_i| ,
    \qquad
    \mathcal F_i\in\mathrm{LOCC}.
    \label{eq:qpd_gamma_def}
\end{equation}
\end{definition}

For a unitary channel $\mathcal U$ such that \(\mathcal{U}(\rho) = U\rho U^{\dagger}\), the robustness of entanglement of its Choi-Jamiolkowski state gives a lower bound on the LOCC-assisted distribution complexity~\cite{Harrow_2003_robustness,zanardi_entanglement_2001,piveteau2024knitting,aram_optimal_2025}. If $\{\alpha_k\}$ are the operator Schmidt coefficients of $U$ across $A|B$, normalized so that $\sum_k \alpha_k^2=1$, then
\begin{equation}
    \gamma_{\mathrm{Choi}}(\mathcal U)
    = 
    2\Bigl(\sum_k \alpha_k\Bigr)^2 - 1, \quad  \gamma_{\mathrm{LOCC}}(\mathcal U)
    \geq \gamma_{\mathrm{Choi}}(\mathcal U).
    \label{eq:choi_bound}
\end{equation}
For generic many-qubit unitaries, evaluating this Choi lower bound requires access to the full Fock-space unitary. As we will see, for the Gaussian orbital rotations appearing in double-factorized electronic simulation, this exponential representation is unnecessary: the same Choi spectrum can be computed from single-particle information (see \cref{thm:gaussian_choi_bound}).

Similar to the LOCC-assisted distribution cost, we can also estimate the \textit{entanglement-assisted} distribution cost of a quantum state or unitary. For a pure state $|\psi\rangle$, the entanglement-assisted distribution complexity is the logarithm of the Schmidt rank, equivalently $S_0(\rho_A)$ for the reduced state~\cite{Yue_2019,theurer2023singleshotentanglementmanipulationstates}. Monotonicity of R\'enyi entropies gives us the following inequalities,
\begin{equation}
    S_0(\rho_A)
    \geq
    S_{1/2}(\rho_A)
    \geq
    S_{\mathrm{vN}}(\rho_A)
    \geq
    S_2(\rho_A) .
    \label{eq:renyi_hierarchy_main}
\end{equation}
The LOCC-assisted distribution complexity is approximately equal to the exponential of the R\'enyi entropy $S_{1/2}(\rho_A)$, which lower bounds the zeroth-order Renyi entropy as per the inequality above~\cite{vidal1999robustness,johnson2026distributedquantumcomputingadaptive}. Therefore, the logarithm of the LOCC-assisted distribution complexity lower bounds the entanglement-assisted distribution complexity for coherent state preparation and/or unitary dynamics. Note that the entanglement-assisted distribution cost, \(S_0(\rho_A)\) is roughly the number of inter-QPU Bell pairs needed to prepare a quantum state/unitary~\cite{Yue_2019,theurer2023singleshotentanglementmanipulationstates}.

% Channel-level and state-dependent extents can differ sharply: a permutation or SWAP-like unitary introduces maximal entanglement as a channel, while its action is easily tracked on product input states~\cite{PhysRevLett.107.180501,Wu2023DistributedEntanglement}.

The distribution-complexity framework we introduce here turns partition selection into a classical preprocessing task. Adaptive circuit knitting follows this principle by choosing cuts with small distribution complexity when the entanglement structure is heterogeneous~\cite{mohseni2026buildquantumsupercomputerscaling,johnson2026distributedquantumcomputingadaptive}. For generic circuits, computing the relevant channel-level extent is limited by the exponential size of the unitary. For double-factorized electronic structure circuits, we are able to avoid this bottleneck by analyzing the structure of the Gaussian and diagonal layers to expose the partitioning data, as we discuss now. 

\section{Distributing Electronic Simulation}
\label{df_intro}

\begin{figure*}[tbhp]
  %\centering
  \begin{tikzpicture}[
    wire/.style={thick},
    gateblock/.style={draw,thick,minimum width=1.4cm,minimum height=3.2cm,
                      rounded corners=2pt,font=\small,align=center},
    dots/.style={font=\Large},
    x=0.72cm, y=0.7cm
  ]
  % Wires
  \foreach \i in {0,...,1} {
    \draw[wire] (-0.5,-\i) -- (20,-\i);
    \node[left,font=\small] at (-1.1,-\i) {$q_{\i}$};
  }
    \foreach \i in {2,...,2} {
    \draw[wire] (-0.5,-\i) -- (20,-\i);
    \node[left,font=\small] at (-0.5,-\i) {$q_{N-2}$};
  }
    \foreach \i in {3,...,3} {
    \draw[wire] (-0.5,-\i) -- (20,-\i);
    \node[left,font=\small] at (-0.5,-\i) {$q_{N-1}$};
  }
  \def\cy{-1.5}

  % Vertical dots
  \node[dots] at (-1.1,\cy + 0.15) {\small{$\vdots$}};
  % Horizontal dots
  \node[dots] at (0.7,\cy) {$\cdots$};

  % Fragment ell
  \node[gateblock, draw=gaussred, fill=gaussred!6]
    (Rl) at (2.8,\cy) {$\hat{R}_\ell$};
  \node[gateblock, draw=isingblue, fill=isingblue!6, minimum width=2.6cm]
    (Dl) at (5.8,\cy) {$e^{-i\tau\lambda_\ell \hat{D}_\ell^2}$};
  \node[gateblock, draw=gaussred, fill=gaussred!6]
    (Rld) at (8.8,\cy) {$\hat{R}_\ell^\dagger$};

  % Fragment ell+1
  \node[gateblock, draw=gaussred, fill=gaussred!6]
    (Rl1) at (11.2,\cy) {$\hat{R}_{\ell+1}$};
  \node[gateblock, draw=isingblue, fill=isingblue!6, minimum width=2.6cm]
    (Dl1) at (14.2,\cy) {$e^{-i\tau\lambda_{\ell+1} \hat{D}_{\ell+1}^2}$};
  \node[gateblock, draw=gaussred, fill=gaussred!6]
    (Rl1d) at (17.2,\cy) {$\hat{R}_{\ell+1}^\dagger$};

  % Right dots
  \node[dots] at (19,\cy) {$\cdots$};

  % % --- NEW: dotted rounded rectangle around \hat{R}_l^\dagger and \hat{R}_{l+1}
  % \draw[dotted, thick, rounded corners=6pt]
  %   ($(Rld.north west)+(-0.3,0.3)$) rectangle
  %   ($(Rl1.south east)+(0.3,-0.3)$);

  % Top labels
  \draw[decorate,decoration={brace,amplitude=4pt,raise=2pt},gaussred,thick]
    (Rl.north west) -- (Rl.north east)
    node[midway,above=6pt,font=\scriptsize,gaussred] {Gaussian};

  \draw[decorate,decoration={brace,amplitude=4pt,raise=2pt},isingblue,thick]
    (Dl.north west) -- (Dl.north east)
    node[midway,above=6pt,font=\scriptsize,isingblue] {Coulomb};

  % --- NEW: combined overbrace (uncolored)
  \draw[decorate,decoration={brace,amplitude=4pt,raise=6pt},thick]
    (Rld.north west) -- (Rl1.north east)
    node[midway,above=8pt,font=\scriptsize]
    {$\hat{U}_\ell = \hat{R}_\ell^\dagger \cdot \hat{R}_{\ell+1}$};

  % Bottom braces (unchanged)
  \draw[decorate,decoration={brace,amplitude=5pt,mirror,raise=3pt},thick]
    (Rl.south west) -- (Rld.south east)
    node[midway,below=8pt,font=\small] {Fragment $\ell$};

  \draw[decorate,decoration={brace,amplitude=5pt,mirror,raise=3pt},thick]
    (Rl1.south west) -- (Rl1d.south east)
    node[midway,below=8pt,font=\small] {Fragment $\ell{+}1$};

\end{tikzpicture}
  \caption{\textbf{Trotter circuits for double-factorized electronic Hamiltonians.} Quantum circuit for a single first-order Trotter step of the double-factorized Hamiltonian. Each fragment \(\ell\) consists of a diagonal Coulomb interaction (\textcolor{RoyalBlue}{blue}), conjugated with a Gaussian orbital rotation (\textcolor{BrickRed}{red}). Adjacent fragment layers contain composite interfragment Gaussian rotations \(\hat{U}_{\ell}=\hat{R}_{\ell}^{\dagger} \cdot \hat{R}_{\ell+1}\). The distribution complexity is therefore controlled by two types of objects: the compatibility of neighboring fragment bases and the diagonal interactions inside each fragment.}
  \label{fig:trotter_circuit}
\end{figure*}
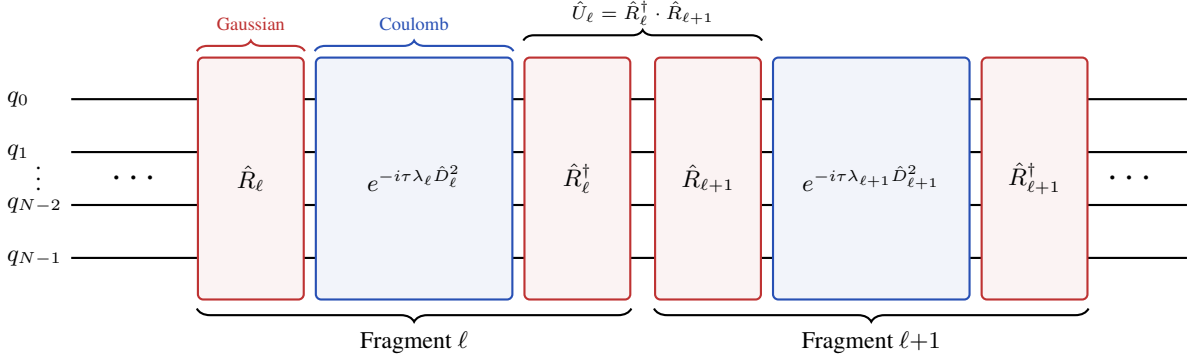

After choosing a finite set of single-particle basis functions, the electronic
Hamiltonian is represented on the fermionic Fock space generated by the
corresponding spin-orbitals.  We denote the fermionic creation and annihilation
operators by \(a_{p\sigma}^\dagger\) and \(a_{p\sigma}\), where \(p\) labels a
spatial orbital and \(\sigma \in \{\uparrow,\downarrow\}\) labels spin.  They
satisfy the canonical anticommutation relations
\( \{a_{p\sigma}, a_{q\tau}^\dagger\} \equiv a_{p\sigma}a_{q\tau}^\dagger + a_{q\tau}^\dagger a_{p\sigma}=\delta_{pq}\delta_{\sigma\tau}\) and
\(\{a_{p\sigma},a_{q\tau}\}=\{a_{p\sigma}^\dagger,a_{q\tau}^\dagger\}=0\).
At fixed electron number, this gives a finite-dimensional many-electron
Hilbert-space approximation to the continuum problem~\cite{szabo1996modern,Helgaker2000,chien2026simulatingfermions}. 
Concrete quantum circuits then require a representation of the fermionic algebra on qubits. Standard one-to-one encodings such as Jordan--Wigner, parity, and Bravyi--Kitaev transformations represent the same fermionic Fock space with different choices of where occupation and parity information are stored~\cite{bravyi2002fermionicquantumcomputation,seeleyBravyiKitaevTransformationQuantum2012,chien2026simulatingfermions}.
In encodings other than Jordan--Wigner, the occupation of a spatially localized orbital is stored across several qubits. This can spread locality and entanglement in qubit space even when the fermionic Hamiltonian has localized orbital structure. 
The Jordan--Wigner encoding has a high Pauli weight, but it keeps orbital occupation information directly in computational-basis qubit occupations, so diagonal Coulomb terms retain the simple and transparent form $n_p:=(1-Z_p)/2$. Despite its high Pauli weight, the regular string-like structure turns out to be quite favorable after circuit synthesis and hardware routing, especially for smaller systems and early FTQC algorithms~\cite{gao2026linearcomplexityfermionic}. As a result, we focus our study of distribution complexity to Jordan-Wigner transformation in this work. 
More sophisticated encodings, including ternary tree, superfast, compact, and succinct fermionic data-structure constructions~\cite{setiaBravyiKitaevSuperfastSimulation2018,jiang_optimal_2020,miller2023bonsai,miller2024treespilation,harrison2024sierpinski,derby2021compact,kirby2021compactencoding,kirby2022secondquantized,carolan2025succinct}, are expected to require encoding-specific algorithms to achieve comparable distribution complexity. 

The electronic Hamiltonian in second quantization with $N$ spatial orbitals is composed of a quadratic effective one-body term and the quartic electron-electron potential term:
\begin{align}
    H_\mathrm{ES} = 
    &\sum_{\sigma\in\{\uparrow,\downarrow\}}
    \sum_{p,q=1}^{N}
    T_{pq}\,
    a_{p\sigma}^\dagger a_{q\sigma}  \nonumber\\
    &+ \ 
    \frac{1}{2}
    \sum_{\alpha,\beta\in\{\uparrow,\downarrow\}}
    \sum_{p,q,r,s=1}^{N}
    V_{pqrs}\,
    a_{p\alpha}^\dagger a_{q\alpha}
    a_{r\beta}^\dagger a_{s\beta}
    \label{eq:electronic_H}
\end{align}
Here $V_{pqrs}=(pq|rs)$ are the spin-independent electronic repulsion integrals (ERI), and the effective one body coefficients 
\begin{equation}
    T_{pq}
    =
    h_{pq}
    -
    \frac{1}{2}\sum_{r=1}^{N} V_{prrq}
    \label{eq:T_eff}
\end{equation}
are given by each electron's kinetic energy, the nuclear-electron potentials, and an effective one-body shift from this ordering of the quartic term \cite{szabo1996modern,lee2021thc}.
The basis set error is fixed once a finite orbital space is chosen. The tensors $T$ and $V$ then  define the operators acting on the $N$-body space spanned by Slater determinants of the single-electron orbitals ~\cite{Helgaker2000,berry2019qubitization,lee2021thc}.
The circuit complexity and distribution problem is dominated by the two-body term; tensor factorizations act on the spatial-orbital indices of the ERI tensor $V$. 
A full fermionic encoding maps the occupation of each spin-orbital to a separate qubit, thus requiring $2N$ qubits for $N$ spatial orbitals ~\cite{bravyi2002fermionicquantumcomputation,seeleyBravyiKitaevTransformationQuantum2012,berry2019qubitization,lee2021thc,chien2026simulatingfermions}. In this work, we focus on restricted closed-shell structure, where the $\alpha$ and $\beta$ sectors share the same spatial orbital basis and every orbital rotation acts identically on both spin sectors\footnote{We could think of partitioning the spin sectors onto two QPUs, after which we would have to partition spatial orbitals within each sector. For a single balanced bipartition for each sector, this would result in four QPUs of size $N/2$ qubits each to simulate the evolution of $2N$ spin-orbitals}. Therefore, we consider the problem of spatial orbital partitions within each spin sector, and so after this point we drop the spin index entirely. We now introduce the \emph{double factorization} of the spatial orbital ERI tensor~\cite{whitten1973coulombic,aquilante2011cholesky,peng2017lowrank,motta2021lowrank}.

For molecular systems that obey time reversal symmetry, the spatial orbital functions can be chosen to be real~\cite{Helgaker2000}. The resulting ERI symmetries allow $V$ to be interpreted as a real symmetric pair-density matrix $W_{(pq),(rs)}=V_{pqrs}$. 
In this space, the quartic terms over the fermionic ladder operators can be written as a quadratic contraction of the pair-wise operators $a^\dagger_p a_q$:
\begin{equation}
    H_\mathrm{two-body} = \sum_{(pq)=1}^{N^2}
    \sum_{(rs)=1}^{N^2} a^\dagger_p a_q \ W_{(pq,rs)} \  a^\dagger_r a_s
\end{equation}
The first factorization takes the spectral form of this ERI matrix
\begin{equation}
    W
    =
    \sum_{\ell=1}^{L}
    \lambda_\ell \ 
    \lvert g^{(\ell)}\rangle\!\langle g^{(\ell)}\rvert,
    \label{eq:W_spectral}
\end{equation}
where $|g^{(\ell)}\rangle$ is the eigenvector on the $N^2$ dimensional space, which also defines a coefficient matrix of a quadratic fermionic Hamiltonian, denoted as $g^{(\ell)}_{p,q} = \langle pq | g^{(\ell)}\rangle$. Doing so allows us to write the two-body part of the electronic Hamiltonian
\begin{equation}
    H_{\mathrm{two-body}} = \frac{1}{2} \sum_{\ell =1}^{L} \lambda_\ell \left(\sum_{p,q} g^{(\ell)}_{pq} a^\dagger_p a_q \right)^2 
    \label{eq:fragments_quadratic}
\end{equation}
where $L$ is the rank of the outer factorization of the ERI matrix $W$. This means we can write the quartic two-body part of the electronic Hamiltonian as a linear combination of squared quadratic fermionic Hamiltonians. 
Each term in this outer factorized sum is a
\emph{tensor fragment} \cite{bellonzi2025qbgsee}, a squared one-body component of the total two-body operator,
which is not to be confused with a molecular fragment, active space, or embedding region of the system~\cite{roos1980complete,knizia2012density,kitaura1999fragment,FMO_review_2021}. For brevity, we drop the qualifier and refer to each term $\ell$ in \cref{eq:fragments_quadratic} as a ``fragment''~\footnote{The semantics vary across the literature, for e.g., ``fermionic fragments''\cite{oumarou2024compressed}, ``leafs'' \cite{oumarou2024compressed}, ``pairwise operators'' for $\hat{G}_\ell$\cite{motta2021lowrank}, ``Cholesky vectors'' for $|g^{(\ell)}\rangle$\cite{motta2019afqmc_lowrank}. Our exact fragment terminology is from \cite{bellonzi2025qbgsee}.}.

The coefficient matrix $g^{(\ell)}$ for each of the $L$-many fragments is guaranteed to be real and symmetric by the symmetries of the ERI tensor, and can be orthogonally diagonalized on the single-particle space via standard free fermion transformations. This yields the second factorization:
\begin{equation}
    g^{(\ell)}
    =
    R_\ell\
    \mathrm{diag}
    (\varepsilon_1^{(\ell)},\ldots,\varepsilon_N^{(\ell)}) \ 
    R_\ell^\dagger,
    \qquad
    R_\ell\in\mathrm{SO}(N),
    \label{eq:eigenmatrix}
\end{equation}
where we explicitly note that $g^{(\ell)}$ and the diagonalizing orbital rotation $R_\ell$ are all $N \times N$ matrices of coefficients. 
We apply distinguish between compact single-particle coefficient tensors and their associated Fock space operators applies via a circumflex. For Gaussian fermionic unitaries and their generators, this means $R_\ell$ for the single-particle rotation matrix, and the circumflex notation $\hat R_\ell$ for the induced Fock-space Gaussian fermionic unitary that implements it as a \(2^N \times 2^N\) transformation \cite{bravyi2005lagrangianflo,kivlichan2018fswap}. 
Putting them together, the double factored form of the two-body Hamiltonian is:
\begin{equation}
    H_{\mathrm{two-body}}
    =
    \frac{1}{2}
    \sum_{\ell=1}^{L}
    \lambda_\ell\,
    \hat R_\ell^\dagger
    \hat{D}_\ell^2
    \hat R_\ell,
    \qquad
    \hat{D}_\ell
    =
    \sum_m
    \varepsilon_m^{(\ell)} a_m^\dagger a_m .
    \label{eq:H_df}
\end{equation}
Double factorization rewrites the quartic two-body Hamiltonian as a linear combination of rotated diagonal fermionic fragments. Each fragment is simple in its own orbital basis. The nontrivial structure results from the distinct rotations $R_\ell\neq R_k$ between fragments, and how much compression is available. 
In many common bases, the coefficients of the ERIs are often linearly dependent. This allows methods such as density fitting, resolution-of-the-identity, and Cholesky constructions to represent the Coulomb interaction using significantly less space than all $O(N^4)$ integrals~\cite{whitten1973coulombic,dunlap1979approximations,vahtras1993integral,beebe1977simplifications,koch2003reduced,aquilante2011cholesky,folkestad2019efficient,aquilante2017innerprojection}. 

For localized Gaussian bases, spatial decay and integral screening make the ERI pair-density matrix highly compressible. In compound and nested low-rank decompositions, one first retains $L$ outer factors and then truncates each retained matrix to an inner rank $\rho_\ell$. Peng and Kowalski reported storage scaling consistent with $O(N^2\log(N))$ for carbon-hydrogen benchmark families~\cite{peng2017lowrank}, and Motta et al. used this behavior to motivate $L\sim O(N)$ and $\langle\rho_\ell\rangle\sim O(\log N)$ for double-factorized Hamiltonian truncations~\cite{motta2021lowrank}. Later on, auxiliary-field quantum Monte Carlo (AFQMC) work argued that $\langle\rho_\ell\rangle$ may even saturate in large localized systems, up to possible logarithmic factors~\cite{motta2019afqmc_lowrank}, although all of these appear to be empirical and representation-dependent.

A standard route for simulating electronic structure Hamiltonians is to
approximate the time evolution via  Trotterization, a product formula
that decomposes the dynamics into implementable one- and two-body
layers~\cite{kivlichan2018fswap,motta2021lowrank,childs2021theory,Campbell_2021}.  For a total simulated
time \(t\), the step size \(\tau=t/n_\mathrm{steps}\) is chosen from the product-formula
order, here first order, and from the allowed product-formula error
\(\epsilon_{\mathrm{PF}}\)~\cite{childs2021theory,mehendale2023estimating}.  The relevant tolerance is
application dependent: real-time dynamics, phase estimation, and
ground-state energy estimation impose different accuracy requirements on
the simulated time evolution. An alternative to Trotterization for quantum simulation is the so-called qubitization or quantum signal processing framework~\cite{low2019hamiltonian,babbush2018encoding}. Qubitization typically treats the Hamiltonian as a \textit{global} operator and washes out its locality. As we will show, distributing a nonlocal operator without additional structure typically incurs a cost that is \textit{exponential} in the number of qubits/modes, and therefore it is important to preserve its local information. We anticipate that techniques such  as the fragment molecular orbital (FMO) method might be able to circumvent some of these limitations \cite{FMO_review_2021}. Moreover, while qubitization and signal processing based methods enjoy asymptotically optimal quantum simulation costs, for many practical applications to condensed matter physics and chemistry, Trotterization still retains better scaling in practice ~\cite{babbush2018encoding,Campbell_2021}. For these reasons, our primary focus in this work is on Trotterization based approaches to quantum simulation. 

Here we analyze the distribution cost of a
single double-factorized (DF) Trotter step, the primitive repeated by the full
simulation circuit.  This layerwise perspective is especially natural for
early fault-tolerant and modular architectures, where nonlocal operations
can be a resource bottleneck~\cite{mohseni2026buildquantumsupercomputerscaling,chung2026partiallyfaulttolerantquantumcomputation}. Denoting the quadratic part sum of \cref{eq:electronic_H} as $\hat{T}_\mathrm{eff}=\sum_{p,q} T_{pq} a^\dagger_p a_q$ and noting that its evolution is also Gaussian, 
the Trotter circuit
\begin{equation}
    \hat{U}_{\mathrm{Trotter}}(\tau)
    =
    e^{-i\tau \, \hat{T}_{\mathrm{eff}}}
    \prod_{\ell=1}^{L}
    \hat R_\ell^\dagger
    \exp\left(-\frac{i\tau\lambda_\ell}{2} \hat{D}_\ell^2\right)
    \hat R_\ell .
    \label{eq:trotter}
\end{equation}
alternates Gaussian orbital rotations with diagonal Coulomb interactions in the fragment eigenbases, illustrated in \cref{fig:trotter_circuit}. The composition of adjacent Gaussian rotations is another Gaussian rotation.
Thus a DF Trotter step contains two boundary Gaussian layers $\hat U_0$ and $\hat U_L$ and \(L-1\)
inter-fragment Gaussian layers $\hat U_\ell$,
\begin{equation}
    \hat U_0
    =
    e^{-i\tau \hat T_{\mathrm{eff}}}\hat R_1,
    \qquad
    \hat U_\ell
    =
    \hat R_\ell^\dagger \hat R_{\ell+1},
    \qquad
    \hat U_L
    =
    \hat R_L^\dagger .
    \label{eq:telescoped}
\end{equation}
If neighboring fragments share an eigenbasis, the corresponding inter-fragment rotation cancels. The distribution complexity of the Gaussian part can be thought of as a direct measure of basis incompatibility between fragments.

Rather than assigning a distribution cost after compiling the full circuit
into elementary two-qubit gates, we estimate the cost of the structured
Gaussian basis changes and diagonal interaction layers produced by the
electronic structure factorization itself.  Using the submultiplicativity of
QPD overheads~\cite{piveteau2024knitting}, the cost of independently cutting
the layers of one Trotter step obeys

\begin{equation}
\begin{split}
    \gamma^{(\mathrm{Trotter})}_\mathrm{LOCC}
    &\leq \prod_{\ell=0}^{L} \gamma_\mathrm{LOCC}(\hat U_\ell) 
    % \\
    % &\quad 
    \cdot \prod_{\ell=1}^{L} 
    \gamma_\mathrm{LOCC}\!\left( e^{-i\tau\lambda_\ell \hat D_\ell^2/2} \right) .
\end{split}
\label{eq:gamma_layers}
\end{equation}

An aspect of classical simulability that emerges in Trotter circuits for double-factorized electronic structure is in their simple alternating form with Gaussian unitaries interleaved with diagonal unitaries. Both classes of unitaries are independently classically simulable: Gaussian unitaries in the covariance matrix form and diagonal unitaries in a low-bond dimension matrix product operator (MPO) form and/or in a compressed diagonal form. However, their interleaved form is nontrivial and typically generates universal quantum circuits and therefore cannot be classically efficiently simulated. Together, Gaussian unitaries and CPhase gates are known to be universal for quantum chemistry simulation and in this sense our distribution complexity formalism is universal \cite{leimkuhler2026exponentialquantumspeedupsnearterm}. In the following two sections we discuss the distribution complexity of the Gaussian and Coulomb layers of the Trotter circuit, respectively.

%=============================================================================
\subsection{Distribution Complexity of Orbital Rotations}
\label{sec:gaussian_overhead}
%=============================================================================

The Choi bound~\cref{eq:choi_bound} gives a channel-level lower bound on the QPD extent from the Schmidt coefficients of the Choi state of a unitary. For a generic $N$-orbital Fock space unitary, even evaluating this bound requires access to an object of dimension $2^N$. However, the Gaussian orbital rotations that appear in double-factorized Trotter circuits are highly structured, and we show that it is possible to compute the relevant bound from the compressed covariance matrix representation rather than by the full many-body unitary~\cite{bravyi2005lagrangianflo,peschel2003rdm,peschel2009reduced}.

We proceed by computing the Choi state directly by introducing one ancilla mode for each physical mode and invoking the Choi-Jamiolkowski isomorphism. We begin by denoting the maximally entangled state in Fock space, in which the physical and ancilla modes are maximally entangled:
\begin{equation}
    \lvert \Omega \rangle
    =
    2^{-N/2}
    \prod_{j=1}^{N}
    \bigl(1+a_j^{\dagger}a_j^{\prime\dagger}\bigr)
    \lvert 0 \rangle 
    \label{eq:omega_bell}
\end{equation}
and observe that this is by definition a Gaussian fermionic state. The Choi state is defined as $\lvert\Phi_R\rangle=(\hat R\otimes I)\lvert\Omega\rangle$ and effectively measures the non-locality of the channel by the entanglement generated across the physical and ancilla sectors---one can understand this as the amount of entanglement between the computational basis states and their images under the action of the unitary~\cite{zanardi_entanglement_2001,piveteau2024knitting,aram_optimal_2025}. Our insight is that since $\lvert\Phi_R\rangle$ is Gaussian, its entanglement spectrum can be extracted from its induced covariance matrix, giving the efficiently calculable formula for the Choi bound $\gamma_\mathrm{Choi}(\hat{R})$ below.

\begin{theorem}[\textbf{Gaussian Choi bound}]
\label{thm:gaussian_choi_bound}
Let $\hat R$ be a particle-preserving fermionic Gaussian unitary induced by an orbital rotation matrix $R\in\mathrm{SO}(N)$, and let $A|B$ be a bipartition of the modes with $|A|=N_A$. If $\sigma_1,\ldots,\sigma_{N_A}$ are the singular values of the submatrix $R[A,A]$\footnote{Where $R[A,A]$ is defined as the $N_A \times N_A$ submatrix whose elements are defined as $R[A,A]_{ij} = R_{A[i],A[j]}$.}, then the Choi lower bound on the LOCC-assisted distribution complexity of $\hat R$ is
\begin{equation}
    \gamma_{\mathrm{Choi}}(\hat R;A)
    =
    2
    \prod_{k=1}^{N_A}
    \bigl(1+\sqrt{1-\sigma_k^2}\bigr)^2
    -
    1 .
    \label{eq:gamma_choi}
\end{equation}
Therefore, the Choi lower bound on the distribution cost of any fermionic Gaussian unitary can be computed from the spectrum of an $N_A\times N_A$ singular-value decomposition in $O(N^3)$ time.
\end{theorem}

This replaces the exponential $2^N$-dimensional Choi-state calculation required for a generic unitary with a polynomial-time computation on the single-particle Gaussian data. In~\cref{app:operator_entanglement_Gaussian_covariance}, we discuss an independent proof of the theorem without using Choi-Jamiolkowski isomorphism via the (operator-space) covariance matrix of Gaussian unitaries. This relies on the fact that the operator entanglement of a unitary can be obtained via the SVD of the realigned matrix. For Gaussian unitaries, one can obtain the realigned covariance matrix from the \(2N \times 2N\) matrix directly.

The covariance construction and the factorization of the Schmidt coefficients are given in \cref{app:covariance,app:factorization_schmidt}. The singular values have a direct interpretation: $\sigma_k=1$ is a mode that remains local to the partition and contributes no distribution cost, while smaller $\sigma_k$ indicates stronger mixing across the cut.

For a single cross-partition mode, $\sigma=\cos\theta$ and \cref{eq:gamma_choi} gives $\gamma=2(1+|\sin\theta|)^2-1$, matching the optimal two-qubit result from the KAK decomposition formula~\cite{tucci2005introductioncartanskakdecomposition,piveteau2024knitting} (See \cref{app:kak_baseline} for more details). For several modes, \cref{eq:gamma_choi} should be read as a lower bound, not as an implementable protocol by itself. Optimally achieving this bound requires that the orthogonal Schmidt operators in the decomposition are themselves product unitaries~\cite{aram_optimal_2025}, while for arbitrary multi-mode Gaussians this is not generically the case.

%-----------------------------------------------------------------------------
\subsubsection{Achievable Protocols and Hierarchy}
\label{sec:protocols}
%-----------------------------------------------------------------------------
The Choi bound in \Cref{thm:gaussian_choi_bound} gives the channel-level scale set by an orbital rotation across a partition. A circuit partitioning implementation still requires an explicit local decomposition, and we show here that some procedures take advantage of the Gaussian structure much better than others.

A standard approach compiles the orbital rotation $R \in \mathrm{SO}(N)$ into 
$N(N{-}1)/2$ two-qubit Givens rotations~\cite{hoffman1972generalized,kivlichan2018fswap}, then cuts
each cross-partition gate independently~\cite{piveteau2024knitting,jones2025distributed}. For a balanced
partition ($N_A = N/2$), $O(N^2)$ cross-partition Givens must be cut,
giving total overhead
\begin{equation}
    \gamma_{\mathrm{gbg}}
    = \prod_{j=1}^{O(N_A^2)} \gamma(G_j)
    \sim \exp\left({c\,N^2}\right)\,.
    \label{eq:gbg}
\end{equation}
This is super-exponential in $N$, and is the gate-by-gate overhead reported in prior distributed electronic structure
studies. This approach reaches impractical values for even shallow fermionic ans\"atze on
small molecular instances~\cite{jones2025distributed}.

We can do better than this and nearly match the Choi lower bound by instead using the cosine-sine (CS) decomposition for the orthogonal rotation matrix~\cite{paige1994history,golub2013matrix} with respect to a partition of the modes. This representation factors out rotations that act locally within subspaces $A$ or $B$, leaving a minimal set of cross-partition two-qubit gates.

\begin{figure}[!thbp]
    \centering
    \includegraphics[width=\columnwidth]{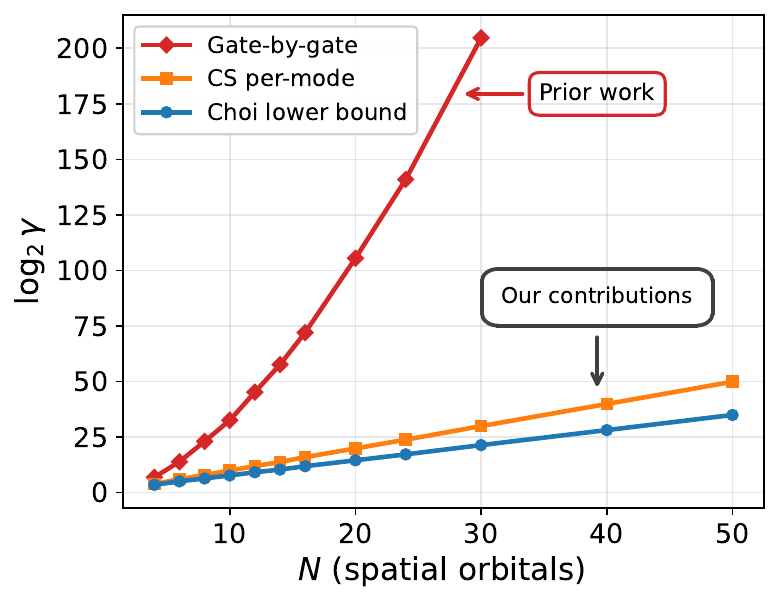}
    \caption{\textbf{Complexity of protocols for distributing random orbital rotations.} Scaling of the log of the distribution complexity with the number of orbitals, \(N\). For each $N$, we plot the median $\log_2(\gamma)$ over an ensemble of Haar random orbital rotations for the three protocols. We note that the baseline gate-by-gate compilation (red) from previous studies~\cite{jones2025distributed,piveteau2024knitting,optimal_joint_two} scales super-exponentially as $O(\exp{(N^2)})$, while our achievable protocol based on the cosine-sine decomposition (orange) almost matches the Choi lower bound (blue), both of which scale as $O(\exp{(N)})$ (with different exponent constants). See \cref{thm:gaussian_cs_protocol} for more details.}
  \label{fig:haar_operator}
\end{figure}

\begin{theorem}[\textbf{Entanglement-aware compilation of Gaussian unitaries}]
\label{thm:gaussian_cs_protocol}
Let $\hat R$ be a particle-preserving fermionic Gaussian unitary induced by $R\in\mathrm{SO}(N)$, and let $A|B$ be a mode partition with $|A|=N_A\leq N_B$. Let $\sigma_k(R[A,A])=\cos\theta_k$ for $k=1,\ldots,N_A$, with $0\leq\theta_k\leq \pi/2$. 

Then $\hat R$ admits an LOCC-assisted cosine-sine protocol across $A|B$ with distribution cost
\begin{equation}
    \gamma_{\mathrm{LOCC}}(\hat R;A)
    \leq
    \gamma_{\mathrm{CS}}(\hat R;A)
    =
    \prod_{k=1}^{N_A}
    \bigl[
        2(1+|\sin\theta_k|)^2 - 1
    \bigr] .
    \label{eq:gamma_cs}
\end{equation}
For a balanced partition, the gate-by-gate compilation decomposes into $O(N^2)$ cross-partition Givens rotations, while the CS protocol leaves only $N/2$. 
Therefore, distributing each independently yields a total distribution cost $\log(\gamma_\mathrm{CS}) = O(N)$.
\end{theorem}

The CS decomposition of an orthogonal matrix $R$ with respect to a partition $A|B$ says that there always exist $N_A\times N_A$ orthogonal matrices $V_1$ and $W_1$, and $N_B \times N_B$ orthogonal matrices $V_2$ and $W_2$, such that,
\begin{equation}
    R =
      \begin{pmatrix} V_1 & 0 \\ 0 & V_2 \end{pmatrix}
      \begin{pmatrix} C & -S \\ S & C' \end{pmatrix}
      \begin{pmatrix} W_1^\dagger & 0 \\ 0 & W_2^\dagger \end{pmatrix}\,,
    \label{eq:cs}
\end{equation}
where $V_1,W_1\in \mathrm{SO}(N_A)$ and $V_2,W_2\in \mathrm{SO}(N_B)$ act within the two processors and therefore require no interprocessor resource. The central block is a product of $N_A$ independent two-qubit gates with $C=\operatorname{diag}(\cos\theta_k)$ and $S=\operatorname{diag}(\sin\theta_k)$. For LOCC-assisted distribution, we can cut the remaining two-qubit gates with minimal extent $2(1+|\sin\theta_k|)^2-1$ using the KAK formula (see \cref{app:kak_baseline} for more details). Composing the cuts gives the following product~\cite{piveteau2024knitting}:
\begin{equation}
    \gamma_{\mathrm{CS}}
    = \prod_{k=1}^{N_A}
      \bigl[2(1 + |\sin\theta_k|)^2 - 1\bigr]
    \sim \exp\left({c'\,N}\right)\,.
    \label{eq:gamma_cs_single_exp}
\end{equation}
For quantum interconnects, the same CS circuit has only $N_A=O(N)$ nonlocal two-mode gates, each of which can be supplied with constant-size entanglement assistance. The interconnect count then scales with the number of CS angles rather than with the number of compiled Givens rotations~\cite{Barral2025DistributedQC,Wu2023DistributedEntanglement}. This yields a quadratic speedup for quantum distribution protocols of Gaussian fermionic unitaries.

Overall, we arrive at a distribution complexity hierarchy:
\begin{equation}
    \gamma_{\mathrm{Choi}} \leq \gamma_\mathrm{LOCC} \leq
    \gamma_{\mathrm{CS}} \leq \gamma_{\mathrm{gbg}}.
    \label{eq:hierarchy}
\end{equation}
In this hierarchy, $\gamma_\mathrm{LOCC}$ is the cost of the most efficient possible distribution protocol, while we show that  $\gamma_\mathrm{Choi}$ is an efficiently computable theoretical lower bound, and $\gamma_\mathrm{CS}$ is an efficiently computable achievable protocol cost. We now turn our attention towards the scaling of these bounds for typical Gaussian fermionic unitaries. 

%-----------------------------------------------------------------------------
\subsubsection{Random Rotations}
\label{sec:haar}
%-----------------------------------------------------------------------------

Before analyzing the fragment rotations of specific molecules, which
carry structure inherited from the molecular geometry and the DF
decomposition, we establish a baseline by studying Haar-random
rotations drawn uniformly from $\mathrm{SO}(N)$.  Such rotations have no
partition-aligned structure, no low-rank fragments, and no preferred
eigenbasis; they provide an unstructured baseline for distributing a
single Gaussian layer~\cite{wachter1980limiting}.

For Haar-random $R \in SO(N)$ with a balanced partition ($m = N/2$), we sample 100
rotations at each~$N$ and compute the mean overhead under each
protocol, see \cref{fig:haar_operator}.  
We recall that the exponent scaling coefficient is
approximately the number of Bell pairs per mode for quantum interconnect protocols, and observe the scaling:
\begin{alignat}{2}
    \log_2\gamma_{\mathrm{gbg}}
      &\approx 0.26\,N^2
      &\qquad&\text{(super-exponential)}\,,
      \nonumber \\
    \log_2\gamma_{\mathrm{CS}}
      &\approx 0.96\,N
      &\qquad&\text{(exponential)}\,,
      \label{eq:scaling_fits} \\
    \log_2\gamma_{\mathrm{Choi}}
      &\approx 0.70\,N
      &\qquad&\text{(lower bound)}\,.
      \nonumber
\end{alignat}
Randomly sampled orbital rotations are exponentially costly to distribute using LOCC protocols, but the cosine-sine decomposition achieves a quadratic speedup in the distribution exponent over naive gate-by-gate compilation schemes. 
The best achievable protocol for arbitrary orbital rotations, whatever it is, has a cost $\gamma_\mathrm{LOCC}$ that is bounded between $\gamma_\mathrm{Choi} \leq \gamma_\mathrm{LOCC} \leq \gamma_\mathrm{CS}$. We therefore provide an efficiently computable minimum possible lower bound, and a nearly-matching achievable upper bound. To the best of our knowledge, such protocols have previously been limited to Clifford unitaries \cite{piveteau2024knitting}.

%-----------------------------------------------------------------------------
\subsection{Distribution Complexity of Coulomb Interactions}
\label{sec:coulomb}
%-----------------------------------------------------------------------------

The diagonal layers of the Trotter circuit are all-to-all connected interactions capturing Coulomb interactions and non-Gaussianity. However, the DF diagonal layers are not actually \textit{generic} diagonal unitaries. Instead, each fragment is the square of a one-body operator, and the cross-boundary phases are locked to a rank-one form $\vartheta_{ij}^{(\ell)} = \tau\lambda_\ell\varepsilon_i^{(\ell)}\varepsilon_j^{(\ell)}$. Moreover, for real molecules, the \(\varepsilon_i^{(\ell)}\)s (in a fixed basis such as STO-3G) are disordered. Together, this rank-one structure and disorder in the "single-particle spectra" considerably limits the operator entanglement growth of these diagonal layers.

\begin{figure*}[!ht]
    \centering
    \includegraphics[width=\textwidth]{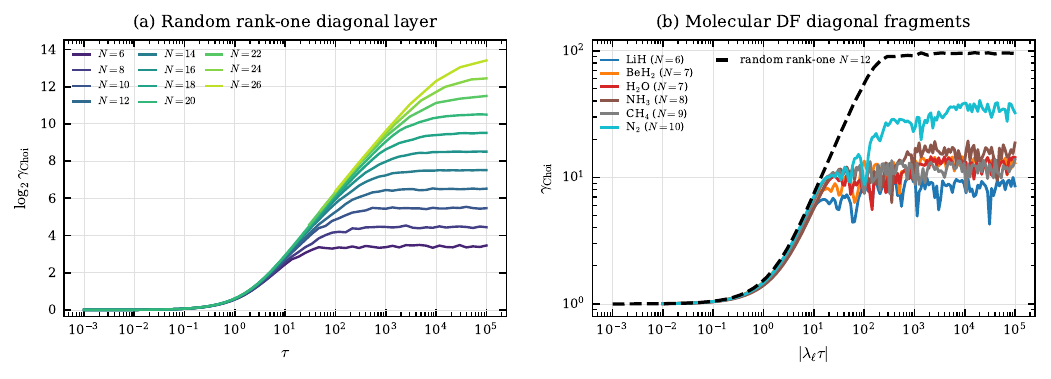}
  \caption{\textbf{Distribution complexity of the diagonal Coulomb layers appearing in Trotter circuits for double-factorized electronic structure simulation.} 
  The evolution in the eigenbasis of each fragment is generated by squares of diagonal one-body operators $\exp(i\tau \lambda_\ell \, \hat{D}_\ell^2)$, where $\hat{D} = \sum_m \varepsilon^{(\ell)}_m a^\dagger_m a_m $. We demonstrate that this structure limits the growth of operator entanglement and therefore makes these layers easy to distribute. \textbf{(a)} We plot the distribution complexity of random rank one diagonal simulation layers $\exp(i\tau \hat{D}^2)$, where each coefficient is randomly chosen $\varepsilon_m \sim \mathcal{N}(0, 1/N)$, corresponding to Gaussian random single-particle energies. We observe a universal scaling behavior with simulation time: with early time growth as $\log(\gamma_\mathrm{Choi}) \sim O(\tau)$, intermediate time growth as $\log(\gamma_\mathrm{Choi}) \sim O(\log(\tau))$ and saturation to an $N$-dependent value, which is not present without the squared one-body operator structure (see \cref{app:random_cphase_saturation} for more details). The slow growth at early- and intermediate-times translates to polynomial distribution complexity of diagonal layers. \textbf{(b)} We illustrate that the random rank-one behavior upper bounds the distribution complexity for diagonal layers in real molecular systems. We compute the DF decomposition of the ERI tensor for small molecules, evaluated in the STO-3G orbital basis, and plot the distribution complexity of their dominant fragments (largest outer eigenvalue $\lambda_\ell$). We then compare this with the distribution complexity of the corresponding random rank one diagonal model and show that the distribution complexity is always bounded by the random rank-one model. We discuss more numerical results of small molecules in \cref{app:small_molecules}. 
}
\label{fig:diagonal_correlation_comparison}
\end{figure*}

To illustrate how this structure constrains entanglement growth, we analyze rank-one diagonal models with Gaussian disorder. Observing the normalization condition from DF eigenmatrices $\sum\varepsilon_i^2=1$, we can parameterize diagonal fragments in terms of their vector of single-particle eigenvalues on the sphere in $N$ dimensions. If an arbitrary diagonal interaction Hamiltonian has
\begin{equation}
    H_\Theta
    =
    \sum_{i<j}
    \Theta_{ij}
    \ \hat{n}_i \ 
   \hat{n}_j ,
    \qquad
    \hat{U}_\Theta(\tau)
    =
    e^{-i\lambda \tau H_\Theta},
    \label{eq:random_diagonal_model}
\end{equation}
then random rank one (RRO) diagonal fragment models of the form we see in DF are given by
\begin{equation}
    H_\mathrm{RRO}(\vec{\varepsilon} \, ) = \left( \sum_{i=1}^N \varepsilon_i \ \hat{n}_i \right)^2,  \qquad \vec\varepsilon \in S^{N-1}
\end{equation}
where $S^{N-1}$ is the hypersphere in $N$ dimensions. For a single diagonal layer, $\lambda$ only rescales the dimensionless evolution strength $\lambda\tau$, but we keep it to explicitly match the DF fragment parameterization in \cref{eq:trotter}.
We observe a universal scaling behavior in the operator entanglement of Coulomb layers by turning to RRO diagonal models, and specifically the case where $\varepsilon$ is distributed uniformly on $S^{N-1}$. This is equivalent to sampling each $\varepsilon_i \sim \mathcal{N}(0,1)$ and then renormalizing the sampled energies $\|\vec\varepsilon\, \|_2$. 

In \cref{fig:diagonal_correlation_comparison}(a) we plot the Choi distribution complexity lower bound as a function of time step $\tau$ for increasing values of $N$ for RRO models as well as for fragments in a handful of small molecules with the STO-3G basis set. In both cases, the entanglement is governed by a universal scaling $\log(\gamma_\mathrm{Choi}) \sim \tau$ for small times (\(\tau \ll 1\)) and \(\log(\gamma_\mathrm{Choi}) \sim \log(\tau)\) for intermediate times. In contrast, the uncorrelated models (without any constraints on the rank-structure) grow in a system-size dependent and much faster rate, as discussed in \cref{app:random_cphase_saturation}.

For quantum simulation, this implies that distributing diagonal Trotter circuits of electronic structure Hamiltonians is in fact efficient, while the majority of the complexity originates from the orbital rotations. This universal slow-growth is consistent with the behavior of small molecules as well. In \cref{fig:diagonal_correlation_comparison}(b), we compare the exact Choi distribution complexity of the largest full-rank diagonal fragment for a few small molecules against a bounding envelope of an arbitrary random one-body operator, as a function of the normalized parameter $\lambda_\ell \tau$. The real molecular fragments are bounded by the same universal behavior exhibited by the worst case random systems. See \cref{app:collective_dephasing} for more details.

To understand this bounded entanglement growth, it is useful to work with the one-half R\'enyi entropy of the Choi state, which is related to the LOCC-assisted distribution complexity as,
\begin{equation}
    S_{1/2}^{\mathrm{Choi}}(\hat{U}) = \log\!\left(\frac{\gamma_{\mathrm{Choi}}(\hat{U})+1}{2}\right).
    \label{eq:choi_half_renyi_coulomb}
\end{equation}
The main observation is that this entropy grows linearly at early times and logarithmically at intermediate times, before saturating to a system-size dependent equilibrium value. This is reminiscent of the growth of entanglement in the phenomenological "\(\ell\)-bits" model of many-body localization (MBL), as we discuss in \cref{subsec:entanglement-qchem}. If the single-particle energies are independent Gaussian variables, then we can analytically explain this behavior. Moreover, in \cref{fig:diagonal_correlation_comparison}(b) we show that real molecules are bounded by this Gaussian random rank-one model. Let us define $S^{\mathrm{RRO}}_{1/2}$ as the one-half R\'enyi entropy of the Choi state a typical instance of the random rank one diagonal unitary,
\begin{equation}
    S^{\mathrm{RRO}}_{1/2}(\lambda \tau) \equiv S^{\mathrm{Choi}}_{1/2} ( e^{i\lambda \tau H_{\mathrm{RRO}}}).
\end{equation}
Then, we have the following theorem which summarizes the bounded distribution complexity for such structured diagonal unitaries.

\begin{theorem}[\textbf{Distribution complexity of rank-one random diagonal unitaries}]
\label{thm:slow-growth-MBL}
The operator entanglement entropies in the random rank-one model of diagonal unitaries with Gaussian disorder has a bounded growth of entanglement. For the \(\alpha=1/2\)-Renyi entropy, the short-time growth is,
\begin{align}
\label{eq:slow-growth-MBL}
S_{1/2}^{\mathrm{RRO}}(\lambda \tau) \approx O(|\lambda \tau|), \qquad \quad &\text{ for } |\lambda \tau| \ll 1, \\
S_{1/2}^{\mathrm{RRO}}(\lambda \tau)\approx O(\log\!\left(|\lambda \tau|\right)), \ \ \ &\text{ for } |\lambda \tau| \gg 1.
\end{align}
Therefore, because $\gamma_\mathrm{Choi} \sim \exp(S_{1/2})$, the distribution complexity is linear in time and independent of system size 
\begin{align}
\gamma_\mathrm{Choi}(e^{i\lambda \tau H_\mathrm{RRO}}) = O(|\lambda \tau|).
\end{align}
\end{theorem}

In fact, we are able to prove a more general version of this result, namely, the bounded growth of operator entanglement for all Renyi-\(\alpha\) entropies. For one diagonal fragment, it is useful to separate the Trotter
normalization from the collective-dephasing scale. We therefore define
\begin{align}
    t_\ell&:=\frac{\tau\lambda_\ell}{2},\\
    v_{A,\ell}^2&:=\frac14\sum_{a\in A}
    \bigl(\varepsilon_a^{(\ell)}\bigr)^2,
    \qquad
    v_{B,\ell}^2:=\frac14\sum_{b\in B}
    \bigl(\varepsilon_b^{(\ell)}\bigr)^2.
\end{align}
and the dimensionless collective-dephasing strength
\begin{equation}
    \omega_\ell:=4|t_\ell|v_{A,\ell}v_{B,\ell}.
    \label{eq:rro_collective_strength}
\end{equation}
For independent Gaussian coefficients of variance $\sigma_\ell^2$,
typical realizations satisfy
$\omega_\ell\simeq |t_\ell|\sigma_\ell^2\sqrt{N_AN_B}$. Thus, for a
balanced cut and a unit-normalized vector
$\boldsymbol\varepsilon^{(\ell)}$, $\omega_\ell$ is \textit{independent} of $N$ to
leading order. The following result explains both the short-time and
intermediate-time regimes; the proof is given in
\cref{app:collective_dephasing}. For the diagonal layer
$\exp[-it_\ell\hat D_\ell^2]$, the finite-system short-time operator
R\'enyi entropies obey
\begin{equation}
S_\alpha=
\begin{cases}
O(\omega_\ell^{2\alpha}), & 0<\alpha<1,\\
O\!\left(\omega_\ell^2\log(1/\omega_\ell)\right), & \alpha=1,\\
O(\omega_\ell^2), & 1<\alpha\leq\infty.
\end{cases}
\label{eq:slow-growth-MBL}
\end{equation}
Here, \(\omega_\ell\) represents the dimensionless collective-dephasing strength as defined above.

In particular, $S_{1/2}=\omega_\ell+o(\omega_\ell)$ is linear at short times,
whereas $S_2=\omega_\ell^2/2+o(\omega_\ell^2)$ is quadratic. In the Gaussian-disorder
approximation, the two entropies are
\begin{equation}
S_{1/2}(\omega_\ell)=\operatorname{arsinh}(\omega_\ell),
    \qquad
S_2(\omega_\ell)=\frac12\log(1+\omega_\ell^2),
\end{equation}
and every fixed-order R\'enyi entropy has the same leading
pre-saturation growth,
\begin{equation}
    S_\alpha(\omega_\ell)=\log \omega_\ell+O_\alpha(1),
    \qquad \omega_\ell\gg1.
\end{equation}
The corresponding Choi lower bound is
\begin{equation}
\gamma_{\mathrm{Choi}}^{(\mathrm G)}
=2\left(\omega_\ell+\sqrt{1+\omega_\ell^2}\right)-1.
\end{equation}
Thus $\gamma_{\mathrm{Choi}}-1=O(\omega_\ell)$ at short times and
$\gamma_{\mathrm{Choi}}^{(\mathrm G)}=O(\omega_\ell)$ throughout the
logarithmic-entanglement regime.

This serves as a bounding envelope for the diagonal layers present in our double factored Trotter circuits. For many molecular systems in a double-factorized form, the norm of the ERI outer eigenvalues $\lambda_\ell$ is often linear with the number of orbitals $O(N)$\cite{lee2021thc,loaiza2023lcucostreduction,oumarou2024compressed,rocca2024symmetry}. Moreover, a typical choice of Trotter step is \(\tau = O(1/N)\) or smaller, and therefore, the distribution overhead is polynomially bounded for all relevant timescales for typical tensor fragments.

Having understood the distribution cost of the Gaussian orbital rotations and the diagonal Coulomb unitaries, we are now ready to estimate the distribution complexity of electronic structure per Trotter step. Recall that a single Trotter step typically involves a number of fragments \(L\) that scales with the number of orbitals \(L \sim \Theta(N)\). Our work reduces the exponent of the LOCC-assisted protocol for every Trotter step from \(O(LN^2)\) to \(O(LN)\), that is linear in both the number of fragments and the number of orbitals. Despite our quadratic reduction via \cref{thm:gaussian_cs_protocol} and the polynomial bound via \cref{thm:slow-growth-MBL}, the LOCC-assisted distribution complexity per Trotter step still scales exponentially. This limits the applicability of circuit partitioning approaches based on conventional HPCs for quantum simulation of strongly correlated electronic systems. Nonetheless, there exist both applications that require finite-depth quantum circuits (such as Trotterized quantum phase estimation to relative energy accuracy~\cite{Campbell_2021}) and molecular systems where the number of dominant fragments scales sublinearly, say as \(\log(N)\) or \(O(1)\), for which these methods can still be useful~\cite{peng2017lowrank}. In contrast, the scaling of resources for quantum interconnects (shared entanglement across QPUs) is much more feasible, although, they may experience other challenges such as frequent distillation of Bell pairs, low fidelities, among others~\cite{ds45-fm9n}.

%=============================================================================
\subsection{Entanglement Propagation in Gaussian and Coulomb layers}
\label{subsec:entanglement-qchem}
%=============================================================================

The total distribution complexity of Trotterized time evolution under electronic structure Hamiltonians is a function of the operator entanglement generated via the Gaussian and Coulomb layers. However, the two layers have distinct entanglement propagation characteristics with the Gaussian layers showing ballistic entanglement propagation, while the disordered Coulomb layers show slow diffusive entanglement growth. Thus the average entanglement propagation (and therefore the average distribution complexity) has a nontrivial interplay generated via the competition between these two layers. Entanglement growth in quenched free-fermionic systems is well understood due to their integrability, which enables translating the entanglement growth in terms of the rate of quasiparticle generation across a boundary~\cite{Calabrese2005,Fagotti2008,Peschel2009,Alba2017}. Typically, in clean non-interacting fermions, both quenched state and operator entanglement grows linearly with time, until it saturates to volume-law entanglement~\cite{Calabrese2005,Dubail2017}. The presence of disorder or randomness can drastically change this growth, for e.g., as is the case for disordered fermions in Anderson or many-body localization (MBL) \cite{Nandkishore2015}. In general, disorder in either space or time results in sub-diffusive growth of entanglement for Gaussian unitaries \cite{Paviglianiti2026}, similar to random Clifford circuits~\cite{PhysRevX.7.031016,PhysRevX.8.021014}.

For systems with disorder, the slow-growth of entanglement often originates from a "dephasing mechanism," similar to that of many-body localization. We first briefly review the key ideas behind MBL and then discuss its implications. Recall that a quantum many-body interacting Hamiltonian is said to be in a \textit{many-body localized} form if it is equivalent to \(H_{\mathrm{MBL}}\) below up to a quasi-local unitary \cite{RevModPhys.91.021001,Basko2006,PhysRevB.82.174411,Nandkishore2015}. Namely, the Hamiltonian of interest is \(H = U^{\dagger} H_{\mathrm{MBL}} U\), where \(U\) is quasi-local. The quasi-locality ensures that the eigenstate entanglement structure of \(H_{\mathrm{MBL}}\) is inherited by the original Hamiltonian \(H\). After the quasi-local transformation, an MBL Hamiltonian generically has the following form, which clearly identifies the local integrals of motion $\sigma^z_k$,
\begin{align}
\label{eq:mbl-diagonal-form}
H_{\mathrm{MBL}}=&\sum_{k} J_{k}^{(1)} \sigma_{k}^{\mathrm{z}}+\sum_{k<l} J_{k, l}^{(2)} \sigma_{k}^{\mathrm{z}} \sigma_{l}^{\mathrm{z}} \nonumber \\
&+\sum_{k<l<m} J_{k, l, m}^{(3)} \sigma_{k}^{\mathrm{z}} \sigma_{l}^{\mathrm{z}} \sigma_{m}^{\mathrm{z}}+\cdots.
\end{align}

Here, the couplings \(J^{(r)}\) are disordered and depend on the original model and typically also decay exponentially. A characteristic feature of MBL Hamiltonians is the slow growth of entanglement, typically as \(\sim \log(t)\). While we do not suggest the same phenomenology, the dephasing mechanism in MBL and the disordered Ising models considered here have the same origin. It is worth emphasizing that the phenomenology of MBL, especially its stability in the thermodynamic limit remains an area of intense debate due to avalanche phenomena \cite{PhysRevB.105.174205,PhysRevB.95.155129,PhysRevB.106.L020202,Lonard2023,PhysRevLett.133.116502}. 

For the purposes of understanding entanglement propagation in electronic structure circuits, we note the following. The dephasing phenomena (which also shows up in MBL) is limited to the Ising-like Coulomb interactions and not necessarily to the entire Trotter circuit. Second, typically, MBL Hamiltonians are equivalent to the \(\ell\)-bit model \cref{eq:mbl-diagonal-form} up to a quasi-local unitary. Since a quasi-local unitary cannot change the entanglement structure of a model (namely an area/volume law phase remains an area/volume law), as a result, the entanglement growth in the original MBL Hamiltonian retains the logarithmic rate predicted by the \(\ell\)-bit model. In contrast, the Trotter circuit for a DF electronic structure alternates between the change of basis (the orbital rotations) and the diagonal disordered evolution. As a result, the effective evolution has a competition between a ballistic propagation of entanglement (the Gaussian unitary from the orbital rotation) versus the slow, diffusive propagation of entanglement (the diagonal Coulomb interaction). We will explore this interplay and related phenomenology in an upcoming manuscript \cite{QChem_MBL_paper}. 

%=============================================================================
\section{Underlying Mechanisms of Distribution Complexity in a Double Factorized Representation}
\label{sec:structure}
%=============================================================================

\begin{figure*}[!ht]
  \centering
  \includegraphics[width=\textwidth]{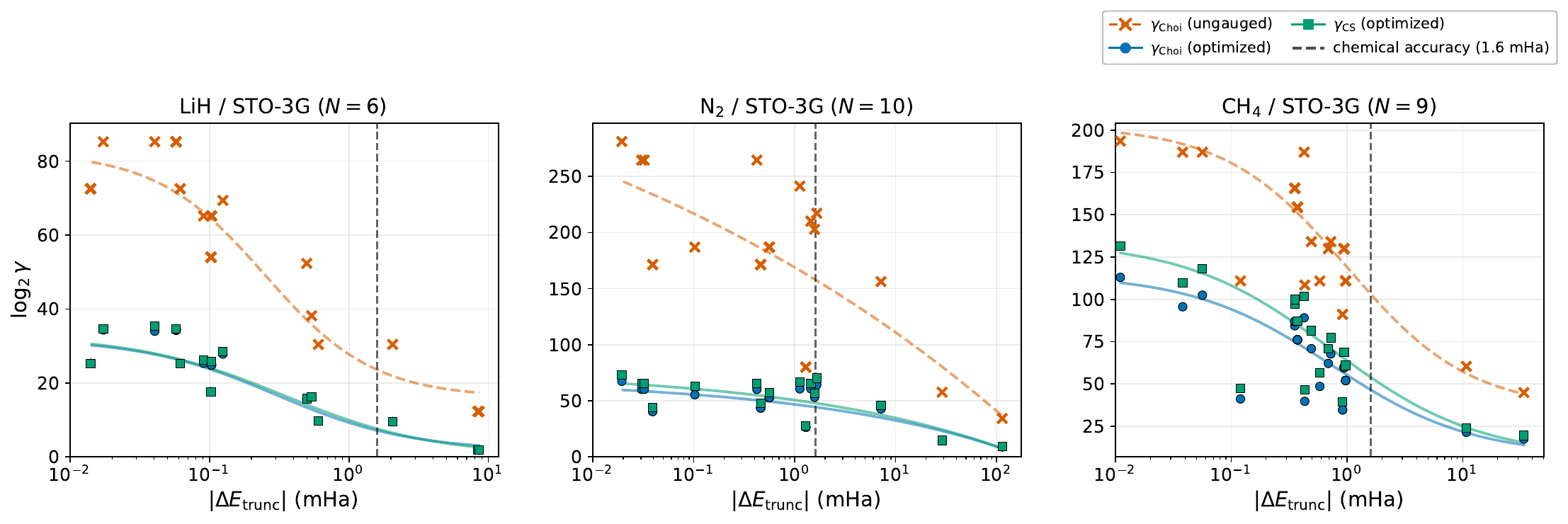}
  \caption{\textbf{Distribution complexity as a function of the truncation error for correlation energy estimation.} We demonstrate the effect of removing weak fragments on distribution complexity for small molecules such as \textbf{(a)} LiH, \textbf{(b)} N$_2$, and \textbf{(c)} CH$_4$, all evaluated in STO-3G orbital bases. We compute the $\log_2 \gamma$ for a single Trotter step, computed by taking the Choi bound for each interfragment Gaussian rotation and diagonal layers independently, and plot this against the induced truncation error $|\Delta E_{\mathrm{trunc}}|$ in the CCSD(T) correlation energy. Each data point is a different energy threshold $\delta$, using the weighted truncation criteria in \cref{eq:weighted_l2}. The vertical dashed line highlights the point at which the induced difference in correlation energy reaches chemical accuracy, $1.6~\mathrm{mHa}$. We compare this for unoptimized (orange) as well as optimized (blue) representations of the circuit, with respect to the degrees of freedom which are the focus of \cref{sec:structure}. Even non-intensive greedy optimization of these parameters leads to dramatic reduction in distribution complexity at nearly all truncation levels, including for the achievable CS protocols (green). In particular, we note that small molecules show reasonable distribution complexity when truncating fragments to achieve truncation error close to chemical accuracy. This highlights the sensitivity of distribution complexity to the algorithmic task (correlation energy estimation) and precision requirements ($|\Delta E_{\mathrm{trunc}}|$). 
  }
  \label{fig:truncation_three_panel}
\end{figure*}

The double-factorized representation contains several choices that leave the
target Hamiltonian unchanged, or change it only within a controlled
approximation, while strongly affecting the distribution cost of the compiled
Trotter step.  We organize these effects into four mechanisms: (i) the retained
outer and inner ranks determine how many diagonal and Gaussian layers appear;
(ii) gauge freedom inside degenerate and nearly degenerate fragment
eigenspaces changes the relative rotations between fragments; (iii) the
ordering of fragments determines which bases are adjacent in the telescoped
circuit; and (iv) the mode partition determines which singular values enter
\cref{eq:gamma_choi}. We first discuss how retained rank trades approximation error against the number and strength of nonlocal layers. The rest of this section then studies representation choices that change channel QPD cost while preserving, or deliberately approximating, the same electronic Hamiltonian.

%-----------------------------------------------------------------------------
\subsection{Low-Rank Compression and Calibration}
\label{sec:truncation}
%-----------------------------------------------------------------------------

In the DF representation of the ERI tensor, truncation occurs at two levels. The first factorization treats $V_{pqrs}$ as a real symmetric pair-density matrix $W_{(pq),(rs)}$ and expands it in tensor fragments with coefficients $\lambda_\ell$. In general, the number of tensor fragments can be as large as $L_{\mathrm{full}}=N(N+1)/2$. Outer truncation retains only $L_{\mathrm{active}}$ fragments with significant $|\lambda_\ell|$, thereby removing the weakest diagonal Coulomb layers and their neighboring Gaussian basis changes from the Trotter circuit. The second factorization diagonalizes each retained fragment $g^{(\ell)}=R_\ell\operatorname{diag}(\varepsilon^{(\ell)})R_\ell^{\dagger}$. Inner truncation keeps only $\rho_\ell$ significant eigenvalues in the $\ell^\mathrm{th}$ fragment, reducing the number of diagonal two-mode phases generated by $\hat D_\ell^2$.

This two-level rank structure is a source of compression exploited by low-rank electronic-structure simulations. Density fitting, resolution-of-the-identity, and Cholesky factorization all decompose $W_{(pq),(rs)}$ to reduce the storage and contraction cost of electronic integrals~\cite{whitten1973coulombic,dunlap1979approximations,vahtras1993integral,aquilante2011cholesky,folkestad2019efficient}. The nested decompositions used in double factorization add a second compression step by truncating the spectrum of each retained fragment matrix $g^{(\ell)}$~\cite{peng2017lowrank,motta2019afqmc_lowrank,motta2021lowrank}. The scaling of the average fragment rank is then dependent on the orbital basis and systems being studied. Compressed and regularized double-factorization methods optimize the retained cores, coefficients, and symmetry shifts to reduce quantum simulation resource estimates, often expressed through qubitization $\lambda$ or related one-norms~\cite{oumarou2024compressed,rocca2024symmetry,caesura2025bliss,loaiza2023lcucostreduction,loaiza2024majoranatensordecomposition}. In this work, we consider only direct truncation.

For DF fragments, the outer coefficient $\lambda_\ell$ carries units of energy, while the normalized inner eigenvalues satisfy $\sum_m(\varepsilon_m^{(\ell)})^2=1$. We therefore use a weighted inner truncation algorithm: after sorting the inner eigenvalues by magnitude, choose the smallest $\rho_\ell$ such that
\begin{equation}
    |\lambda_\ell|
    \sum_{k>\rho_\ell}
    \bigl(\varepsilon_k^{(\ell)}\bigr)^2
    <
    \delta ,
    \label{eq:weighted_l2}
\end{equation}
where $\delta$ is an energy threshold. This rule truncates weak fragments more aggressively and treats important fragments more conservatively. It also removes fragments entirely when $|\lambda_\ell|<\delta$. Moreover, note that, as the inner eigenvalues \(\varepsilon^{(\ell)}\) are normalized, the norm of the Hamiltonian is determined by the outer eigenvalues, \(\lambda_\ell\), which controls the scaling of qubitization-based simulation via linear combination of unitaries \cite{low2019hamiltonian}. In an upcoming work, we explore the distributed simulation of DF Hamiltonians via qubitization and its scaling as a function of the truncation of inner and outer eigenvalues \cite{QChem_MBL_paper}.

The weighted tail in \cref{eq:weighted_l2} is a tensor-space error budget. It
orders the DF data by the energy-weighted Frobenius weight discarded from each
fragment, and therefore gives a simple one-parameter path of truncated
Hamiltonians $V_{\mathrm{trunc}}(\delta)$. This is a rigorous bound on the operator error of the approximation, but it does not directly translate to specific observables of interest. Weak tensor directions can enter with small amplitudes or cancel, so the induced energy error can be much smaller than the raw Frobenius residual.

In \cref{fig:truncation_three_panel}, we illustrate this by calibrating the truncation path according to the induced error in the
post-Hartree--Fock correlation energy. For each threshold $\delta$, we rebuild
the Hamiltonian from $V_{\mathrm{trunc}}(\delta)$ and compute
\begin{equation}
    \Delta E_{\mathrm{trunc}}(\delta)
    =
    \left|
    E_{\mathrm{CCSD(T)}}(V)
    -
    E_{\mathrm{CCSD(T)}}(V_{\mathrm{trunc}}(\delta))
    \right| ,
    \label{eq:delta_e_trunc}
\end{equation}
using the same molecular orbitals before and after truncation. Lee \emph{et al.}\cite{lee2021thc} argued that CCSD(T) correlation-energy error is a more practical criterion that truncated operator norms because those coefficient errors overestimate the differences in observable estimations. Similarly, Caesura
\emph{et al.} use $L_2$ error to optimize the factorization but choose the retained rank using CCSD(T) and DMRG correlation-energy
checks~\cite{caesura2025bliss}. The vertical line in
\cref{fig:truncation_three_panel} marks the point at which the induced error reaches chemical accuracy,
$\Delta E_{\mathrm{trunc}}=1.6~\mathrm{mHa}$. We plot this tradeoff for three STO-3G examples, computing exactly how much error in correlation energy is incurred with increasing truncation, as well as how much can be saved in distribution cost of the circuit.
Meanwhile, the significant separation between the ungauged Choi curve
and the optimized Choi/CS curves shows that rank-based truncation does not accurately represent the full picture

% the full picture alone does not tell the whole story. \masoud{last statement is a bit too informal. I suggest " ... that rank alone does represent the full picture."  or " ... that rank alone does not contain the full information." or something similar}

%-----------------------------------------------------------------------------
\subsection{Gauge Freedom in Fragment Representations}
\label{sec:gauge}
%-----------------------------------------------------------------------------

The outer factorization of the ERI tensor fixes each fragment matrix $g^{(\ell)}$, but does not always uniquely determine the circuit representation necessary to implement that fragment's evolution. If
\begin{equation}
    g^{(\ell)}
    =
    R_\ell\,\mathrm{diag}(\varepsilon^{(\ell)})\,R_\ell^\dagger ,
    \label{eq:fragment_diagonalization}
\end{equation}
has degenerate subspaces, then there is standard numerical instability in the eigendecomposition. Nearly degenerate blocks can rotate strongly under small changes of a matrix~\cite{davis1970rotation,stewart1990matrixperturbation}, and for our purposes, this means that any rotation within degenerate subspaces is a gauge freedom that results in the same DF Hamiltonian. This also means that truncation leading to gap closures can dramatically affect the observed interfragment rotation strengths. For low-rank DF fragments, the nullspace is especially important. If only $\rho_\ell<N$ directions are retained in the diagonal factor, the basis on the remaining $N-\rho_\ell$ directions is not fixed by that fragment Hamiltonian~\cite{gunther2026phaseestimation}. Modes absent from $\hat{D}_\ell$ can therefore be mixed without changing $g^{(\ell)}$, but those choices still appear in the relative rotations between neighboring Trotter layers.

Let the distinct eigenspaces of $g^{(\ell)}$ have dimensions $n_k$. The gauge group of the fragment can be written as
\begin{equation}
    \mathcal{G}(g^{(\ell)})
    =
    \left\{
    % G = 
    \left(\bigoplus_k Q_k\right)P
    \quad \Bigg\vert \quad 
    Q_k \in SO(n_k),\; P\in S_N
    \right\},
    \label{eq:gauge_group}
\end{equation}
where the $Q_k$ rotate within equal-eigenvalue subspaces and $P$ relabels the eigenmode to qubit index map. For any $G\in\mathcal{G}(g^{(\ell)})$,
\begin{equation}
    R_\ell G\,P^\dagger \mathrm{diag}(\varepsilon^{(\ell)})
    P \, G^\dagger R_\ell^\dagger
    =
    R_\ell \, \mathrm{diag}(\varepsilon^{(\ell)})\, R_\ell^\dagger
\end{equation}
with any choice of gauge defined above leaves the simulated Hamiltonian unchanged, as long as we relabel the diagonal layers correspondingly, but can dramatically change the interfragment Gaussian rotations. Under $R_\ell\mapsto R_\ell G_\ell$, an interior layer transforms as
\begin{equation}
    R_{\ell+1}^\dagger R_\ell
    \;\longmapsto\;
    G_{\ell+1}^\dagger
    R_{\ell+1}^\dagger R_\ell
    G_\ell .
    \label{eq:gauge_telescope}
\end{equation}
Since $\gamma_{\mathrm{Choi}}$ depends on the singular values of the partitioned submatrix of this relative rotation, two diagonalizers of the same fragment can give different distribution costs for the same Hamiltonian evolution. In practical implementations, it can be valuable to attempt first-order block-diagonalization to work around instabilities and search for good gauge choices~\cite{bartels1972solution}.

%-----------------------------------------------------------------------------
\subsection{Fragment Ordering and Mode Partition}
\label{sec:ordering_partition}
%-----------------------------------------------------------------------------

After truncation and gauge choice, the remaining freedom is how the retained fragments are composed. A first-order Trotter step may apply the retained DF fragments in any order. The Hamiltonian terms are the same, but the interfragment Gaussian layers are not: placing fragment $\ell$ next to fragment $k$ produces the inter-fragment rotation $R_k^\dagger R_\ell$. For a fixed partition $A|B$, define a pseudo-distance between each fragment
\begin{equation}
    d(\ell,k)
    =
    \log \gamma_{\mathrm{Choi}}\!\left(R_k^\dagger R_\ell\right),
    \label{eq:ordering_distance}
\end{equation}
and optimize by finding an ordering with small total path length. This is a pathfinding problem through fragment bases, with distances given by distribution cost rather than geometry.

The mode partition is an independent choice. For $N$ fermionic modes, a balanced distribution selects a subset $A$ with $|A|=N/2$ and assigns the remaining modes to $B$. The same circuit can have very different QPD distribution complexity for different choices of $A$, because every Gaussian factor in \cref{eq:gamma_choi} depends on the singular values of $R[A,A]$. Spatial locality, embedding, and orbital optimization already exploit the fact that useful electronic degrees of freedom often live in weakly hybridized subspaces
~\cite{knizia2012density,FMO_review_2021,otten2022localized,li2025emo,zhang2026coreoptimizedorbitals}. In our setting, a useful partition is one that makes the retained fragment rotations close to block diagonal across $A|B$.

%-----------------------------------------------------------------------------
\subsection{CS-Spectrum Governing Entanglement Boundaries}
\label{sec:synthetic_cs_spectrum}
%-----------------------------------------------------------------------------

\begin{figure}[t]
  \centering
  \includegraphics[width=\columnwidth]{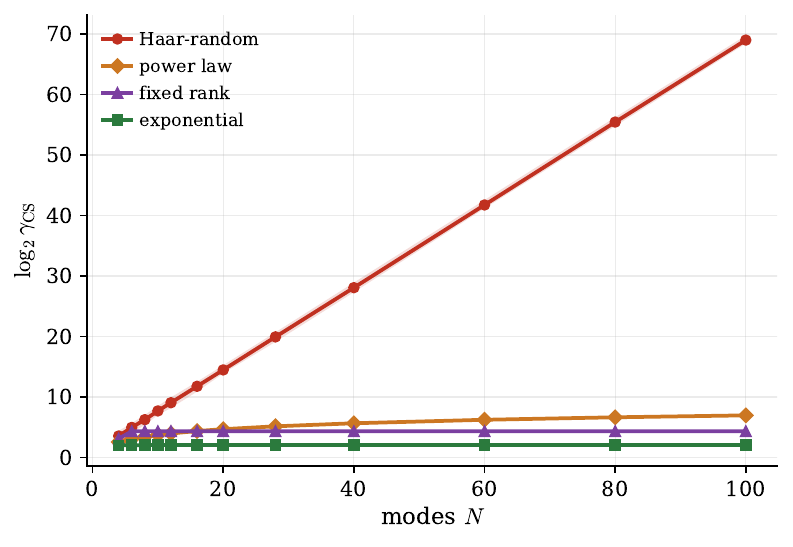}
  \caption{\textbf{Distribution cost of orbital rotations as a function of the operator entanglement spectrum.} Single-channel coupled-region regimes with the interface specified by
  its CS-angle spectrum. Exponentially decaying and fixed-rank crossing spectra
  give bounded Gaussian distribution exponents, a power-law spectrum gives slow
  growth, and a uniform spectrum gives the unstructured volume-law limit.}
  \label{fig:synthetic_cs_envelopes}
\end{figure}

The cosine-sine decomposition gives an operational criterion~\cite{paige1994history,golub2013matrix} for when a Gaussian orbital rotation can be distributed with a given QPD sampling complexity. Let $\sigma_k(R[A,A]) = \cos\theta_k$ be the CS spectrum of the one-particle rotation across the processor cut. The nonlocal part of the CS circuit is a product of cross-partition Givens rotations with angles $\theta_k$, so the LOCC/QPD cost grows with the number and size of the appreciable crossing angles. A bounded, low-rank, or summable crossing spectrum gives an area-law distribution boundary. The worst case is a maximally non-local rotation with uniformly large crossing spectrum, and this gives a volume-law boundary.

For a balanced cut $A|B$, choose the fragment eigenbasis in CS form, 
\begin{equation}
    R
    =
    (Q_A\oplus Q_B)\,C(\theta)\,(V_A\oplus V_B)^{\dagger},
    \qquad
    \sigma_k(R[A,A])=\cos\theta_k .
    \label{eq:synthetic_cs_main}
\end{equation}
The local factors act within the two processors and do not change the
distribution cost. The crossing block $C(\theta)$ is the interface. In this
single-channel limit of the coupled-region model, prescribing $\theta_k$ is the same as prescribing how many one-particle modes cross the boundary.

The question is when a DF representation realizes one of these favorable crossing spectra. The unstructured comparison is the Haar random rotation baseline, where many crossing angles remain appreciable~\cite{wachter1980limiting,truncated_unitaries_2015}. However, if the rotation is aligned with the partition so that only a limited number of cross-partition beamsplitter angles are appreciable, then the Gaussian layer has bounded distribution complexity, as we illustrate in \cref{fig:synthetic_cs_envelopes}. We plot the cumulative distribution cost for an orbital rotation given its cosine-sine decomposition yields crossing mode angles which satisfy (i) power law $\theta_k \propto k^{-1}$, (ii) exponential decay $\theta_k \propto \exp(-k)$, and (iii) fixed rank, where only a constant number of angles are nonzero. These are compared to the Haar random baseline, in which all cross-partition modes contribute with random angles. The natural question that follows is whether we can define systems which are structured with respect to the mechanisms we have outlined in a way that ultimately satisfies some of these bounded entanglement growth envelopes.

%=============================================================================
\section{Characterizing Distribution Complexity in Model Systems}
\label{sec:synthetic}
%=============================================================================

Here we analyze model Hamiltonians, two of which are designed to serve as application-oriented and planted-instance benchmarks, as verifiable problem instances of tunable hardness~\cite{lubinski2021applicationoriented,bellonzi2025qbgsee,wang2026plantedsolutions,watts2024fullerene}. Although they are not derived from specific molecules, they represent controlled DF Hamiltonian data sets that preserve the algebraic structure of electronic Hamiltonians while isolating the features that determine distribution cost: the mode partition, the gauge of low-rank or degenerate eigenspaces, the Trotter order dependence, and ultimately the cumulative cosine-sine crossing spectrum. These model systems allow us to discuss computational resources in terms of ERI tensor properties using the symmetries present in real fermionic systems. Their specific form enables exploring parameter regimes which would allow for the quantum simulation of large, chemically relevant systems on many small QPUs with bounded distribution cost.

We introduce two models that are explicit in their DF parameterizations, which we call the \textit{Localized Fragments} and \textit{Drifting Fragments} models. The Fermi--Hubbard model is distinct from these two in that, while the interactions are local on the lattice, the Fourier basis change that diagonalizes the hopping terms is highly non-local. The so-called "split-operator" technique used to simulate strongly-interacting Hamiltonians relies on using a fermionic fast Fourier transform (FFFT) \cite{PhysRevA.79.032316}. We demonstrate with these models that the optimal mode partitions, fragment gauge choices, and Trotter order are explicit circuit representation choices that have dramatic consequences for simulating the same Hamiltonian data. Together, these models span different distribution complexity regimes for LOCC-assisted or purely quantum distribution protocols.

%-----------------------------------------------------------------------------
\subsection{Localized Fragments: Mode Partitions and Sensitivity}
\label{sec:model_localized_fragments}
%-----------------------------------------------------------------------------

\begin{figure}[t]
  \centering
  \includegraphics[width=\columnwidth]{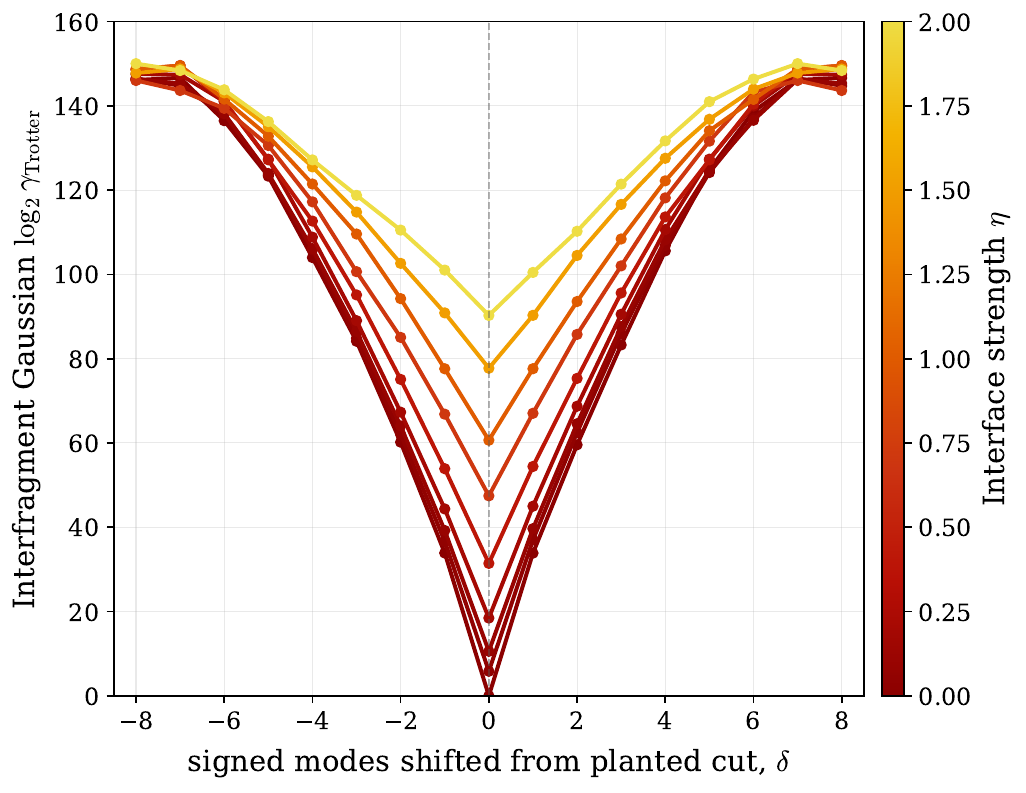}
  \caption{\textbf{Distribution cost of localized fragment Hamiltonians with weak cross-partition coupling.} Here we compute the mode partition sensitivity for nearly block-diagonal, localized fragment Hamiltonians defined in \cref{eq:localized_fragments}, with $N=30$ and $L=8$. We sweep cyclic balanced partition $A=\{j,j+1,\ldots,j+N/2-1\}$ over starting positions. The distribution complexity is clearly minimized when the cut aligns with the planted block structure and rises exponentially with $\eta$ as the partition is shifted away. This quantifies explicitly how locality helps distribution only when the hardware partition exposes the weak interface.}
  \label{fig:partition_heatmap}
\end{figure}

Electronic structure methods exploit locality in several different ways. Active space and downfolding methods identify a correlated subspace, or derive an effective Hamiltonian for one, so that the most important many-body correlations can be treated inside a smaller problem~\cite{roos1980complete, bauman2019downfolding}. Embedding and fragment molecular orbital methods (not to be confused with double factorization fragments discussed here) instead keep a chosen region coupled to an approximate environment, replacing the full molecule by regions whose couplings to the rest of the system are treated as weakly correlated~\cite{knizia2012density,kitaura1999fragment,FMO_review_2021,vorwerk2022quantumembedding,ma2020quantumsimulations,otten2022localized}. Recent quantum-classical chemistry workflows use the same logic at application scale, where QPUs or high-accuracy electronic structure solvers act on embedded molecular subproblems inside a larger HPC pipeline~\cite{merz2026crossing12000atom,gunther2026biomolecularfreeenergies}.

We introduce what we call the \emph{Localized Fragments} model to cleanly isolate the effect of spatial localization and orbital fragmentation in the double-factorized circuits. The question is whether there exists a single-particle basis and a mode partition $A|B$ for which the fragment coefficient matrices are close to block local. In that case, the distribution cost is controlled by the small cross-block part of each fragment rather than by the full one-particle dimension. The model therefore gives a controlled setting in which the layerwise Gaussian and diagonal cost formulas that we have derived identify when a DF Trotter step can be assigned across processors with manageable overhead.

\begin{figure*}[!ht]
  \centering
  \includegraphics[width=\textwidth]{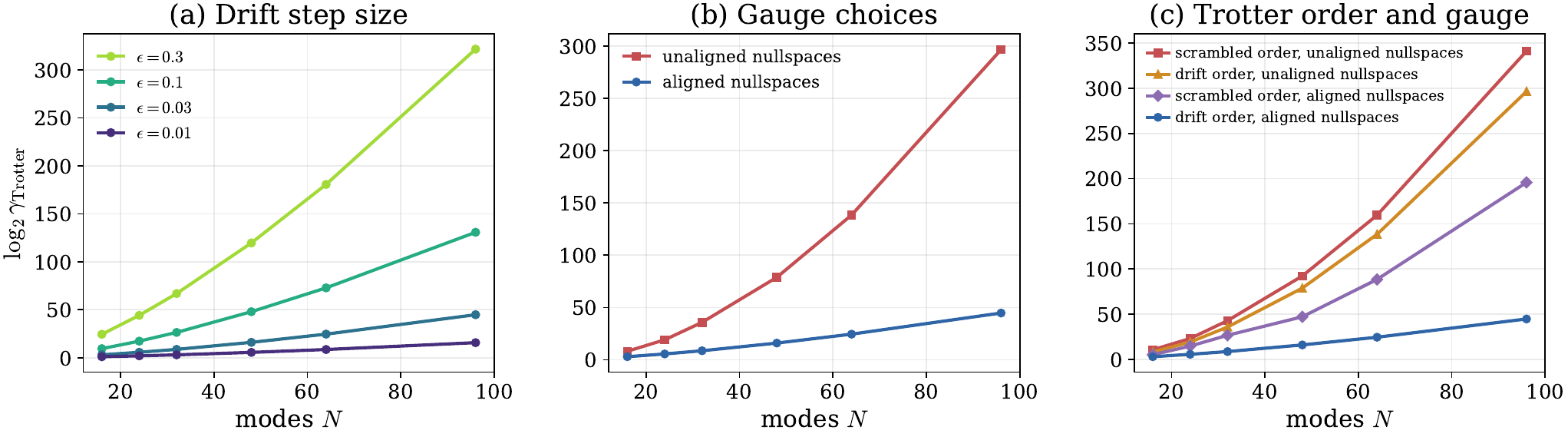}
  \caption{\textbf{Drifting Fragments model and the Trotter-order and gauge dependence of Gaussian distribution complexity.} We consider a Drifting Fragments model with $L=N$ double-factorized fragments and inner rank $\rho_\ell=\lceil\log_2 N\rceil$. We plot the Trotter distribution complexity bound $\gamma_{\mathrm{Trotter}} \leq \prod_\ell \gamma(\hat{U}_\ell) \cdot \gamma(e^{-i\tau \hat{D}_\ell^2})$, treating each of the interfragment orbital rotations $\hat{U}_\ell = \hat{R}_{\ell+1}^\dagger \hat{R}_\ell$ as well as diagonal evolution layers independently. \textbf{(a)} Increasing the drift step size raises the interfragment cost, because adjacent fragment eigenbases differ by larger cross-partition rotations. \textbf{(b)} 
  Low-rank fragments have nullspace gauge freedom: an arbitrary basis in the unused subspace can introduce avoidable crossings by artificially decreasing similarity between adjacent fragment eigenbases, while an aligned gauge keeps those modes local. Here we plot drifting fragments models with fixed drift step size $\epsilon=0.03$, in the case where we take random nullspace bases versus the case where we keep the nullspaces aligned to the drift order.
  \textbf{(c)} The same fragment set is much cheaper when applied in gradual drift order than in scrambled order. The lowest curve uses both drift order and aligned nullspaces, while scrambling the order or leaving nullspaces unaligned dramatically increases the telescoped Gaussian cost.}
  \label{fig:drifting_fragments_mech}
\end{figure*}

Local orbital bases and other orbital-optimization schemes implicitly try to find exactly this kind of representation choice~\cite{lowdin1950nonorthogonality,boys1960construction,pipek1989fast,li2025emo, quiqbox,
zhang2026coreoptimizedorbitals}. In the Localized Fragments model, the cross-block interface of $g^{(\ell)}$ controls how many eigenmodes hybridize across the partition, and whether these crossings are compressible or not. For a given partition of the modes into regions $A$ and $B$, we define the block matrix form of the Localized Fragments model as
\begin{align}
    g^{(\ell)}
    &=
    \begin{pmatrix} A_\ell & \eta\,X_\ell \\ \eta\,X_\ell^{\dagger} & B_\ell \end{pmatrix},
     \nonumber \\
    H_{\mathrm{Loc}}
    &= \frac12\sum_\ell \lambda_\ell
      \Bigl(\sum_{pq} g^{(\ell)}_{pq} a_p^\dagger a_q\Bigr)^2 .
    \label{eq:localized_fragments}
\end{align}
where $A_\ell$ and $B_\ell$ are local block matrix interactions with respect to the mode partitions, and $X_\ell$ is a normalized $\| X_\ell\|_2 = 1$ cross-region interface. This allows us to discuss the interface strength $\eta$ separately from its rank and compressibility.

With no interaction $\eta=0$, every eigenbasis factorizes as $R_\ell=R_{\ell,A}\oplus R_{\ell,B}$ and the simulation becomes embarrassingly parallelizable. By varying the interfaces $X_\ell$ and the coupling strength $\eta$, we can observe the eigenmodes of the fragment hybridize, leading to rotations that are more costly to distribute. In the perturbative regime of weak $\eta$, the number of appreciable cross-partition modes becomes boundary data rather than volume data.

\cref{fig:partition_heatmap} shows the sensitivity of distribution overhead with respect to the choice of mode partition in the case of the localization described above. We generate random local interactions $A_\ell$ and $B_\ell$, as well random interfaces $X_\ell$, and intuitively, we see a sharp dependence on the distance of a cut from the optimal partition. As the strength of the interaction $\eta$ grows, we see an exponential saturation to the worst case behavior, which ultimately saturates when $\eta$ is equal to $ \|A_\ell\|$ and $\|B_\ell\|$ (when inter-region couplings are equal to the intra-region couplings). 

We also consider the case of interfaces with small effective dimensions. Although the absolute interaction strength $\eta$ is not negligible, the interactions can be carried by low-rank or compressible ``channels.'' This models a boundary where only a few orbital combinations carry most of the cross-region coupling. This mechanism is distinct from weak coupling. The parameter $\eta$ sets the strength of the cross-region block, while the rank and singular-value decay of $X_\ell$ determine how many independent directions can mix across the partition. If $\operatorname{rank}(X_\ell)=r_\ell$, then at most an $r_\ell$-dimensional set of orbital combinations is directly coupled across $A|B$ at leading order. If the singular values of $X_\ell$ decay rapidly, the same statement holds approximately. These active directions are the ones that generate appreciable cross-partition angles in the CS decomposition, so an interface can have finite strength while still producing small distribution cost when it is low rank or compressible.

We take a moment to discuss numerical stability and the gauge choice, which is particularly easy to understand in the Localized Fragments model. When the spectra of $A_\ell$ and $B_\ell$ are well separated, standard numerical linear algebra implies that rotations of the fragment eigenbasis are well conditioned~\cite{davis1970rotation,stewart1990matrixperturbation}. In the case of degeneracies between $A_\ell$ and $B_\ell$, such as in the case of dimers like H$_2$, Cr$_2$, and so on in symmetric bases, the exact orbital eigenmodes become unstable under these cross-partition perturbations, and care needs to be taken to optimize over these perturbative representations~\cite{paige1994history,golub2013matrix}.

%-----------------------------------------------------------------------------
\subsection{Drifting Fragments: Gauge Representations and Trotter Order}
\label{sec:representation_choices}
%-----------------------------------------------------------------------------

A different kind of structure appears when the double-factorized fragment eigenbases are correlated with one another, even while individual fragments may be delocalized. There may be interaction channels that are ordered by some collective coordinate. Low-rank, double-factorized, and tensor-hypercontracted electronic structure methods show that useful interaction representations have empirical rank and compression structure rather than behaving like independent random matrices~\cite{peng2017lowrank,motta2019afqmc_lowrank,motta2021lowrank,lee2021thc}. We explicitly consider the case where there is cross-fragment correlation beyond just outer compressibility.
The Drifting Fragments model considers the presence of these correlations within a single electronic system by taking a sequence of nearly identical fragment eigenbases:
\begin{align}
    \hat{R}_{\ell+1}&=\hat{R}_\ell\,\hat{F}_\ell; \nonumber \\ 
    H_{\mathrm{Drift}}
    &= \frac12\sum_{\ell=1}^{L} \lambda_\ell
        \ \hat{R}_\ell \left(\sum_m \varepsilon^{(\ell)}_m a^\dagger_m a_m\right)^2 \hat{R}^\dagger_\ell
    \label{eq:drifting_basis}
\end{align}
where at each step the relative Gaussian rotation $\hat{F}_\ell$, which we call the ``drift operator'', is close to the identity, and the resulting fragment matrices $g^{(\ell)}$ can be variable rank. Applying the fragments in their natural order therefore keeps the Gaussian distribution path short, where the Trotter circuit takes a ``smooth'' path through the different eigenbases only rotating a small amount at a time, while scrambling the same fragments replaces nearby basis changes with long jumps through the space of orbital eigenbases. 
If the ERI tensor decomposes into components with nearly simultaneous eigenbases like this, even when no single fragment basis is localized, the efficient computation of interfragment rotations' distribution overhead once again allows the entanglement hypervisor to discover an efficient protocol.

\begin{figure*}[!ht]
  \centering
  \includegraphics[width=0.8\textwidth]{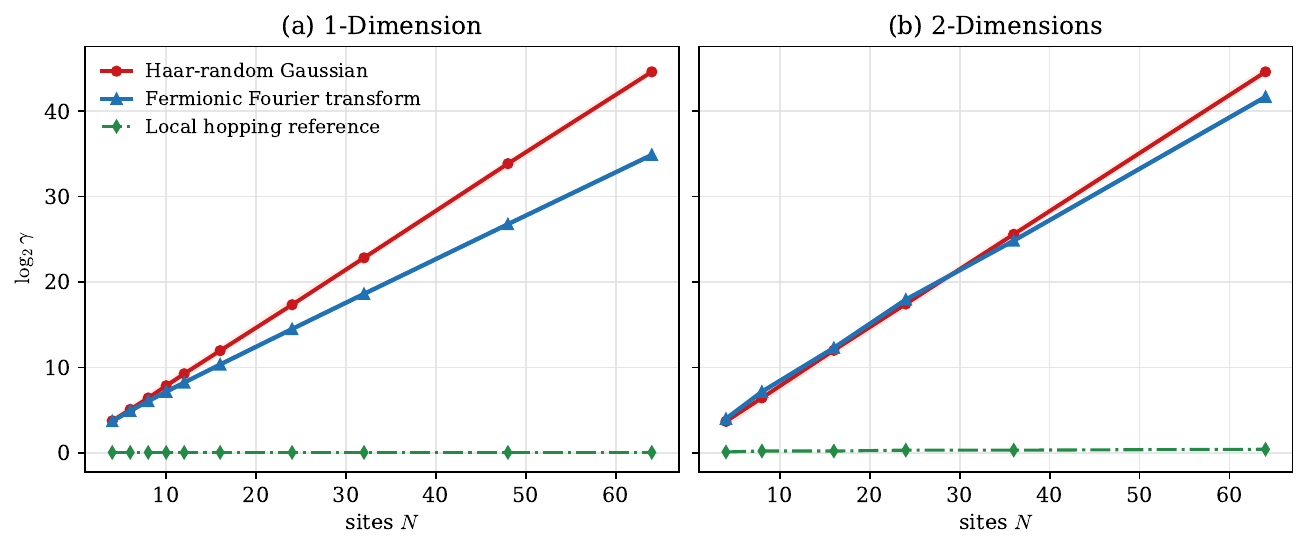}
  \caption{\textbf{Distribution complexity of Fourier orbital rotations across contiguous real-space cuts.} 
    We use the Fermi--Hubbard kinetic term as a model example where the hopping Hamiltonian is diagonal in momentum space, but the basis change to momentum space is not local on a spatially distributed architecture. For each lattice size $N$, we plot the Gaussian Choi distribution complexity $\log_2\gamma$ across a contiguous half-system real-space cut. \textbf{(a)} In one dimension, the fermionic Fourier transform has a rapidly growing distribution exponent across the contiguous cut, close to the median over $100$ Haar-random number-conserving Gaussian rotations on the same number of modes. The local hopping reference is an open-boundary nearest-neighbor hopping layer at $\tau =0.01$; because only the hopping term crossing the cut contributes, its cost stays near zero. \textbf{(b)} In two dimensions, the same comparison holds for a contiguous spatial cut through the lattice: the Fourier transform remains close to the Haar-random Gaussian scale, while the local hopping reference grows only with the length of the spatial boundary.}
  \label{fig:hubbard_fourier_hardness}
\end{figure*}

We use an isotropic random walk in the orbital basis, and vary the step size between adjacent fragment eigenbases. Specifically, we use unbiased generators in the Lie algebra $\mathfrak{so}(N)$ to generate the small relative rotation matrices:
\begin{equation}
    F_\ell = \exp(\epsilon K_\ell)
    \label{eq:drifting_step}
\end{equation}
where $K_\ell$ is a Frobenius-normalized Gaussian skew-symmetric matrix, so that $\epsilon$ alone sets the drift step size. Choosing generators this way ensures that the step between adjacent eigenbasis is isotropic, rather than being tied to a particular orbital subspace as in the Localized Fragments model. This controls the complexity of the interfragment orbital rotations in the Gaussian layers $\hat{U}_\ell = \hat{R}_{\ell}^\dagger \cdot \hat{R}_{\ell+1}$, 
even when the individual fragment bases are globally delocalized or arbitrary rank. 
\cref{fig:drifting_fragments_mech} illustrates three ways the smooth-path structure enters the distribution cost of a Trotter circuit representation. In panel (a), increasing the drift step size $\epsilon$ increases the mismatch between adjacent eigenbases, leading to an exponential increase in complexity until saturation to the Haar-random level. In panel (b) we explicitly consider low rank fragments, illustrating the impact of arbitrary choices for the nullspace as discussed in \cref{sec:gauge}. Naive numerical eigensolvers can produce arbitrary bases for these subspaces and therefore introduce significant inter-fragment rotations that are not necessary by the physical interactions. An aligned gauge across low-rank fragments can remove this artificial cost. In panel (c), we combine this insight on nullspaces with the Trotter order with respect to the drift operators. The scrambled order destroys the small relative rotation, but putting the fragments in the most compatible order still requires nullspace alignment (and generally the correct gauge choices) to fully recover slow-growing distribution complexity. 

When the fragments are completely non-degenerate, the only gauge choice is a permutation choice (mapping mode $i$ in the eigenvector order of fragment $\ell$ to mode $j$ in the eigenvector order of fragment $\ell+1$), which can still introduce dramatic artificial entanglement.  Even this discrete choice can introduce large crossings, essentially inserting unnecessary SWAPs across the cut, which are among the most expensive two-qubit operations in the QPD benchmarks (see \cref{app:kak_baseline} for more details). These representation choices (and not just the parameters of the Hamiltonian) determine how much nonlocal Gaussian rotations appear in the compiled circuit and ultimately the total distribution cost for the circuit.

%-----------------------------------------------------------------------------
\subsection{The Role of Orbital Rotations: Fermionic Fourier Transforms vs Random Unitaries}
\label{sec:hubbard_fourier}
%-----------------------------------------------------------------------------

The Fermi--Hubbard model is a standard minimal model of strongly-correlated lattice
fermions~\cite{hubbard1963electron,Arovas2022},
\begin{equation}
    H_{\mathrm{FH}}
    = -t\sum_{\langle i,j\rangle,\sigma}
      \bigl(a_{i\sigma}^\dagger a_{j\sigma}+\mathrm{h.c.}\bigr)
    + U\sum_i\hat{n}_{i\uparrow} \hat{n}_{i\downarrow}.
    \label{eq:hubbard}
\end{equation}
Here, \(t\) denotes the hopping strength and \(U\) is the fermion-fermion repulsion. The Fermi-Hubbard model is extremely versatile, in the sense that realizations of this model can capture relevant physics of material properties such as (anti)ferromagnetism, unconventional superconductivity, topological order (e.g., spin-liquid phases) etc.~\cite{Arovas2022}.

Classical simulation of the Fermi-Hubbard model is especially challenging due to several competing phases separated by small energy scales in the ground-state. Standard techniques such as tensor networks struggle with the growth of entanglement while quantum Monte Carlo struggles with the sign problem; although there is some progress in solving the ground-state properties of this model via fermionic Projected Entangled Pair States~\cite{Liu2025}. Given the classical hardness of simulating this model, as well as its versatility in capturing the physics of materials, quantum simulation of the Fermi-Hubbard model is anticipated to be one of the key applications of quantum computers in the near-term~\cite{mohseni2026buildquantumsupercomputerscaling,PhysRevApplied.9.044036,alam2025programmabledigitalquantumsimulation}.

Note that the Fermi-Hubbard model can be thought of as a simplified version of the electronic structure Hamiltonian considered in this work as it can be obtained by imposing additional constraints such as only a single localized orbital per site, a single-particle matrix of the tight-binding form, and an ERI tensor with only on-site Coulomb elements. A simplified Fermi-Hubbard model only has a hopping term and an interaction term. For quantum simulation, the leading approach is to use the so-called "split operator" Trotterization: to diagonalize the kinetic and potential terms independently using a fermionic Fourier transform \cite{Verstraete2009}. Our first order Trotterization followed by distribution approach closely resembles the same intuition. Namely, the interaction is diagonal in the real-space basis and the kinetic term is
diagonal in the momentum basis, allowing simulation which alternates between the two bases~\cite{wecker2015solving,babbush2018lowdepth}. Then the simulated interactions themselves are all local, and the nonlocal operation is the change of basis. The basis change here is achieved by the fermionic Fourier transform, which is by itself a Gaussian circuit, but which has a large Choi exponent across a contiguous real-space partition. See \cref{fig:hubbard_fourier_hardness} for a discussion of the distribution complexity of fermionic Fourier transforms.

%-----------------------------------------------------------------------------
\section{Complexity Landscape: Classical Simulability vs. Distribution Scaling}
\label{sec:flo}
%-----------------------------------------------------------------------------
\begin{figure*}[!ht]
  \centering
  \includegraphics[width=0.8\textwidth]{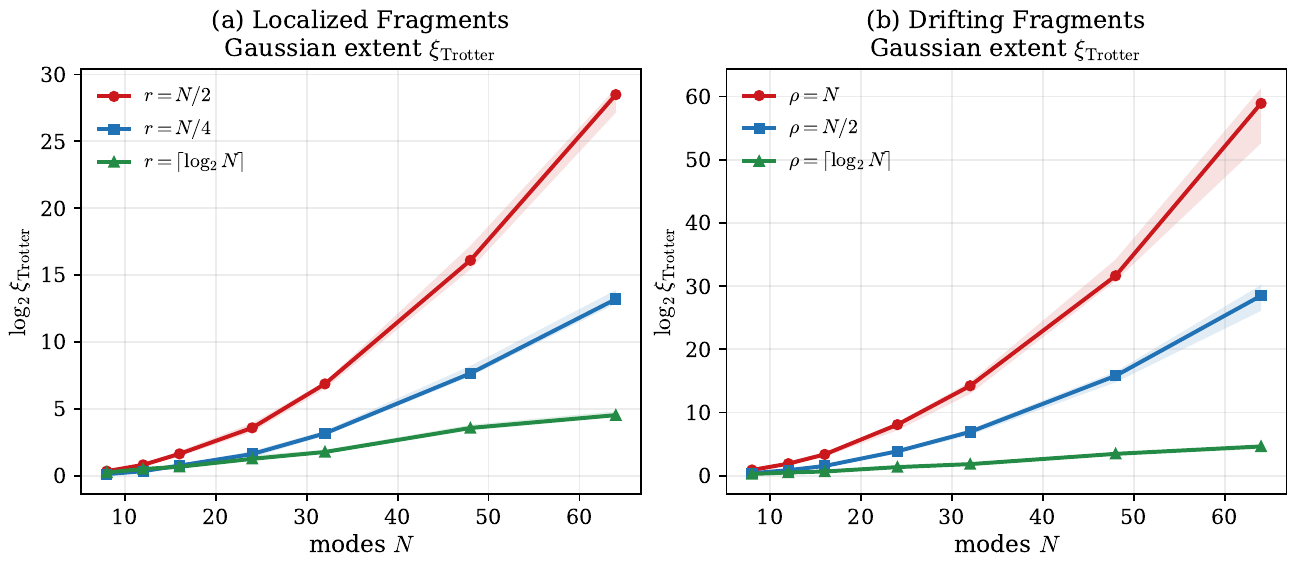}
  \caption{\textbf{(Distribution complexity is independent from classical simulability.)} 
    In both panels, we plot the diagonal-layer Gaussian extent $\log_2\xi_{\mathrm{Trotter}}^*$, which determines the classical simulation overhead of the Gaussian-plus-CPhase algorithm, for one DF Trotter step at fixed small $\tau=0.05$.
    \textbf{(a)} In the Localized Fragments model, increasing the intra-region diagonal rank $r$ (the rank of block matrices $A_\ell$ and $B_\ell$ in \cref{eq:localized_fragments}, not to be confused with the rank of the interfaces $X_\ell$) adds local non-Gaussian CPhase structure inside each region while leaving unchanged the cross-region interfaces $\eta X_\ell$ as the mechanism controlling distribution cost. 
    \textbf{(b)} In the Drifting Fragments model, increasing the inner fragment rank $\rho_\ell$ adds diagonal CPhase structure while the Gaussian distribution cost is controlled separately by the path of adjacent fragment eigenbases.
    Full-rank diagonal spectra result in exponential complexity for classical simulability via Gaussian extent methods, while weakly hybridizing or well-ordered fragment eigenbases can keep distribution complexity bounded even as we increase rank and the number of orbitals.
  }
  \label{fig:model_decoupling}
\end{figure*}

The distribution complexity, characterized by $\gamma$, measures the difficulty of distributing quantum simulations across multiple quantum processors. This is very different from the classical simulation cost of the same circuit, and it is worth analyzing them both for a given application instance. 
Classical simulation algorithms exist that are designed for circuits dominated by fermionic Gaussian unitaries (also known as "matchgates"), whose evolution is efficiently tracked using the covariance matrix representation~\cite{Valiant2001,PhysRevA.65.032325}. Non-Gaussian gates in the circuit disrupt this structure; however there exist classical simulation techniques techniques that can tolerate logarithmically many non-Gaussian gates~\cite{PhysRevA.102.052604,PhysRevLett.123.080503}. We choose a method in which the only non-Gaussian gates are assumed to be Controlled-Phase (CPhase) gates, which is the case for our double-factorized Trotter circuits. Note that Gaussian unitaries together with CPhase gates are universal for electronic structure simulation (and quantum computation in general)~\cite{leimkuhler2026exponentialquantumspeedupsnearterm}.

Let $\hat{U}_\mathrm{GCP}$ be a quantum circuit which is made up of only Gaussian fermionic unitaries interspersed with $N_\mathrm{CPhase}$-many CPhase gates. Each $\operatorname{CPhase}(\theta) \equiv \exp(i \theta \, \hat{n}_i \hat{n}_j)$ can be written as a linear combination of two Gaussian fermionic unitaries
\begin{equation}
    \operatorname{CPhase}(\theta)
    =
    e^{i\theta/4}
    \left[
    \cos\!\left(\frac{\theta}{4}\right)\hat{R}_0(\theta)
    +
    i\sin\!\left(\frac{\theta}{4}\right)\hat{R}_1(\theta)
    \right],
    \label{eq:cphase_gaussian_decomp}
\end{equation}
which can be done independently and multiplied out, such that the total circuit can be written as a linear combination of $2^{N_\mathrm{CPhase}}$-many Gaussian fermionic unitaries (all of the possible choices between $\hat{R}_0$ or $\hat{R}_1$ for each CPhase). 
\begin{equation}
    \hat{U}_\mathrm{GCP} = \sum_{b=1}^{2^{N_\mathrm{CPhase}}} a_b \, \hat{R}_b 
\end{equation}
Then the goal is to classically simulate this decomposition using only its Gaussian fermionic unitaries $\hat{R}_b$, just as the goal of distribution protocols is to simulate the quasiprobability decompositions of a quantum circuit using only its LOCC operators. However, since this is trivially exponential in the number of CPhase gates in the quantum circuit, it is important to find improved and approximate simulation methods.

Recently a number of algorithms have been developed that use the Gaussian fermionic unitaries as the free resource theory in the same way that stabilizer simulation methods use Clifford unitaries. Clifford gates are also efficiently simulable, and the exponential part of the simulation cost is controlled by the non-Clifford resource, which can be quantified by the number of $T$ gates, stabilizer rank, or stabilizer extent~\cite{gottesman1998heisenberg,aaronson2004improved,bravyi2016trading,bravyi2016improved,bravyi2019simulation}. Similarly, for Gaussian plus CPhase circuits, we focus on Gaussian trajectory sampling algorithms, for which the relevant parameter is the \emph{Gaussian extent}, which is defined for each CPhase gate in terms of its angle:
\begin{equation}
    \xi[{\mathrm{CPhase}(\theta)}]
    =
    \left[
    \cos\!\left(\frac{|\theta|}{4}\right)
    +
    \sin\!\left(\frac{|\theta|}{4}\right)
    \right]^2 .
    \label{eq:gaussian_extent_cphase}
\end{equation}
Then the direct composition described above is also multiplicative for the quantum circuit $\hat{U}_\mathrm{GCP}$
\begin{equation}
    \xi^* = \prod_{i=1}^{N_\mathrm{CPhase}} \left[
    \cos\!\left(\frac{|\theta_i|}{4}\right)
    +
    \sin\!\left(\frac{|\theta_i|}{4}\right)
    \right]^2
\end{equation}
we write $\xi^*$ for the aggregate circuit-level extent~\cite{reardonsmith2024improved,hassman2025flo}, as opposed to the extent of an individual gate. We note that despite the name, the Gaussian extent is by definition a measure of the \emph{non-Gaussianity} of a unitary. As in stabilizer and Gaussian extent, $\xi^*$ is already a squared coefficient one-norm~\cite{bravyi2019simulation,dias2024nongaussian,cudby2023gaussian};
this differs from our QPD convention, where $\gamma$ is the unsquared
quasiprobability one-norm and sampled circuit executions scale as
$\gamma^2$~\cite{mitarai2021overhead,piveteau2024knitting}. In the extended matchgate
simulation literature this quantity is also called the Fermionic Linear Optical (FLO) extent~\cite{hassman2025flo}, while closely related resource-theoretic papers use the names fermionic Gaussian extent or just Gaussian extent~\cite{dias2024nongaussian,cudby2023gaussian,reardonsmith2024improved}, as we adopt here.

In a double-factorized Trotter step, these CPhase gates arise from the
off-diagonal terms in the square of each diagonal one-body operator. A
rank-$\rho_\ell$ diagonal fragment contributes $\binom{\rho_\ell}{2}$ CPhase
gates, resulting in $O(L\cdot \langle \rho_\ell^2\rangle)$ total CPhase gates, with angles
\begin{equation}
    \theta_{mw}^{(\ell)}
    =
    \tau\lambda_\ell
    \varepsilon_m^{(\ell)}\varepsilon_w^{(\ell)},
    \qquad m<w .
\end{equation}
Thus the Gaussian extent of one retained DF Trotter step is
\begin{equation}
    \xi_{\mathrm{Trotter}}^*
    =
    \prod_{\ell=1}^L
    \prod_{m<w}^{\rho_\ell}
    \left(1+
        \sin\!\left(
        \frac{|\tau\lambda_\ell
        \varepsilon_m^{(\ell)}\varepsilon_w^{(\ell)}|}{2}
        \right)
    \right).
    \label{eq:gaussian_extent_df}
\end{equation}

This expression explicitly depends on the diagonal spectra, fragment weights, retained inner ranks, and simulation time. It does not depend on the interfragment orbital rotations at all, because the Gaussian-trajectory simulation tracks those rotations efficiently by construction. We illustrate the scaling behavior of $\xi^*$ with respect to the fragment ranks of both the Localized Fragments model and Drifting Fragments model in \cref{fig:model_decoupling}. Rank dependence also enters the product over the $\binom{\rho_\ell}{2}$ CPhase angles in each retained fragment. Therefore, $\xi_\mathrm{Trotter}^*$ grows exponentially with the average fragment rank, and each full-rank fragment contributes a factor of up to $O(\exp(N^2))$ to the Gaussian extent. Meanwhile, we showed that the contribution of the diagonal layers to the distribution cost $\gamma_\mathrm{Trotter}$ is bounded even for full rank fragments. In \cref{sec:coulomb}, we showed that each diagonal fragment layer contributes a factor of $O(|\lambda_\ell \tau|)$. The squared one-body operator structure of the evolution leads to this scaling, saturating regardless of $N$ and $\rho_\ell$.

\begin{table*}[!htbp]
\footnotesize
\providecommand{\distyes}{\textcolor{green!50!black}{\ding{51}}}
\providecommand{\distno}{\textcolor{red!70!black}{\ding{55}}}
\caption{\textbf{Regimes of simulatability and distributability of DF Hamiltonian model.} The table classifies parameter regimes by whether the Gaussian extent $\xi_{\mathrm{Trotter}}^*$ or LOCC/QPD distribution complexity $\gamma_{\mathrm{Trotter}}$ are polynomial in the number of modes $N$. ``Gauss. sim.'' refers to classical simulation of the diagonal layers by Gaussian+CPhase trajectory sampling. ``LOCC dist.'' refers to circuit knitting over classical communication. ``Quantum dist.'' refers to coherent distribution with quantum interconnects, which is feasible for the regimes shown whenever $\log_2\gamma_{\mathrm{Trotter}}=\mathrm{poly}(N)$.}
\label{tab:complexity_landscape}

\resizebox{\textwidth}{!}{%
\renewcommand{\arraystretch}{1.18}
\begin{tabular}{@{}lllll@{}}
\toprule
& Parameter regimes
& \begin{tabular}[t]{@{}c@{}}Classical simulation\\Gaussian extent $\xi_{\mathrm{Trotter}}^*$\end{tabular}
& \begin{tabular}[t]{@{}c@{}}LOCC distribution\\complexity $\gamma_{\mathrm{Trotter}}$\end{tabular}
& \begin{tabular}[t]{@{}c@{}}Feasibility\\\end{tabular} \\
\midrule
\begin{tabular}[t]{@{}l@{}}\textbf{Simulable}\\\textbf{and distributable}\end{tabular}
& \begin{tabular}[t]{@{}l@{}}
\textit{Localized Fragments}: weak interface and low-rank $X_\ell$;\\
\quad low diagonal rank.\\
\textit{Drifting Fragments}: $L=O(N)$ and $\langle\rho\rangle=O(\log N)$;\\
\quad controlled weighted diagonal phases, motivated by low-rank\\
\quad Gaussian-basis compression~\cite{peng2017lowrank,motta2019afqmc_lowrank,motta2021lowrank}.
\end{tabular}
& $\mathrm{poly}(N)$
& $\mathrm{poly}(N)$
& \begin{tabular}[t]{@{}l@{}}Gauss. sim.~\distyes\\LOCC dist.~\distyes\\Quantum dist.~\distyes\end{tabular} \\
\midrule
\begin{tabular}[t]{@{}l@{}}\textbf{Simulable but}\\\textbf{distribution-hard}\end{tabular}
& \begin{tabular}[t]{@{}l@{}}
\textit{Fourier/Hubbard Gaussian step}: global Fourier basis change\\
\quad across a spatial cut.\\
\textit{Unstructured Gaussian step}: volume-law crossing spectrum\\
\quad with absent or compressed diagonal phases.
\end{tabular}
& $\mathrm{poly}(N)$
& $\exp[\Omega(N)]$
& \begin{tabular}[t]{@{}l@{}}Gauss. sim.~\distyes\\LOCC dist.~\distno\\Quantum dist.~\distyes\end{tabular} \\
\midrule
\begin{tabular}[t]{@{}l@{}}\textbf{Large Gaussian extent}\\\textbf{but distributable}\end{tabular}
& \begin{tabular}[t]{@{}l@{}}
\textit{Localized Fragments}: high intra-region diagonal rank;\\
\quad weak cross-region interface.\\
\textit{Drifting Fragments}: high inner rank;\\
\quad smooth eigenbasis path.\\
\textit{Rank-one diagonal DF layer}: $\Theta_{ij}=\lambda\varepsilon_i\varepsilon_j$.
\end{tabular}
& $\exp[\mathrm{poly}(N)]$
& $\mathrm{poly}(N)$
& \begin{tabular}[t]{@{}l@{}}Gauss. sim.~\distno\\LOCC dist.~\distyes\\Quantum dist.~\distyes\end{tabular} \\
\midrule
\begin{tabular}[t]{@{}l@{}}\textbf{Large Gaussian extent}\\\textbf{and distribution-hard}\end{tabular}
& \begin{tabular}[t]{@{}l@{}}
\textit{Localized Fragments}: strong and full-rank interface.\\
\textit{Drifting Fragments}: full-rank and scrambled fragments.\\
\textit{Unstructured baseline}: full-rank bases with effectively\\
\quad independent cross-boundary phases.
\end{tabular}
& $\exp[\mathrm{poly}(N)]$
& $\exp[\mathrm{poly}(N)]$
& \begin{tabular}[t]{@{}l@{}}Gauss. sim.~\distno\\LOCC dist.~\distno\\Quantum dist.~\distyes\end{tabular} \\
\bottomrule
\end{tabular}%
}
\end{table*}

Of course, this is only one choice of classical simulation technique. The Gaussian extent is a cost parameter where fermionic Gaussian evolution is treated efficiently and CPhase gates are treated as the branching non-Gaussian resource. Other nearly-Gaussian simulation methods organize the resource differently. For example, matchgate circuits also become universal when combined with the SWAP gate or other non-matchgates~\cite{
jozsa2008matchgates,brod2011extending,
mocherla2024extendingmatchgatesimulationmethods}, while Majorana Propagation tracks observables in a Majorana-monomial basis and truncates operator growth rather than tracking state evolution via Gaussian trajectory sampling~\cite{miller2025majoranapropagation}. While polynomial Gaussian extent $\xi_{\mathrm{Trotter}}^*$ is sufficient to show that a quantum circuit can be efficiently simulated, large Gaussian extent is not itself a universal hardness result. This is analogous to the stabilizer-rank simulation scheme where large stabilizer/extent only implies hardness of simulation in the stabilizer formalism and not necessarily by other techniques such as tensor networks.

For our double-factorized Trotter circuits, however, it is a natural classical comparison. The Gaussian extent and distribution by LOCC implementable operators both exploit the linear combination over the resource unitaries, and yield sample complexity with respect to the one-norm of the linear coefficients~\cite{arunachalam2026tomographyquantumstatesbounded}. The contrast between $\xi_{\mathrm{Trotter}}^*$ and $\gamma_{\mathrm{Trotter}}$ compares the cost of classical Gaussian-trajectory sampling versus distributed quantum execution.

We summarize different regimes of classical simulability and distribution complexity  characterized by the Localized Fragments and Drifting Fragments models in \Cref{tab:complexity_landscape}.
The Fourier basis and uniform crossing spectra are classically easy but
distribution-hard. Block-local regions, decaying crossing spectra, and naturally
ordered drifting bases are favorable on all axes. Rank-one diagonal layers and
full-rank ordered drifting bases are classically hard but distribution-cheap. An
unstructured basis with full-rank diagonal phases is hard on both the
classical-simulation and classical-interconnect axes. A coherent multi-QPU quantum supercomputer with quantum  interconnects can distribute across all regimes discussed here, with a Bell-pair cost of $\log\gamma$ that scales polynomially in number of orbitals.

%=============================================================================
\section{Discussion and Outlook}
\label{sec:conclusion}
%=============================================================================

Simulating industry-scale quantum chemistry and material science problems on fault-tolerant quantum computers will eventually necessitate distributing fermionic simulation across a network of QPUs. However, the inherent nonlocality of interacting fermions makes optimally partitioning and merging quantum circuits a highly nontrivial task. In this work, we adopt operator entanglement as a measure to capture the boundaries of highly correlated subproblems or sub-circuits, leading to novel workflows for hybrid distributed quantum-classical simulations. In particular, we have introduced preprocessing algorithms that efficiently characterize the entanglement structure of Trotter circuits for electronic Hamiltonians before compilation into elementary two-qubit gates. This information can then be used for building a \textit{entanglement-aware hypervisor} to manage workloads across QPUs and GPUs via quantum interconnects (distributed Bell pairs across QPUs) and/or conventional classical HPC interconnects. 

We primarily focus on quantum circuits originating from Trotter-Suzuki expansion for double-factorized electronic structure, which consists of alternating orbital rotations and Coulomb interactions. For Gaussian orbital rotations, we show how to compute the distribution cost for a tensor fragment via its single-particle rotation matrix, replacing an exponentially large Fock-space diagonalization with an \(O(N^3)\) calculation. This exponential reduction in estimating the distribution cost arises from carefully utilizing the Gaussian unitary structure of orbital rotations, which standard gate-by-gate methods ignore. Our~\cref{thm:gaussian_cs_protocol} further provides a quadratic improvement in compiling Gaussian unitaries so that only \(O(N)\) two-qubit gates cross a partition in the worst-case, achieving a scaling close to the optimal lower bound. It is worth highlighting that our distribution complexity results are independent of the \textit{computational complexity} classification of solving static or dynamic electronic structure problems~\cite{Schuch2009,OGorman2022}.

Similarly, for the diagonal Coulomb interactions, we show that their rank-one structure together with disorder in the fragment spectrum results in a dephasing-induced localization within each tensor fragment. This dephasing mechanism implies a slow, logarithmic growth of operator entanglement, reminiscent of many-body localization. We provide analytical bounds on the distribution cost of such diagonal unitaries and numerically show this bound is satisfied by small molecules. Together, these analytical and numerical results enable us to lower bound the distribution cost of a Coulomb layer to \textit{at most} polynomial for all relevant Trotter steps and \(O(1)\) for most times of interest. In an upcoming work, we will show how the bounded distribution cost translates into an efficient tensor network representation for the Coulomb layers~\cite{QChem_MBL_paper}.

Our work identifies various mechanisms underlying distribution complexity in electronic structure and how they can be optimized for application-specific instances. As an example, for correlation energy estimation to chemical accuracy, we show that our optimizations enable truncating weak fragments in a systematic way to reduce distribution complexity by orders of magnitude. We numerically show the impact of this for small molecules such as LiH, N$_2$, and CH$_4$ in STO-3G basis. It will be interesting to see if the techniques introduced here can be generalized to larger problems such as distributing the active-space of FeMoCo~\cite{zhai2026classicalsolutionfemocofactormodel}.

Our algorithms and analyses are primarily focused on Trotterization since it preserves locality, including the entanglement structure behind Gaussian and Coulomb layers, enabling  independent and efficient estimation of their distribution complexity. Another promising approach to quantum simulation is based on qubitization/quantum signal processing. However, unlike Trotterization, qubitization block-encodes the Hamiltonian into a \textit{global} walk operator. This typically washes out the local structure and makes its distribution less transparent, since the low-entanglement boundaries are not revealed as readily. While qubitization is resource-optimal in the fault-tolerant framework, this is primarily optimized for monolithic architectures and its distribution compatibility remains unclear. Our work suggests an alternative path towards utility-scale quantum computing, where algorithms compatible with distributed quantum simulation are preferred, similar to current day classical supercomputing.

Future work will focus on generalizing the algorithms introduced here to other fermionic ans\"atze such as Unitary Coupled-Cluster Singles and Doubles (UCCSD) and local unitary cluster Jastrow (LUCJ); which encode the locality and interactions in electronic structure. For near-term quantum simulation, it will be useful to perform hardware-specific optimizations aimed at reconfigurable quantum systems such as neutral atoms \cite{Bluvstein2023} and trapped ion architectures \cite{Schwerdt2024}. Reconfigurable systems can implement the fast fermionic Fourier transform efficiently~\cite{maskara2025fastsimulationfermionsreconfigurable} as well as implement the diagonal long-range Coulomb interactions natively and may be particularly well-suited for electronic structure. 

The framework introduced here could provide a  workload management and resource-estimation toolbox for heterogeneous quantum supercomputing architectures. This includes QPUs interacting with heterogeneous classical compute such as CPUs, GPUs, and probabilistic processing units (PPUs); as well as heterogeneous qubit modalities which may be optimized for fast gates/measurements, quantum memory, and long-range connectivity. In such heterogeneous QPU networks, our entanglement-aware hypervisor could decide not only between quantum vs. classical interconnects but also which hardware modality should host which part of the electronic structure simulation.

It will also be interesting to explore the feasibility of distributing partially fault-tolerant quantum computation. Partial FTQC typically requires an order of magnitude fewer physical qubits but requires small angle rotations, which naturally appear in Trotterized quantum simulation~\cite{chung2026partiallyfaulttolerantquantumcomputation}. There is also the possibility to identify heterogeneous distribution of quantum correlations across the quantum circuit, and deploy quantum interconnects even more surgically to minimize the inter-QPU entanglement for distributed simulation~\cite{hu2026entanglementassistedcircuitknittingdistributed}. Together, the advances in entanglement-aware compilation and distribution, application-centric quantum architectures and algorithms, and distributed quantum error correction will be key in enabling useful applications to quantum chemistry, material science, and condensed matter physics.\\

%=============================================================================
\section*{Acknowledgements}
We thank Andrew Projansky, Weishi Wang, Brent Harrison, and Sam Alterman for useful discussions.
We acknowledge support from
DARPA's Quantum Benchmarking Initiative (QBI), contract no.\ HR00112590116.  JN and JDW are supported by the Army Research Office under award W911NF2410043 as well as the Department of Energy under award DE‐SC0019374.
JDW holds concurrent appointments at Dartmouth College and as an Amazon
Visiting Academic. This paper describes work performed at Dartmouth College
and is not associated with Amazon.

\bibliography{refs}

\begin{thebibliography}{166}%
\makeatletter
\providecommand \@ifxundefined [1]{%
 \@ifx{#1\undefined}
}%
\providecommand \@ifnum [1]{%
 \ifnum #1\expandafter \@firstoftwo
 \else \expandafter \@secondoftwo
 \fi
}%
\providecommand \@ifx [1]{%
 \ifx #1\expandafter \@firstoftwo
 \else \expandafter \@secondoftwo
 \fi
}%
\providecommand \natexlab [1]{#1}%
\providecommand \enquote  [1]{``#1''}%
\providecommand \bibnamefont  [1]{#1}%
\providecommand \bibfnamefont [1]{#1}%
\providecommand \citenamefont [1]{#1}%
\providecommand \href@noop [0]{\@secondoftwo}%
\providecommand \href [0]{\begingroup \@sanitize@url \@href}%
\providecommand \@href[1]{\@@startlink{#1}\@@href}%
\providecommand \@@href[1]{\endgroup#1\@@endlink}%
\providecommand \@sanitize@url [0]{\catcode `\\12\catcode `\$12\catcode `\&12\catcode `\#12\catcode `\^12\catcode `\_12\catcode `\%12\relax}%
\providecommand \@@startlink[1]{}%
\providecommand \@@endlink[0]{}%
\providecommand \url  [0]{\begingroup\@sanitize@url \@url }%
\providecommand \@url [1]{\endgroup\@href {#1}{\urlprefix }}%
\providecommand \urlprefix  [0]{URL }%
\providecommand \Eprint [0]{\href }%
\providecommand \doibase [0]{https://doi.org/}%
\providecommand \selectlanguage [0]{\@gobble}%
\providecommand \bibinfo  [0]{\@secondoftwo}%
\providecommand \bibfield  [0]{\@secondoftwo}%
\providecommand \translation [1]{[#1]}%
\providecommand \BibitemOpen [0]{}%
\providecommand \bibitemStop [0]{}%
\providecommand \bibitemNoStop [0]{.\EOS\space}%
\providecommand \EOS [0]{\spacefactor3000\relax}%
\providecommand \BibitemShut  [1]{\csname bibitem#1\endcsname}%
\let\auto@bib@innerbib\@empty
%</preamble>
\bibitem [{\citenamefont {Mohseni}\ \emph {et~al.}(2026)\citenamefont {Mohseni}, \citenamefont {Scherer}, \citenamefont {Johnson}, \citenamefont {Wertheim}, \citenamefont {Otten}, \citenamefont {Anand}, \citenamefont {Aadit}, \citenamefont {Alexeev}, \citenamefont {Ben-Shach}, \citenamefont {Bresniker}, \citenamefont {Camsari}, \citenamefont {Chapman}, \citenamefont {Chatterjee}, \citenamefont {Chowdhury}, \citenamefont {Dagnew}, \citenamefont {Dvir}, \citenamefont {Esposito}, \citenamefont {Fahim}, \citenamefont {Ferguson}, \citenamefont {Fiorentino}, \citenamefont {Gajjar}, \citenamefont {Gratsea}, \citenamefont {Gyawali}, \citenamefont {Heiter}, \citenamefont {Kavaki}, \citenamefont {Khalid}, \citenamefont {Kong}, \citenamefont {Kulchytskyy}, \citenamefont {Kyoseva}, \citenamefont {Li}, \citenamefont {Lott}, \citenamefont {Markov}, \citenamefont {McDermott}, \citenamefont {Morais}, \citenamefont {Pedretti}, \citenamefont {Rao}, \citenamefont {Rieffel}, \citenamefont {Silva}, \citenamefont {Sorebo},
  \citenamefont {Spentzouris}, \citenamefont {Steiner}, \citenamefont {Torosov}, \citenamefont {Venturelli}, \citenamefont {Visser}, \citenamefont {Webb}, \citenamefont {Zhan}, \citenamefont {Cohen}, \citenamefont {Ronagh}, \citenamefont {Ho}, \citenamefont {Beausoleil},\ and\ \citenamefont {Martinis}}]{mohseni2026buildquantumsupercomputerscaling}%
  \BibitemOpen
  \bibfield  {author} {\bibinfo {author} {\bibfnamefont {M.}~\bibnamefont {Mohseni}}, \bibinfo {author} {\bibfnamefont {A.}~\bibnamefont {Scherer}}, \bibinfo {author} {\bibfnamefont {K.~G.}\ \bibnamefont {Johnson}}, \bibinfo {author} {\bibfnamefont {O.}~\bibnamefont {Wertheim}}, \bibinfo {author} {\bibfnamefont {M.}~\bibnamefont {Otten}}, \bibinfo {author} {\bibfnamefont {N.}~\bibnamefont {Anand}}, \bibinfo {author} {\bibfnamefont {N.~A.}\ \bibnamefont {Aadit}}, \bibinfo {author} {\bibfnamefont {Y.}~\bibnamefont {Alexeev}}, \bibinfo {author} {\bibfnamefont {G.}~\bibnamefont {Ben-Shach}}, \bibinfo {author} {\bibfnamefont {K.~M.}\ \bibnamefont {Bresniker}}, \bibinfo {author} {\bibfnamefont {K.~Y.}\ \bibnamefont {Camsari}}, \bibinfo {author} {\bibfnamefont {B.}~\bibnamefont {Chapman}}, \bibinfo {author} {\bibfnamefont {S.}~\bibnamefont {Chatterjee}}, \bibinfo {author} {\bibfnamefont {S.}~\bibnamefont {Chowdhury}}, \bibinfo {author} {\bibfnamefont {G.~A.}\ \bibnamefont {Dagnew}}, \bibinfo {author} {\bibfnamefont
  {T.}~\bibnamefont {Dvir}}, \bibinfo {author} {\bibfnamefont {A.}~\bibnamefont {Esposito}}, \bibinfo {author} {\bibfnamefont {F.}~\bibnamefont {Fahim}}, \bibinfo {author} {\bibfnamefont {M.}~\bibnamefont {Ferguson}}, \bibinfo {author} {\bibfnamefont {M.}~\bibnamefont {Fiorentino}}, \bibinfo {author} {\bibfnamefont {A.}~\bibnamefont {Gajjar}}, \bibinfo {author} {\bibfnamefont {K.}~\bibnamefont {Gratsea}}, \bibinfo {author} {\bibfnamefont {G.}~\bibnamefont {Gyawali}}, \bibinfo {author} {\bibfnamefont {C.}~\bibnamefont {Heiter}}, \bibinfo {author} {\bibfnamefont {A.~H.~Z.}\ \bibnamefont {Kavaki}}, \bibinfo {author} {\bibfnamefont {A.}~\bibnamefont {Khalid}}, \bibinfo {author} {\bibfnamefont {X.}~\bibnamefont {Kong}}, \bibinfo {author} {\bibfnamefont {B.}~\bibnamefont {Kulchytskyy}}, \bibinfo {author} {\bibfnamefont {E.}~\bibnamefont {Kyoseva}}, \bibinfo {author} {\bibfnamefont {R.}~\bibnamefont {Li}}, \bibinfo {author} {\bibfnamefont {P.~A.}\ \bibnamefont {Lott}}, \bibinfo {author} {\bibfnamefont {I.~L.}\
  \bibnamefont {Markov}}, \bibinfo {author} {\bibfnamefont {R.~F.}\ \bibnamefont {McDermott}}, \bibinfo {author} {\bibfnamefont {L.}~\bibnamefont {Morais}}, \bibinfo {author} {\bibfnamefont {G.}~\bibnamefont {Pedretti}}, \bibinfo {author} {\bibfnamefont {P.}~\bibnamefont {Rao}}, \bibinfo {author} {\bibfnamefont {E.}~\bibnamefont {Rieffel}}, \bibinfo {author} {\bibfnamefont {A.}~\bibnamefont {Silva}}, \bibinfo {author} {\bibfnamefont {J.}~\bibnamefont {Sorebo}}, \bibinfo {author} {\bibfnamefont {P.}~\bibnamefont {Spentzouris}}, \bibinfo {author} {\bibfnamefont {Z.}~\bibnamefont {Steiner}}, \bibinfo {author} {\bibfnamefont {B.}~\bibnamefont {Torosov}}, \bibinfo {author} {\bibfnamefont {D.}~\bibnamefont {Venturelli}}, \bibinfo {author} {\bibfnamefont {R.~J.}\ \bibnamefont {Visser}}, \bibinfo {author} {\bibfnamefont {Z.}~\bibnamefont {Webb}}, \bibinfo {author} {\bibfnamefont {X.}~\bibnamefont {Zhan}}, \bibinfo {author} {\bibfnamefont {Y.}~\bibnamefont {Cohen}}, \bibinfo {author} {\bibfnamefont {P.}~\bibnamefont
  {Ronagh}}, \bibinfo {author} {\bibfnamefont {A.}~\bibnamefont {Ho}}, \bibinfo {author} {\bibfnamefont {R.~G.}\ \bibnamefont {Beausoleil}},\ and\ \bibinfo {author} {\bibfnamefont {J.~M.}\ \bibnamefont {Martinis}},\ }\href {https://arxiv.org/abs/2411.10406} {\bibinfo {title} {How to build a quantum supercomputer: Scaling from hundreds to millions of qubits}} (\bibinfo {year} {2026}),\ \Eprint {https://arxiv.org/abs/2411.10406} {arXiv:2411.10406 [quant-ph]} \BibitemShut {NoStop}%
\bibitem [{\citenamefont {Chung}\ \emph {et~al.}(2026)\citenamefont {Chung}, \citenamefont {Kavaki}, \citenamefont {Scherer}, \citenamefont {Khalid}, \citenamefont {Kong}, \citenamefont {Kawakubo}, \citenamefont {Anand}, \citenamefont {Dagnew}, \citenamefont {Webb}, \citenamefont {Silva}, \citenamefont {Gyawali}, \citenamefont {Yan}, \citenamefont {Fujii}, \citenamefont {Ho}, \citenamefont {Mohseni}, \citenamefont {Ronagh},\ and\ \citenamefont {Martinis}}]{chung2026partiallyfaulttolerantquantumcomputation}%
  \BibitemOpen
  \bibfield  {author} {\bibinfo {author} {\bibfnamefont {M.-Z.}\ \bibnamefont {Chung}}, \bibinfo {author} {\bibfnamefont {A.~H.~Z.}\ \bibnamefont {Kavaki}}, \bibinfo {author} {\bibfnamefont {A.}~\bibnamefont {Scherer}}, \bibinfo {author} {\bibfnamefont {A.}~\bibnamefont {Khalid}}, \bibinfo {author} {\bibfnamefont {X.}~\bibnamefont {Kong}}, \bibinfo {author} {\bibfnamefont {T.}~\bibnamefont {Kawakubo}}, \bibinfo {author} {\bibfnamefont {N.}~\bibnamefont {Anand}}, \bibinfo {author} {\bibfnamefont {G.~A.}\ \bibnamefont {Dagnew}}, \bibinfo {author} {\bibfnamefont {Z.}~\bibnamefont {Webb}}, \bibinfo {author} {\bibfnamefont {A.}~\bibnamefont {Silva}}, \bibinfo {author} {\bibfnamefont {G.}~\bibnamefont {Gyawali}}, \bibinfo {author} {\bibfnamefont {T.}~\bibnamefont {Yan}}, \bibinfo {author} {\bibfnamefont {K.}~\bibnamefont {Fujii}}, \bibinfo {author} {\bibfnamefont {A.}~\bibnamefont {Ho}}, \bibinfo {author} {\bibfnamefont {M.}~\bibnamefont {Mohseni}}, \bibinfo {author} {\bibfnamefont {P.}~\bibnamefont {Ronagh}},\
  and\ \bibinfo {author} {\bibfnamefont {J.}~\bibnamefont {Martinis}},\ }\href {https://arxiv.org/abs/2603.13093} {\bibinfo {title} {Partially fault-tolerant quantum computation for megaquop applications}} (\bibinfo {year} {2026}),\ \Eprint {https://arxiv.org/abs/2603.13093} {arXiv:2603.13093 [quant-ph]} \BibitemShut {NoStop}%
\bibitem [{\citenamefont {Babbush}\ \emph {et~al.}(2026)\citenamefont {Babbush}, \citenamefont {Zalcman}, \citenamefont {Gidney}, \citenamefont {Broughton}, \citenamefont {Khattar}, \citenamefont {Neven}, \citenamefont {Bergamaschi}, \citenamefont {Drake},\ and\ \citenamefont {Boneh}}]{babbush2026securingellipticcurvecryptocurrencies}%
  \BibitemOpen
  \bibfield  {author} {\bibinfo {author} {\bibfnamefont {R.}~\bibnamefont {Babbush}}, \bibinfo {author} {\bibfnamefont {A.}~\bibnamefont {Zalcman}}, \bibinfo {author} {\bibfnamefont {C.}~\bibnamefont {Gidney}}, \bibinfo {author} {\bibfnamefont {M.}~\bibnamefont {Broughton}}, \bibinfo {author} {\bibfnamefont {T.}~\bibnamefont {Khattar}}, \bibinfo {author} {\bibfnamefont {H.}~\bibnamefont {Neven}}, \bibinfo {author} {\bibfnamefont {T.}~\bibnamefont {Bergamaschi}}, \bibinfo {author} {\bibfnamefont {J.}~\bibnamefont {Drake}},\ and\ \bibinfo {author} {\bibfnamefont {D.}~\bibnamefont {Boneh}},\ }\href {https://arxiv.org/abs/2603.28846} {\bibinfo {title} {Securing elliptic curve cryptocurrencies against quantum vulnerabilities: Resource estimates and mitigations}} (\bibinfo {year} {2026}),\ \Eprint {https://arxiv.org/abs/2603.28846} {arXiv:2603.28846 [quant-ph]} \BibitemShut {NoStop}%
\bibitem [{\citenamefont {Webster}\ \emph {et~al.}(2026)\citenamefont {Webster}, \citenamefont {Berent}, \citenamefont {Chandra}, \citenamefont {Hockings}, \citenamefont {Baspin}, \citenamefont {Thomsen}, \citenamefont {Smith},\ and\ \citenamefont {Cohen}}]{webster2026pinnaclearchitecturereducingcost}%
  \BibitemOpen
  \bibfield  {author} {\bibinfo {author} {\bibfnamefont {P.}~\bibnamefont {Webster}}, \bibinfo {author} {\bibfnamefont {L.}~\bibnamefont {Berent}}, \bibinfo {author} {\bibfnamefont {O.}~\bibnamefont {Chandra}}, \bibinfo {author} {\bibfnamefont {E.~T.}\ \bibnamefont {Hockings}}, \bibinfo {author} {\bibfnamefont {N.}~\bibnamefont {Baspin}}, \bibinfo {author} {\bibfnamefont {F.}~\bibnamefont {Thomsen}}, \bibinfo {author} {\bibfnamefont {S.~C.}\ \bibnamefont {Smith}},\ and\ \bibinfo {author} {\bibfnamefont {L.~Z.}\ \bibnamefont {Cohen}},\ }\href {https://arxiv.org/abs/2602.11457} {\bibinfo {title} {The pinnacle architecture: Reducing the cost of breaking rsa-2048 to 100 000 physical qubits using quantum ldpc codes}} (\bibinfo {year} {2026}),\ \Eprint {https://arxiv.org/abs/2602.11457} {arXiv:2602.11457 [quant-ph]} \BibitemShut {NoStop}%
\bibitem [{\citenamefont {Cain}\ \emph {et~al.}(2026)\citenamefont {Cain}, \citenamefont {Xu}, \citenamefont {King}, \citenamefont {Picard}, \citenamefont {Levine}, \citenamefont {Endres}, \citenamefont {Preskill}, \citenamefont {Huang},\ and\ \citenamefont {Bluvstein}}]{cain2026shorsalgorithmpossible10000}%
  \BibitemOpen
  \bibfield  {author} {\bibinfo {author} {\bibfnamefont {M.}~\bibnamefont {Cain}}, \bibinfo {author} {\bibfnamefont {Q.}~\bibnamefont {Xu}}, \bibinfo {author} {\bibfnamefont {R.}~\bibnamefont {King}}, \bibinfo {author} {\bibfnamefont {L.~R.~B.}\ \bibnamefont {Picard}}, \bibinfo {author} {\bibfnamefont {H.}~\bibnamefont {Levine}}, \bibinfo {author} {\bibfnamefont {M.}~\bibnamefont {Endres}}, \bibinfo {author} {\bibfnamefont {J.}~\bibnamefont {Preskill}}, \bibinfo {author} {\bibfnamefont {H.-Y.}\ \bibnamefont {Huang}},\ and\ \bibinfo {author} {\bibfnamefont {D.}~\bibnamefont {Bluvstein}},\ }\href {https://arxiv.org/abs/2603.28627} {\bibinfo {title} {Shor's algorithm is possible with as few as 10,000 reconfigurable atomic qubits}} (\bibinfo {year} {2026}),\ \Eprint {https://arxiv.org/abs/2603.28627} {arXiv:2603.28627 [quant-ph]} \BibitemShut {NoStop}%
\bibitem [{\citenamefont {Barral}\ \emph {et~al.}(2025)\citenamefont {Barral}, \citenamefont {Cardama}, \citenamefont {Díaz-Camacho}, \citenamefont {Faílde}, \citenamefont {Llovo}, \citenamefont {Mussa-Juane}, \citenamefont {Vázquez-Pérez}, \citenamefont {Villasuso}, \citenamefont {Piñeiro}, \citenamefont {Costas}, \citenamefont {Pichel}, \citenamefont {Pena},\ and\ \citenamefont {Gómez}}]{Barral2025DistributedQC}%
  \BibitemOpen
  \bibfield  {author} {\bibinfo {author} {\bibfnamefont {D.}~\bibnamefont {Barral}}, \bibinfo {author} {\bibfnamefont {F.~J.}\ \bibnamefont {Cardama}}, \bibinfo {author} {\bibfnamefont {G.}~\bibnamefont {Díaz-Camacho}}, \bibinfo {author} {\bibfnamefont {D.}~\bibnamefont {Faílde}}, \bibinfo {author} {\bibfnamefont {I.~F.}\ \bibnamefont {Llovo}}, \bibinfo {author} {\bibfnamefont {M.}~\bibnamefont {Mussa-Juane}}, \bibinfo {author} {\bibfnamefont {J.}~\bibnamefont {Vázquez-Pérez}}, \bibinfo {author} {\bibfnamefont {J.}~\bibnamefont {Villasuso}}, \bibinfo {author} {\bibfnamefont {C.}~\bibnamefont {Piñeiro}}, \bibinfo {author} {\bibfnamefont {N.}~\bibnamefont {Costas}}, \bibinfo {author} {\bibfnamefont {J.~C.}\ \bibnamefont {Pichel}}, \bibinfo {author} {\bibfnamefont {T.~F.}\ \bibnamefont {Pena}},\ and\ \bibinfo {author} {\bibfnamefont {A.}~\bibnamefont {Gómez}},\ }\href {https://doi.org/10.1016/j.cosrev.2025.100747} {\bibfield  {journal} {\bibinfo  {journal} {Computer Science Review}\ }\textbf {\bibinfo
  {volume} {57}},\ \bibinfo {pages} {100747} (\bibinfo {year} {2025})}\BibitemShut {NoStop}%
\bibitem [{\citenamefont {Cao}\ \emph {et~al.}(2019)\citenamefont {Cao}, \citenamefont {Romero}, \citenamefont {Olson}, \citenamefont {Degroote}, \citenamefont {Johnson}, \citenamefont {Kieferová}, \citenamefont {Kivlichan}, \citenamefont {Menke}, \citenamefont {Peropadre}, \citenamefont {Sawaya}, \citenamefont {Sim}, \citenamefont {Veis},\ and\ \citenamefont {Aspuru-Guzik}}]{Cao2019QChemQC}%
  \BibitemOpen
  \bibfield  {author} {\bibinfo {author} {\bibfnamefont {Y.}~\bibnamefont {Cao}}, \bibinfo {author} {\bibfnamefont {J.}~\bibnamefont {Romero}}, \bibinfo {author} {\bibfnamefont {J.~P.}\ \bibnamefont {Olson}}, \bibinfo {author} {\bibfnamefont {M.}~\bibnamefont {Degroote}}, \bibinfo {author} {\bibfnamefont {P.~D.}\ \bibnamefont {Johnson}}, \bibinfo {author} {\bibfnamefont {M.}~\bibnamefont {Kieferová}}, \bibinfo {author} {\bibfnamefont {I.~D.}\ \bibnamefont {Kivlichan}}, \bibinfo {author} {\bibfnamefont {T.}~\bibnamefont {Menke}}, \bibinfo {author} {\bibfnamefont {B.}~\bibnamefont {Peropadre}}, \bibinfo {author} {\bibfnamefont {N.~P.~D.}\ \bibnamefont {Sawaya}}, \bibinfo {author} {\bibfnamefont {S.}~\bibnamefont {Sim}}, \bibinfo {author} {\bibfnamefont {L.}~\bibnamefont {Veis}},\ and\ \bibinfo {author} {\bibfnamefont {A.}~\bibnamefont {Aspuru-Guzik}},\ }\href {https://doi.org/10.1021/acs.chemrev.8b00803} {\bibfield  {journal} {\bibinfo  {journal} {Chemical Reviews}\ }\textbf {\bibinfo {volume} {119}},\
  \bibinfo {pages} {10856–10915} (\bibinfo {year} {2019})}\BibitemShut {NoStop}%
\bibitem [{\citenamefont {Motta}\ \emph {et~al.}(2021)\citenamefont {Motta}, \citenamefont {Ye}, \citenamefont {McClean}, \citenamefont {Li}, \citenamefont {Minnich}, \citenamefont {Babbush},\ and\ \citenamefont {Chan}}]{motta2021lowrank}%
  \BibitemOpen
  \bibfield  {author} {\bibinfo {author} {\bibfnamefont {M.}~\bibnamefont {Motta}}, \bibinfo {author} {\bibfnamefont {E.}~\bibnamefont {Ye}}, \bibinfo {author} {\bibfnamefont {J.~R.}\ \bibnamefont {McClean}}, \bibinfo {author} {\bibfnamefont {Z.}~\bibnamefont {Li}}, \bibinfo {author} {\bibfnamefont {A.~J.}\ \bibnamefont {Minnich}}, \bibinfo {author} {\bibfnamefont {R.}~\bibnamefont {Babbush}},\ and\ \bibinfo {author} {\bibfnamefont {G.~K.-L.}\ \bibnamefont {Chan}},\ }\href {https://doi.org/10.1038/s41534-021-00416-z} {\bibfield  {journal} {\bibinfo  {journal} {npj Quantum Information}\ }\textbf {\bibinfo {volume} {7}},\ \bibinfo {pages} {83} (\bibinfo {year} {2021})},\ \bibinfo {note} {arXiv:1808.02625}\BibitemShut {NoStop}%
\bibitem [{\citenamefont {Lee}\ \emph {et~al.}(2021)\citenamefont {Lee}, \citenamefont {Berry}, \citenamefont {Gidney}, \citenamefont {Huggins}, \citenamefont {McClean}, \citenamefont {Wiebe},\ and\ \citenamefont {Babbush}}]{lee2021thc}%
  \BibitemOpen
  \bibfield  {author} {\bibinfo {author} {\bibfnamefont {J.}~\bibnamefont {Lee}}, \bibinfo {author} {\bibfnamefont {D.~W.}\ \bibnamefont {Berry}}, \bibinfo {author} {\bibfnamefont {C.}~\bibnamefont {Gidney}}, \bibinfo {author} {\bibfnamefont {W.~J.}\ \bibnamefont {Huggins}}, \bibinfo {author} {\bibfnamefont {J.~R.}\ \bibnamefont {McClean}}, \bibinfo {author} {\bibfnamefont {N.}~\bibnamefont {Wiebe}},\ and\ \bibinfo {author} {\bibfnamefont {R.}~\bibnamefont {Babbush}},\ }\href {https://doi.org/10.1103/PRXQuantum.2.030305} {\bibfield  {journal} {\bibinfo  {journal} {PRX Quantum}\ }\textbf {\bibinfo {volume} {2}},\ \bibinfo {pages} {030305} (\bibinfo {year} {2021})},\ \Eprint {https://arxiv.org/abs/2011.03494} {arXiv:2011.03494 [quant-ph]} \BibitemShut {NoStop}%
\bibitem [{\citenamefont {Feng}\ \emph {et~al.}(2024)\citenamefont {Feng}, \citenamefont {Xu}, \citenamefont {Yu}, \citenamefont {Ye}, \citenamefont {Yao},\ and\ \citenamefont {Zhao}}]{feng2024distributedquantumsimulation}%
  \BibitemOpen
  \bibfield  {author} {\bibinfo {author} {\bibfnamefont {T.}~\bibnamefont {Feng}}, \bibinfo {author} {\bibfnamefont {J.}~\bibnamefont {Xu}}, \bibinfo {author} {\bibfnamefont {W.}~\bibnamefont {Yu}}, \bibinfo {author} {\bibfnamefont {Z.}~\bibnamefont {Ye}}, \bibinfo {author} {\bibfnamefont {P.}~\bibnamefont {Yao}},\ and\ \bibinfo {author} {\bibfnamefont {Q.}~\bibnamefont {Zhao}},\ }\href {https://arxiv.org/abs/2411.02881} {\bibinfo {title} {Distributed quantum simulation}} (\bibinfo {year} {2024}),\ \Eprint {https://arxiv.org/abs/2411.02881} {arXiv:2411.02881 [quant-ph]} \BibitemShut {NoStop}%
\bibitem [{\citenamefont {Peng}\ \emph {et~al.}(2020)\citenamefont {Peng}, \citenamefont {Harrow}, \citenamefont {Ozols},\ and\ \citenamefont {Wu}}]{peng2020simulating}%
  \BibitemOpen
  \bibfield  {author} {\bibinfo {author} {\bibfnamefont {T.}~\bibnamefont {Peng}}, \bibinfo {author} {\bibfnamefont {A.~W.}\ \bibnamefont {Harrow}}, \bibinfo {author} {\bibfnamefont {M.}~\bibnamefont {Ozols}},\ and\ \bibinfo {author} {\bibfnamefont {X.}~\bibnamefont {Wu}},\ }\href {https://doi.org/10.1103/PhysRevLett.125.150504} {\bibfield  {journal} {\bibinfo  {journal} {Physical Review Letters}\ }\textbf {\bibinfo {volume} {125}},\ \bibinfo {pages} {150504} (\bibinfo {year} {2020})},\ \bibinfo {note} {arXiv:1904.00102}\BibitemShut {NoStop}%
\bibitem [{\citenamefont {Tang}\ \emph {et~al.}(2021)\citenamefont {Tang}, \citenamefont {Tomesh}, \citenamefont {Suchara}, \citenamefont {Larson},\ and\ \citenamefont {Martonosi}}]{Tang2021}%
  \BibitemOpen
  \bibfield  {author} {\bibinfo {author} {\bibfnamefont {W.}~\bibnamefont {Tang}}, \bibinfo {author} {\bibfnamefont {T.}~\bibnamefont {Tomesh}}, \bibinfo {author} {\bibfnamefont {M.}~\bibnamefont {Suchara}}, \bibinfo {author} {\bibfnamefont {J.}~\bibnamefont {Larson}},\ and\ \bibinfo {author} {\bibfnamefont {M.}~\bibnamefont {Martonosi}},\ }in\ \href {https://doi.org/10.1145/3445814.3446758} {\emph {\bibinfo {booktitle} {Proceedings of the 26th ACM International Conference on Architectural Support for Programming Languages and Operating Systems}}},\ \bibinfo {series and number} {ASPLOS ’21}\ (\bibinfo  {publisher} {ACM},\ \bibinfo {year} {2021})\ p.\ \bibinfo {pages} {473–486}\BibitemShut {NoStop}%
\bibitem [{\citenamefont {Bravyi}\ \emph {et~al.}(2022)\citenamefont {Bravyi}, \citenamefont {Dial}, \citenamefont {Gambetta}, \citenamefont {Gil},\ and\ \citenamefont {Nazario}}]{Bravyi2022}%
  \BibitemOpen
  \bibfield  {author} {\bibinfo {author} {\bibfnamefont {S.}~\bibnamefont {Bravyi}}, \bibinfo {author} {\bibfnamefont {O.}~\bibnamefont {Dial}}, \bibinfo {author} {\bibfnamefont {J.~M.}\ \bibnamefont {Gambetta}}, \bibinfo {author} {\bibfnamefont {D.}~\bibnamefont {Gil}},\ and\ \bibinfo {author} {\bibfnamefont {Z.}~\bibnamefont {Nazario}},\ }\bibfield  {journal} {\bibinfo  {journal} {Journal of Applied Physics}\ }\textbf {\bibinfo {volume} {132}},\ \href {https://doi.org/10.1063/5.0082975} {10.1063/5.0082975} (\bibinfo {year} {2022})\BibitemShut {NoStop}%
\bibitem [{\citenamefont {Lowe}\ \emph {et~al.}(2023)\citenamefont {Lowe}, \citenamefont {Medvidovi{\'{c}}}, \citenamefont {Hayes}, \citenamefont {O'Riordan}, \citenamefont {Bromley}, \citenamefont {Arrazola},\ and\ \citenamefont {Killoran}}]{Lowe2023fastquantumcircuit}%
  \BibitemOpen
  \bibfield  {author} {\bibinfo {author} {\bibfnamefont {A.}~\bibnamefont {Lowe}}, \bibinfo {author} {\bibfnamefont {M.}~\bibnamefont {Medvidovi{\'{c}}}}, \bibinfo {author} {\bibfnamefont {A.}~\bibnamefont {Hayes}}, \bibinfo {author} {\bibfnamefont {L.~J.}\ \bibnamefont {O'Riordan}}, \bibinfo {author} {\bibfnamefont {T.~R.}\ \bibnamefont {Bromley}}, \bibinfo {author} {\bibfnamefont {J.~M.}\ \bibnamefont {Arrazola}},\ and\ \bibinfo {author} {\bibfnamefont {N.}~\bibnamefont {Killoran}},\ }\href {https://doi.org/10.22331/q-2023-03-02-934} {\bibfield  {journal} {\bibinfo  {journal} {{Quantum}}\ }\textbf {\bibinfo {volume} {7}},\ \bibinfo {pages} {934} (\bibinfo {year} {2023})}\BibitemShut {NoStop}%
\bibitem [{\citenamefont {Carrera~Vazquez}\ \emph {et~al.}(2024)\citenamefont {Carrera~Vazquez}, \citenamefont {Tornow}, \citenamefont {Ristè}, \citenamefont {Woerner}, \citenamefont {Takita},\ and\ \citenamefont {Egger}}]{CarreraVazquez2024}%
  \BibitemOpen
  \bibfield  {author} {\bibinfo {author} {\bibfnamefont {A.}~\bibnamefont {Carrera~Vazquez}}, \bibinfo {author} {\bibfnamefont {C.}~\bibnamefont {Tornow}}, \bibinfo {author} {\bibfnamefont {D.}~\bibnamefont {Ristè}}, \bibinfo {author} {\bibfnamefont {S.}~\bibnamefont {Woerner}}, \bibinfo {author} {\bibfnamefont {M.}~\bibnamefont {Takita}},\ and\ \bibinfo {author} {\bibfnamefont {D.~J.}\ \bibnamefont {Egger}},\ }\href {https://doi.org/10.1038/s41586-024-08178-2} {\bibfield  {journal} {\bibinfo  {journal} {Nature}\ }\textbf {\bibinfo {volume} {636}},\ \bibinfo {pages} {75–79} (\bibinfo {year} {2024})}\BibitemShut {NoStop}%
\bibitem [{\citenamefont {Piveteau}\ and\ \citenamefont {Sutter}(2024{\natexlab{a}})}]{Piveteau2024}%
  \BibitemOpen
  \bibfield  {author} {\bibinfo {author} {\bibfnamefont {C.}~\bibnamefont {Piveteau}}\ and\ \bibinfo {author} {\bibfnamefont {D.}~\bibnamefont {Sutter}},\ }\href {https://doi.org/10.1109/tit.2023.3310797} {\bibfield  {journal} {\bibinfo  {journal} {IEEE Transactions on Information Theory}\ }\textbf {\bibinfo {volume} {70}},\ \bibinfo {pages} {2734–2745} (\bibinfo {year} {2024}{\natexlab{a}})}\BibitemShut {NoStop}%
\bibitem [{\citenamefont {Piveteau}\ and\ \citenamefont {Sutter}(2024{\natexlab{b}})}]{piveteau2024knitting}%
  \BibitemOpen
  \bibfield  {author} {\bibinfo {author} {\bibfnamefont {C.}~\bibnamefont {Piveteau}}\ and\ \bibinfo {author} {\bibfnamefont {D.}~\bibnamefont {Sutter}},\ }\href {https://doi.org/10.22331/q-2024-10-24-1500} {\bibfield  {journal} {\bibinfo  {journal} {Quantum}\ }\textbf {\bibinfo {volume} {8}},\ \bibinfo {pages} {1500} (\bibinfo {year} {2024}{\natexlab{b}})},\ \bibinfo {note} {arXiv:2205.00016}\BibitemShut {NoStop}%
\bibitem [{\citenamefont {Harrow}\ and\ \citenamefont {Lowe}(2025)}]{aram_optimal_2025}%
  \BibitemOpen
  \bibfield  {author} {\bibinfo {author} {\bibfnamefont {A.~W.}\ \bibnamefont {Harrow}}\ and\ \bibinfo {author} {\bibfnamefont {A.}~\bibnamefont {Lowe}},\ }\href {https://doi.org/10.1103/PRXQuantum.6.010316} {\bibfield  {journal} {\bibinfo  {journal} {PRX Quantum}\ }\textbf {\bibinfo {volume} {6}},\ \bibinfo {pages} {010316} (\bibinfo {year} {2025})}\BibitemShut {NoStop}%
\bibitem [{\citenamefont {Jing}\ \emph {et~al.}(2025)\citenamefont {Jing}, \citenamefont {Zhu},\ and\ \citenamefont {Wang}}]{jing2024entanglement}%
  \BibitemOpen
  \bibfield  {author} {\bibinfo {author} {\bibfnamefont {M.}~\bibnamefont {Jing}}, \bibinfo {author} {\bibfnamefont {C.}~\bibnamefont {Zhu}},\ and\ \bibinfo {author} {\bibfnamefont {X.}~\bibnamefont {Wang}},\ }\href {https://doi.org/10.1103/PhysRevA.111.012433} {\bibfield  {journal} {\bibinfo  {journal} {Physical Review A}\ }\textbf {\bibinfo {volume} {111}},\ \bibinfo {pages} {012433} (\bibinfo {year} {2025})},\ \Eprint {https://arxiv.org/abs/2404.03619} {arXiv:2404.03619 [quant-ph]} \BibitemShut {NoStop}%
\bibitem [{\citenamefont {Zanardi}(2001)}]{zanardi_entanglement_2001}%
  \BibitemOpen
  \bibfield  {author} {\bibinfo {author} {\bibfnamefont {P.}~\bibnamefont {Zanardi}},\ }\href {https://doi.org/10.1103/PhysRevA.63.040304} {\bibfield  {journal} {\bibinfo  {journal} {Phys. Rev. A}\ }\textbf {\bibinfo {volume} {63}},\ \bibinfo {pages} {040304} (\bibinfo {year} {2001})}\BibitemShut {NoStop}%
\bibitem [{\citenamefont {Styliaris}\ \emph {et~al.}(2021)\citenamefont {Styliaris}, \citenamefont {Anand},\ and\ \citenamefont {Zanardi}}]{styliaris2020informationscrambling}%
  \BibitemOpen
  \bibfield  {author} {\bibinfo {author} {\bibfnamefont {G.}~\bibnamefont {Styliaris}}, \bibinfo {author} {\bibfnamefont {N.}~\bibnamefont {Anand}},\ and\ \bibinfo {author} {\bibfnamefont {P.}~\bibnamefont {Zanardi}},\ }\href {https://doi.org/10.1103/PhysRevLett.126.030601} {\bibfield  {journal} {\bibinfo  {journal} {Physical Review Letters}\ }\textbf {\bibinfo {volume} {126}},\ \bibinfo {pages} {030601} (\bibinfo {year} {2021})},\ \Eprint {https://arxiv.org/abs/2007.08570} {arXiv:2007.08570 [quant-ph]} \BibitemShut {NoStop}%
\bibitem [{\citenamefont {Jones}\ and\ \citenamefont {Jacobsen}(2025)}]{jones2025distributed}%
  \BibitemOpen
  \bibfield  {author} {\bibinfo {author} {\bibfnamefont {G.~M.}\ \bibnamefont {Jones}}\ and\ \bibinfo {author} {\bibfnamefont {H.-A.}\ \bibnamefont {Jacobsen}},\ }\href {https://arxiv.org/abs/2507.01902} {\bibinfo {title} {Analyzing common electronic structure theory algorithms for distributed quantum computing}} (\bibinfo {year} {2025}),\ \Eprint {https://arxiv.org/abs/2507.01902} {arXiv:2507.01902 [quant-ph]} \BibitemShut {NoStop}%
\bibitem [{\citenamefont {Motta}\ \emph {et~al.}(2023)\citenamefont {Motta}, \citenamefont {Sung}, \citenamefont {Whaley}, \citenamefont {Head-Gordon},\ and\ \citenamefont {Shee}}]{Motta2023LUCJ}%
  \BibitemOpen
  \bibfield  {author} {\bibinfo {author} {\bibfnamefont {M.}~\bibnamefont {Motta}}, \bibinfo {author} {\bibfnamefont {K.~J.}\ \bibnamefont {Sung}}, \bibinfo {author} {\bibfnamefont {K.~B.}\ \bibnamefont {Whaley}}, \bibinfo {author} {\bibfnamefont {M.}~\bibnamefont {Head-Gordon}},\ and\ \bibinfo {author} {\bibfnamefont {J.}~\bibnamefont {Shee}},\ }\href {https://doi.org/10.1039/d3sc02516k} {\bibfield  {journal} {\bibinfo  {journal} {Chemical Science}\ }\textbf {\bibinfo {volume} {14}},\ \bibinfo {pages} {11213–11227} (\bibinfo {year} {2023})}\BibitemShut {NoStop}%
\bibitem [{\citenamefont {Shirakawa}\ \emph {et~al.}(2025)\citenamefont {Shirakawa}, \citenamefont {Robledo-Moreno}, \citenamefont {Itoko}, \citenamefont {Tripathi}, \citenamefont {Ueda}, \citenamefont {Kawashima}, \citenamefont {Broers}, \citenamefont {Kirby}, \citenamefont {Pathak}, \citenamefont {Paik}, \citenamefont {Tsuji}, \citenamefont {Kodama}, \citenamefont {Sato}, \citenamefont {Evangelinos}, \citenamefont {Seelam}, \citenamefont {Walkup}, \citenamefont {Yunoki}, \citenamefont {Motta}, \citenamefont {Jurcevic}, \citenamefont {Horii},\ and\ \citenamefont {Mezzacapo}}]{Shirakawa2025IBM}%
  \BibitemOpen
  \bibfield  {author} {\bibinfo {author} {\bibfnamefont {T.}~\bibnamefont {Shirakawa}}, \bibinfo {author} {\bibfnamefont {J.}~\bibnamefont {Robledo-Moreno}}, \bibinfo {author} {\bibfnamefont {T.}~\bibnamefont {Itoko}}, \bibinfo {author} {\bibfnamefont {V.}~\bibnamefont {Tripathi}}, \bibinfo {author} {\bibfnamefont {K.}~\bibnamefont {Ueda}}, \bibinfo {author} {\bibfnamefont {Y.}~\bibnamefont {Kawashima}}, \bibinfo {author} {\bibfnamefont {L.}~\bibnamefont {Broers}}, \bibinfo {author} {\bibfnamefont {W.}~\bibnamefont {Kirby}}, \bibinfo {author} {\bibfnamefont {H.}~\bibnamefont {Pathak}}, \bibinfo {author} {\bibfnamefont {H.}~\bibnamefont {Paik}}, \bibinfo {author} {\bibfnamefont {M.}~\bibnamefont {Tsuji}}, \bibinfo {author} {\bibfnamefont {Y.}~\bibnamefont {Kodama}}, \bibinfo {author} {\bibfnamefont {M.}~\bibnamefont {Sato}}, \bibinfo {author} {\bibfnamefont {C.}~\bibnamefont {Evangelinos}}, \bibinfo {author} {\bibfnamefont {S.}~\bibnamefont {Seelam}}, \bibinfo {author} {\bibfnamefont {R.}~\bibnamefont
  {Walkup}}, \bibinfo {author} {\bibfnamefont {S.}~\bibnamefont {Yunoki}}, \bibinfo {author} {\bibfnamefont {M.}~\bibnamefont {Motta}}, \bibinfo {author} {\bibfnamefont {P.}~\bibnamefont {Jurcevic}}, \bibinfo {author} {\bibfnamefont {H.}~\bibnamefont {Horii}},\ and\ \bibinfo {author} {\bibfnamefont {A.}~\bibnamefont {Mezzacapo}},\ }\href {https://doi.org/10.48550/ARXIV.2511.00224} {\bibinfo {title} {Closed-loop calculations of electronic structure on a quantum processor and a classical supercomputer at full scale}} (\bibinfo {year} {2025})\BibitemShut {NoStop}%
\bibitem [{\citenamefont {G{\"u}nther}\ \emph {et~al.}(2026{\natexlab{a}})\citenamefont {G{\"u}nther}, \citenamefont {Weymuth}, \citenamefont {Bensberg}, \citenamefont {Witteveen}, \citenamefont {Teynor}, \citenamefont {Thomasen}, \citenamefont {Sora}, \citenamefont {Bro-J{\o}rgensen}, \citenamefont {Husistein}, \citenamefont {Erakovic}, \citenamefont {Miller}, \citenamefont {Weisburn}, \citenamefont {Cho}, \citenamefont {Eckhoff}, \citenamefont {Harrow}, \citenamefont {Krogh}, \citenamefont {Van~Voorhis}, \citenamefont {Lindorff-Larsen}, \citenamefont {Solomon}, \citenamefont {Reiher},\ and\ \citenamefont {Christandl}}]{gunther2026biomolecularfreeenergies}%
  \BibitemOpen
  \bibfield  {author} {\bibinfo {author} {\bibfnamefont {J.}~\bibnamefont {G{\"u}nther}}, \bibinfo {author} {\bibfnamefont {T.}~\bibnamefont {Weymuth}}, \bibinfo {author} {\bibfnamefont {M.}~\bibnamefont {Bensberg}}, \bibinfo {author} {\bibfnamefont {F.}~\bibnamefont {Witteveen}}, \bibinfo {author} {\bibfnamefont {M.~S.}\ \bibnamefont {Teynor}}, \bibinfo {author} {\bibfnamefont {F.~E.}\ \bibnamefont {Thomasen}}, \bibinfo {author} {\bibfnamefont {V.}~\bibnamefont {Sora}}, \bibinfo {author} {\bibfnamefont {W.}~\bibnamefont {Bro-J{\o}rgensen}}, \bibinfo {author} {\bibfnamefont {R.~T.}\ \bibnamefont {Husistein}}, \bibinfo {author} {\bibfnamefont {M.}~\bibnamefont {Erakovic}}, \bibinfo {author} {\bibfnamefont {M.}~\bibnamefont {Miller}}, \bibinfo {author} {\bibfnamefont {L.}~\bibnamefont {Weisburn}}, \bibinfo {author} {\bibfnamefont {M.}~\bibnamefont {Cho}}, \bibinfo {author} {\bibfnamefont {M.}~\bibnamefont {Eckhoff}}, \bibinfo {author} {\bibfnamefont {A.~W.}\ \bibnamefont {Harrow}}, \bibinfo {author}
  {\bibfnamefont {A.}~\bibnamefont {Krogh}}, \bibinfo {author} {\bibfnamefont {T.}~\bibnamefont {Van~Voorhis}}, \bibinfo {author} {\bibfnamefont {K.}~\bibnamefont {Lindorff-Larsen}}, \bibinfo {author} {\bibfnamefont {G.}~\bibnamefont {Solomon}}, \bibinfo {author} {\bibfnamefont {M.}~\bibnamefont {Reiher}},\ and\ \bibinfo {author} {\bibfnamefont {M.}~\bibnamefont {Christandl}},\ }\href {https://doi.org/10.1021/acs.jctc.5c02088} {\bibfield  {journal} {\bibinfo  {journal} {Journal of Chemical Theory and Computation}\ }\textbf {\bibinfo {volume} {22}},\ \bibinfo {pages} {4329} (\bibinfo {year} {2026}{\natexlab{a}})}\BibitemShut {NoStop}%
\bibitem [{\citenamefont {Kirby}\ \emph {et~al.}(2026)\citenamefont {Kirby}, \citenamefont {Pokharel}, \citenamefont {Moreno}, \citenamefont {Smith}, \citenamefont {Bravyi}, \citenamefont {Deshpande}, \citenamefont {Evangelinos}, \citenamefont {Fuller}, \citenamefont {Garrison}, \citenamefont {Jaderberg}, \citenamefont {Johnson}, \citenamefont {Jurcevic}, \citenamefont {Lee}, \citenamefont {Martiel}, \citenamefont {Motta}, \citenamefont {Seelam}, \citenamefont {Shtanko}, \citenamefont {Sung}, \citenamefont {Tran}, \citenamefont {Tripathi}, \citenamefont {Seki}, \citenamefont {Shinjo}, \citenamefont {Xu}, \citenamefont {Broers}, \citenamefont {Shirakawa}, \citenamefont {Yunoki}, \citenamefont {Sharma},\ and\ \citenamefont {Mezzacapo}}]{Kirby2026IBM}%
  \BibitemOpen
  \bibfield  {author} {\bibinfo {author} {\bibfnamefont {W.}~\bibnamefont {Kirby}}, \bibinfo {author} {\bibfnamefont {B.}~\bibnamefont {Pokharel}}, \bibinfo {author} {\bibfnamefont {J.~R.}\ \bibnamefont {Moreno}}, \bibinfo {author} {\bibfnamefont {K.~C.}\ \bibnamefont {Smith}}, \bibinfo {author} {\bibfnamefont {S.}~\bibnamefont {Bravyi}}, \bibinfo {author} {\bibfnamefont {A.}~\bibnamefont {Deshpande}}, \bibinfo {author} {\bibfnamefont {C.}~\bibnamefont {Evangelinos}}, \bibinfo {author} {\bibfnamefont {B.}~\bibnamefont {Fuller}}, \bibinfo {author} {\bibfnamefont {J.~R.}\ \bibnamefont {Garrison}}, \bibinfo {author} {\bibfnamefont {B.}~\bibnamefont {Jaderberg}}, \bibinfo {author} {\bibfnamefont {C.}~\bibnamefont {Johnson}}, \bibinfo {author} {\bibfnamefont {P.}~\bibnamefont {Jurcevic}}, \bibinfo {author} {\bibfnamefont {S.-u.}\ \bibnamefont {Lee}}, \bibinfo {author} {\bibfnamefont {S.}~\bibnamefont {Martiel}}, \bibinfo {author} {\bibfnamefont {M.}~\bibnamefont {Motta}}, \bibinfo {author} {\bibfnamefont
  {S.}~\bibnamefont {Seelam}}, \bibinfo {author} {\bibfnamefont {O.}~\bibnamefont {Shtanko}}, \bibinfo {author} {\bibfnamefont {K.~J.}\ \bibnamefont {Sung}}, \bibinfo {author} {\bibfnamefont {M.}~\bibnamefont {Tran}}, \bibinfo {author} {\bibfnamefont {V.}~\bibnamefont {Tripathi}}, \bibinfo {author} {\bibfnamefont {K.}~\bibnamefont {Seki}}, \bibinfo {author} {\bibfnamefont {K.}~\bibnamefont {Shinjo}}, \bibinfo {author} {\bibfnamefont {H.}~\bibnamefont {Xu}}, \bibinfo {author} {\bibfnamefont {L.}~\bibnamefont {Broers}}, \bibinfo {author} {\bibfnamefont {T.}~\bibnamefont {Shirakawa}}, \bibinfo {author} {\bibfnamefont {S.}~\bibnamefont {Yunoki}}, \bibinfo {author} {\bibfnamefont {K.}~\bibnamefont {Sharma}},\ and\ \bibinfo {author} {\bibfnamefont {A.}~\bibnamefont {Mezzacapo}},\ }\href {https://doi.org/10.48550/ARXIV.2603.03496} {\bibinfo {title} {Observation of improved accuracy over classical sparse ground-state solvers using a quantum computer}} (\bibinfo {year} {2026})\BibitemShut {NoStop}%
\bibitem [{\citenamefont {Yue}\ and\ \citenamefont {Chitambar}(2019)}]{Yue_2019}%
  \BibitemOpen
  \bibfield  {author} {\bibinfo {author} {\bibfnamefont {Q.}~\bibnamefont {Yue}}\ and\ \bibinfo {author} {\bibfnamefont {E.}~\bibnamefont {Chitambar}},\ }\bibfield  {journal} {\bibinfo  {journal} {Journal of Mathematical Physics}\ }\textbf {\bibinfo {volume} {60}},\ \href {https://doi.org/10.1063/1.5087815} {10.1063/1.5087815} (\bibinfo {year} {2019}),\ \Eprint {https://arxiv.org/abs/1808.10516} {1808.10516} \BibitemShut {NoStop}%
\bibitem [{\citenamefont {Wu}\ \emph {et~al.}(2023)\citenamefont {Wu}, \citenamefont {Matsui}, \citenamefont {Forrer}, \citenamefont {Soeda}, \citenamefont {Andr{\'{e}}s-Mart{\'{i}}nez}, \citenamefont {Mills}, \citenamefont {Henaut},\ and\ \citenamefont {Murao}}]{Wu2023DistributedEntanglement}%
  \BibitemOpen
  \bibfield  {author} {\bibinfo {author} {\bibfnamefont {J.-Y.}\ \bibnamefont {Wu}}, \bibinfo {author} {\bibfnamefont {K.}~\bibnamefont {Matsui}}, \bibinfo {author} {\bibfnamefont {T.}~\bibnamefont {Forrer}}, \bibinfo {author} {\bibfnamefont {A.}~\bibnamefont {Soeda}}, \bibinfo {author} {\bibfnamefont {P.}~\bibnamefont {Andr{\'{e}}s-Mart{\'{i}}nez}}, \bibinfo {author} {\bibfnamefont {D.}~\bibnamefont {Mills}}, \bibinfo {author} {\bibfnamefont {L.}~\bibnamefont {Henaut}},\ and\ \bibinfo {author} {\bibfnamefont {M.}~\bibnamefont {Murao}},\ }\href {https://doi.org/10.22331/q-2023-12-05-1196} {\bibfield  {journal} {\bibinfo  {journal} {{Quantum}}\ }\textbf {\bibinfo {volume} {7}},\ \bibinfo {pages} {1196} (\bibinfo {year} {2023})}\BibitemShut {NoStop}%
\bibitem [{\citenamefont {Mitarai}\ and\ \citenamefont {Fujii}(2021)}]{mitarai2021overhead}%
  \BibitemOpen
  \bibfield  {author} {\bibinfo {author} {\bibfnamefont {K.}~\bibnamefont {Mitarai}}\ and\ \bibinfo {author} {\bibfnamefont {K.}~\bibnamefont {Fujii}},\ }\href {https://doi.org/10.22331/q-2021-01-28-388} {\bibfield  {journal} {\bibinfo  {journal} {Quantum}\ }\textbf {\bibinfo {volume} {5}},\ \bibinfo {pages} {388} (\bibinfo {year} {2021})},\ \bibinfo {note} {arXiv:2006.11174}\BibitemShut {NoStop}%
\bibitem [{\citenamefont {Johnson}\ \emph {et~al.}(2026)\citenamefont {Johnson}, \citenamefont {Esposito}, \citenamefont {Gyawali}, \citenamefont {Zhan}, \citenamefont {Ganti}, \citenamefont {Anand}, \citenamefont {Beausoleil},\ and\ \citenamefont {Mohseni}}]{johnson2026distributedquantumcomputingadaptive}%
  \BibitemOpen
  \bibfield  {author} {\bibinfo {author} {\bibfnamefont {K.~G.}\ \bibnamefont {Johnson}}, \bibinfo {author} {\bibfnamefont {A.}~\bibnamefont {Esposito}}, \bibinfo {author} {\bibfnamefont {G.}~\bibnamefont {Gyawali}}, \bibinfo {author} {\bibfnamefont {X.}~\bibnamefont {Zhan}}, \bibinfo {author} {\bibfnamefont {R.}~\bibnamefont {Ganti}}, \bibinfo {author} {\bibfnamefont {N.}~\bibnamefont {Anand}}, \bibinfo {author} {\bibfnamefont {R.~G.}\ \bibnamefont {Beausoleil}},\ and\ \bibinfo {author} {\bibfnamefont {M.}~\bibnamefont {Mohseni}},\ }\href {https://arxiv.org/abs/2603.12411} {\bibinfo {title} {Distributed quantum computing via adaptive circuit knitting}} (\bibinfo {year} {2026}),\ \Eprint {https://arxiv.org/abs/2603.12411} {arXiv:2603.12411 [quant-ph]} \BibitemShut {NoStop}%
\bibitem [{\citenamefont {Piveteau}\ \emph {et~al.}(2025)\citenamefont {Piveteau}, \citenamefont {Schmitt},\ and\ \citenamefont {Sutter}}]{piveteau2025side}%
  \BibitemOpen
  \bibfield  {author} {\bibinfo {author} {\bibfnamefont {C.}~\bibnamefont {Piveteau}}, \bibinfo {author} {\bibfnamefont {L.}~\bibnamefont {Schmitt}},\ and\ \bibinfo {author} {\bibfnamefont {D.}~\bibnamefont {Sutter}},\ }\href {https://doi.org/10.1103/38mx-36k6} {\bibfield  {journal} {\bibinfo  {journal} {Phys. Rev. Res.}\ }\textbf {\bibinfo {volume} {7}},\ \bibinfo {pages} {033063} (\bibinfo {year} {2025})}\BibitemShut {NoStop}%
\bibitem [{\citenamefont {Harrow}\ and\ \citenamefont {Nielsen}(2003)}]{Harrow_2003_robustness}%
  \BibitemOpen
  \bibfield  {author} {\bibinfo {author} {\bibfnamefont {A.~W.}\ \bibnamefont {Harrow}}\ and\ \bibinfo {author} {\bibfnamefont {M.~A.}\ \bibnamefont {Nielsen}},\ }\bibfield  {journal} {\bibinfo  {journal} {Physical Review A}\ }\textbf {\bibinfo {volume} {68}},\ \href {https://doi.org/10.1103/physreva.68.012308} {10.1103/physreva.68.012308} (\bibinfo {year} {2003})\BibitemShut {NoStop}%
\bibitem [{\citenamefont {Theurer}\ \emph {et~al.}(2023)\citenamefont {Theurer}, \citenamefont {Fang},\ and\ \citenamefont {Gour}}]{theurer2023singleshotentanglementmanipulationstates}%
  \BibitemOpen
  \bibfield  {author} {\bibinfo {author} {\bibfnamefont {T.}~\bibnamefont {Theurer}}, \bibinfo {author} {\bibfnamefont {K.}~\bibnamefont {Fang}},\ and\ \bibinfo {author} {\bibfnamefont {G.}~\bibnamefont {Gour}},\ }\href {https://arxiv.org/abs/2312.17088} {\bibinfo {title} {Single-shot entanglement manipulation of states and channels revisited}} (\bibinfo {year} {2023}),\ \Eprint {https://arxiv.org/abs/2312.17088} {arXiv:2312.17088 [quant-ph]} \BibitemShut {NoStop}%
\bibitem [{\citenamefont {Vidal}\ and\ \citenamefont {Tarrach}(1999)}]{vidal1999robustness}%
  \BibitemOpen
  \bibfield  {author} {\bibinfo {author} {\bibfnamefont {G.}~\bibnamefont {Vidal}}\ and\ \bibinfo {author} {\bibfnamefont {R.}~\bibnamefont {Tarrach}},\ }\href {https://doi.org/10.1103/PhysRevA.59.141} {\bibfield  {journal} {\bibinfo  {journal} {Physical Review A}\ }\textbf {\bibinfo {volume} {59}},\ \bibinfo {pages} {141} (\bibinfo {year} {1999})},\ \Eprint {https://arxiv.org/abs/quant-ph/9806094} {arXiv:quant-ph/9806094} \BibitemShut {NoStop}%
\bibitem [{\citenamefont {Szabo}\ and\ \citenamefont {Ostlund}(1996)}]{szabo1996modern}%
  \BibitemOpen
  \bibfield  {author} {\bibinfo {author} {\bibfnamefont {A.}~\bibnamefont {Szabo}}\ and\ \bibinfo {author} {\bibfnamefont {N.~S.}\ \bibnamefont {Ostlund}},\ }\href@noop {} {\emph {\bibinfo {title} {Modern Quantum Chemistry: Introduction to Advanced Electronic Structure Theory}}}\ (\bibinfo  {publisher} {Dover Publications},\ \bibinfo {address} {Mineola, NY},\ \bibinfo {year} {1996})\BibitemShut {NoStop}%
\bibitem [{\citenamefont {Helgaker}\ \emph {et~al.}(2000)\citenamefont {Helgaker}, \citenamefont {J{\o}rgensen},\ and\ \citenamefont {Olsen}}]{Helgaker2000}%
  \BibitemOpen
  \bibfield  {author} {\bibinfo {author} {\bibfnamefont {T.}~\bibnamefont {Helgaker}}, \bibinfo {author} {\bibfnamefont {P.}~\bibnamefont {J{\o}rgensen}},\ and\ \bibinfo {author} {\bibfnamefont {J.}~\bibnamefont {Olsen}},\ }\href@noop {} {\emph {\bibinfo {title} {Molecular Electronic-Structure Theory}}}\ (\bibinfo  {publisher} {John Wiley \& Sons},\ \bibinfo {year} {2000})\BibitemShut {NoStop}%
\bibitem [{\citenamefont {Chien}\ \emph {et~al.}(2026)\citenamefont {Chien}, \citenamefont {Chiew}, \citenamefont {Harrison}, \citenamefont {Necaise}, \citenamefont {Wang}, \citenamefont {Mudassar}, \citenamefont {McLauchlan}, \citenamefont {Henderson}, \citenamefont {Scuseria}, \citenamefont {Strelchuk},\ and\ \citenamefont {Whitfield}}]{chien2026simulatingfermions}%
  \BibitemOpen
  \bibfield  {author} {\bibinfo {author} {\bibfnamefont {R.~W.}\ \bibnamefont {Chien}}, \bibinfo {author} {\bibfnamefont {M.}~\bibnamefont {Chiew}}, \bibinfo {author} {\bibfnamefont {B.}~\bibnamefont {Harrison}}, \bibinfo {author} {\bibfnamefont {J.}~\bibnamefont {Necaise}}, \bibinfo {author} {\bibfnamefont {W.}~\bibnamefont {Wang}}, \bibinfo {author} {\bibfnamefont {M.}~\bibnamefont {Mudassar}}, \bibinfo {author} {\bibfnamefont {C.}~\bibnamefont {McLauchlan}}, \bibinfo {author} {\bibfnamefont {T.~M.}\ \bibnamefont {Henderson}}, \bibinfo {author} {\bibfnamefont {G.~E.}\ \bibnamefont {Scuseria}}, \bibinfo {author} {\bibfnamefont {S.}~\bibnamefont {Strelchuk}},\ and\ \bibinfo {author} {\bibfnamefont {J.~D.}\ \bibnamefont {Whitfield}},\ }\href {https://doi.org/10.1038/s42254-025-00914-5} {\bibfield  {journal} {\bibinfo  {journal} {Nature Reviews Physics}\ }\textbf {\bibinfo {volume} {8}},\ \bibinfo {pages} {131} (\bibinfo {year} {2026})}\BibitemShut {NoStop}%
\bibitem [{\citenamefont {Bravyi}\ and\ \citenamefont {Kitaev}(2002)}]{bravyi2002fermionicquantumcomputation}%
  \BibitemOpen
  \bibfield  {author} {\bibinfo {author} {\bibfnamefont {S.~B.}\ \bibnamefont {Bravyi}}\ and\ \bibinfo {author} {\bibfnamefont {A.~Y.}\ \bibnamefont {Kitaev}},\ }\href {https://doi.org/10.1006/aphy.2002.6254} {\bibfield  {journal} {\bibinfo  {journal} {Annals of Physics}\ }\textbf {\bibinfo {volume} {298}},\ \bibinfo {pages} {210} (\bibinfo {year} {2002})},\ \Eprint {https://arxiv.org/abs/quant-ph/0003137} {arXiv:quant-ph/0003137 [quant-ph]} \BibitemShut {NoStop}%
\bibitem [{\citenamefont {Seeley}\ \emph {et~al.}(2012)\citenamefont {Seeley}, \citenamefont {Richard},\ and\ \citenamefont {Love}}]{seeleyBravyiKitaevTransformationQuantum2012}%
  \BibitemOpen
  \bibfield  {author} {\bibinfo {author} {\bibfnamefont {J.~T.}\ \bibnamefont {Seeley}}, \bibinfo {author} {\bibfnamefont {M.~J.}\ \bibnamefont {Richard}},\ and\ \bibinfo {author} {\bibfnamefont {P.~J.}\ \bibnamefont {Love}},\ }\href {https://doi.org/10.1063/1.4768229} {\bibfield  {journal} {\bibinfo  {journal} {The Journal of Chemical Physics}\ }\textbf {\bibinfo {volume} {137}},\ \bibinfo {pages} {224109} (\bibinfo {year} {2012})}\BibitemShut {NoStop}%
\bibitem [{\citenamefont {Gao}\ \emph {et~al.}(2026)\citenamefont {Gao}, \citenamefont {Li}, \citenamefont {Li}, \citenamefont {Li}, \citenamefont {Huang}, \citenamefont {Iancu},\ and\ \citenamefont {Zhang}}]{gao2026linearcomplexityfermionic}%
  \BibitemOpen
  \bibfield  {author} {\bibinfo {author} {\bibfnamefont {X.}~\bibnamefont {Gao}}, \bibinfo {author} {\bibfnamefont {W.}~\bibnamefont {Li}}, \bibinfo {author} {\bibfnamefont {J.}~\bibnamefont {Li}}, \bibinfo {author} {\bibfnamefont {Z.}~\bibnamefont {Li}}, \bibinfo {author} {\bibfnamefont {Y.}~\bibnamefont {Huang}}, \bibinfo {author} {\bibfnamefont {C.}~\bibnamefont {Iancu}},\ and\ \bibinfo {author} {\bibfnamefont {E.~Z.}\ \bibnamefont {Zhang}},\ }\href {https://doi.org/10.48550/arXiv.2606.00982} {\bibinfo {title} {Linear complexity fermionic simulation on quantum devices with hardware connectivity constraints}} (\bibinfo {year} {2026}),\ \Eprint {https://arxiv.org/abs/2606.00982} {arXiv:2606.00982 [cs.AR]} \BibitemShut {NoStop}%
\bibitem [{\citenamefont {Setia}\ and\ \citenamefont {Whitfield}(2018)}]{setiaBravyiKitaevSuperfastSimulation2018}%
  \BibitemOpen
  \bibfield  {author} {\bibinfo {author} {\bibfnamefont {K.}~\bibnamefont {Setia}}\ and\ \bibinfo {author} {\bibfnamefont {J.~D.}\ \bibnamefont {Whitfield}},\ }\href {https://doi.org/10.1063/1.5019371} {\bibfield  {journal} {\bibinfo  {journal} {The Journal of Chemical Physics}\ }\textbf {\bibinfo {volume} {148}},\ \bibinfo {pages} {164104} (\bibinfo {year} {2018})}\BibitemShut {NoStop}%
\bibitem [{\citenamefont {Jiang}\ \emph {et~al.}(2020)\citenamefont {Jiang}, \citenamefont {Kalev}, \citenamefont {Mruczkiewicz},\ and\ \citenamefont {Neven}}]{jiang_optimal_2020}%
  \BibitemOpen
  \bibfield  {author} {\bibinfo {author} {\bibfnamefont {Z.}~\bibnamefont {Jiang}}, \bibinfo {author} {\bibfnamefont {A.}~\bibnamefont {Kalev}}, \bibinfo {author} {\bibfnamefont {W.}~\bibnamefont {Mruczkiewicz}},\ and\ \bibinfo {author} {\bibfnamefont {H.}~\bibnamefont {Neven}},\ }\href {https://doi.org/10.22331/q-2020-06-04-276} {\bibfield  {journal} {\bibinfo  {journal} {Quantum}\ }\textbf {\bibinfo {volume} {4}},\ \bibinfo {pages} {276} (\bibinfo {year} {2020})},\ \Eprint {https://arxiv.org/abs/1910.10746} {arXiv:1910.10746 [quant-ph]} \BibitemShut {NoStop}%
\bibitem [{\citenamefont {Miller}\ \emph {et~al.}(2023)\citenamefont {Miller}, \citenamefont {Zimboras}, \citenamefont {Knecht}, \citenamefont {Maniscalco},\ and\ \citenamefont {Garcia-Perez}}]{miller2023bonsai}%
  \BibitemOpen
  \bibfield  {author} {\bibinfo {author} {\bibfnamefont {A.}~\bibnamefont {Miller}}, \bibinfo {author} {\bibfnamefont {Z.}~\bibnamefont {Zimboras}}, \bibinfo {author} {\bibfnamefont {S.}~\bibnamefont {Knecht}}, \bibinfo {author} {\bibfnamefont {S.}~\bibnamefont {Maniscalco}},\ and\ \bibinfo {author} {\bibfnamefont {G.}~\bibnamefont {Garcia-Perez}},\ }\href {https://doi.org/10.1103/PRXQuantum.4.030314} {\bibfield  {journal} {\bibinfo  {journal} {PRX Quantum}\ }\textbf {\bibinfo {volume} {4}},\ \bibinfo {pages} {030314} (\bibinfo {year} {2023})}\BibitemShut {NoStop}%
\bibitem [{\citenamefont {Miller}\ \emph {et~al.}(2026)\citenamefont {Miller}, \citenamefont {Glos},\ and\ \citenamefont {Zimbor{\'a}s}}]{miller2024treespilation}%
  \BibitemOpen
  \bibfield  {author} {\bibinfo {author} {\bibfnamefont {A.}~\bibnamefont {Miller}}, \bibinfo {author} {\bibfnamefont {A.}~\bibnamefont {Glos}},\ and\ \bibinfo {author} {\bibfnamefont {Z.}~\bibnamefont {Zimbor{\'a}s}},\ }\href {https://doi.org/10.1038/s41534-025-01170-2} {\bibfield  {journal} {\bibinfo  {journal} {npj Quantum Information}\ }\textbf {\bibinfo {volume} {12}},\ \bibinfo {pages} {26} (\bibinfo {year} {2026})},\ \Eprint {https://arxiv.org/abs/2403.03992} {arXiv:2403.03992 [quant-ph]} \BibitemShut {NoStop}%
\bibitem [{\citenamefont {Harrison}\ \emph {et~al.}(2024)\citenamefont {Harrison}, \citenamefont {Chiew}, \citenamefont {Necaise}, \citenamefont {Projansky}, \citenamefont {Strelchuk},\ and\ \citenamefont {Whitfield}}]{harrison2024sierpinski}%
  \BibitemOpen
  \bibfield  {author} {\bibinfo {author} {\bibfnamefont {B.}~\bibnamefont {Harrison}}, \bibinfo {author} {\bibfnamefont {M.}~\bibnamefont {Chiew}}, \bibinfo {author} {\bibfnamefont {J.}~\bibnamefont {Necaise}}, \bibinfo {author} {\bibfnamefont {A.~M.}\ \bibnamefont {Projansky}}, \bibinfo {author} {\bibfnamefont {S.}~\bibnamefont {Strelchuk}},\ and\ \bibinfo {author} {\bibfnamefont {J.~D.}\ \bibnamefont {Whitfield}},\ }\href@noop {} {\bibinfo {title} {A {Sierpinski} triangle fermion-to-qubit transform}} (\bibinfo {year} {2024}),\ \Eprint {https://arxiv.org/abs/2409.04348} {arXiv:2409.04348 [quant-ph]} \BibitemShut {NoStop}%
\bibitem [{\citenamefont {Derby}\ \emph {et~al.}(2021)\citenamefont {Derby}, \citenamefont {Klassen}, \citenamefont {Bausch},\ and\ \citenamefont {Cubitt}}]{derby2021compact}%
  \BibitemOpen
  \bibfield  {author} {\bibinfo {author} {\bibfnamefont {C.}~\bibnamefont {Derby}}, \bibinfo {author} {\bibfnamefont {J.}~\bibnamefont {Klassen}}, \bibinfo {author} {\bibfnamefont {J.}~\bibnamefont {Bausch}},\ and\ \bibinfo {author} {\bibfnamefont {T.}~\bibnamefont {Cubitt}},\ }\href {https://doi.org/10.1103/PhysRevB.104.035118} {\bibfield  {journal} {\bibinfo  {journal} {Physical Review B}\ }\textbf {\bibinfo {volume} {104}},\ \bibinfo {pages} {035118} (\bibinfo {year} {2021})},\ \Eprint {https://arxiv.org/abs/2003.06939} {arXiv:2003.06939 [quant-ph]} \BibitemShut {NoStop}%
\bibitem [{\citenamefont {Kirby}\ \emph {et~al.}(2021)\citenamefont {Kirby}, \citenamefont {Hadi}, \citenamefont {Kreshchuk},\ and\ \citenamefont {Love}}]{kirby2021compactencoding}%
  \BibitemOpen
  \bibfield  {author} {\bibinfo {author} {\bibfnamefont {W.~M.}\ \bibnamefont {Kirby}}, \bibinfo {author} {\bibfnamefont {S.}~\bibnamefont {Hadi}}, \bibinfo {author} {\bibfnamefont {M.}~\bibnamefont {Kreshchuk}},\ and\ \bibinfo {author} {\bibfnamefont {P.~J.}\ \bibnamefont {Love}},\ }\href {https://doi.org/10.1103/PhysRevA.104.042607} {\bibfield  {journal} {\bibinfo  {journal} {Physical Review A}\ }\textbf {\bibinfo {volume} {104}},\ \bibinfo {pages} {042607} (\bibinfo {year} {2021})},\ \Eprint {https://arxiv.org/abs/2105.10941} {arXiv:2105.10941 [quant-ph]} \BibitemShut {NoStop}%
\bibitem [{\citenamefont {Kirby}\ \emph {et~al.}(2022)\citenamefont {Kirby}, \citenamefont {Fuller}, \citenamefont {Hadfield},\ and\ \citenamefont {Mezzacapo}}]{kirby2022secondquantized}%
  \BibitemOpen
  \bibfield  {author} {\bibinfo {author} {\bibfnamefont {W.}~\bibnamefont {Kirby}}, \bibinfo {author} {\bibfnamefont {B.}~\bibnamefont {Fuller}}, \bibinfo {author} {\bibfnamefont {C.}~\bibnamefont {Hadfield}},\ and\ \bibinfo {author} {\bibfnamefont {A.}~\bibnamefont {Mezzacapo}},\ }\href {https://doi.org/10.1103/PRXQuantum.3.020351} {\bibfield  {journal} {\bibinfo  {journal} {PRX Quantum}\ }\textbf {\bibinfo {volume} {3}},\ \bibinfo {pages} {020351} (\bibinfo {year} {2022})},\ \Eprint {https://arxiv.org/abs/2109.14465} {arXiv:2109.14465 [quant-ph]} \BibitemShut {NoStop}%
\bibitem [{\citenamefont {Carolan}\ and\ \citenamefont {Schaeffer}(2025)}]{carolan2025succinct}%
  \BibitemOpen
  \bibfield  {author} {\bibinfo {author} {\bibfnamefont {J.}~\bibnamefont {Carolan}}\ and\ \bibinfo {author} {\bibfnamefont {L.}~\bibnamefont {Schaeffer}},\ }in\ \href {https://doi.org/10.4230/LIPIcs.ITCS.2025.32} {\emph {\bibinfo {booktitle} {16th Innovations in Theoretical Computer Science Conference (ITCS 2025)}}},\ \bibinfo {series} {Leibniz International Proceedings in Informatics (LIPIcs)}, Vol.\ \bibinfo {volume} {325}\ (\bibinfo  {publisher} {Schloss Dagstuhl -- Leibniz-Zentrum f{\"u}r Informatik},\ \bibinfo {year} {2025})\ pp.\ \bibinfo {pages} {32:1--32:21},\ \Eprint {https://arxiv.org/abs/2410.04015} {arXiv:2410.04015 [quant-ph]} \BibitemShut {NoStop}%
\bibitem [{\citenamefont {Berry}\ \emph {et~al.}(2019)\citenamefont {Berry}, \citenamefont {Gidney}, \citenamefont {Motta}, \citenamefont {McClean},\ and\ \citenamefont {Babbush}}]{berry2019qubitization}%
  \BibitemOpen
  \bibfield  {author} {\bibinfo {author} {\bibfnamefont {D.~W.}\ \bibnamefont {Berry}}, \bibinfo {author} {\bibfnamefont {C.}~\bibnamefont {Gidney}}, \bibinfo {author} {\bibfnamefont {M.}~\bibnamefont {Motta}}, \bibinfo {author} {\bibfnamefont {J.~R.}\ \bibnamefont {McClean}},\ and\ \bibinfo {author} {\bibfnamefont {R.}~\bibnamefont {Babbush}},\ }\href {https://doi.org/10.22331/q-2019-12-02-208} {\bibfield  {journal} {\bibinfo  {journal} {Quantum}\ }\textbf {\bibinfo {volume} {3}},\ \bibinfo {pages} {208} (\bibinfo {year} {2019})},\ \Eprint {https://arxiv.org/abs/1902.02134} {arXiv:1902.02134} \BibitemShut {NoStop}%
\bibitem [{Note1()}]{Note1}%
  \BibitemOpen
  \bibinfo {note} {We could think of partitioning the spin sectors onto two QPUs, after which we would have to partition spatial orbitals within each sector. For a single balanced bipartition for each sector, this would result in four QPUs of size $N/2$ qubits each to simulate the evolution of $2N$ spin-orbitals}\BibitemShut {NoStop}%
\bibitem [{\citenamefont {Whitten}(1973)}]{whitten1973coulombic}%
  \BibitemOpen
  \bibfield  {author} {\bibinfo {author} {\bibfnamefont {J.~L.}\ \bibnamefont {Whitten}},\ }\href {https://doi.org/10.1063/1.1679012} {\bibfield  {journal} {\bibinfo  {journal} {The Journal of Chemical Physics}\ }\textbf {\bibinfo {volume} {58}},\ \bibinfo {pages} {4496} (\bibinfo {year} {1973})}\BibitemShut {NoStop}%
\bibitem [{\citenamefont {Aquilante}\ \emph {et~al.}(2011)\citenamefont {Aquilante}, \citenamefont {Boman}, \citenamefont {Bostr{\"o}m}, \citenamefont {Koch}, \citenamefont {Lindh}, \citenamefont {S{\'a}nchez~de Mer{\'a}s},\ and\ \citenamefont {Pedersen}}]{aquilante2011cholesky}%
  \BibitemOpen
  \bibfield  {author} {\bibinfo {author} {\bibfnamefont {F.}~\bibnamefont {Aquilante}}, \bibinfo {author} {\bibfnamefont {L.}~\bibnamefont {Boman}}, \bibinfo {author} {\bibfnamefont {J.}~\bibnamefont {Bostr{\"o}m}}, \bibinfo {author} {\bibfnamefont {H.}~\bibnamefont {Koch}}, \bibinfo {author} {\bibfnamefont {R.}~\bibnamefont {Lindh}}, \bibinfo {author} {\bibfnamefont {A.}~\bibnamefont {S{\'a}nchez~de Mer{\'a}s}},\ and\ \bibinfo {author} {\bibfnamefont {T.~B.}\ \bibnamefont {Pedersen}},\ }in\ \href {https://doi.org/10.1007/978-90-481-2853-2_13} {\emph {\bibinfo {booktitle} {Linear-Scaling Techniques in Computational Chemistry and Physics}}}\ (\bibinfo  {publisher} {Springer},\ \bibinfo {year} {2011})\ pp.\ \bibinfo {pages} {301--343}\BibitemShut {NoStop}%
\bibitem [{\citenamefont {Peng}\ and\ \citenamefont {Kowalski}(2017)}]{peng2017lowrank}%
  \BibitemOpen
  \bibfield  {author} {\bibinfo {author} {\bibfnamefont {B.}~\bibnamefont {Peng}}\ and\ \bibinfo {author} {\bibfnamefont {K.}~\bibnamefont {Kowalski}},\ }\href {https://doi.org/10.1021/acs.jctc.7b00605} {\bibfield  {journal} {\bibinfo  {journal} {Journal of Chemical Theory and Computation}\ }\textbf {\bibinfo {volume} {13}},\ \bibinfo {pages} {4179} (\bibinfo {year} {2017})}\BibitemShut {NoStop}%
\bibitem [{\citenamefont {Bellonzi}\ \emph {et~al.}(2025)\citenamefont {Bellonzi}, \citenamefont {Cantin}, \citenamefont {Jangrouei}, \citenamefont {Kunitsa}, \citenamefont {Necaise}, \citenamefont {Nguyen}, \citenamefont {Penuel}, \citenamefont {Radin}, \citenamefont {Fontalvo}, \citenamefont {Sundareswara}, \citenamefont {Wang}, \citenamefont {Watts}, \citenamefont {Zhou}, \citenamefont {Garrett}, \citenamefont {Holmes}, \citenamefont {Izmaylov},\ and\ \citenamefont {Otten}}]{bellonzi2025qbgsee}%
  \BibitemOpen
  \bibfield  {author} {\bibinfo {author} {\bibfnamefont {N.}~\bibnamefont {Bellonzi}}, \bibinfo {author} {\bibfnamefont {J.~T.}\ \bibnamefont {Cantin}}, \bibinfo {author} {\bibfnamefont {M.~R.}\ \bibnamefont {Jangrouei}}, \bibinfo {author} {\bibfnamefont {A.}~\bibnamefont {Kunitsa}}, \bibinfo {author} {\bibfnamefont {J.}~\bibnamefont {Necaise}}, \bibinfo {author} {\bibfnamefont {N.}~\bibnamefont {Nguyen}}, \bibinfo {author} {\bibfnamefont {J.}~\bibnamefont {Penuel}}, \bibinfo {author} {\bibfnamefont {M.~D.}\ \bibnamefont {Radin}}, \bibinfo {author} {\bibfnamefont {J.~R.}\ \bibnamefont {Fontalvo}}, \bibinfo {author} {\bibfnamefont {R.}~\bibnamefont {Sundareswara}}, \bibinfo {author} {\bibfnamefont {L.}~\bibnamefont {Wang}}, \bibinfo {author} {\bibfnamefont {T.}~\bibnamefont {Watts}}, \bibinfo {author} {\bibfnamefont {Y.}~\bibnamefont {Zhou}}, \bibinfo {author} {\bibfnamefont {M.~C.}\ \bibnamefont {Garrett}}, \bibinfo {author} {\bibfnamefont {A.}~\bibnamefont {Holmes}}, \bibinfo {author} {\bibfnamefont {A.~F.}\
  \bibnamefont {Izmaylov}},\ and\ \bibinfo {author} {\bibfnamefont {M.}~\bibnamefont {Otten}},\ }\href {https://doi.org/10.48550/arXiv.2508.10873} {\bibinfo {title} {{QB} ground state energy estimation benchmark}} (\bibinfo {year} {2025}),\ \Eprint {https://arxiv.org/abs/2508.10873} {arXiv:2508.10873 [quant-ph]} \BibitemShut {NoStop}%
\bibitem [{\citenamefont {Roos}\ \emph {et~al.}(1980)\citenamefont {Roos}, \citenamefont {Taylor},\ and\ \citenamefont {Siegbahn}}]{roos1980complete}%
  \BibitemOpen
  \bibfield  {author} {\bibinfo {author} {\bibfnamefont {B.~O.}\ \bibnamefont {Roos}}, \bibinfo {author} {\bibfnamefont {P.~R.}\ \bibnamefont {Taylor}},\ and\ \bibinfo {author} {\bibfnamefont {P.~E.~M.}\ \bibnamefont {Siegbahn}},\ }\href {https://doi.org/10.1016/0301-0104(80)80045-0} {\bibfield  {journal} {\bibinfo  {journal} {Chemical Physics}\ }\textbf {\bibinfo {volume} {48}},\ \bibinfo {pages} {157} (\bibinfo {year} {1980})}\BibitemShut {NoStop}%
\bibitem [{\citenamefont {Knizia}\ and\ \citenamefont {Chan}(2012)}]{knizia2012density}%
  \BibitemOpen
  \bibfield  {author} {\bibinfo {author} {\bibfnamefont {G.}~\bibnamefont {Knizia}}\ and\ \bibinfo {author} {\bibfnamefont {G.~K.-L.}\ \bibnamefont {Chan}},\ }\href {https://doi.org/10.1103/PhysRevLett.109.186404} {\bibfield  {journal} {\bibinfo  {journal} {Physical Review Letters}\ }\textbf {\bibinfo {volume} {109}},\ \bibinfo {pages} {186404} (\bibinfo {year} {2012})},\ \Eprint {https://arxiv.org/abs/1204.5783} {arXiv:1204.5783 [cond-mat.str-el]} \BibitemShut {NoStop}%
\bibitem [{\citenamefont {Kitaura}\ \emph {et~al.}(1999)\citenamefont {Kitaura}, \citenamefont {Ikeo}, \citenamefont {Asada}, \citenamefont {Nakano},\ and\ \citenamefont {Uebayasi}}]{kitaura1999fragment}%
  \BibitemOpen
  \bibfield  {author} {\bibinfo {author} {\bibfnamefont {K.}~\bibnamefont {Kitaura}}, \bibinfo {author} {\bibfnamefont {E.}~\bibnamefont {Ikeo}}, \bibinfo {author} {\bibfnamefont {T.}~\bibnamefont {Asada}}, \bibinfo {author} {\bibfnamefont {T.}~\bibnamefont {Nakano}},\ and\ \bibinfo {author} {\bibfnamefont {M.}~\bibnamefont {Uebayasi}},\ }\href {https://doi.org/10.1016/S0009-2614(99)00874-X} {\bibfield  {journal} {\bibinfo  {journal} {Chemical Physics Letters}\ }\textbf {\bibinfo {volume} {313}},\ \bibinfo {pages} {701} (\bibinfo {year} {1999})}\BibitemShut {NoStop}%
\bibitem [{\citenamefont {Mochizuki}\ \emph {et~al.}(2021)\citenamefont {Mochizuki}, \citenamefont {Tanaka},\ and\ \citenamefont {Fukuzawa}}]{FMO_review_2021}%
  \BibitemOpen
  \bibinfo {editor} {\bibfnamefont {Y.}~\bibnamefont {Mochizuki}}, \bibinfo {editor} {\bibfnamefont {S.}~\bibnamefont {Tanaka}},\ and\ \bibinfo {editor} {\bibfnamefont {K.}~\bibnamefont {Fukuzawa}},\ eds.,\ \href {https://doi.org/10.1007/978-981-15-9235-5} {\emph {\bibinfo {title} {Recent Advances of the Fragment Molecular Orbital Method: Enhanced Performance and Applicability}}}\ (\bibinfo  {publisher} {Springer Singapore},\ \bibinfo {year} {2021})\BibitemShut {NoStop}%
\bibitem [{Note2()}]{Note2}%
  \BibitemOpen
  \bibinfo {note} {The semantics vary across the literature, for e.g., ``fermionic fragments''\cite {oumarou2024compressed}, ``leafs'' \cite {oumarou2024compressed}, ``pairwise operators'' for $\protect \hat {G}_\ell $\cite {motta2021lowrank}, ``Cholesky vectors'' for $|g^{(\ell )}\rangle $\cite {motta2019afqmc_lowrank}. Our exact fragment terminology is from \cite {bellonzi2025qbgsee}.}\BibitemShut {Stop}%
\bibitem [{\citenamefont {Bravyi}(2005)}]{bravyi2005lagrangianflo}%
  \BibitemOpen
  \bibfield  {author} {\bibinfo {author} {\bibfnamefont {S.}~\bibnamefont {Bravyi}},\ }\href {https://doi.org/10.26421/QIC5.3-3} {\bibfield  {journal} {\bibinfo  {journal} {Quantum Information and Computation}\ }\textbf {\bibinfo {volume} {5}},\ \bibinfo {pages} {216} (\bibinfo {year} {2005})},\ \Eprint {https://arxiv.org/abs/quant-ph/0404180} {arXiv:quant-ph/0404180 [quant-ph]} \BibitemShut {NoStop}%
\bibitem [{\citenamefont {Kivlichan}\ \emph {et~al.}(2018)\citenamefont {Kivlichan}, \citenamefont {McClean}, \citenamefont {Wiebe}, \citenamefont {Gidney}, \citenamefont {Aspuru-Guzik}, \citenamefont {Chan},\ and\ \citenamefont {Babbush}}]{kivlichan2018fswap}%
  \BibitemOpen
  \bibfield  {author} {\bibinfo {author} {\bibfnamefont {I.~D.}\ \bibnamefont {Kivlichan}}, \bibinfo {author} {\bibfnamefont {J.}~\bibnamefont {McClean}}, \bibinfo {author} {\bibfnamefont {N.}~\bibnamefont {Wiebe}}, \bibinfo {author} {\bibfnamefont {C.}~\bibnamefont {Gidney}}, \bibinfo {author} {\bibfnamefont {A.}~\bibnamefont {Aspuru-Guzik}}, \bibinfo {author} {\bibfnamefont {G.~K.-L.}\ \bibnamefont {Chan}},\ and\ \bibinfo {author} {\bibfnamefont {R.}~\bibnamefont {Babbush}},\ }\href {https://doi.org/10.1103/PhysRevLett.120.110501} {\bibfield  {journal} {\bibinfo  {journal} {Physical Review Letters}\ }\textbf {\bibinfo {volume} {120}},\ \bibinfo {pages} {110501} (\bibinfo {year} {2018})},\ \Eprint {https://arxiv.org/abs/1711.04789} {arXiv:1711.04789 [quant-ph]} \BibitemShut {NoStop}%
\bibitem [{\citenamefont {Dunlap}\ \emph {et~al.}(1979)\citenamefont {Dunlap}, \citenamefont {Connolly},\ and\ \citenamefont {Sabin}}]{dunlap1979approximations}%
  \BibitemOpen
  \bibfield  {author} {\bibinfo {author} {\bibfnamefont {B.~I.}\ \bibnamefont {Dunlap}}, \bibinfo {author} {\bibfnamefont {J.~W.~D.}\ \bibnamefont {Connolly}},\ and\ \bibinfo {author} {\bibfnamefont {J.~R.}\ \bibnamefont {Sabin}},\ }\href {https://doi.org/10.1063/1.438728} {\bibfield  {journal} {\bibinfo  {journal} {The Journal of Chemical Physics}\ }\textbf {\bibinfo {volume} {71}},\ \bibinfo {pages} {3396} (\bibinfo {year} {1979})}\BibitemShut {NoStop}%
\bibitem [{\citenamefont {Vahtras}\ \emph {et~al.}(1993)\citenamefont {Vahtras}, \citenamefont {Alml{\"o}f},\ and\ \citenamefont {Feyereisen}}]{vahtras1993integral}%
  \BibitemOpen
  \bibfield  {author} {\bibinfo {author} {\bibfnamefont {O.}~\bibnamefont {Vahtras}}, \bibinfo {author} {\bibfnamefont {J.}~\bibnamefont {Alml{\"o}f}},\ and\ \bibinfo {author} {\bibfnamefont {M.~W.}\ \bibnamefont {Feyereisen}},\ }\href {https://doi.org/10.1016/0009-2614(93)89151-7} {\bibfield  {journal} {\bibinfo  {journal} {Chemical Physics Letters}\ }\textbf {\bibinfo {volume} {213}},\ \bibinfo {pages} {514} (\bibinfo {year} {1993})}\BibitemShut {NoStop}%
\bibitem [{\citenamefont {Beebe}\ and\ \citenamefont {Linderberg}(1977)}]{beebe1977simplifications}%
  \BibitemOpen
  \bibfield  {author} {\bibinfo {author} {\bibfnamefont {N.~H.~F.}\ \bibnamefont {Beebe}}\ and\ \bibinfo {author} {\bibfnamefont {J.}~\bibnamefont {Linderberg}},\ }\href {https://doi.org/10.1002/qua.560120408} {\bibfield  {journal} {\bibinfo  {journal} {International Journal of Quantum Chemistry}\ }\textbf {\bibinfo {volume} {12}},\ \bibinfo {pages} {683} (\bibinfo {year} {1977})}\BibitemShut {NoStop}%
\bibitem [{\citenamefont {Koch}\ \emph {et~al.}(2003)\citenamefont {Koch}, \citenamefont {S{\'a}nchez~de Mer{\'a}s},\ and\ \citenamefont {Pedersen}}]{koch2003reduced}%
  \BibitemOpen
  \bibfield  {author} {\bibinfo {author} {\bibfnamefont {H.}~\bibnamefont {Koch}}, \bibinfo {author} {\bibfnamefont {A.}~\bibnamefont {S{\'a}nchez~de Mer{\'a}s}},\ and\ \bibinfo {author} {\bibfnamefont {T.~B.}\ \bibnamefont {Pedersen}},\ }\href {https://doi.org/10.1063/1.1578621} {\bibfield  {journal} {\bibinfo  {journal} {The Journal of Chemical Physics}\ }\textbf {\bibinfo {volume} {118}},\ \bibinfo {pages} {9481} (\bibinfo {year} {2003})}\BibitemShut {NoStop}%
\bibitem [{\citenamefont {Folkestad}\ \emph {et~al.}(2019)\citenamefont {Folkestad}, \citenamefont {Kj{\o}nstad},\ and\ \citenamefont {Koch}}]{folkestad2019efficient}%
  \BibitemOpen
  \bibfield  {author} {\bibinfo {author} {\bibfnamefont {S.~D.}\ \bibnamefont {Folkestad}}, \bibinfo {author} {\bibfnamefont {E.~F.}\ \bibnamefont {Kj{\o}nstad}},\ and\ \bibinfo {author} {\bibfnamefont {H.}~\bibnamefont {Koch}},\ }\href {https://doi.org/10.1063/1.5083802} {\bibfield  {journal} {\bibinfo  {journal} {The Journal of Chemical Physics}\ }\textbf {\bibinfo {volume} {150}},\ \bibinfo {pages} {194112} (\bibinfo {year} {2019})},\ \Eprint {https://arxiv.org/abs/1811.12890} {arXiv:1811.12890} \BibitemShut {NoStop}%
\bibitem [{\citenamefont {Aquilante}\ \emph {et~al.}(2017)\citenamefont {Aquilante}, \citenamefont {Delcey}, \citenamefont {Pedersen}, \citenamefont {Fdez.~Galv{\'a}n},\ and\ \citenamefont {Lindh}}]{aquilante2017innerprojection}%
  \BibitemOpen
  \bibfield  {author} {\bibinfo {author} {\bibfnamefont {F.}~\bibnamefont {Aquilante}}, \bibinfo {author} {\bibfnamefont {M.~G.}\ \bibnamefont {Delcey}}, \bibinfo {author} {\bibfnamefont {T.~B.}\ \bibnamefont {Pedersen}}, \bibinfo {author} {\bibfnamefont {I.}~\bibnamefont {Fdez.~Galv{\'a}n}},\ and\ \bibinfo {author} {\bibfnamefont {R.}~\bibnamefont {Lindh}},\ }\href {https://doi.org/10.1080/00268976.2017.1284354} {\bibfield  {journal} {\bibinfo  {journal} {Molecular Physics}\ }\textbf {\bibinfo {volume} {115}},\ \bibinfo {pages} {2052} (\bibinfo {year} {2017})}\BibitemShut {NoStop}%
\bibitem [{\citenamefont {Motta}\ \emph {et~al.}(2019)\citenamefont {Motta}, \citenamefont {Shee}, \citenamefont {Zhang},\ and\ \citenamefont {Chan}}]{motta2019afqmc_lowrank}%
  \BibitemOpen
  \bibfield  {author} {\bibinfo {author} {\bibfnamefont {M.}~\bibnamefont {Motta}}, \bibinfo {author} {\bibfnamefont {J.}~\bibnamefont {Shee}}, \bibinfo {author} {\bibfnamefont {S.}~\bibnamefont {Zhang}},\ and\ \bibinfo {author} {\bibfnamefont {G.~K.-L.}\ \bibnamefont {Chan}},\ }\href {https://doi.org/10.1021/acs.jctc.8b00996} {\bibfield  {journal} {\bibinfo  {journal} {Journal of Chemical Theory and Computation}\ }\textbf {\bibinfo {volume} {15}},\ \bibinfo {pages} {3510} (\bibinfo {year} {2019})},\ \Eprint {https://arxiv.org/abs/1810.01549} {arXiv:1810.01549} \BibitemShut {NoStop}%
\bibitem [{\citenamefont {Childs}\ \emph {et~al.}(2021)\citenamefont {Childs}, \citenamefont {Su}, \citenamefont {Tran}, \citenamefont {Wiebe},\ and\ \citenamefont {Zhu}}]{childs2021theory}%
  \BibitemOpen
  \bibfield  {author} {\bibinfo {author} {\bibfnamefont {A.~M.}\ \bibnamefont {Childs}}, \bibinfo {author} {\bibfnamefont {Y.}~\bibnamefont {Su}}, \bibinfo {author} {\bibfnamefont {M.~C.}\ \bibnamefont {Tran}}, \bibinfo {author} {\bibfnamefont {N.}~\bibnamefont {Wiebe}},\ and\ \bibinfo {author} {\bibfnamefont {S.}~\bibnamefont {Zhu}},\ }\href {https://doi.org/10.1103/PhysRevX.11.011020} {\bibfield  {journal} {\bibinfo  {journal} {Physical Review X}\ }\textbf {\bibinfo {volume} {11}},\ \bibinfo {pages} {011020} (\bibinfo {year} {2021})},\ \Eprint {https://arxiv.org/abs/1912.08854} {arXiv:1912.08854 [quant-ph]} \BibitemShut {NoStop}%
\bibitem [{\citenamefont {Campbell}(2021)}]{Campbell_2021}%
  \BibitemOpen
  \bibfield  {author} {\bibinfo {author} {\bibfnamefont {E.~T.}\ \bibnamefont {Campbell}},\ }\href {https://doi.org/10.1088/2058-9565/ac3110} {\bibfield  {journal} {\bibinfo  {journal} {Quantum Science and Technology}\ }\textbf {\bibinfo {volume} {7}},\ \bibinfo {pages} {015007} (\bibinfo {year} {2021})},\ \Eprint {https://arxiv.org/abs/2012.09238} {2012.09238} \BibitemShut {NoStop}%
\bibitem [{\citenamefont {Mehendale}\ \emph {et~al.}(2025)\citenamefont {Mehendale}, \citenamefont {Mart{\'i}nez-Mart{\'i}nez}, \citenamefont {Kamath},\ and\ \citenamefont {Izmaylov}}]{mehendale2023estimating}%
  \BibitemOpen
  \bibfield  {author} {\bibinfo {author} {\bibfnamefont {S.~G.}\ \bibnamefont {Mehendale}}, \bibinfo {author} {\bibfnamefont {L.~A.}\ \bibnamefont {Mart{\'i}nez-Mart{\'i}nez}}, \bibinfo {author} {\bibfnamefont {P.~D.}\ \bibnamefont {Kamath}},\ and\ \bibinfo {author} {\bibfnamefont {A.~F.}\ \bibnamefont {Izmaylov}},\ }\href {https://doi.org/10.1039/D5DD00185D} {\bibfield  {journal} {\bibinfo  {journal} {Digital Discovery}\ }\textbf {\bibinfo {volume} {4}},\ \bibinfo {pages} {3540} (\bibinfo {year} {2025})},\ \Eprint {https://arxiv.org/abs/2312.13282} {arXiv:2312.13282 [quant-ph]} \BibitemShut {NoStop}%
\bibitem [{\citenamefont {Low}\ and\ \citenamefont {Chuang}(2019)}]{low2019hamiltonian}%
  \BibitemOpen
  \bibfield  {author} {\bibinfo {author} {\bibfnamefont {G.~H.}\ \bibnamefont {Low}}\ and\ \bibinfo {author} {\bibfnamefont {I.~L.}\ \bibnamefont {Chuang}},\ }\href {https://doi.org/10.22331/q-2019-07-12-163} {\bibfield  {journal} {\bibinfo  {journal} {Quantum}\ }\textbf {\bibinfo {volume} {3}},\ \bibinfo {pages} {163} (\bibinfo {year} {2019})},\ \Eprint {https://arxiv.org/abs/1610.06546} {arXiv:1610.06546 [quant-ph]} \BibitemShut {NoStop}%
\bibitem [{\citenamefont {Babbush}\ \emph {et~al.}(2018{\natexlab{a}})\citenamefont {Babbush}, \citenamefont {Gidney}, \citenamefont {Berry}, \citenamefont {Wiebe}, \citenamefont {McClean}, \citenamefont {Paler}, \citenamefont {Fowler},\ and\ \citenamefont {Neven}}]{babbush2018encoding}%
  \BibitemOpen
  \bibfield  {author} {\bibinfo {author} {\bibfnamefont {R.}~\bibnamefont {Babbush}}, \bibinfo {author} {\bibfnamefont {C.}~\bibnamefont {Gidney}}, \bibinfo {author} {\bibfnamefont {D.~W.}\ \bibnamefont {Berry}}, \bibinfo {author} {\bibfnamefont {N.}~\bibnamefont {Wiebe}}, \bibinfo {author} {\bibfnamefont {J.}~\bibnamefont {McClean}}, \bibinfo {author} {\bibfnamefont {A.}~\bibnamefont {Paler}}, \bibinfo {author} {\bibfnamefont {A.}~\bibnamefont {Fowler}},\ and\ \bibinfo {author} {\bibfnamefont {H.}~\bibnamefont {Neven}},\ }\href {https://doi.org/10.1103/PhysRevX.8.041015} {\bibfield  {journal} {\bibinfo  {journal} {Physical Review X}\ }\textbf {\bibinfo {volume} {8}},\ \bibinfo {pages} {041015} (\bibinfo {year} {2018}{\natexlab{a}})},\ \Eprint {https://arxiv.org/abs/1805.03662} {arXiv:1805.03662 [quant-ph]} \BibitemShut {NoStop}%
\bibitem [{\citenamefont {Leimkuhler}\ and\ \citenamefont {Whaley}(2026)}]{leimkuhler2026exponentialquantumspeedupsnearterm}%
  \BibitemOpen
  \bibfield  {author} {\bibinfo {author} {\bibfnamefont {O.}~\bibnamefont {Leimkuhler}}\ and\ \bibinfo {author} {\bibfnamefont {K.~B.}\ \bibnamefont {Whaley}},\ }\href {https://arxiv.org/abs/2503.21041} {\bibinfo {title} {Exponential quantum speedups for near-term molecular electronic structure methods}} (\bibinfo {year} {2026}),\ \Eprint {https://arxiv.org/abs/2503.21041} {arXiv:2503.21041 [quant-ph]} \BibitemShut {NoStop}%
\bibitem [{\citenamefont {Peschel}(2003)}]{peschel2003rdm}%
  \BibitemOpen
  \bibfield  {author} {\bibinfo {author} {\bibfnamefont {I.}~\bibnamefont {Peschel}},\ }\href {https://doi.org/10.1088/0305-4470/36/14/101} {\bibfield  {journal} {\bibinfo  {journal} {Journal of Physics A: Mathematical and General}\ }\textbf {\bibinfo {volume} {36}},\ \bibinfo {pages} {L205} (\bibinfo {year} {2003})}\BibitemShut {NoStop}%
\bibitem [{\citenamefont {Peschel}\ and\ \citenamefont {Eisler}(2009{\natexlab{a}})}]{peschel2009reduced}%
  \BibitemOpen
  \bibfield  {author} {\bibinfo {author} {\bibfnamefont {I.}~\bibnamefont {Peschel}}\ and\ \bibinfo {author} {\bibfnamefont {V.}~\bibnamefont {Eisler}},\ }\href {https://doi.org/10.1088/1751-8113/42/50/504003} {\bibfield  {journal} {\bibinfo  {journal} {Journal of Physics A: Mathematical and Theoretical}\ }\textbf {\bibinfo {volume} {42}},\ \bibinfo {pages} {504003} (\bibinfo {year} {2009}{\natexlab{a}})}\BibitemShut {NoStop}%
\bibitem [{Note3()}]{Note3}%
  \BibitemOpen
  \bibinfo {note} {Where $R[A,A]$ is defined as the $N_A \times N_A$ submatrix whose elements are defined as $R[A,A]_{ij} = R_{A[i],A[j]}$.}\BibitemShut {Stop}%
\bibitem [{\citenamefont {Tucci}(2005)}]{tucci2005introductioncartanskakdecomposition}%
  \BibitemOpen
  \bibfield  {author} {\bibinfo {author} {\bibfnamefont {R.~R.}\ \bibnamefont {Tucci}},\ }\href {https://arxiv.org/abs/quant-ph/0507171} {\bibinfo {title} {An introduction to cartan's kak decomposition for qc programmers}} (\bibinfo {year} {2005}),\ \Eprint {https://arxiv.org/abs/quant-ph/0507171} {arXiv:quant-ph/0507171 [quant-ph]} \BibitemShut {NoStop}%
\bibitem [{\citenamefont {Hoffman}\ \emph {et~al.}(1972)\citenamefont {Hoffman}, \citenamefont {Raffenetti},\ and\ \citenamefont {Ruedenberg}}]{hoffman1972generalized}%
  \BibitemOpen
  \bibfield  {author} {\bibinfo {author} {\bibfnamefont {K.}~\bibnamefont {Hoffman}}, \bibinfo {author} {\bibfnamefont {R.~C.}\ \bibnamefont {Raffenetti}},\ and\ \bibinfo {author} {\bibfnamefont {K.}~\bibnamefont {Ruedenberg}},\ }\href {https://doi.org/10.1063/1.1666011} {\bibfield  {journal} {\bibinfo  {journal} {Journal of Mathematical Physics}\ }\textbf {\bibinfo {volume} {13}},\ \bibinfo {pages} {528} (\bibinfo {year} {1972})}\BibitemShut {NoStop}%
\bibitem [{\citenamefont {Paige}\ and\ \citenamefont {Wei}(1994)}]{paige1994history}%
  \BibitemOpen
  \bibfield  {author} {\bibinfo {author} {\bibfnamefont {C.~C.}\ \bibnamefont {Paige}}\ and\ \bibinfo {author} {\bibfnamefont {M.}~\bibnamefont {Wei}},\ }\href {https://doi.org/10.1016/0024-3795(94)90446-4} {\bibfield  {journal} {\bibinfo  {journal} {Linear Algebra and its Applications}\ }\textbf {\bibinfo {volume} {208--209}},\ \bibinfo {pages} {303} (\bibinfo {year} {1994})}\BibitemShut {NoStop}%
\bibitem [{\citenamefont {Golub}\ and\ \citenamefont {Van~Loan}(2013)}]{golub2013matrix}%
  \BibitemOpen
  \bibfield  {author} {\bibinfo {author} {\bibfnamefont {G.~H.}\ \bibnamefont {Golub}}\ and\ \bibinfo {author} {\bibfnamefont {C.~F.}\ \bibnamefont {Van~Loan}},\ }\href@noop {} {\emph {\bibinfo {title} {Matrix Computations}}},\ \bibinfo {edition} {4th}\ ed.\ (\bibinfo  {publisher} {Johns Hopkins University Press},\ \bibinfo {year} {2013})\BibitemShut {NoStop}%
\bibitem [{\citenamefont {Ufrecht}\ \emph {et~al.}(2024)\citenamefont {Ufrecht}, \citenamefont {Herzog}, \citenamefont {Scherer}, \citenamefont {Periyasamy}, \citenamefont {Rietsch}, \citenamefont {Plinge},\ and\ \citenamefont {Mutschler}}]{optimal_joint_two}%
  \BibitemOpen
  \bibfield  {author} {\bibinfo {author} {\bibfnamefont {C.}~\bibnamefont {Ufrecht}}, \bibinfo {author} {\bibfnamefont {L.~S.}\ \bibnamefont {Herzog}}, \bibinfo {author} {\bibfnamefont {D.~D.}\ \bibnamefont {Scherer}}, \bibinfo {author} {\bibfnamefont {M.}~\bibnamefont {Periyasamy}}, \bibinfo {author} {\bibfnamefont {S.}~\bibnamefont {Rietsch}}, \bibinfo {author} {\bibfnamefont {A.}~\bibnamefont {Plinge}},\ and\ \bibinfo {author} {\bibfnamefont {C.}~\bibnamefont {Mutschler}},\ }\href {https://doi.org/10.1103/PhysRevA.109.052440} {\bibfield  {journal} {\bibinfo  {journal} {Phys. Rev. A}\ }\textbf {\bibinfo {volume} {109}},\ \bibinfo {pages} {052440} (\bibinfo {year} {2024})}\BibitemShut {NoStop}%
\bibitem [{\citenamefont {Wachter}(1980)}]{wachter1980limiting}%
  \BibitemOpen
  \bibfield  {author} {\bibinfo {author} {\bibfnamefont {K.~W.}\ \bibnamefont {Wachter}},\ }\href {https://doi.org/10.1214/aos/1176345134} {\bibfield  {journal} {\bibinfo  {journal} {The Annals of Statistics}\ }\textbf {\bibinfo {volume} {8}},\ \bibinfo {pages} {937} (\bibinfo {year} {1980})}\BibitemShut {NoStop}%
\bibitem [{\citenamefont {Loaiza}\ \emph {et~al.}(2023)\citenamefont {Loaiza}, \citenamefont {Marefat~Khah}, \citenamefont {Wiebe},\ and\ \citenamefont {Izmaylov}}]{loaiza2023lcucostreduction}%
  \BibitemOpen
  \bibfield  {author} {\bibinfo {author} {\bibfnamefont {I.}~\bibnamefont {Loaiza}}, \bibinfo {author} {\bibfnamefont {A.}~\bibnamefont {Marefat~Khah}}, \bibinfo {author} {\bibfnamefont {N.}~\bibnamefont {Wiebe}},\ and\ \bibinfo {author} {\bibfnamefont {A.~F.}\ \bibnamefont {Izmaylov}},\ }\href {https://doi.org/10.1088/2058-9565/acd577} {\bibfield  {journal} {\bibinfo  {journal} {Quantum Science and Technology}\ }\textbf {\bibinfo {volume} {8}},\ \bibinfo {pages} {035019} (\bibinfo {year} {2023})},\ \Eprint {https://arxiv.org/abs/2208.08272} {arXiv:2208.08272 [quant-ph]} \BibitemShut {NoStop}%
\bibitem [{\citenamefont {Oumarou}\ \emph {et~al.}(2024)\citenamefont {Oumarou}, \citenamefont {Scheurer}, \citenamefont {Parrish}, \citenamefont {Hohenstein},\ and\ \citenamefont {Gogolin}}]{oumarou2024compressed}%
  \BibitemOpen
  \bibfield  {author} {\bibinfo {author} {\bibfnamefont {O.}~\bibnamefont {Oumarou}}, \bibinfo {author} {\bibfnamefont {M.}~\bibnamefont {Scheurer}}, \bibinfo {author} {\bibfnamefont {R.~M.}\ \bibnamefont {Parrish}}, \bibinfo {author} {\bibfnamefont {E.~G.}\ \bibnamefont {Hohenstein}},\ and\ \bibinfo {author} {\bibfnamefont {C.}~\bibnamefont {Gogolin}},\ }\href {https://doi.org/10.22331/q-2024-06-13-1371} {\bibfield  {journal} {\bibinfo  {journal} {{Quantum}}\ }\textbf {\bibinfo {volume} {8}},\ \bibinfo {pages} {1371} (\bibinfo {year} {2024})}\BibitemShut {NoStop}%
\bibitem [{\citenamefont {Rocca}\ \emph {et~al.}(2024)\citenamefont {Rocca}, \citenamefont {Cortes}, \citenamefont {Gonthier}, \citenamefont {Ollitrault}, \citenamefont {Parrish}, \citenamefont {Anselmetti}, \citenamefont {Degroote}, \citenamefont {Moll}, \citenamefont {Santagati},\ and\ \citenamefont {Streif}}]{rocca2024symmetry}%
  \BibitemOpen
  \bibfield  {author} {\bibinfo {author} {\bibfnamefont {D.}~\bibnamefont {Rocca}}, \bibinfo {author} {\bibfnamefont {C.~L.}\ \bibnamefont {Cortes}}, \bibinfo {author} {\bibfnamefont {J.~F.}\ \bibnamefont {Gonthier}}, \bibinfo {author} {\bibfnamefont {P.~J.}\ \bibnamefont {Ollitrault}}, \bibinfo {author} {\bibfnamefont {R.~M.}\ \bibnamefont {Parrish}}, \bibinfo {author} {\bibfnamefont {G.-L.}\ \bibnamefont {Anselmetti}}, \bibinfo {author} {\bibfnamefont {M.}~\bibnamefont {Degroote}}, \bibinfo {author} {\bibfnamefont {N.}~\bibnamefont {Moll}}, \bibinfo {author} {\bibfnamefont {R.}~\bibnamefont {Santagati}},\ and\ \bibinfo {author} {\bibfnamefont {M.}~\bibnamefont {Streif}},\ }\href {https://doi.org/10.1021/acs.jctc.4c00352} {\bibfield  {journal} {\bibinfo  {journal} {Journal of Chemical Theory and Computation}\ }\textbf {\bibinfo {volume} {20}},\ \bibinfo {pages} {4639} (\bibinfo {year} {2024})},\ \Eprint {https://arxiv.org/abs/2403.03502} {arXiv:2403.03502 [quant-ph]} \BibitemShut {NoStop}%
\bibitem [{\citenamefont {Marqversen}\ \emph {et~al.}(2026)\citenamefont {Marqversen}, \citenamefont {Baranes}, \citenamefont {Sirotin},\ and\ \citenamefont {Borregaard}}]{ds45-fm9n}%
  \BibitemOpen
  \bibfield  {author} {\bibinfo {author} {\bibfnamefont {F.~K.}\ \bibnamefont {Marqversen}}, \bibinfo {author} {\bibfnamefont {G.}~\bibnamefont {Baranes}}, \bibinfo {author} {\bibfnamefont {M.}~\bibnamefont {Sirotin}},\ and\ \bibinfo {author} {\bibfnamefont {J.}~\bibnamefont {Borregaard}},\ }\href {https://doi.org/10.1103/ds45-fm9n} {\bibfield  {journal} {\bibinfo  {journal} {Phys. Rev. Res.}\ }\textbf {\bibinfo {volume} {8}},\ \bibinfo {pages} {023097} (\bibinfo {year} {2026})}\BibitemShut {NoStop}%
\bibitem [{\citenamefont {Calabrese}\ and\ \citenamefont {Cardy}(2005)}]{Calabrese2005}%
  \BibitemOpen
  \bibfield  {author} {\bibinfo {author} {\bibfnamefont {P.}~\bibnamefont {Calabrese}}\ and\ \bibinfo {author} {\bibfnamefont {J.}~\bibnamefont {Cardy}},\ }\href {https://doi.org/10.1088/1742-5468/2005/04/p04010} {\bibfield  {journal} {\bibinfo  {journal} {Journal of Statistical Mechanics: Theory and Experiment}\ }\textbf {\bibinfo {volume} {2005}},\ \bibinfo {pages} {P04010} (\bibinfo {year} {2005})}\BibitemShut {NoStop}%
\bibitem [{\citenamefont {Fagotti}\ and\ \citenamefont {Calabrese}(2008)}]{Fagotti2008}%
  \BibitemOpen
  \bibfield  {author} {\bibinfo {author} {\bibfnamefont {M.}~\bibnamefont {Fagotti}}\ and\ \bibinfo {author} {\bibfnamefont {P.}~\bibnamefont {Calabrese}},\ }\bibfield  {journal} {\bibinfo  {journal} {Physical Review A}\ }\textbf {\bibinfo {volume} {78}},\ \href {https://doi.org/10.1103/physreva.78.010306} {10.1103/physreva.78.010306} (\bibinfo {year} {2008})\BibitemShut {NoStop}%
\bibitem [{\citenamefont {Peschel}\ and\ \citenamefont {Eisler}(2009{\natexlab{b}})}]{Peschel2009}%
  \BibitemOpen
  \bibfield  {author} {\bibinfo {author} {\bibfnamefont {I.}~\bibnamefont {Peschel}}\ and\ \bibinfo {author} {\bibfnamefont {V.}~\bibnamefont {Eisler}},\ }\href {https://doi.org/10.1088/1751-8113/42/50/504003} {\bibfield  {journal} {\bibinfo  {journal} {Journal of Physics A: Mathematical and Theoretical}\ }\textbf {\bibinfo {volume} {42}},\ \bibinfo {pages} {504003} (\bibinfo {year} {2009}{\natexlab{b}})}\BibitemShut {NoStop}%
\bibitem [{\citenamefont {Alba}\ and\ \citenamefont {Calabrese}(2017)}]{Alba2017}%
  \BibitemOpen
  \bibfield  {author} {\bibinfo {author} {\bibfnamefont {V.}~\bibnamefont {Alba}}\ and\ \bibinfo {author} {\bibfnamefont {P.}~\bibnamefont {Calabrese}},\ }\href {https://doi.org/10.1073/pnas.1703516114} {\bibfield  {journal} {\bibinfo  {journal} {Proceedings of the National Academy of Sciences}\ }\textbf {\bibinfo {volume} {114}},\ \bibinfo {pages} {7947–7951} (\bibinfo {year} {2017})}\BibitemShut {NoStop}%
\bibitem [{\citenamefont {Dubail}(2017)}]{Dubail2017}%
  \BibitemOpen
  \bibfield  {author} {\bibinfo {author} {\bibfnamefont {J.}~\bibnamefont {Dubail}},\ }\href {https://doi.org/10.1088/1751-8121/aa6f38} {\bibfield  {journal} {\bibinfo  {journal} {Journal of Physics A: Mathematical and Theoretical}\ }\textbf {\bibinfo {volume} {50}},\ \bibinfo {pages} {234001} (\bibinfo {year} {2017})}\BibitemShut {NoStop}%
\bibitem [{\citenamefont {Nandkishore}\ and\ \citenamefont {Huse}(2015)}]{Nandkishore2015}%
  \BibitemOpen
  \bibfield  {author} {\bibinfo {author} {\bibfnamefont {R.}~\bibnamefont {Nandkishore}}\ and\ \bibinfo {author} {\bibfnamefont {D.~A.}\ \bibnamefont {Huse}},\ }\href {https://doi.org/10.1146/annurev-conmatphys-031214-014726} {\bibfield  {journal} {\bibinfo  {journal} {Annual Review of Condensed Matter Physics}\ }\textbf {\bibinfo {volume} {6}},\ \bibinfo {pages} {15–38} (\bibinfo {year} {2015})}\BibitemShut {NoStop}%
\bibitem [{\citenamefont {Paviglianiti}\ \emph {et~al.}(2026)\citenamefont {Paviglianiti}, \citenamefont {Lumia}, \citenamefont {Tirrito}, \citenamefont {Silva}, \citenamefont {Collura}, \citenamefont {Turkeshi},\ and\ \citenamefont {Lami}}]{Paviglianiti2026}%
  \BibitemOpen
  \bibfield  {author} {\bibinfo {author} {\bibfnamefont {A.}~\bibnamefont {Paviglianiti}}, \bibinfo {author} {\bibfnamefont {L.}~\bibnamefont {Lumia}}, \bibinfo {author} {\bibfnamefont {E.}~\bibnamefont {Tirrito}}, \bibinfo {author} {\bibfnamefont {A.}~\bibnamefont {Silva}}, \bibinfo {author} {\bibfnamefont {M.}~\bibnamefont {Collura}}, \bibinfo {author} {\bibfnamefont {X.}~\bibnamefont {Turkeshi}},\ and\ \bibinfo {author} {\bibfnamefont {G.}~\bibnamefont {Lami}},\ }\bibfield  {journal} {\bibinfo  {journal} {Physical Review Letters}\ }\textbf {\bibinfo {volume} {136}},\ \href {https://doi.org/10.1103/w97w-7zny} {10.1103/w97w-7zny} (\bibinfo {year} {2026})\BibitemShut {NoStop}%
\bibitem [{\citenamefont {Nahum}\ \emph {et~al.}(2017)\citenamefont {Nahum}, \citenamefont {Ruhman}, \citenamefont {Vijay},\ and\ \citenamefont {Haah}}]{PhysRevX.7.031016}%
  \BibitemOpen
  \bibfield  {author} {\bibinfo {author} {\bibfnamefont {A.}~\bibnamefont {Nahum}}, \bibinfo {author} {\bibfnamefont {J.}~\bibnamefont {Ruhman}}, \bibinfo {author} {\bibfnamefont {S.}~\bibnamefont {Vijay}},\ and\ \bibinfo {author} {\bibfnamefont {J.}~\bibnamefont {Haah}},\ }\href {https://doi.org/10.1103/PhysRevX.7.031016} {\bibfield  {journal} {\bibinfo  {journal} {Phys. Rev. X}\ }\textbf {\bibinfo {volume} {7}},\ \bibinfo {pages} {031016} (\bibinfo {year} {2017})}\BibitemShut {NoStop}%
\bibitem [{\citenamefont {Nahum}\ \emph {et~al.}(2018)\citenamefont {Nahum}, \citenamefont {Vijay},\ and\ \citenamefont {Haah}}]{PhysRevX.8.021014}%
  \BibitemOpen
  \bibfield  {author} {\bibinfo {author} {\bibfnamefont {A.}~\bibnamefont {Nahum}}, \bibinfo {author} {\bibfnamefont {S.}~\bibnamefont {Vijay}},\ and\ \bibinfo {author} {\bibfnamefont {J.}~\bibnamefont {Haah}},\ }\href {https://doi.org/10.1103/PhysRevX.8.021014} {\bibfield  {journal} {\bibinfo  {journal} {Phys. Rev. X}\ }\textbf {\bibinfo {volume} {8}},\ \bibinfo {pages} {021014} (\bibinfo {year} {2018})}\BibitemShut {NoStop}%
\bibitem [{\citenamefont {Abanin}\ \emph {et~al.}(2019)\citenamefont {Abanin}, \citenamefont {Altman}, \citenamefont {Bloch},\ and\ \citenamefont {Serbyn}}]{RevModPhys.91.021001}%
  \BibitemOpen
  \bibfield  {author} {\bibinfo {author} {\bibfnamefont {D.~A.}\ \bibnamefont {Abanin}}, \bibinfo {author} {\bibfnamefont {E.}~\bibnamefont {Altman}}, \bibinfo {author} {\bibfnamefont {I.}~\bibnamefont {Bloch}},\ and\ \bibinfo {author} {\bibfnamefont {M.}~\bibnamefont {Serbyn}},\ }\href {https://doi.org/10.1103/RevModPhys.91.021001} {\bibfield  {journal} {\bibinfo  {journal} {Rev. Mod. Phys.}\ }\textbf {\bibinfo {volume} {91}},\ \bibinfo {pages} {021001} (\bibinfo {year} {2019})}\BibitemShut {NoStop}%
\bibitem [{\citenamefont {Basko}\ \emph {et~al.}(2006)\citenamefont {Basko}, \citenamefont {Aleiner},\ and\ \citenamefont {Altshuler}}]{Basko2006}%
  \BibitemOpen
  \bibfield  {author} {\bibinfo {author} {\bibfnamefont {D.}~\bibnamefont {Basko}}, \bibinfo {author} {\bibfnamefont {I.}~\bibnamefont {Aleiner}},\ and\ \bibinfo {author} {\bibfnamefont {B.}~\bibnamefont {Altshuler}},\ }\href {https://doi.org/10.1016/j.aop.2005.11.014} {\bibfield  {journal} {\bibinfo  {journal} {Annals of Physics}\ }\textbf {\bibinfo {volume} {321}},\ \bibinfo {pages} {1126–1205} (\bibinfo {year} {2006})}\BibitemShut {NoStop}%
\bibitem [{\citenamefont {Pal}\ and\ \citenamefont {Huse}(2010)}]{PhysRevB.82.174411}%
  \BibitemOpen
  \bibfield  {author} {\bibinfo {author} {\bibfnamefont {A.}~\bibnamefont {Pal}}\ and\ \bibinfo {author} {\bibfnamefont {D.~A.}\ \bibnamefont {Huse}},\ }\href {https://doi.org/10.1103/PhysRevB.82.174411} {\bibfield  {journal} {\bibinfo  {journal} {Phys. Rev. B}\ }\textbf {\bibinfo {volume} {82}},\ \bibinfo {pages} {174411} (\bibinfo {year} {2010})}\BibitemShut {NoStop}%
\bibitem [{\citenamefont {Morningstar}\ \emph {et~al.}(2022)\citenamefont {Morningstar}, \citenamefont {Colmenarez}, \citenamefont {Khemani}, \citenamefont {Luitz},\ and\ \citenamefont {Huse}}]{PhysRevB.105.174205}%
  \BibitemOpen
  \bibfield  {author} {\bibinfo {author} {\bibfnamefont {A.}~\bibnamefont {Morningstar}}, \bibinfo {author} {\bibfnamefont {L.}~\bibnamefont {Colmenarez}}, \bibinfo {author} {\bibfnamefont {V.}~\bibnamefont {Khemani}}, \bibinfo {author} {\bibfnamefont {D.~J.}\ \bibnamefont {Luitz}},\ and\ \bibinfo {author} {\bibfnamefont {D.~A.}\ \bibnamefont {Huse}},\ }\href {https://doi.org/10.1103/PhysRevB.105.174205} {\bibfield  {journal} {\bibinfo  {journal} {Phys. Rev. B}\ }\textbf {\bibinfo {volume} {105}},\ \bibinfo {pages} {174205} (\bibinfo {year} {2022})}\BibitemShut {NoStop}%
\bibitem [{\citenamefont {De~Roeck}\ and\ \citenamefont {Huveneers}(2017)}]{PhysRevB.95.155129}%
  \BibitemOpen
  \bibfield  {author} {\bibinfo {author} {\bibfnamefont {W.}~\bibnamefont {De~Roeck}}\ and\ \bibinfo {author} {\bibfnamefont {F.~m.~c.}\ \bibnamefont {Huveneers}},\ }\href {https://doi.org/10.1103/PhysRevB.95.155129} {\bibfield  {journal} {\bibinfo  {journal} {Phys. Rev. B}\ }\textbf {\bibinfo {volume} {95}},\ \bibinfo {pages} {155129} (\bibinfo {year} {2017})}\BibitemShut {NoStop}%
\bibitem [{\citenamefont {Sels}(2022)}]{PhysRevB.106.L020202}%
  \BibitemOpen
  \bibfield  {author} {\bibinfo {author} {\bibfnamefont {D.}~\bibnamefont {Sels}},\ }\href {https://doi.org/10.1103/PhysRevB.106.L020202} {\bibfield  {journal} {\bibinfo  {journal} {Phys. Rev. B}\ }\textbf {\bibinfo {volume} {106}},\ \bibinfo {pages} {L020202} (\bibinfo {year} {2022})}\BibitemShut {NoStop}%
\bibitem [{\citenamefont {Léonard}\ \emph {et~al.}(2023)\citenamefont {Léonard}, \citenamefont {Kim}, \citenamefont {Rispoli}, \citenamefont {Lukin}, \citenamefont {Schittko}, \citenamefont {Kwan}, \citenamefont {Demler}, \citenamefont {Sels},\ and\ \citenamefont {Greiner}}]{Lonard2023}%
  \BibitemOpen
  \bibfield  {author} {\bibinfo {author} {\bibfnamefont {J.}~\bibnamefont {Léonard}}, \bibinfo {author} {\bibfnamefont {S.}~\bibnamefont {Kim}}, \bibinfo {author} {\bibfnamefont {M.}~\bibnamefont {Rispoli}}, \bibinfo {author} {\bibfnamefont {A.}~\bibnamefont {Lukin}}, \bibinfo {author} {\bibfnamefont {R.}~\bibnamefont {Schittko}}, \bibinfo {author} {\bibfnamefont {J.}~\bibnamefont {Kwan}}, \bibinfo {author} {\bibfnamefont {E.}~\bibnamefont {Demler}}, \bibinfo {author} {\bibfnamefont {D.}~\bibnamefont {Sels}},\ and\ \bibinfo {author} {\bibfnamefont {M.}~\bibnamefont {Greiner}},\ }\href {https://doi.org/10.1038/s41567-022-01887-3} {\bibfield  {journal} {\bibinfo  {journal} {Nature Physics}\ }\textbf {\bibinfo {volume} {19}},\ \bibinfo {pages} {481–485} (\bibinfo {year} {2023})}\BibitemShut {NoStop}%
\bibitem [{\citenamefont {Colbois}\ \emph {et~al.}(2024)\citenamefont {Colbois}, \citenamefont {Alet},\ and\ \citenamefont {Laflorencie}}]{PhysRevLett.133.116502}%
  \BibitemOpen
  \bibfield  {author} {\bibinfo {author} {\bibfnamefont {J.}~\bibnamefont {Colbois}}, \bibinfo {author} {\bibfnamefont {F.}~\bibnamefont {Alet}},\ and\ \bibinfo {author} {\bibfnamefont {N.}~\bibnamefont {Laflorencie}},\ }\href {https://doi.org/10.1103/PhysRevLett.133.116502} {\bibfield  {journal} {\bibinfo  {journal} {Phys. Rev. Lett.}\ }\textbf {\bibinfo {volume} {133}},\ \bibinfo {pages} {116502} (\bibinfo {year} {2024})}\BibitemShut {NoStop}%
\bibitem [{\citenamefont {Anand}\ and\ \citenamefont {et. al.}(2026)}]{QChem_MBL_paper}%
  \BibitemOpen
  \bibfield  {author} {\bibinfo {author} {\bibfnamefont {N.}~\bibnamefont {Anand}}\ and\ \bibinfo {author} {\bibnamefont {et. al.}},\ }\href@noop {} {\bibinfo {title} {{Signatures of Localization in Electronic Structure Simulation}}} (\bibinfo {year} {2026}),\ \bibinfo {note} {{\textit{In Preparation}}}\BibitemShut {NoStop}%
\bibitem [{\citenamefont {Caesura}\ \emph {et~al.}(2025)\citenamefont {Caesura}, \citenamefont {Cortes}, \citenamefont {Pol}, \citenamefont {Sim}, \citenamefont {Steudtner}, \citenamefont {Anselmetti}, \citenamefont {Degroote}, \citenamefont {Moll}, \citenamefont {Santagati}, \citenamefont {Streif},\ and\ \citenamefont {Tautermann}}]{caesura2025bliss}%
  \BibitemOpen
  \bibfield  {author} {\bibinfo {author} {\bibfnamefont {A.}~\bibnamefont {Caesura}}, \bibinfo {author} {\bibfnamefont {C.~L.}\ \bibnamefont {Cortes}}, \bibinfo {author} {\bibfnamefont {W.}~\bibnamefont {Pol}}, \bibinfo {author} {\bibfnamefont {S.}~\bibnamefont {Sim}}, \bibinfo {author} {\bibfnamefont {M.}~\bibnamefont {Steudtner}}, \bibinfo {author} {\bibfnamefont {G.-L.~R.}\ \bibnamefont {Anselmetti}}, \bibinfo {author} {\bibfnamefont {M.}~\bibnamefont {Degroote}}, \bibinfo {author} {\bibfnamefont {N.}~\bibnamefont {Moll}}, \bibinfo {author} {\bibfnamefont {R.}~\bibnamefont {Santagati}}, \bibinfo {author} {\bibfnamefont {M.}~\bibnamefont {Streif}},\ and\ \bibinfo {author} {\bibfnamefont {C.~S.}\ \bibnamefont {Tautermann}},\ }\href {https://doi.org/10.1103/PRXQuantum.6.030337} {\bibfield  {journal} {\bibinfo  {journal} {PRX Quantum}\ }\textbf {\bibinfo {volume} {6}},\ \bibinfo {pages} {030337} (\bibinfo {year} {2025})},\ \Eprint {https://arxiv.org/abs/2501.06165} {arXiv:2501.06165 [quant-ph]} \BibitemShut
  {NoStop}%
\bibitem [{\citenamefont {Loaiza}\ \emph {et~al.}(2025)\citenamefont {Loaiza}, \citenamefont {Brahmachari},\ and\ \citenamefont {Izmaylov}}]{loaiza2024majoranatensordecomposition}%
  \BibitemOpen
  \bibfield  {author} {\bibinfo {author} {\bibfnamefont {I.}~\bibnamefont {Loaiza}}, \bibinfo {author} {\bibfnamefont {A.~S.}\ \bibnamefont {Brahmachari}},\ and\ \bibinfo {author} {\bibfnamefont {A.~F.}\ \bibnamefont {Izmaylov}},\ }\href {https://doi.org/10.1088/2058-9565/add9c1} {\bibfield  {journal} {\bibinfo  {journal} {Quantum Science and Technology}\ }\textbf {\bibinfo {volume} {10}},\ \bibinfo {pages} {035035} (\bibinfo {year} {2025})},\ \Eprint {https://arxiv.org/abs/2407.06571} {arXiv:2407.06571 [quant-ph]} \BibitemShut {NoStop}%
\bibitem [{\citenamefont {Davis}\ and\ \citenamefont {Kahan}(1970)}]{davis1970rotation}%
  \BibitemOpen
  \bibfield  {author} {\bibinfo {author} {\bibfnamefont {C.}~\bibnamefont {Davis}}\ and\ \bibinfo {author} {\bibfnamefont {W.~M.}\ \bibnamefont {Kahan}},\ }\href {https://doi.org/10.1137/0707001} {\bibfield  {journal} {\bibinfo  {journal} {SIAM Journal on Numerical Analysis}\ }\textbf {\bibinfo {volume} {7}},\ \bibinfo {pages} {1} (\bibinfo {year} {1970})}\BibitemShut {NoStop}%
\bibitem [{\citenamefont {Stewart}\ and\ \citenamefont {Sun}(1990)}]{stewart1990matrixperturbation}%
  \BibitemOpen
  \bibfield  {author} {\bibinfo {author} {\bibfnamefont {G.~W.}\ \bibnamefont {Stewart}}\ and\ \bibinfo {author} {\bibfnamefont {J.-g.}\ \bibnamefont {Sun}},\ }\href@noop {} {\emph {\bibinfo {title} {Matrix Perturbation Theory}}}\ (\bibinfo  {publisher} {Academic Press},\ \bibinfo {year} {1990})\BibitemShut {NoStop}%
\bibitem [{\citenamefont {G{\"u}nther}\ \emph {et~al.}(2026{\natexlab{b}})\citenamefont {G{\"u}nther}, \citenamefont {Witteveen}, \citenamefont {Schmidhuber}, \citenamefont {Miller}, \citenamefont {Christandl},\ and\ \citenamefont {Harrow}}]{gunther2026phaseestimation}%
  \BibitemOpen
  \bibfield  {author} {\bibinfo {author} {\bibfnamefont {J.}~\bibnamefont {G{\"u}nther}}, \bibinfo {author} {\bibfnamefont {F.}~\bibnamefont {Witteveen}}, \bibinfo {author} {\bibfnamefont {A.}~\bibnamefont {Schmidhuber}}, \bibinfo {author} {\bibfnamefont {M.}~\bibnamefont {Miller}}, \bibinfo {author} {\bibfnamefont {M.}~\bibnamefont {Christandl}},\ and\ \bibinfo {author} {\bibfnamefont {A.~W.}\ \bibnamefont {Harrow}},\ }\href {https://doi.org/10.1103/ynxb-p2xq} {\bibfield  {journal} {\bibinfo  {journal} {PRX Quantum}\ }\textbf {\bibinfo {volume} {7}},\ \bibinfo {pages} {020332} (\bibinfo {year} {2026}{\natexlab{b}})},\ \Eprint {https://arxiv.org/abs/2503.05647} {arXiv:2503.05647 [quant-ph]} \BibitemShut {NoStop}%
\bibitem [{\citenamefont {Bartels}\ and\ \citenamefont {Stewart}(1972)}]{bartels1972solution}%
  \BibitemOpen
  \bibfield  {author} {\bibinfo {author} {\bibfnamefont {R.~H.}\ \bibnamefont {Bartels}}\ and\ \bibinfo {author} {\bibfnamefont {G.~W.}\ \bibnamefont {Stewart}},\ }\href {https://doi.org/10.1145/361573.361582} {\bibfield  {journal} {\bibinfo  {journal} {Communications of the ACM}\ }\textbf {\bibinfo {volume} {15}},\ \bibinfo {pages} {820} (\bibinfo {year} {1972})}\BibitemShut {NoStop}%
\bibitem [{\citenamefont {Otten}\ \emph {et~al.}(2022)\citenamefont {Otten}, \citenamefont {Hermes}, \citenamefont {Pandharkar}, \citenamefont {Alexeev}, \citenamefont {Gray},\ and\ \citenamefont {Gagliardi}}]{otten2022localized}%
  \BibitemOpen
  \bibfield  {author} {\bibinfo {author} {\bibfnamefont {M.}~\bibnamefont {Otten}}, \bibinfo {author} {\bibfnamefont {M.~R.}\ \bibnamefont {Hermes}}, \bibinfo {author} {\bibfnamefont {R.}~\bibnamefont {Pandharkar}}, \bibinfo {author} {\bibfnamefont {Y.}~\bibnamefont {Alexeev}}, \bibinfo {author} {\bibfnamefont {S.~K.}\ \bibnamefont {Gray}},\ and\ \bibinfo {author} {\bibfnamefont {L.}~\bibnamefont {Gagliardi}},\ }\href {https://doi.org/10.1021/acs.jctc.2c00388} {\bibfield  {journal} {\bibinfo  {journal} {Journal of Chemical Theory and Computation}\ }\textbf {\bibinfo {volume} {18}},\ \bibinfo {pages} {7205} (\bibinfo {year} {2022})},\ \Eprint {https://arxiv.org/abs/2203.02012} {arXiv:2203.02012 [quant-ph]} \BibitemShut {NoStop}%
\bibitem [{\citenamefont {Li}(2025)}]{li2025emo}%
  \BibitemOpen
  \bibfield  {author} {\bibinfo {author} {\bibfnamefont {Z.}~\bibnamefont {Li}},\ }\href {https://arxiv.org/abs/2506.13386} {\bibinfo {title} {Entanglement-minimized orbitals enable faster quantum simulation of molecules}} (\bibinfo {year} {2025}),\ \Eprint {https://arxiv.org/abs/2506.13386} {arXiv:2506.13386 [quant-ph]} \BibitemShut {NoStop}%
\bibitem [{\citenamefont {Zhang}\ and\ \citenamefont {Otten}(2026)}]{zhang2026coreoptimizedorbitals}%
  \BibitemOpen
  \bibfield  {author} {\bibinfo {author} {\bibfnamefont {H.}~\bibnamefont {Zhang}}\ and\ \bibinfo {author} {\bibfnamefont {M.}~\bibnamefont {Otten}},\ }\href {https://doi.org/10.48550/arXiv.2605.22977} {\bibinfo {title} {Absorbing many-body correlations into core-optimized orbitals}} (\bibinfo {year} {2026}),\ \Eprint {https://arxiv.org/abs/2605.22977} {arXiv:2605.22977 [physics.chem-ph]} \BibitemShut {NoStop}%
\bibitem [{\citenamefont {Kieburg}\ \emph {et~al.}(2016)\citenamefont {Kieburg}, \citenamefont {Kuijlaars},\ and\ \citenamefont {Stivigny}}]{truncated_unitaries_2015}%
  \BibitemOpen
  \bibfield  {author} {\bibinfo {author} {\bibfnamefont {M.}~\bibnamefont {Kieburg}}, \bibinfo {author} {\bibfnamefont {A.~B.~J.}\ \bibnamefont {Kuijlaars}},\ and\ \bibinfo {author} {\bibfnamefont {D.}~\bibnamefont {Stivigny}},\ }\href {https://doi.org/10.1093/imrn/rnv242} {\bibfield  {journal} {\bibinfo  {journal} {International Mathematics Research Notices}\ }\textbf {\bibinfo {volume} {2016}},\ \bibinfo {pages} {3392} (\bibinfo {year} {2016})},\ \Eprint {https://arxiv.org/abs/1501.03910} {arXiv:1501.03910 [math-ph]} \BibitemShut {NoStop}%
\bibitem [{\citenamefont {Lubinski}\ \emph {et~al.}(2023)\citenamefont {Lubinski}, \citenamefont {Johri}, \citenamefont {Varosy}, \citenamefont {Coleman}, \citenamefont {Zhao}, \citenamefont {Necaise}, \citenamefont {Baldwin}, \citenamefont {Mayer},\ and\ \citenamefont {Proctor}}]{lubinski2021applicationoriented}%
  \BibitemOpen
  \bibfield  {author} {\bibinfo {author} {\bibfnamefont {T.}~\bibnamefont {Lubinski}}, \bibinfo {author} {\bibfnamefont {S.}~\bibnamefont {Johri}}, \bibinfo {author} {\bibfnamefont {P.}~\bibnamefont {Varosy}}, \bibinfo {author} {\bibfnamefont {J.}~\bibnamefont {Coleman}}, \bibinfo {author} {\bibfnamefont {L.}~\bibnamefont {Zhao}}, \bibinfo {author} {\bibfnamefont {J.}~\bibnamefont {Necaise}}, \bibinfo {author} {\bibfnamefont {C.~H.}\ \bibnamefont {Baldwin}}, \bibinfo {author} {\bibfnamefont {K.}~\bibnamefont {Mayer}},\ and\ \bibinfo {author} {\bibfnamefont {T.}~\bibnamefont {Proctor}},\ }\href {https://doi.org/10.1109/TQE.2023.3253761} {\bibfield  {journal} {\bibinfo  {journal} {IEEE Transactions on Quantum Engineering}\ }\textbf {\bibinfo {volume} {4}},\ \bibinfo {pages} {1} (\bibinfo {year} {2023})},\ \Eprint {https://arxiv.org/abs/2110.03137} {arXiv:2110.03137 [quant-ph]} \BibitemShut {NoStop}%
\bibitem [{\citenamefont {Wang}\ \emph {et~al.}(2025)\citenamefont {Wang}, \citenamefont {Cantin}, \citenamefont {Patel}, \citenamefont {Loaiza}, \citenamefont {Huang},\ and\ \citenamefont {Izmaylov}}]{wang2026plantedsolutions}%
  \BibitemOpen
  \bibfield  {author} {\bibinfo {author} {\bibfnamefont {L.}~\bibnamefont {Wang}}, \bibinfo {author} {\bibfnamefont {J.~T.}\ \bibnamefont {Cantin}}, \bibinfo {author} {\bibfnamefont {S.}~\bibnamefont {Patel}}, \bibinfo {author} {\bibfnamefont {I.}~\bibnamefont {Loaiza}}, \bibinfo {author} {\bibfnamefont {R.}~\bibnamefont {Huang}},\ and\ \bibinfo {author} {\bibfnamefont {A.~F.}\ \bibnamefont {Izmaylov}},\ }\href {https://doi.org/10.48550/arXiv.2507.15166} {\bibinfo {title} {Planted solutions in quantum chemistry: Generating non-trivial hamiltonians with known ground states}} (\bibinfo {year} {2025}),\ \Eprint {https://arxiv.org/abs/2507.15166} {arXiv:2507.15166 [quant-ph]} \BibitemShut {NoStop}%
\bibitem [{\citenamefont {Watts}\ \emph {et~al.}(2024)\citenamefont {Watts}, \citenamefont {Otten}, \citenamefont {Necaise}, \citenamefont {Nguyen}, \citenamefont {Link}, \citenamefont {Williams}, \citenamefont {Sanders}, \citenamefont {Elman}, \citenamefont {Kieferova}, \citenamefont {Bremner}, \citenamefont {Morrell}, \citenamefont {Elenewski}, \citenamefont {Johnson}, \citenamefont {Mathieson}, \citenamefont {Obenland}, \citenamefont {Sundareswara},\ and\ \citenamefont {Holmes}}]{watts2024fullerene}%
  \BibitemOpen
  \bibfield  {author} {\bibinfo {author} {\bibfnamefont {T.~W.}\ \bibnamefont {Watts}}, \bibinfo {author} {\bibfnamefont {M.}~\bibnamefont {Otten}}, \bibinfo {author} {\bibfnamefont {J.~T.}\ \bibnamefont {Necaise}}, \bibinfo {author} {\bibfnamefont {N.}~\bibnamefont {Nguyen}}, \bibinfo {author} {\bibfnamefont {B.}~\bibnamefont {Link}}, \bibinfo {author} {\bibfnamefont {K.~S.}\ \bibnamefont {Williams}}, \bibinfo {author} {\bibfnamefont {Y.~R.}\ \bibnamefont {Sanders}}, \bibinfo {author} {\bibfnamefont {S.~J.}\ \bibnamefont {Elman}}, \bibinfo {author} {\bibfnamefont {M.}~\bibnamefont {Kieferova}}, \bibinfo {author} {\bibfnamefont {M.~J.}\ \bibnamefont {Bremner}}, \bibinfo {author} {\bibfnamefont {K.~J.}\ \bibnamefont {Morrell}}, \bibinfo {author} {\bibfnamefont {J.~E.}\ \bibnamefont {Elenewski}}, \bibinfo {author} {\bibfnamefont {S.~D.}\ \bibnamefont {Johnson}}, \bibinfo {author} {\bibfnamefont {L.}~\bibnamefont {Mathieson}}, \bibinfo {author} {\bibfnamefont {K.~M.}\ \bibnamefont {Obenland}}, \bibinfo {author}
  {\bibfnamefont {R.}~\bibnamefont {Sundareswara}},\ and\ \bibinfo {author} {\bibfnamefont {A.}~\bibnamefont {Holmes}},\ }\href {https://doi.org/10.48550/arXiv.2408.13244} {\bibinfo {title} {Fullerene-encapsulated cyclic ozone for the next generation of nano-sized propellants via quantum computation}} (\bibinfo {year} {2024}),\ \Eprint {https://arxiv.org/abs/2408.13244} {arXiv:2408.13244 [quant-ph]} \BibitemShut {NoStop}%
\bibitem [{\citenamefont {Verstraete}\ \emph {et~al.}(2009{\natexlab{a}})\citenamefont {Verstraete}, \citenamefont {Cirac},\ and\ \citenamefont {Latorre}}]{PhysRevA.79.032316}%
  \BibitemOpen
  \bibfield  {author} {\bibinfo {author} {\bibfnamefont {F.}~\bibnamefont {Verstraete}}, \bibinfo {author} {\bibfnamefont {J.~I.}\ \bibnamefont {Cirac}},\ and\ \bibinfo {author} {\bibfnamefont {J.~I.}\ \bibnamefont {Latorre}},\ }\href {https://doi.org/10.1103/PhysRevA.79.032316} {\bibfield  {journal} {\bibinfo  {journal} {Phys. Rev. A}\ }\textbf {\bibinfo {volume} {79}},\ \bibinfo {pages} {032316} (\bibinfo {year} {2009}{\natexlab{a}})}\BibitemShut {NoStop}%
\bibitem [{\citenamefont {Bauman}\ \emph {et~al.}(2019)\citenamefont {Bauman}, \citenamefont {Bylaska}, \citenamefont {Krishnamoorthy}, \citenamefont {Low}, \citenamefont {Wiebe}, \citenamefont {Granade}, \citenamefont {Roetteler}, \citenamefont {Troyer},\ and\ \citenamefont {Kowalski}}]{bauman2019downfolding}%
  \BibitemOpen
  \bibfield  {author} {\bibinfo {author} {\bibfnamefont {N.~P.}\ \bibnamefont {Bauman}}, \bibinfo {author} {\bibfnamefont {E.~J.}\ \bibnamefont {Bylaska}}, \bibinfo {author} {\bibfnamefont {S.}~\bibnamefont {Krishnamoorthy}}, \bibinfo {author} {\bibfnamefont {G.~H.}\ \bibnamefont {Low}}, \bibinfo {author} {\bibfnamefont {N.}~\bibnamefont {Wiebe}}, \bibinfo {author} {\bibfnamefont {C.~E.}\ \bibnamefont {Granade}}, \bibinfo {author} {\bibfnamefont {M.}~\bibnamefont {Roetteler}}, \bibinfo {author} {\bibfnamefont {M.}~\bibnamefont {Troyer}},\ and\ \bibinfo {author} {\bibfnamefont {K.}~\bibnamefont {Kowalski}},\ }\href {https://doi.org/10.1063/1.5094643} {\bibfield  {journal} {\bibinfo  {journal} {The Journal of Chemical Physics}\ }\textbf {\bibinfo {volume} {151}},\ \bibinfo {pages} {014107} (\bibinfo {year} {2019})},\ \Eprint {https://arxiv.org/abs/1902.01553} {arXiv:1902.01553 [quant-ph]} \BibitemShut {NoStop}%
\bibitem [{\citenamefont {Vorwerk}\ \emph {et~al.}(2022)\citenamefont {Vorwerk}, \citenamefont {Sheng}, \citenamefont {Govoni}, \citenamefont {Huang},\ and\ \citenamefont {Galli}}]{vorwerk2022quantumembedding}%
  \BibitemOpen
  \bibfield  {author} {\bibinfo {author} {\bibfnamefont {C.}~\bibnamefont {Vorwerk}}, \bibinfo {author} {\bibfnamefont {N.}~\bibnamefont {Sheng}}, \bibinfo {author} {\bibfnamefont {M.}~\bibnamefont {Govoni}}, \bibinfo {author} {\bibfnamefont {B.}~\bibnamefont {Huang}},\ and\ \bibinfo {author} {\bibfnamefont {G.}~\bibnamefont {Galli}},\ }\href {https://doi.org/10.1038/s43588-022-00279-0} {\bibfield  {journal} {\bibinfo  {journal} {Nature Computational Science}\ }\textbf {\bibinfo {volume} {2}},\ \bibinfo {pages} {424} (\bibinfo {year} {2022})},\ \Eprint {https://arxiv.org/abs/2105.04736} {arXiv:2105.04736 [quant-ph]} \BibitemShut {NoStop}%
\bibitem [{\citenamefont {Ma}\ \emph {et~al.}(2020)\citenamefont {Ma}, \citenamefont {Govoni},\ and\ \citenamefont {Galli}}]{ma2020quantumsimulations}%
  \BibitemOpen
  \bibfield  {author} {\bibinfo {author} {\bibfnamefont {H.}~\bibnamefont {Ma}}, \bibinfo {author} {\bibfnamefont {M.}~\bibnamefont {Govoni}},\ and\ \bibinfo {author} {\bibfnamefont {G.}~\bibnamefont {Galli}},\ }\href {https://doi.org/10.1038/s41524-020-00353-z} {\bibfield  {journal} {\bibinfo  {journal} {npj Computational Materials}\ }\textbf {\bibinfo {volume} {6}},\ \bibinfo {pages} {85} (\bibinfo {year} {2020})},\ \Eprint {https://arxiv.org/abs/2002.11173} {arXiv:2002.11173 [quant-ph]} \BibitemShut {NoStop}%
\bibitem [{\citenamefont {Merz}\ \emph {et~al.}(2026)\citenamefont {Merz}, \citenamefont {Shajan}, \citenamefont {Kaliakin}, \citenamefont {Liang}, \citenamefont {Otsuka}, \citenamefont {Shirakawa}, \citenamefont {Broers}, \citenamefont {Xu}, \citenamefont {Tsuji}, \citenamefont {Sato}, \citenamefont {Yunoki}, \citenamefont {Wakizaka}, \citenamefont {Kawashima}, \citenamefont {Doi}, \citenamefont {Itoko}, \citenamefont {Horii}, \citenamefont {Pellegrini}, \citenamefont {Robledo~Moreno}, \citenamefont {Sung}, \citenamefont {Fejer}, \citenamefont {Walkup}, \citenamefont {Seelam},\ and\ \citenamefont {Motta}}]{merz2026crossing12000atom}%
  \BibitemOpen
  \bibfield  {author} {\bibinfo {author} {\bibfnamefont {K.~M.~J.}\ \bibnamefont {Merz}}, \bibinfo {author} {\bibfnamefont {A.}~\bibnamefont {Shajan}}, \bibinfo {author} {\bibfnamefont {D.}~\bibnamefont {Kaliakin}}, \bibinfo {author} {\bibfnamefont {F.}~\bibnamefont {Liang}}, \bibinfo {author} {\bibfnamefont {Y.}~\bibnamefont {Otsuka}}, \bibinfo {author} {\bibfnamefont {T.}~\bibnamefont {Shirakawa}}, \bibinfo {author} {\bibfnamefont {L.}~\bibnamefont {Broers}}, \bibinfo {author} {\bibfnamefont {H.}~\bibnamefont {Xu}}, \bibinfo {author} {\bibfnamefont {M.}~\bibnamefont {Tsuji}}, \bibinfo {author} {\bibfnamefont {M.}~\bibnamefont {Sato}}, \bibinfo {author} {\bibfnamefont {S.}~\bibnamefont {Yunoki}}, \bibinfo {author} {\bibfnamefont {R.}~\bibnamefont {Wakizaka}}, \bibinfo {author} {\bibfnamefont {Y.}~\bibnamefont {Kawashima}}, \bibinfo {author} {\bibfnamefont {J.}~\bibnamefont {Doi}}, \bibinfo {author} {\bibfnamefont {T.}~\bibnamefont {Itoko}}, \bibinfo {author} {\bibfnamefont {H.}~\bibnamefont {Horii}},
  \bibinfo {author} {\bibfnamefont {T.}~\bibnamefont {Pellegrini}}, \bibinfo {author} {\bibfnamefont {J.}~\bibnamefont {Robledo~Moreno}}, \bibinfo {author} {\bibfnamefont {K.~J.}\ \bibnamefont {Sung}}, \bibinfo {author} {\bibfnamefont {E.}~\bibnamefont {Fejer}}, \bibinfo {author} {\bibfnamefont {R.}~\bibnamefont {Walkup}}, \bibinfo {author} {\bibfnamefont {S.}~\bibnamefont {Seelam}},\ and\ \bibinfo {author} {\bibfnamefont {M.}~\bibnamefont {Motta}},\ }\href {https://doi.org/10.48550/arXiv.2605.01138} {\bibinfo {title} {Crossing the 12,000-atom barrier with heterogeneous quantum-classical supercomputing: quantum chemistry of protein-ligand complexes}} (\bibinfo {year} {2026}),\ \Eprint {https://arxiv.org/abs/2605.01138} {arXiv:2605.01138 [quant-ph]} \BibitemShut {NoStop}%
\bibitem [{\citenamefont {L{\"o}wdin}(1950)}]{lowdin1950nonorthogonality}%
  \BibitemOpen
  \bibfield  {author} {\bibinfo {author} {\bibfnamefont {P.-O.}\ \bibnamefont {L{\"o}wdin}},\ }\href {https://doi.org/10.1063/1.1747632} {\bibfield  {journal} {\bibinfo  {journal} {The Journal of Chemical Physics}\ }\textbf {\bibinfo {volume} {18}},\ \bibinfo {pages} {365} (\bibinfo {year} {1950})}\BibitemShut {NoStop}%
\bibitem [{\citenamefont {Boys}(1960)}]{boys1960construction}%
  \BibitemOpen
  \bibfield  {author} {\bibinfo {author} {\bibfnamefont {S.~F.}\ \bibnamefont {Boys}},\ }\href {https://doi.org/10.1103/RevModPhys.32.296} {\bibfield  {journal} {\bibinfo  {journal} {Reviews of Modern Physics}\ }\textbf {\bibinfo {volume} {32}},\ \bibinfo {pages} {296} (\bibinfo {year} {1960})}\BibitemShut {NoStop}%
\bibitem [{\citenamefont {Pipek}\ and\ \citenamefont {Mezey}(1989)}]{pipek1989fast}%
  \BibitemOpen
  \bibfield  {author} {\bibinfo {author} {\bibfnamefont {J.}~\bibnamefont {Pipek}}\ and\ \bibinfo {author} {\bibfnamefont {P.~G.}\ \bibnamefont {Mezey}},\ }\href {https://doi.org/10.1063/1.456588} {\bibfield  {journal} {\bibinfo  {journal} {The Journal of Chemical Physics}\ }\textbf {\bibinfo {volume} {90}},\ \bibinfo {pages} {4916} (\bibinfo {year} {1989})}\BibitemShut {NoStop}%
\bibitem [{\citenamefont {Wang}\ and\ \citenamefont {Whitfield}(2023)}]{quiqbox}%
  \BibitemOpen
  \bibfield  {author} {\bibinfo {author} {\bibfnamefont {W.}~\bibnamefont {Wang}}\ and\ \bibinfo {author} {\bibfnamefont {J.~D.}\ \bibnamefont {Whitfield}},\ }\href {https://doi.org/10.1021/acs.jctc.3c00011} {\bibfield  {journal} {\bibinfo  {journal} {Journal of Chemical Theory and Computation}\ }\textbf {\bibinfo {volume} {19}},\ \bibinfo {pages} {8032} (\bibinfo {year} {2023})},\ \bibinfo {note} {pMID: 37924295},\ \Eprint {https://arxiv.org/abs/https://doi.org/10.1021/acs.jctc.3c00011} {https://doi.org/10.1021/acs.jctc.3c00011} \BibitemShut {NoStop}%
\bibitem [{\citenamefont {Hubbard}(1963)}]{hubbard1963electron}%
  \BibitemOpen
  \bibfield  {author} {\bibinfo {author} {\bibfnamefont {J.}~\bibnamefont {Hubbard}},\ }\href {https://doi.org/10.1098/rspa.1963.0204} {\bibfield  {journal} {\bibinfo  {journal} {Proceedings of the Royal Society of London. Series A. Mathematical and Physical Sciences}\ }\textbf {\bibinfo {volume} {276}},\ \bibinfo {pages} {238} (\bibinfo {year} {1963})}\BibitemShut {NoStop}%
\bibitem [{\citenamefont {Arovas}\ \emph {et~al.}(2022)\citenamefont {Arovas}, \citenamefont {Berg}, \citenamefont {Kivelson},\ and\ \citenamefont {Raghu}}]{Arovas2022}%
  \BibitemOpen
  \bibfield  {author} {\bibinfo {author} {\bibfnamefont {D.~P.}\ \bibnamefont {Arovas}}, \bibinfo {author} {\bibfnamefont {E.}~\bibnamefont {Berg}}, \bibinfo {author} {\bibfnamefont {S.~A.}\ \bibnamefont {Kivelson}},\ and\ \bibinfo {author} {\bibfnamefont {S.}~\bibnamefont {Raghu}},\ }\href {https://doi.org/10.1146/annurev-conmatphys-031620-102024} {\bibfield  {journal} {\bibinfo  {journal} {Annual Review of Condensed Matter Physics}\ }\textbf {\bibinfo {volume} {13}},\ \bibinfo {pages} {239–274} (\bibinfo {year} {2022})}\BibitemShut {NoStop}%
\bibitem [{\citenamefont {Liu}\ \emph {et~al.}(2025)\citenamefont {Liu}, \citenamefont {Zhai}, \citenamefont {Peng}, \citenamefont {Gu},\ and\ \citenamefont {Chan}}]{Liu2025}%
  \BibitemOpen
  \bibfield  {author} {\bibinfo {author} {\bibfnamefont {W.-Y.}\ \bibnamefont {Liu}}, \bibinfo {author} {\bibfnamefont {H.}~\bibnamefont {Zhai}}, \bibinfo {author} {\bibfnamefont {R.}~\bibnamefont {Peng}}, \bibinfo {author} {\bibfnamefont {Z.-C.}\ \bibnamefont {Gu}},\ and\ \bibinfo {author} {\bibfnamefont {G.~K.-L.}\ \bibnamefont {Chan}},\ }\bibfield  {journal} {\bibinfo  {journal} {Physical Review Letters}\ }\textbf {\bibinfo {volume} {134}},\ \href {https://doi.org/10.1103/r4q9-4yvj} {10.1103/r4q9-4yvj} (\bibinfo {year} {2025})\BibitemShut {NoStop}%
\bibitem [{\citenamefont {Jiang}\ \emph {et~al.}(2018)\citenamefont {Jiang}, \citenamefont {Sung}, \citenamefont {Kechedzhi}, \citenamefont {Smelyanskiy},\ and\ \citenamefont {Boixo}}]{PhysRevApplied.9.044036}%
  \BibitemOpen
  \bibfield  {author} {\bibinfo {author} {\bibfnamefont {Z.}~\bibnamefont {Jiang}}, \bibinfo {author} {\bibfnamefont {K.~J.}\ \bibnamefont {Sung}}, \bibinfo {author} {\bibfnamefont {K.}~\bibnamefont {Kechedzhi}}, \bibinfo {author} {\bibfnamefont {V.~N.}\ \bibnamefont {Smelyanskiy}},\ and\ \bibinfo {author} {\bibfnamefont {S.}~\bibnamefont {Boixo}},\ }\href {https://doi.org/10.1103/PhysRevApplied.9.044036} {\bibfield  {journal} {\bibinfo  {journal} {Phys. Rev. Appl.}\ }\textbf {\bibinfo {volume} {9}},\ \bibinfo {pages} {044036} (\bibinfo {year} {2018})}\BibitemShut {NoStop}%
\bibitem [{\citenamefont {Alam}\ \emph {et~al.}(2025)\citenamefont {Alam}, \citenamefont {Bosse}, \citenamefont {Čepaitė}, \citenamefont {Chapman}, \citenamefont {Clinton}, \citenamefont {Crichigno}, \citenamefont {Crosson}, \citenamefont {Cubitt}, \citenamefont {Derby}, \citenamefont {Dowinton}, \citenamefont {Faehrmann}, \citenamefont {Flammia}, \citenamefont {Flynn}, \citenamefont {Gambetta}, \citenamefont {García-Patrón}, \citenamefont {Hunter-Gordon}, \citenamefont {Jones}, \citenamefont {Khedkar}, \citenamefont {Klassen}, \citenamefont {Kreshchuk}, \citenamefont {McMullan}, \citenamefont {Mineh}, \citenamefont {Montanaro}, \citenamefont {Mora}, \citenamefont {Morton}, \citenamefont {Patel}, \citenamefont {Rolph}, \citenamefont {Santos}, \citenamefont {Seddon}, \citenamefont {Sheridan}, \citenamefont {Somogyi}, \citenamefont {Svensson}, \citenamefont {Vaishnav}, \citenamefont {Wang},\ and\ \citenamefont {Wright}}]{alam2025programmabledigitalquantumsimulation}%
  \BibitemOpen
  \bibfield  {author} {\bibinfo {author} {\bibfnamefont {F.}~\bibnamefont {Alam}}, \bibinfo {author} {\bibfnamefont {J.~L.}\ \bibnamefont {Bosse}}, \bibinfo {author} {\bibfnamefont {I.}~\bibnamefont {Čepaitė}}, \bibinfo {author} {\bibfnamefont {A.}~\bibnamefont {Chapman}}, \bibinfo {author} {\bibfnamefont {L.}~\bibnamefont {Clinton}}, \bibinfo {author} {\bibfnamefont {M.}~\bibnamefont {Crichigno}}, \bibinfo {author} {\bibfnamefont {E.}~\bibnamefont {Crosson}}, \bibinfo {author} {\bibfnamefont {T.}~\bibnamefont {Cubitt}}, \bibinfo {author} {\bibfnamefont {C.}~\bibnamefont {Derby}}, \bibinfo {author} {\bibfnamefont {O.}~\bibnamefont {Dowinton}}, \bibinfo {author} {\bibfnamefont {P.~K.}\ \bibnamefont {Faehrmann}}, \bibinfo {author} {\bibfnamefont {S.}~\bibnamefont {Flammia}}, \bibinfo {author} {\bibfnamefont {B.}~\bibnamefont {Flynn}}, \bibinfo {author} {\bibfnamefont {F.~M.}\ \bibnamefont {Gambetta}}, \bibinfo {author} {\bibfnamefont {R.}~\bibnamefont {García-Patrón}}, \bibinfo {author} {\bibfnamefont
  {M.}~\bibnamefont {Hunter-Gordon}}, \bibinfo {author} {\bibfnamefont {G.}~\bibnamefont {Jones}}, \bibinfo {author} {\bibfnamefont {A.}~\bibnamefont {Khedkar}}, \bibinfo {author} {\bibfnamefont {J.}~\bibnamefont {Klassen}}, \bibinfo {author} {\bibfnamefont {M.}~\bibnamefont {Kreshchuk}}, \bibinfo {author} {\bibfnamefont {E.~H.}\ \bibnamefont {McMullan}}, \bibinfo {author} {\bibfnamefont {L.}~\bibnamefont {Mineh}}, \bibinfo {author} {\bibfnamefont {A.}~\bibnamefont {Montanaro}}, \bibinfo {author} {\bibfnamefont {C.}~\bibnamefont {Mora}}, \bibinfo {author} {\bibfnamefont {J.~J.~L.}\ \bibnamefont {Morton}}, \bibinfo {author} {\bibfnamefont {D.}~\bibnamefont {Patel}}, \bibinfo {author} {\bibfnamefont {P.}~\bibnamefont {Rolph}}, \bibinfo {author} {\bibfnamefont {R.~A.}\ \bibnamefont {Santos}}, \bibinfo {author} {\bibfnamefont {J.~R.}\ \bibnamefont {Seddon}}, \bibinfo {author} {\bibfnamefont {E.}~\bibnamefont {Sheridan}}, \bibinfo {author} {\bibfnamefont {W.}~\bibnamefont {Somogyi}}, \bibinfo {author}
  {\bibfnamefont {M.}~\bibnamefont {Svensson}}, \bibinfo {author} {\bibfnamefont {N.}~\bibnamefont {Vaishnav}}, \bibinfo {author} {\bibfnamefont {S.~Y.}\ \bibnamefont {Wang}},\ and\ \bibinfo {author} {\bibfnamefont {G.}~\bibnamefont {Wright}},\ }\href {https://arxiv.org/abs/2510.26845} {\bibinfo {title} {Programmable digital quantum simulation of 2d fermi-hubbard dynamics using 72 superconducting qubits}} (\bibinfo {year} {2025}),\ \Eprint {https://arxiv.org/abs/2510.26845} {arXiv:2510.26845 [quant-ph]} \BibitemShut {NoStop}%
\bibitem [{\citenamefont {Verstraete}\ \emph {et~al.}(2009{\natexlab{b}})\citenamefont {Verstraete}, \citenamefont {Cirac},\ and\ \citenamefont {Latorre}}]{Verstraete2009}%
  \BibitemOpen
  \bibfield  {author} {\bibinfo {author} {\bibfnamefont {F.}~\bibnamefont {Verstraete}}, \bibinfo {author} {\bibfnamefont {J.~I.}\ \bibnamefont {Cirac}},\ and\ \bibinfo {author} {\bibfnamefont {J.~I.}\ \bibnamefont {Latorre}},\ }\bibfield  {journal} {\bibinfo  {journal} {Physical Review A}\ }\textbf {\bibinfo {volume} {79}},\ \href {https://doi.org/10.1103/physreva.79.032316} {10.1103/physreva.79.032316} (\bibinfo {year} {2009}{\natexlab{b}})\BibitemShut {NoStop}%
\bibitem [{\citenamefont {Wecker}\ \emph {et~al.}(2015)\citenamefont {Wecker}, \citenamefont {Hastings}, \citenamefont {Wiebe}, \citenamefont {Clark}, \citenamefont {Nayak},\ and\ \citenamefont {Troyer}}]{wecker2015solving}%
  \BibitemOpen
  \bibfield  {author} {\bibinfo {author} {\bibfnamefont {D.}~\bibnamefont {Wecker}}, \bibinfo {author} {\bibfnamefont {M.~B.}\ \bibnamefont {Hastings}}, \bibinfo {author} {\bibfnamefont {N.}~\bibnamefont {Wiebe}}, \bibinfo {author} {\bibfnamefont {B.~K.}\ \bibnamefont {Clark}}, \bibinfo {author} {\bibfnamefont {C.}~\bibnamefont {Nayak}},\ and\ \bibinfo {author} {\bibfnamefont {M.}~\bibnamefont {Troyer}},\ }\href {https://doi.org/10.1103/PhysRevA.92.062318} {\bibfield  {journal} {\bibinfo  {journal} {Physical Review A}\ }\textbf {\bibinfo {volume} {92}},\ \bibinfo {pages} {062318} (\bibinfo {year} {2015})},\ \Eprint {https://arxiv.org/abs/1506.05135} {arXiv:1506.05135 [quant-ph]} \BibitemShut {NoStop}%
\bibitem [{\citenamefont {Babbush}\ \emph {et~al.}(2018{\natexlab{b}})\citenamefont {Babbush}, \citenamefont {Wiebe}, \citenamefont {McClean}, \citenamefont {McClain}, \citenamefont {Neven},\ and\ \citenamefont {Chan}}]{babbush2018lowdepth}%
  \BibitemOpen
  \bibfield  {author} {\bibinfo {author} {\bibfnamefont {R.}~\bibnamefont {Babbush}}, \bibinfo {author} {\bibfnamefont {N.}~\bibnamefont {Wiebe}}, \bibinfo {author} {\bibfnamefont {J.}~\bibnamefont {McClean}}, \bibinfo {author} {\bibfnamefont {J.}~\bibnamefont {McClain}}, \bibinfo {author} {\bibfnamefont {H.}~\bibnamefont {Neven}},\ and\ \bibinfo {author} {\bibfnamefont {G.~K.-L.}\ \bibnamefont {Chan}},\ }\href {https://doi.org/10.1103/PhysRevX.8.011044} {\bibfield  {journal} {\bibinfo  {journal} {Physical Review X}\ }\textbf {\bibinfo {volume} {8}},\ \bibinfo {pages} {011044} (\bibinfo {year} {2018}{\natexlab{b}})}\BibitemShut {NoStop}%
\bibitem [{\citenamefont {Valiant}(2001)}]{Valiant2001}%
  \BibitemOpen
  \bibfield  {author} {\bibinfo {author} {\bibfnamefont {L.~G.}\ \bibnamefont {Valiant}},\ }in\ \href {https://doi.org/10.1145/380752.380785} {\emph {\bibinfo {booktitle} {Proceedings of the thirty-third annual ACM symposium on Theory of computing}}},\ \bibinfo {series and number} {STOC01}\ (\bibinfo  {publisher} {ACM},\ \bibinfo {year} {2001})\ p.\ \bibinfo {pages} {114–123}\BibitemShut {NoStop}%
\bibitem [{\citenamefont {Terhal}\ and\ \citenamefont {DiVincenzo}(2002)}]{PhysRevA.65.032325}%
  \BibitemOpen
  \bibfield  {author} {\bibinfo {author} {\bibfnamefont {B.~M.}\ \bibnamefont {Terhal}}\ and\ \bibinfo {author} {\bibfnamefont {D.~P.}\ \bibnamefont {DiVincenzo}},\ }\href {https://doi.org/10.1103/PhysRevA.65.032325} {\bibfield  {journal} {\bibinfo  {journal} {Phys. Rev. A}\ }\textbf {\bibinfo {volume} {65}},\ \bibinfo {pages} {032325} (\bibinfo {year} {2002})}\BibitemShut {NoStop}%
\bibitem [{\citenamefont {Hebenstreit}\ \emph {et~al.}(2020)\citenamefont {Hebenstreit}, \citenamefont {Jozsa}, \citenamefont {Kraus},\ and\ \citenamefont {Strelchuk}}]{PhysRevA.102.052604}%
  \BibitemOpen
  \bibfield  {author} {\bibinfo {author} {\bibfnamefont {M.}~\bibnamefont {Hebenstreit}}, \bibinfo {author} {\bibfnamefont {R.}~\bibnamefont {Jozsa}}, \bibinfo {author} {\bibfnamefont {B.}~\bibnamefont {Kraus}},\ and\ \bibinfo {author} {\bibfnamefont {S.}~\bibnamefont {Strelchuk}},\ }\href {https://doi.org/10.1103/PhysRevA.102.052604} {\bibfield  {journal} {\bibinfo  {journal} {Phys. Rev. A}\ }\textbf {\bibinfo {volume} {102}},\ \bibinfo {pages} {052604} (\bibinfo {year} {2020})}\BibitemShut {NoStop}%
\bibitem [{\citenamefont {Hebenstreit}\ \emph {et~al.}(2019)\citenamefont {Hebenstreit}, \citenamefont {Jozsa}, \citenamefont {Kraus}, \citenamefont {Strelchuk},\ and\ \citenamefont {Yoganathan}}]{PhysRevLett.123.080503}%
  \BibitemOpen
  \bibfield  {author} {\bibinfo {author} {\bibfnamefont {M.}~\bibnamefont {Hebenstreit}}, \bibinfo {author} {\bibfnamefont {R.}~\bibnamefont {Jozsa}}, \bibinfo {author} {\bibfnamefont {B.}~\bibnamefont {Kraus}}, \bibinfo {author} {\bibfnamefont {S.}~\bibnamefont {Strelchuk}},\ and\ \bibinfo {author} {\bibfnamefont {M.}~\bibnamefont {Yoganathan}},\ }\href {https://doi.org/10.1103/PhysRevLett.123.080503} {\bibfield  {journal} {\bibinfo  {journal} {Phys. Rev. Lett.}\ }\textbf {\bibinfo {volume} {123}},\ \bibinfo {pages} {080503} (\bibinfo {year} {2019})}\BibitemShut {NoStop}%
\bibitem [{\citenamefont {Gottesman}(1998)}]{gottesman1998heisenberg}%
  \BibitemOpen
  \bibfield  {author} {\bibinfo {author} {\bibfnamefont {D.}~\bibnamefont {Gottesman}},\ }\href {https://arxiv.org/abs/quant-ph/9807006} {\bibinfo {title} {The heisenberg representation of quantum computers}} (\bibinfo {year} {1998}),\ \Eprint {https://arxiv.org/abs/quant-ph/9807006} {arXiv:quant-ph/9807006 [quant-ph]} \BibitemShut {NoStop}%
\bibitem [{\citenamefont {Aaronson}\ and\ \citenamefont {Gottesman}(2004)}]{aaronson2004improved}%
  \BibitemOpen
  \bibfield  {author} {\bibinfo {author} {\bibfnamefont {S.}~\bibnamefont {Aaronson}}\ and\ \bibinfo {author} {\bibfnamefont {D.}~\bibnamefont {Gottesman}},\ }\href {https://doi.org/10.1103/PhysRevA.70.052328} {\bibfield  {journal} {\bibinfo  {journal} {Physical Review A}\ }\textbf {\bibinfo {volume} {70}},\ \bibinfo {pages} {052328} (\bibinfo {year} {2004})},\ \Eprint {https://arxiv.org/abs/quant-ph/0406196} {arXiv:quant-ph/0406196 [quant-ph]} \BibitemShut {NoStop}%
\bibitem [{\citenamefont {Bravyi}\ \emph {et~al.}(2016)\citenamefont {Bravyi}, \citenamefont {Smith},\ and\ \citenamefont {Smolin}}]{bravyi2016trading}%
  \BibitemOpen
  \bibfield  {author} {\bibinfo {author} {\bibfnamefont {S.}~\bibnamefont {Bravyi}}, \bibinfo {author} {\bibfnamefont {G.}~\bibnamefont {Smith}},\ and\ \bibinfo {author} {\bibfnamefont {J.~A.}\ \bibnamefont {Smolin}},\ }\href {https://doi.org/10.1103/PhysRevX.6.021043} {\bibfield  {journal} {\bibinfo  {journal} {Physical Review X}\ }\textbf {\bibinfo {volume} {6}},\ \bibinfo {pages} {021043} (\bibinfo {year} {2016})},\ \Eprint {https://arxiv.org/abs/1506.01396} {arXiv:1506.01396 [quant-ph]} \BibitemShut {NoStop}%
\bibitem [{\citenamefont {Bravyi}\ and\ \citenamefont {Gosset}(2016)}]{bravyi2016improved}%
  \BibitemOpen
  \bibfield  {author} {\bibinfo {author} {\bibfnamefont {S.}~\bibnamefont {Bravyi}}\ and\ \bibinfo {author} {\bibfnamefont {D.}~\bibnamefont {Gosset}},\ }\href {https://doi.org/10.1103/PhysRevLett.116.250501} {\bibfield  {journal} {\bibinfo  {journal} {Physical Review Letters}\ }\textbf {\bibinfo {volume} {116}},\ \bibinfo {pages} {250501} (\bibinfo {year} {2016})},\ \Eprint {https://arxiv.org/abs/1601.07601} {arXiv:1601.07601 [quant-ph]} \BibitemShut {NoStop}%
\bibitem [{\citenamefont {Bravyi}\ \emph {et~al.}(2019)\citenamefont {Bravyi}, \citenamefont {Browne}, \citenamefont {Calpin}, \citenamefont {Campbell}, \citenamefont {Gosset},\ and\ \citenamefont {Howard}}]{bravyi2019simulation}%
  \BibitemOpen
  \bibfield  {author} {\bibinfo {author} {\bibfnamefont {S.}~\bibnamefont {Bravyi}}, \bibinfo {author} {\bibfnamefont {D.}~\bibnamefont {Browne}}, \bibinfo {author} {\bibfnamefont {P.}~\bibnamefont {Calpin}}, \bibinfo {author} {\bibfnamefont {E.}~\bibnamefont {Campbell}}, \bibinfo {author} {\bibfnamefont {D.}~\bibnamefont {Gosset}},\ and\ \bibinfo {author} {\bibfnamefont {M.}~\bibnamefont {Howard}},\ }\href {https://doi.org/10.22331/q-2019-09-02-181} {\bibfield  {journal} {\bibinfo  {journal} {Quantum}\ }\textbf {\bibinfo {volume} {3}},\ \bibinfo {pages} {181} (\bibinfo {year} {2019})},\ \Eprint {https://arxiv.org/abs/1808.00128} {arXiv:1808.00128 [quant-ph]} \BibitemShut {NoStop}%
\bibitem [{\citenamefont {Reardon-Smith}\ \emph {et~al.}(2024)\citenamefont {Reardon-Smith}, \citenamefont {Oszmaniec},\ and\ \citenamefont {Korzekwa}}]{reardonsmith2024improved}%
  \BibitemOpen
  \bibfield  {author} {\bibinfo {author} {\bibfnamefont {O.}~\bibnamefont {Reardon-Smith}}, \bibinfo {author} {\bibfnamefont {M.}~\bibnamefont {Oszmaniec}},\ and\ \bibinfo {author} {\bibfnamefont {K.}~\bibnamefont {Korzekwa}},\ }\href {https://doi.org/10.22331/q-2024-12-04-1549} {\bibfield  {journal} {\bibinfo  {journal} {Quantum}\ }\textbf {\bibinfo {volume} {8}},\ \bibinfo {pages} {1549} (\bibinfo {year} {2024})},\ \Eprint {https://arxiv.org/abs/2307.12702} {arXiv:2307.12702 [quant-ph]} \BibitemShut {NoStop}%
\bibitem [{\citenamefont {Hassman}\ \emph {et~al.}(2025)\citenamefont {Hassman}, \citenamefont {Reardon-Smith}, \citenamefont {Ravi}, \citenamefont {Chong},\ and\ \citenamefont {Sung}}]{hassman2025flo}%
  \BibitemOpen
  \bibfield  {author} {\bibinfo {author} {\bibfnamefont {Z.}~\bibnamefont {Hassman}}, \bibinfo {author} {\bibfnamefont {O.}~\bibnamefont {Reardon-Smith}}, \bibinfo {author} {\bibfnamefont {G.~S.}\ \bibnamefont {Ravi}}, \bibinfo {author} {\bibfnamefont {F.~T.}\ \bibnamefont {Chong}},\ and\ \bibinfo {author} {\bibfnamefont {K.~J.}\ \bibnamefont {Sung}},\ }\href@noop {} {\bibinfo {title} {Enhancing chemistry on quantum computers with fermionic linear optical simulation}} (\bibinfo {year} {2025}),\ \Eprint {https://arxiv.org/abs/2511.12416} {arXiv:2511.12416 [quant-ph]} \BibitemShut {NoStop}%
\bibitem [{\citenamefont {Dias}\ and\ \citenamefont {Koenig}(2024)}]{dias2024nongaussian}%
  \BibitemOpen
  \bibfield  {author} {\bibinfo {author} {\bibfnamefont {B.}~\bibnamefont {Dias}}\ and\ \bibinfo {author} {\bibfnamefont {R.}~\bibnamefont {Koenig}},\ }\href {https://doi.org/10.22331/q-2024-05-21-1350} {\bibfield  {journal} {\bibinfo  {journal} {Quantum}\ }\textbf {\bibinfo {volume} {8}},\ \bibinfo {pages} {1350} (\bibinfo {year} {2024})},\ \Eprint {https://arxiv.org/abs/2307.12912} {arXiv:2307.12912 [quant-ph]} \BibitemShut {NoStop}%
\bibitem [{\citenamefont {Cudby}\ and\ \citenamefont {Strelchuk}(2025)}]{cudby2023gaussian}%
  \BibitemOpen
  \bibfield  {author} {\bibinfo {author} {\bibfnamefont {J.}~\bibnamefont {Cudby}}\ and\ \bibinfo {author} {\bibfnamefont {S.}~\bibnamefont {Strelchuk}},\ }\href {https://arxiv.org/abs/2307.12654} {\bibinfo {title} {Gaussian decomposition of magic states for matchgate computations}} (\bibinfo {year} {2025}),\ \Eprint {https://arxiv.org/abs/2307.12654} {arXiv:2307.12654 [quant-ph]} \BibitemShut {NoStop}%
\bibitem [{\citenamefont {Jozsa}\ and\ \citenamefont {Miyake}(2008)}]{jozsa2008matchgates}%
  \BibitemOpen
  \bibfield  {author} {\bibinfo {author} {\bibfnamefont {R.}~\bibnamefont {Jozsa}}\ and\ \bibinfo {author} {\bibfnamefont {A.}~\bibnamefont {Miyake}},\ }\href {https://doi.org/10.1098/rspa.2008.0189} {\bibfield  {journal} {\bibinfo  {journal} {Proceedings of the Royal Society A: Mathematical, Physical and Engineering Sciences}\ }\textbf {\bibinfo {volume} {464}},\ \bibinfo {pages} {3089} (\bibinfo {year} {2008})},\ \Eprint {https://arxiv.org/abs/0804.4050} {arXiv:0804.4050 [quant-ph]} \BibitemShut {NoStop}%
\bibitem [{\citenamefont {Brod}\ and\ \citenamefont {Galv{\~a}o}(2011)}]{brod2011extending}%
  \BibitemOpen
  \bibfield  {author} {\bibinfo {author} {\bibfnamefont {D.~J.}\ \bibnamefont {Brod}}\ and\ \bibinfo {author} {\bibfnamefont {E.~F.}\ \bibnamefont {Galv{\~a}o}},\ }\href {https://doi.org/10.1103/PhysRevA.84.022310} {\bibfield  {journal} {\bibinfo  {journal} {Physical Review A}\ }\textbf {\bibinfo {volume} {84}},\ \bibinfo {pages} {022310} (\bibinfo {year} {2011})},\ \Eprint {https://arxiv.org/abs/1106.1863} {arXiv:1106.1863 [quant-ph]} \BibitemShut {NoStop}%
\bibitem [{\citenamefont {Mocherla}\ \emph {et~al.}(2024)\citenamefont {Mocherla}, \citenamefont {Lao},\ and\ \citenamefont {Browne}}]{mocherla2024extendingmatchgatesimulationmethods}%
  \BibitemOpen
  \bibfield  {author} {\bibinfo {author} {\bibfnamefont {A.}~\bibnamefont {Mocherla}}, \bibinfo {author} {\bibfnamefont {L.}~\bibnamefont {Lao}},\ and\ \bibinfo {author} {\bibfnamefont {D.~E.}\ \bibnamefont {Browne}},\ }\href {https://arxiv.org/abs/2302.02654} {\bibinfo {title} {Extending matchgate simulation methods to universal quantum circuits}} (\bibinfo {year} {2024}),\ \Eprint {https://arxiv.org/abs/2302.02654} {arXiv:2302.02654 [quant-ph]} \BibitemShut {NoStop}%
\bibitem [{\citenamefont {Miller}\ \emph {et~al.}(2025)\citenamefont {Miller}, \citenamefont {Favre}, \citenamefont {Holmes}, \citenamefont {Salehi}, \citenamefont {Chakraborty}, \citenamefont {Nyk{\"a}nen}, \citenamefont {Zimbor{\'a}s}, \citenamefont {Glos},\ and\ \citenamefont {Garc{\'i}a-P{\'e}rez}}]{miller2025majoranapropagation}%
  \BibitemOpen
  \bibfield  {author} {\bibinfo {author} {\bibfnamefont {A.}~\bibnamefont {Miller}}, \bibinfo {author} {\bibfnamefont {J.}~\bibnamefont {Favre}}, \bibinfo {author} {\bibfnamefont {Z.}~\bibnamefont {Holmes}}, \bibinfo {author} {\bibfnamefont {{\"O}.}~\bibnamefont {Salehi}}, \bibinfo {author} {\bibfnamefont {R.}~\bibnamefont {Chakraborty}}, \bibinfo {author} {\bibfnamefont {A.}~\bibnamefont {Nyk{\"a}nen}}, \bibinfo {author} {\bibfnamefont {Z.}~\bibnamefont {Zimbor{\'a}s}}, \bibinfo {author} {\bibfnamefont {A.}~\bibnamefont {Glos}},\ and\ \bibinfo {author} {\bibfnamefont {G.}~\bibnamefont {Garc{\'i}a-P{\'e}rez}},\ }\href@noop {} {\bibinfo {title} {Simulation of fermionic circuits using majorana propagation}} (\bibinfo {year} {2025}),\ \Eprint {https://arxiv.org/abs/2503.18939} {arXiv:2503.18939 [quant-ph]} \BibitemShut {NoStop}%
\bibitem [{\citenamefont {Arunachalam}\ and\ \citenamefont {Dutt}(2026)}]{arunachalam2026tomographyquantumstatesbounded}%
  \BibitemOpen
  \bibfield  {author} {\bibinfo {author} {\bibfnamefont {S.}~\bibnamefont {Arunachalam}}\ and\ \bibinfo {author} {\bibfnamefont {A.}~\bibnamefont {Dutt}},\ }\href {https://arxiv.org/abs/2606.07425} {\bibinfo {title} {Tomography of quantum states with bounded extent}} (\bibinfo {year} {2026}),\ \Eprint {https://arxiv.org/abs/2606.07425} {arXiv:2606.07425 [quant-ph]} \BibitemShut {NoStop}%
\bibitem [{\citenamefont {Schuch}\ and\ \citenamefont {Verstraete}(2009)}]{Schuch2009}%
  \BibitemOpen
  \bibfield  {author} {\bibinfo {author} {\bibfnamefont {N.}~\bibnamefont {Schuch}}\ and\ \bibinfo {author} {\bibfnamefont {F.}~\bibnamefont {Verstraete}},\ }\href {https://doi.org/10.1038/nphys1370} {\bibfield  {journal} {\bibinfo  {journal} {Nature Physics}\ }\textbf {\bibinfo {volume} {5}},\ \bibinfo {pages} {732–735} (\bibinfo {year} {2009})}\BibitemShut {NoStop}%
\bibitem [{\citenamefont {O’Gorman}\ \emph {et~al.}(2022)\citenamefont {O’Gorman}, \citenamefont {Irani}, \citenamefont {Whitfield},\ and\ \citenamefont {Fefferman}}]{OGorman2022}%
  \BibitemOpen
  \bibfield  {author} {\bibinfo {author} {\bibfnamefont {B.}~\bibnamefont {O’Gorman}}, \bibinfo {author} {\bibfnamefont {S.}~\bibnamefont {Irani}}, \bibinfo {author} {\bibfnamefont {J.}~\bibnamefont {Whitfield}},\ and\ \bibinfo {author} {\bibfnamefont {B.}~\bibnamefont {Fefferman}},\ }\bibfield  {journal} {\bibinfo  {journal} {PRX Quantum}\ }\textbf {\bibinfo {volume} {3}},\ \href {https://doi.org/10.1103/prxquantum.3.020322} {10.1103/prxquantum.3.020322} (\bibinfo {year} {2022})\BibitemShut {NoStop}%
\bibitem [{\citenamefont {Zhai}\ \emph {et~al.}(2026)\citenamefont {Zhai}, \citenamefont {Li}, \citenamefont {Zhang}, \citenamefont {Li}, \citenamefont {Lee},\ and\ \citenamefont {Chan}}]{zhai2026classicalsolutionfemocofactormodel}%
  \BibitemOpen
  \bibfield  {author} {\bibinfo {author} {\bibfnamefont {H.}~\bibnamefont {Zhai}}, \bibinfo {author} {\bibfnamefont {C.}~\bibnamefont {Li}}, \bibinfo {author} {\bibfnamefont {X.}~\bibnamefont {Zhang}}, \bibinfo {author} {\bibfnamefont {Z.}~\bibnamefont {Li}}, \bibinfo {author} {\bibfnamefont {S.}~\bibnamefont {Lee}},\ and\ \bibinfo {author} {\bibfnamefont {G.~K.-L.}\ \bibnamefont {Chan}},\ }\href {https://arxiv.org/abs/2601.04621} {\bibinfo {title} {Classical solution of the femo-cofactor model to chemical accuracy and its implications}} (\bibinfo {year} {2026}),\ \Eprint {https://arxiv.org/abs/2601.04621} {arXiv:2601.04621 [physics.chem-ph]} \BibitemShut {NoStop}%
\bibitem [{\citenamefont {Bluvstein}\ \emph {et~al.}(2023)\citenamefont {Bluvstein}, \citenamefont {Evered}, \citenamefont {Geim}, \citenamefont {Li}, \citenamefont {Zhou}, \citenamefont {Manovitz}, \citenamefont {Ebadi}, \citenamefont {Cain}, \citenamefont {Kalinowski}, \citenamefont {Hangleiter}, \citenamefont {Bonilla~Ataides}, \citenamefont {Maskara}, \citenamefont {Cong}, \citenamefont {Gao}, \citenamefont {Sales~Rodriguez}, \citenamefont {Karolyshyn}, \citenamefont {Semeghini}, \citenamefont {Gullans}, \citenamefont {Greiner}, \citenamefont {Vuletić},\ and\ \citenamefont {Lukin}}]{Bluvstein2023}%
  \BibitemOpen
  \bibfield  {author} {\bibinfo {author} {\bibfnamefont {D.}~\bibnamefont {Bluvstein}}, \bibinfo {author} {\bibfnamefont {S.~J.}\ \bibnamefont {Evered}}, \bibinfo {author} {\bibfnamefont {A.~A.}\ \bibnamefont {Geim}}, \bibinfo {author} {\bibfnamefont {S.~H.}\ \bibnamefont {Li}}, \bibinfo {author} {\bibfnamefont {H.}~\bibnamefont {Zhou}}, \bibinfo {author} {\bibfnamefont {T.}~\bibnamefont {Manovitz}}, \bibinfo {author} {\bibfnamefont {S.}~\bibnamefont {Ebadi}}, \bibinfo {author} {\bibfnamefont {M.}~\bibnamefont {Cain}}, \bibinfo {author} {\bibfnamefont {M.}~\bibnamefont {Kalinowski}}, \bibinfo {author} {\bibfnamefont {D.}~\bibnamefont {Hangleiter}}, \bibinfo {author} {\bibfnamefont {J.~P.}\ \bibnamefont {Bonilla~Ataides}}, \bibinfo {author} {\bibfnamefont {N.}~\bibnamefont {Maskara}}, \bibinfo {author} {\bibfnamefont {I.}~\bibnamefont {Cong}}, \bibinfo {author} {\bibfnamefont {X.}~\bibnamefont {Gao}}, \bibinfo {author} {\bibfnamefont {P.}~\bibnamefont {Sales~Rodriguez}}, \bibinfo {author} {\bibfnamefont
  {T.}~\bibnamefont {Karolyshyn}}, \bibinfo {author} {\bibfnamefont {G.}~\bibnamefont {Semeghini}}, \bibinfo {author} {\bibfnamefont {M.~J.}\ \bibnamefont {Gullans}}, \bibinfo {author} {\bibfnamefont {M.}~\bibnamefont {Greiner}}, \bibinfo {author} {\bibfnamefont {V.}~\bibnamefont {Vuletić}},\ and\ \bibinfo {author} {\bibfnamefont {M.~D.}\ \bibnamefont {Lukin}},\ }\href {https://doi.org/10.1038/s41586-023-06927-3} {\bibfield  {journal} {\bibinfo  {journal} {Nature}\ }\textbf {\bibinfo {volume} {626}},\ \bibinfo {pages} {58–65} (\bibinfo {year} {2023})}\BibitemShut {NoStop}%
\bibitem [{\citenamefont {Schwerdt}\ \emph {et~al.}(2024)\citenamefont {Schwerdt}, \citenamefont {Peleg}, \citenamefont {Shapira}, \citenamefont {Priel}, \citenamefont {Florshaim}, \citenamefont {Gross}, \citenamefont {Zalic}, \citenamefont {Afek}, \citenamefont {Akerman}, \citenamefont {Stern}, \citenamefont {Kish},\ and\ \citenamefont {Ozeri}}]{Schwerdt2024}%
  \BibitemOpen
  \bibfield  {author} {\bibinfo {author} {\bibfnamefont {D.}~\bibnamefont {Schwerdt}}, \bibinfo {author} {\bibfnamefont {L.}~\bibnamefont {Peleg}}, \bibinfo {author} {\bibfnamefont {Y.}~\bibnamefont {Shapira}}, \bibinfo {author} {\bibfnamefont {N.}~\bibnamefont {Priel}}, \bibinfo {author} {\bibfnamefont {Y.}~\bibnamefont {Florshaim}}, \bibinfo {author} {\bibfnamefont {A.}~\bibnamefont {Gross}}, \bibinfo {author} {\bibfnamefont {A.}~\bibnamefont {Zalic}}, \bibinfo {author} {\bibfnamefont {G.}~\bibnamefont {Afek}}, \bibinfo {author} {\bibfnamefont {N.}~\bibnamefont {Akerman}}, \bibinfo {author} {\bibfnamefont {A.}~\bibnamefont {Stern}}, \bibinfo {author} {\bibfnamefont {A.~B.}\ \bibnamefont {Kish}},\ and\ \bibinfo {author} {\bibfnamefont {R.}~\bibnamefont {Ozeri}},\ }\bibfield  {journal} {\bibinfo  {journal} {Physical Review X}\ }\textbf {\bibinfo {volume} {14}},\ \href {https://doi.org/10.1103/physrevx.14.041017} {10.1103/physrevx.14.041017} (\bibinfo {year} {2024})\BibitemShut {NoStop}%
\bibitem [{\citenamefont {Maskara}\ \emph {et~al.}(2025)\citenamefont {Maskara}, \citenamefont {Kalinowski}, \citenamefont {Gonzalez-Cuadra},\ and\ \citenamefont {Lukin}}]{maskara2025fastsimulationfermionsreconfigurable}%
  \BibitemOpen
  \bibfield  {author} {\bibinfo {author} {\bibfnamefont {N.}~\bibnamefont {Maskara}}, \bibinfo {author} {\bibfnamefont {M.}~\bibnamefont {Kalinowski}}, \bibinfo {author} {\bibfnamefont {D.}~\bibnamefont {Gonzalez-Cuadra}},\ and\ \bibinfo {author} {\bibfnamefont {M.~D.}\ \bibnamefont {Lukin}},\ }\href {https://arxiv.org/abs/2509.08898} {\bibinfo {title} {Fast simulation of fermions with reconfigurable qubits}} (\bibinfo {year} {2025}),\ \Eprint {https://arxiv.org/abs/2509.08898} {arXiv:2509.08898 [quant-ph]} \BibitemShut {NoStop}%
\bibitem [{\citenamefont {Hu}\ \emph {et~al.}(2026)\citenamefont {Hu}, \citenamefont {Liu},\ and\ \citenamefont {Wu}}]{hu2026entanglementassistedcircuitknittingdistributed}%
  \BibitemOpen
  \bibfield  {author} {\bibinfo {author} {\bibfnamefont {S.-H.}\ \bibnamefont {Hu}}, \bibinfo {author} {\bibfnamefont {P.-S.}\ \bibnamefont {Liu}},\ and\ \bibinfo {author} {\bibfnamefont {J.-Y.}\ \bibnamefont {Wu}},\ }\href {https://arxiv.org/abs/2510.26789} {\bibinfo {title} {Entanglement-assisted circuit knitting: Distributed quantum computing using limited entanglement resources}} (\bibinfo {year} {2026}),\ \Eprint {https://arxiv.org/abs/2510.26789} {arXiv:2510.26789 [quant-ph]} \BibitemShut {NoStop}%
\bibitem [{\citenamefont {Nielsen}(1999)}]{nielsen_conditions_1999}%
  \BibitemOpen
  \bibfield  {author} {\bibinfo {author} {\bibfnamefont {M.~A.}\ \bibnamefont {Nielsen}},\ }\href {https://doi.org/10.1103/PhysRevLett.83.436} {\bibfield  {journal} {\bibinfo  {journal} {Physical Review Letters}\ }\textbf {\bibinfo {volume} {83}},\ \bibinfo {pages} {436} (\bibinfo {year} {1999})}\BibitemShut {NoStop}%
\bibitem [{\citenamefont {R{\'e}nyi}(1961)}]{renyi1961measures}%
  \BibitemOpen
  \bibfield  {author} {\bibinfo {author} {\bibfnamefont {A.}~\bibnamefont {R{\'e}nyi}},\ }in\ \href@noop {} {\emph {\bibinfo {booktitle} {Proceedings of the Fourth Berkeley Symposium on Mathematical Statistics and Probability, Volume 1: Contributions to the Theory of Statistics}}}\ (\bibinfo {organization} {The Regents of the University of California},\ \bibinfo {year} {1961})\BibitemShut {NoStop}%
\bibitem [{\citenamefont {Jamio{\l}kowski}(1972)}]{jamiolkowskiLinearTransformationsWhich1972}%
  \BibitemOpen
  \bibfield  {author} {\bibinfo {author} {\bibfnamefont {A.}~\bibnamefont {Jamio{\l}kowski}},\ }\href {https://doi.org/10.1016/0034-4877(72)90011-0} {\bibfield  {journal} {\bibinfo  {journal} {Reports on Mathematical Physics}\ }\textbf {\bibinfo {volume} {3}},\ \bibinfo {pages} {275} (\bibinfo {year} {1972})}\BibitemShut {NoStop}%
\bibitem [{\citenamefont {Choi}(1975)}]{choiCompletelyPositiveLinear1975}%
  \BibitemOpen
  \bibfield  {author} {\bibinfo {author} {\bibfnamefont {M.-D.}\ \bibnamefont {Choi}},\ }\href {https://doi.org/10.1016/0024-3795(75)90075-0} {\bibfield  {journal} {\bibinfo  {journal} {Linear Algebra and its Applications}\ }\textbf {\bibinfo {volume} {10}},\ \bibinfo {pages} {285} (\bibinfo {year} {1975})}\BibitemShut {NoStop}%
\bibitem [{\citenamefont {Gittsovich}\ \emph {et~al.}(2008)\citenamefont {Gittsovich}, \citenamefont {G\"{u}hne}, \citenamefont {Hyllus},\ and\ \citenamefont {Eisert}}]{Gittsovich2008}%
  \BibitemOpen
  \bibfield  {author} {\bibinfo {author} {\bibfnamefont {O.}~\bibnamefont {Gittsovich}}, \bibinfo {author} {\bibfnamefont {O.}~\bibnamefont {G\"{u}hne}}, \bibinfo {author} {\bibfnamefont {P.}~\bibnamefont {Hyllus}},\ and\ \bibinfo {author} {\bibfnamefont {J.}~\bibnamefont {Eisert}},\ }\bibfield  {journal} {\bibinfo  {journal} {Physical Review A}\ }\textbf {\bibinfo {volume} {78}},\ \href {https://doi.org/10.1103/physreva.78.052319} {10.1103/physreva.78.052319} (\bibinfo {year} {2008})\BibitemShut {NoStop}%
\end{thebibliography}%
\PRLrefsep
\clearpage
\onecolumngrid

%%%% APPENDICES
\section*{Appendices}
\let\addcontentsline\oldaddcontentsline
\appendix

% resetting counters
\renewcommand{\thesection}{A\arabic{section}}
\setcounter{section}{0}
\renewcommand{\thetable}{A\arabic{table}}
\setcounter{table}{0}
\renewcommand{\thefigure}{A\arabic{figure}}
\setcounter{figure}{0}
\renewcommand{\theequation}{A\arabic{equation}}
\setcounter{equation}{0}

\makeatletter
\@starttoc{toc}
% REVTeX resets equations at each appendix section; remove this reset.
\@removefromreset{equation}{section}
\makeatother

\section{Framework and background results for distribution complexity}
\label{app:distribution_complexity_details}
In the main text, we discussed the distribution complexity of two ways of implementing a nonlocal circuit across separate quantum processors. With only classical communication, a crossing operation is replaced by a quasiprobability decomposition (QPD) over local operations and classical communication (LOCC). With a coherent quantum interconnect, the processors can instead consume shared entanglement, such as Bell pairs, to implement or prepare the nonlocal resource directly. These two settings are controlled by related, but distinct, entanglement measures. In this section, we give a theoretical background necessary to understand our distribution complexity metrics. First first we discuss the optimal two-qubit \emph{LOCC-assisted} distribution cost used for gate-by-gate compiled baselines, and follow with with a discussion of of \emph{entanglement-assisted} cost, and the R\'enyi-entropy hierarchy that relates these two.

\subsection{Two-qubit formulas used for compiled baselines}
\label{app:kak_baseline}

For two-qubit gates, the optimal quasiprobability extent, i.e. distribution cost, can be expressed through the KAK decomposition~\cite{tucci2005introductioncartanskakdecomposition,piveteau2024knitting}. Up to single-qubit rotations, any two-qubit unitary can be expressed as
\begin{equation}
    U_{AB}
    =
    \exp\left(
    i\theta_X X\otimes X
    +
    i\theta_Y Y\otimes Y
    +
    i\theta_Z Z\otimes Z
    \right),
    \label{eq:kak_supp}
\end{equation}
with real parameters ordered as $|\theta_Z|\leq \theta_Y\leq \theta_X\leq \pi/4$. For the Givens and diagonal CPhase gates appearing in the double-factorized Trotter circuit, $\theta_Z=0$, and the optimal two-qubit QPD distribution complexity is
\begin{equation}
    \gamma_{\mathrm{LOCC}}(\mathcal U)
    =
    1
    +
    2|\sin 2\theta_X|
    +
    2|\sin 2\theta_Y|
    +
    2|\sin 2\theta_X|\,|\sin 2\theta_Y| .
    \label{eq:tight_2q_supp}
\end{equation}
A Givens rotation $G(\theta)$ has $\theta_X=\theta_Y=\theta/2$, giving $\gamma=2(1+|\sin\theta|)^2-1$. A diagonal $ZZ$ rotation with phase angle $\vartheta$ has one nonzero KAK angle and gives $\gamma=1+2|\sin(\vartheta/2)|$.

For sequential layers, cutting layers independently gives a multiplicative upper bound on the total QPD distribution complexity, and hence a squared multiplicative contribution to sampling overhead~\cite{piveteau2024knitting}. This multiplicative estimate is the compiled-gate baseline for the Gaussian and diagonal contributions analyzed in the main text.

\subsection{Bell-pair costs and R\'enyi hierarchy}
\label{app:bell_pair_renyi}

Instead of sampling many branches of product unitaries across QPUs, the processors may consume shared entanglement to prepare the relevant nonlocal state or channel coherently.
For pure bipartite states, where the same Schmidt spectrum controls both the quasiprobability robustness and the minimum zero-error entanglement resource. We write the Schmidt decomposition of the state $|\psi\rangle\in\mathcal H_A\otimes\mathcal H_B$ as
\begin{equation}
    |\psi\rangle
    =
    \sum_{j=1}^{r}
    \lambda_j
    |\phi_j\rangle_A
    \otimes
    |\widetilde{\phi}_j\rangle_B,
    \label{eq:schmidt_supp}
\end{equation}
where $\ket{\phi_j}_A$ and $\ket{\tilde \phi_j}_B$ are the orthonormal basis states for A and B respectively. Likewise, $\lambda_j$ are non-negative real numbers satisfying $\sum_j \lambda_j^2=1$ known as Schmidt coefficients.
Nielsen's theorem characterizes pure-state LOCC transformations by majorization of Schmidt vectors~\cite{nielsen_conditions_1999}. In particular, a maximally entangled two-qudit state can be converted by LOCC into any pure state of Schmidt rank at most $d$, making Bell pairs the natural resource for coherent distributed state preparation.

The zero-error entanglement cost can be written as~\cite{Yue_2019,theurer2023singleshotentanglementmanipulationstates}
\begin{equation}
    E_c^{(0)}(\rho)
    =
    \min
    \left\{
    \log d:
    \inf_{\Lambda\in\mathcal L}
    \left\|
    \Lambda\left(|\Phi_d\rangle\langle\Phi_d|\right)
    -
    \rho
    \right\|_1
    =
    0
    \right\},
    \label{eq:zero_error_cost_supp}
\end{equation}
where $\mathcal L$ denotes LOCC maps, and $\ket{\Phi_d}$ is the two-qudit maximally entangled state defined as $\ket{\Phi_d} = \frac{1}{\sqrt{d}} \sum_{j=1}^{d} \ket{\phi_j}_A \otimes \ket{\tilde\phi_j}_B$. For pure states this reduces to $S_0(\rho_A)=\log\mathrm{rank}(\rho_A)$, the zeroth R\'enyi entropy. The R\'enyi entropy of order $\alpha$ is defined as
\begin{equation}
    S_\alpha(\rho)
    =
    \frac{1}{1-\alpha}
    \log\operatorname{Tr}(\rho^\alpha),
    \qquad
    \alpha\neq 1.
    \label{eq:renyi_def_supp}
\end{equation}
The R\'enyi entropies obey the hierarchy in \cref{eq:renyi_hierarchy_main}~\cite{renyi1961measures}. The robustness of entanglement, $R = (\sum_j\lambda_j)^2-1$, and equivalently the LOCC-assisted distribution complexity, $\gamma_{\mathrm{state}}=2R +1$, are related to the half-R\'enyi entropy as follows: 
\begin{equation}
    S_{1/2}(\rho_A)
    =
    \log\left(
    \frac{\gamma_{\mathrm{state}}+1}{2}
    \right).
    \label{eq:state_gamma_renyi_supp}
\end{equation}
Due to the hierarchy of R\'enyi entropies, $S_{1/2}$ gives a lower bound on the zero-error Bell-pair cost, while the von Neumann entropy $S_{\mathrm{vN}}(\rho) = \lim_{\alpha \to 1} S_\alpha(\rho)$ represents the asymptotic entanglement distillation cost and $S_2$ provides the experimentally accessible entanglement via two copies of a quantum state.

%=============================================================================
\section{Covariance matrix of the maximally entangled state}
\label{app:covariance}
%=============================================================================

The distribution complexity lower bound in \cref{thm:gaussian_choi_bound} rests on the fact that the Choi state of a fermionic Gaussian unitary is itself a fermionic Gaussian state. This means that we don't need to construct the $2^N \times 2^N$ unitary or its dense Choi statevector, because a $2N \times 2N$ covariance matrix of moments provides a complete description of fermionic Gaussian states. The Schmidt coefficients of such states are accessible, and their one-norm is efficiently computable, from the eigenvalues of reduced submatrices of this covariance matrix, which we now compute specifically for fermionic Gaussian Choi states.

In the Choi--Jamiolkowski construction, this reference state entangles the basis states of Hilbert space with an identical auxiliary copy~\cite{jamiolkowskiLinearTransformationsWhich1972,choiCompletelyPositiveLinear1975}. 
We do this in the fermionic setting by taking the $N$ physical fermionic modes $a^\dagger_1, \ldots a^\dagger_N$ and pairing each with a secondary set of auxiliary modes $a^{\dagger}_1{'}, \ldots a^{\dagger}_N{'}$, such that:
\begin{equation}
    | \Omega \rangle
    =
    2^{-N/2}
    \prod_{j=1}^{N}
    \left(1+a_j^\dagger a_j^{\dagger}{'}\right)
    |\mathrm{vac}\rangle .
\end{equation}
Equivalently, we can think of $|\Omega\rangle$ as initializing each physical mode $j$ and its auxiliary partner $j'$ in a two-mode fermionic Bell state $(|00\rangle + |11\rangle)/\sqrt{2}$, with a fixed interleaved ordering convention $(a_1,a'_1,a_2,a'_2,\ldots)$. For each mode pair, define the four Majorana operators
\begin{equation}
    c_{4j-3}=a_j+a_j^\dagger,\qquad
    c_{4j-2}=i(a_j-a_j^\dagger),\qquad
    c_{4j-1}=a'_j+a_j^{\dagger}{'},\qquad
    c_{4j}=i(a'_j-a_j^{\dagger}{'}).
\end{equation}
Then the covariance matrix is defined as 
\begin{equation}
    [\Gamma]_{kl}
    =
    \frac{i}{2}
    \langle [c_k,c_l]\rangle ,
\end{equation}
which characterizes a fermionic Gaussian state through its two-point
Majorana correlations~\cite{bravyi2005lagrangianflo,peschel2003rdm,peschel2009reduced}.

Since $|\Omega\rangle$ is a product over independent physical-auxiliary pairs,
$\Gamma^\Omega$ is a direct sum of identical two-mode blocks. It is enough to
compute one pair in the state
$|\phi\rangle=(|00\rangle+|11\rangle)/\sqrt{2}$, where
$|11\rangle=a^\dagger a^{\prime\dagger}|\mathrm{vac}\rangle$.
The nonzero expectation values are
\begin{equation}
    \langle a^\dagger a\rangle
    =
    \langle a^{\dagger}{'}a'\rangle
    =
    \frac{1}{2},
    \qquad
    \langle a^\dagger a^{\dagger}{'}\rangle
    =
    \langle a'a\rangle
    =
    \frac{1}{2},
\end{equation}
with $\langle aa'\rangle=\langle a^{\dagger}{'}a^\dagger\rangle=-1/2$ by
fermionic ordering. The cross terms
$\langle a^\dagger a'\rangle$ and $\langle aa^{\dagger}{'}\rangle$ vanish. Using linearity of expectations,
\begin{align}
\Gamma_{13}
&=
i\left\langle (a+a^\dagger)(a'+a^{\prime\dagger})\right\rangle,
&
\Gamma_{14}
&=
-\left\langle (a+a^\dagger)(a'-a^{\prime\dagger})\right\rangle,
\\
\Gamma_{23}
&=
-\left\langle (a-a^\dagger)(a'+a^{\prime\dagger})\right\rangle,
&
\Gamma_{24}
&=
-i\left\langle (a-a^\dagger)(a'-a^{\prime\dagger})\right\rangle,
\end{align}
and simplifying, we get 
\begin{equation}
    \Gamma_{\mathrm{pair}}
    =
    \begin{pmatrix}
       0 &  0 &  0 &  1 \\
       0 &  0 &  1 &  0 \\
       0 & -1 &  0 &  0 \\
      -1 &  0 &  0 &  0
    \end{pmatrix}, \qquad \Gamma^\Omega
    =
    \bigoplus_{j=1}^{N}
    \Gamma_{\mathrm{pair}} .
\end{equation}

%=============================================================================
\section{Efficient One-norm for Gaussian States}
\label{app:factorization_schmidt}
%=============================================================================

For a pure fermionic Gaussian state $|\psi\rangle$ on $N$ modes, the restricted
covariance matrix $\Gamma_{\mathcal{A}}$ for a subsystem $N_A$ modes has eigenvalues $\pm i\nu_k$ with $\nu_k \in [0,1]$, $k = 1  ,\ldots,N_A$~\cite{bravyi2005lagrangianflo,peschel2003rdm,peschel2009reduced}.
Then the $2^{N_A}$ many eigenvalues of the reduced density matrix $\rho_A = \mathrm{Tr}_B(|\psi\rangle\langle\psi|)$ are given as the products of the factors:
\begin{equation}
    p_{\vec{n}} = \prod_{k=1}^{N_A}
    \biggl(\frac{1+\nu_k}{2}\biggr)^{\!n_k}
    \biggl(\frac{1-\nu_k}{2}\biggr)^{\!1-n_k},
    \qquad \vec{n} \in \{0,1\}^{N_A} .
    \label{eq:gaussian_spectrum}
\end{equation}
Since the Schmidt coefficients are given by $\alpha_{\vec{n}}=\sqrt{p_{\vec{n}}}$, we can show that the one-norm summation $\sum_i \alpha_i$ factorizes:
\begin{align}
    \sum_{\vec{n}} \alpha_{\vec{n}}
    &= \sum_{\vec{n}} \prod_{k=1}^{N_A}
       \sqrt{\biggl(\frac{1+\nu_k}{2}\biggr)^{\!n_k}
             \biggl(\frac{1-\nu_k}{2}\biggr)^{\!1-n_k}}
    \nonumber \\
    &= \prod_{k=1}^{N_A} \sum_{n_k \in \{0,1\}}
       \sqrt{\biggl(\frac{1+\nu_k}{2}\biggr)^{\!n_k}
             \biggl(\frac{1-\nu_k}{2}\biggr)^{\!1-n_k}}
    \nonumber \\
    &= \prod_{k=1}^{N_A}
       \biggl(\sqrt{\frac{1+\nu_k}{2}} +
              \sqrt{\frac{1-\nu_k}{2}}\biggr).
    \label{eq:schmidt_factorization}
\end{align}
The second equality exchanges the sum over $2^{N_A}$ binary strings with
a product of $N_A$ binary sums, using the product form of $p_{\vec n}$.

Squaring the single-mode factor gives
\begin{equation}
    \biggl(\sqrt{\frac{1+s}{2}}+\sqrt{\frac{1-s}{2}}\biggr)^2
    =
    1+\sqrt{1-s^2}.
\end{equation}
Thus a generic Gaussian state has
\begin{equation}
    \sum_{\vec n}\alpha_{\vec n}
    =
    \prod_{k=1}^{N_A}
    \bigl(1+\sqrt{1-\nu_k^2}\bigr)^{1/2}.
    \label{eq:generic_gaussian_nuclear_norm}
\end{equation}
For the Choi state used in \cref{eq:gamma_choi}, the restricted covariance spectrum has a doubled structure: each singular value $\sigma_k$ of $R[A,A]$ appears twice among the $\nu_k$. Applying \cref{eq:generic_gaussian_nuclear_norm} to these two identical modes removes the square root and gives
\begin{equation}
    S
    =
    \sum_{\vec n}\alpha_{\vec n}
    =
    \prod_{k=1}^{N_A}
    \bigl(1+\sqrt{1-\sigma_k^2}\bigr).
    \label{eq:choi_schmidt_sum}
\end{equation}
Substituting $S$ into the Choi/robustness expression $\gamma_{\mathrm{Choi}}=2S^2-1$ gives \cref{eq:gamma_choi}.

%=============================================================================
\section{Uncorrelated diagonal interactions and random CPhase saturation}
\label{app:random_cphase_saturation}
%=============================================================================

\begin{figure*}[!thbp]
    \centering
    \includegraphics[width=\textwidth]{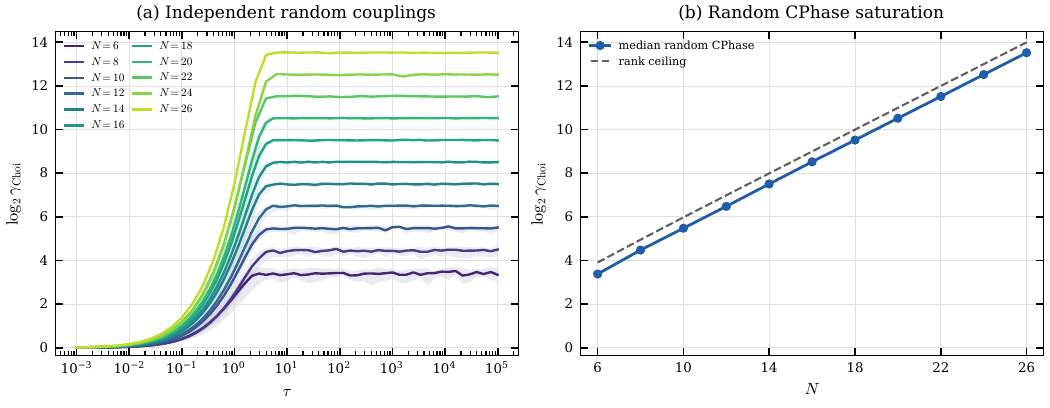}
    \caption{\textbf{Generic diagonal phases rapidly approach the random CPhase limit.}
    Here we show what happens when the correlated phase structure of a squared one-body diagonal layer is removed, with two generic diagonal models. In both panels, the plotted quantity is the exact Choi lower bound $\log_2\gamma_{\mathrm{Choi}}$ across a balanced bipartition. \textbf{(a)} We sample dense all-to-all diagonal Hamiltonians $H=\sum_{p<q}\Theta_{pq}\hat n_p\hat n_q$ with independent random couplings, normalized to one as in the random rank-one model discussed in~\cref{sec:coulomb}, and plot the median $\log_2\gamma_{\mathrm{Choi}}$ as a function of evolution time $\tau$. Shaded bands show the $10$--$90\%$ range over random instances. Unlike the rank-one model, which shows a slow growth of entanglement after an initial regime, the independent-coupling ensemble rapidly approaches an $N$-dependent saturation value. \textbf{(b)} We plot the  corresponding infinite-time random-phase limit: an all-to-all cross-partition CPhase circuit with independent angles $\theta_{ij}\sim\mathrm{Uniform}(0,2\pi)$. The blue curve shows the median over random angle samples, and the dashed black curve is the maximum possible Choi value for the balanced cut. Together, the panels show that slow diagonal-layer growth is not a generic consequence of diagonality, but comes from the rank-one phase locking present in double-factorized fragments.}
    \label{fig:random_diagonal_baselines}
    \label{fig:appendix_uncorrelated_iid_diagonal}
    \label{fig:cphase_gamma_scaling}
\end{figure*}

The Coulomb layers in a double-factorized Trotter step are generated by squares of one-body diagonal operators.  If
$\hat D_\ell=\sum_p\varepsilon_p^{(\ell)}\hat n_p$, then
\begin{equation}
    \frac{\lambda_\ell}{2}\hat D_\ell^2
    =
    \frac{\lambda_\ell}{2}\sum_p
    \bigl(\varepsilon_p^{(\ell)}\bigr)^2\hat n_p
    +
    \lambda_\ell
    \sum_{p<q}
    \varepsilon_p^{(\ell)}\varepsilon_q^{(\ell)}
    \hat n_p\hat n_q .
    \label{eq:df_diagonal_pair_expansion}
\end{equation}
The first term gives one-mode phases.  For a fixed processor partition $A|B$,
the two-mode terms with $p,q$ both in $A$ or both in $B$ are also local.  Thus
the part of the diagonal layer that can create operator entanglement across
the cut is the cross-boundary CPhase layer with angles
$\theta_{pq}^{(\ell)}=\tau\lambda_\ell
\varepsilon_p^{(\ell)}\varepsilon_q^{(\ell)}$, for $p\in A$ and $q\in B$.
The sign is immaterial for the Choi spectrum.

To separate this squared-one-body structure from the generic behavior of
diagonal two-body phases, we compare it with an uncorrelated random diagonal
Hamiltonian
\begin{equation}
    \hat H_\Theta
    =
    \sum_{p<q}
    \Theta_{pq}\hat n_p\hat n_q,
    \qquad
    \hat U_\Theta(\tau)
    =
    \exp(-i\tau\hat H_\Theta).
    \label{eq:random_pair_diagonal_hamiltonian}
\end{equation}
In \cref{fig:appendix_uncorrelated_iid_diagonal}, we consider models with independent couplings sampled uniformly $X_{pq}\sim\operatorname{Uniform}(-1,1)$, symmetrize, and normalize according to the Frobenius norm to match the random rank-one models in \cref{sec:coulomb}, 
\begin{equation}
    \Theta_{pq}
    =
    \frac{X_{pq}+X_{qp}}{\|X\|_F},
    \qquad p<q .
    \label{eq:appendix_random_theta_distribution}
\end{equation}
This normalization fixes the horizontal scale of $\tau$. 
For any fixed partition, only the cross block of $\Theta$ affects the
operator-Schmidt spectrum.  The cross-boundary unitary is
\begin{equation}
    \hat U_{AB}(\tau)
    =
    \prod_{p\in A,\,q\in B}
    \exp\!\left(-i\tau\Theta_{pq}\hat n_p\hat n_q\right),
    \label{eq:cross_cphase_factorization}
\end{equation}
where the product order is irrelevant because all factors commute. In \cref{fig:random_diagonal_baselines}(a), we plot the $\gamma_\mathrm{Choi}$ of the above unitaries by explicitly computing the dense Schmidt decomposition of the Choi state. 
For a fixed set of generic nonzero couplings, increasing $\tau$ eventually causes the
angles $\tau\Theta_{pq}$ to wind around the circle modulo $2\pi$, where the coupling values act as CPhase angle frequencies. 
At large times, these wrapped phases behave like independent random angles. At sufficiently large $\tau$, the value of $\gamma_\mathrm{Choi}$ saturates to an $N$-dependent plateau. Indeed, the distribution complexity of the random rank-one model in \cref{sec:coulomb} saturates same value for each $N$.
% , where all of the couplings are of the form $\varepsilon_p \varepsilon_q$ for a single vector $\vec{\varepsilon}$ and which transitions to the slow intermediate growth discussed in \cref{thm:slow-growth-MBL}. 
However, unlike the random rank-one model, there is no slow growth regime for random independent couplings, and instead we observe superpolynomial growth of $\gamma_\mathrm{Choi}$ until saturation.

We also explicitly consider the distribution bound for all-to-all CPhase circuits, in which we directly sample
\begin{equation}
    \theta_{pq}\sim\mathrm{Uniform}(0,2\pi),
    \qquad p\in A,\ q\in B,
\end{equation}
and the corresponding random CPhase unitary
\begin{equation}
    \hat{U}_\mathrm{RCP} = \prod_{p\in A,\,q\in B}
    \exp(i\theta_{pq}\hat n_p\hat n_q).
\end{equation}
We then plot the median Choi lower bound for many random iterations for increasing $N$, always taking balanced bipartitions, and observe the exponential growth of $\gamma_\mathrm{Choi}(\hat{U}_\mathrm{RCP})$ in \cref{fig:cphase_gamma_scaling}. We observe that its values closely match the long-time plateaus of both the random independent couplings model and the random rank-one model. This shows that the slow growth of entanglement is not a consequence of the general diagonal form of the Coulomb interactions, but specifically due to the structure of being squared one-body operators.

%=============================================================================
\section{Small molecule details and distribution representation optimization}
\label{app:small_molecules}
%=============================================================================

Here we use small STO-3G molecular systems to test whether the representation levers identified in \cref{sec:structure} for double-factorized electronic
structure data. The closed-form Gaussian diagnostics make it scalable to
compare truncation, partition, fragment-order, and gauge choices directly,
without constructing the full $2^N$ dimensional unitary.

For each molecule we start from the spinless double-factorized Hamiltonian in a
canonical molecular-orbital basis.  The truncation threshold is chosen along the
weighted-\(L_2\) path of \cref{eq:weighted_l2}, and we retain the most aggressive
threshold whose induced CCSD(T) correlation-energy error satisfies
\(\Delta E_{\mathrm{corr}}\le 1.6~\mathrm{mHa}\).  This follows the
chemistry-facing calibration used in recent tensor-factorization resource
estimates, where raw tensor norms are useful construction objectives but are not
by themselves the observable error of interest~\cite{lee2021thc,caesura2025bliss}.

We optimize the Gaussian part of the DF
Trotter step, since the diagonal rank-one Coulomb layers are discussed separately, and are shown to be bounded for small step sizes. For a fixed balanced partition \(A|B\), the Gaussian lower bound cost is
the product of Choi costs for the first fragment basis, adjacent interfragment basis changes, and the final fragment basis.
After truncation, we search for the best mode partition, fragment order, and eigenspace gauges for each fragment that results in the smallest distribution overhead. 

For \(N\le 10\) we search over all balanced partitions, while for the three
largest STO-3G entries below it is a greedy single-mode-swap search. 
At each partition, we search for fragment orders using the ungauged pairwise Choi cost as a distance matrix and then apply 2-opt local refinements. For each partition and order, we check optimize for each gauge choice using only the degenerate and
null eigenspaces of retained fragments.

\begin{table*}[t]
    \centering
    \scriptsize
    \setlength{\tabcolsep}{3.5pt}
    \caption{\textbf{Gaussian costs for molecular systems are strongly representation dependent.}
    For each STO-3G molecule, $N$ is the number of spatial orbitals,
    $L_{\mathrm{full}}=N(N+1)/2$ is the upper bound on the number of fragments, and
    $L_{\mathrm{active}}$ is the number of retained fragments after
    choosing the most aggressive weighted-$L_2$ truncation threshold whose
    induced CCSD(T) correlation-energy error satisfies
    $|\Delta E_{\mathrm{corr}}| < 1.6~\mathrm{mHa}$. The cost columns report the total distribution cost  DF Trotter step, rounded to the nearest 
    power of two. The costs measured treat each diagonal and Gaussian layers independently, meaning that the total cost is multiplicative across all layers. The Choi columns
    are lower-bound diagnostics for the corresponding Gaussian representation,
    while the CS column is the cost of the constructive cosine-sine protocol evaluated on the
    same optimized representation.}
    \label{tab:molecular_gaussian_appendix}
    \begin{tabular}{@{}lrrrrrrr@{}}
        \toprule
        Molecule
            & $N$
            & $L_{\mathrm{full}}$
            & $L_{\mathrm{active}}$
            & \shortstack{untruncated default\\$\log_2(\gamma_{\mathrm{Choi}}^{(\mathrm{Trotter})})$}
            & \shortstack{truncated default\\$\log_2(\gamma_{\mathrm{Choi}}^{(\mathrm{Trotter})})$}
            & \shortstack{optimized\\$\log_2(\gamma_{\mathrm{Choi}}^{(\mathrm{Trotter})})$}
            & \shortstack{optimized\\$\log_2(\gamma_{\mathrm{CS}}^{(\mathrm{Trotter})})$} \\
        \midrule
        H$_2$    &  2 &   3 &  2 &   10 &   7 &  7 &  7 \\
        LiH      &  6 &  21 &  6 &  101 &  30 & 10 & 10 \\
        HF       &  6 &  21 & 11 &  101 &  52 & 18 & 18 \\
        BeH$_2$  &  7 &  28 &  7 &  151 &  42 & 10 & 10 \\
        H$_2$O   &  7 &  28 & 12 &  149 &  69 & 22 & 24 \\
        NH$_3$   &  8 &  36 & 14 &  230 &  89 & 27 & 31 \\
        CH$_4$   &  9 &  45 & 14 &  285 &  91 & 35 & 40 \\
        CO       & 10 &  55 & 10 &  417 &  74 & 24 & 28 \\
        HCl      & 10 &  55 & 24 &  411 & 176 & 21 & 23 \\
        N$_2$    & 10 &  55 & 10 &  422 &  80 & 26 & 28 \\
        F$_2$    & 10 &  55 & 20 &  420 & 152 & 37 & 40 \\
        H$_2$S   & 11 &  66 & 25 &  543 & 214 & 53 & 59 \\
        H$_2$CO  & 12 &  78 & 23 &  685 & 202 & 50 & 56 \\
        CO$_2$   & 15 & 120 & 14 & 1288 & 149 & 39 & 44 \\
        \bottomrule
    \end{tabular}
\end{table*}

\Cref{tab:molecular_gaussian_appendix} shows that chemistry-calibrated
truncation can substantially shorten the Gaussian path, and that representation
choices within the retained Hamiltonian can reduce the exponent further. The
optimized columns lower the distribution cost by tens to hundreds of powers of two across the benchmark
set, corresponding to large powers in the distribution cost. The
optimized CS exponents also track the optimized Choi diagnostics closely for
these instances, suggesting that the constructive Gaussian protocol captures
most of the improvement indicated by the lower-bound calculation. The searches used in \cref{tab:molecular_gaussian_appendix} are also relatively
modest. Broader searches over orbital bases, partitions, fragment orderings, and
gauges could identify still more efficient representations. The main point is
that distribution complexity becomes a computable objective for choosing among
electronic-structure representations before compiling the simulation into
elementary gates.

%=============================================================================
\section{Local orbital bases and projection}
\label{app:localized_orbitals}
%=============================================================================

The fermionic modes used to define the ERI tensor can be any orthonormal basis of a finite single-particle space. Two orthonormal bases can describe the same Fock space, but express different ERI tensor elements and therefore have very different double-factorized fragment rotations. This appendix illustrates that basis choice can either expose or obscure locality across a processor cut. This is especially true for spatially localized systems
% Working in one basis over another might dramatically change the properties of the ERI tensor, and therefore the Gaussian layers between tensor fragments, that contribute directly to the distribution complexity of the Hamiltonian, especially spatially localized systems. 
If $\{a^\dagger_j\}$ are the creation operators for a given set of orbitals $\{\chi_i\}$, then the transformation to 
\begin{equation}
    b^\dagger_i = \sum_j U_{ji} \, a^\dagger_j 
\end{equation}
is another set of creation operators for the set of orbitals $\{\phi_i\}$ which span the same single-particle space as $\{\chi_i\}$ if and only if $U$ is a unitary matrix. Orbital rotations are exactly implemented in Fock space by Gaussian fermionic unitaries. The point is that working in one basis or another transforms the electronic Hamiltonian coefficients by the consistent transformation:
\begin{align}
    h'_{pq}
      &=
      \sum_{ij} U^*_{ip} h_{ij} U_{jq},\\
    V'_{pqrs}
      &=
      \sum_{ijkl}
      U^*_{ip}U_{jq}U^*_{kr}U_{ls} V_{ijkl}.
\end{align}
where $h'$ and $V'$ are the one- and two-body integrals in the $\{\phi_i\}$ orbital basis in terms of the corresponding integrals $h$ and $V$ in the original $\{\chi_i\}$ basis.
Working in one basis or another might cause the decomposition of the ERI tensor to have different spectral properties. This is not to be confused with working in basis sets which span inequivalent Fock spaces: for example, STO-3G and cc-pVDZ denote distinct common orbital wavefunction approximations and are not related by orbital rotations (and generally have different numbers of orbitals for the same system of atoms and electrons). 

Atomic-orbital functions are usually not orthogonal, and need to be orthogonalized to arrive at orthonormal fermionic modes to define Fock space. Orbitals can be orthogonalized with different objectives in mind, a prime example being locality preservation. For example, Löwdin orthogonalization preserves the atom-centered character of the resulting basis as much as possible under a symmetric orthogonalization~\cite{lowdin1950nonorthogonality}. Boys and Pipek--Mezey localized orbitals provide a variational procedure that rotate an orthonormal molecular-orbital basis to optimize a localization criterion~\cite{boys1960construction,pipek1989fast}. In all cases, preserving locality can concentrate large electronic integrals within subsets of nearby modes, which we illustrate in \cref{fig:h2_lowdin_locality_simple}.

A common choice of basis for quantum algorithms is the canonical Hartree--Fock molecular orbitals. These are the single-particle eigenmodes of the self-consistent Fock operator for a system, which is the effective quadratic Hamiltonian which best approximates the system. In this basis, the Hartree--Fock reference state is a computational basis state, which makes it convenient to work in for state preparation when using mean-field approximations as a starting point. However, these are generally delocalized orbitals that extend across the molecular system, and they may not align well with any single processor partition.

\begin{figure}[t]
    \centering
    \includegraphics[width=\textwidth]{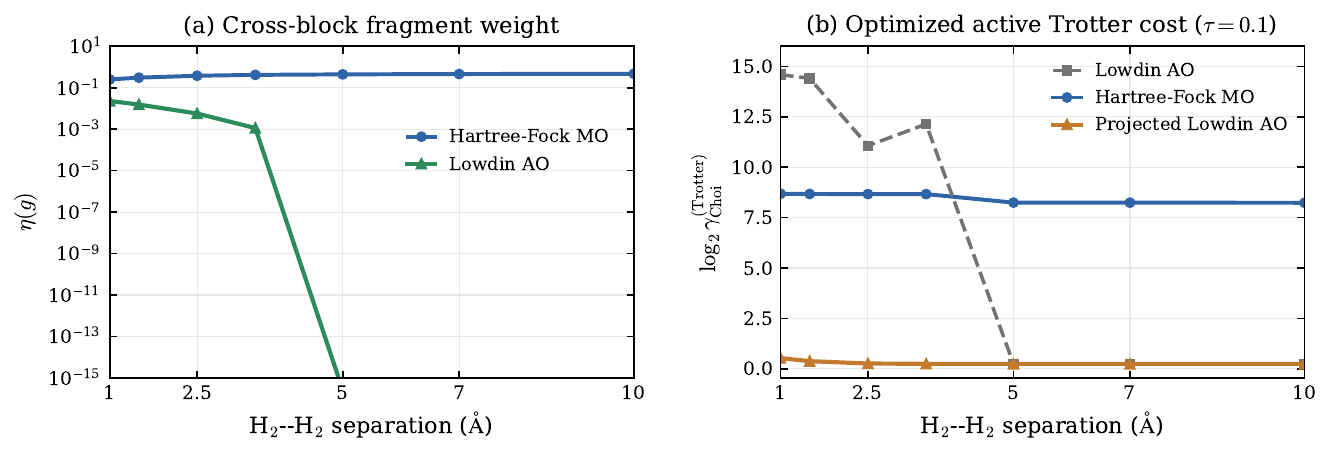}
    \caption{\textbf{Local orbital bases and cross-boundary near-degeneracies.}
    Molecular orbital bases can be highly nonlocal and exhibit much greater cross partition strength than other orthogonalization procedures despite spanning the same single-particle Fock space. We illustrate this for two H$_2$ molecules in the STO-3G basis, with one molecule assigned to each side of the processor cut. \textbf{(a)} We plot the $|\lambda_\ell|$-weighted cross-block fragment weight $\eta = \sum \lambda_\ell \, \eta(g^{(\ell)}) / (\sum \lambda_\ell)$ for the active double-factorized fragments, where $\eta(g^{(\ell)})=\|g^{(\ell)}-\Pi_{A|B}(g^{(\ell)})\|_F/\|g^{(\ell)}\|_F$. The Hartree--Fock molecular orbitals delocalize across the two molecules even as the intermolecule separation increases, while the L\"owdin atomic-orbital representation rapidly becomes block local as the molecules separate. \textbf{(b)} The optimized distribution cost for the Trotter step of this system at step size $\tau = 0.1$, including the Gaussian and diagonal layers. The projected L\"owdin curve removes the cross-block part of each fragment and serves as a maximally local reference. For the Hartree--Fock and even unprojected L\"owdin orbital bases at small separation, the distribution remains high.}
\label{fig:h2_lowdin_locality_simple}
\end{figure}

In a highly localized basis, the double factorization fragment matrices $g^{(\ell)}$ may have very weak cross-partition elements. For a spatial mode partition $A|B$, define the block projection
\begin{equation}
    \Pi_{A|B}(g^{(\ell)})
    =
    P_A g^{(\ell)} P_A + P_B g^{(\ell)} P_B ,
    \qquad P_B=I-P_A .
    \label{eq:block_projection_appendix}
\end{equation}
and the normalized cross-block weight
\begin{equation}
    \eta(g^{(\ell)})
    =
    \frac{\|g^{(\ell)}-\Pi_{A|B}(g^{(\ell)})\|_F}{\|g^{(\ell)}\|_F}.
\end{equation}
In \cref{fig:h2_lowdin_locality_simple}(a), we plot the corresponding $|\lambda_\ell|$-weighted average over the active double-factorized fragments. This gives a basis-level diagnostic for how much of the fragment representation crosses the processor cut before any circuit decomposition is chosen. 
Each projected fragment matrix can be diagonalized via a block-local rotation,
$R_\ell^{\mathrm{blk}}=R_{\ell,A}\oplus R_{\ell,B}$, so its Gaussian
diagonalization layer has $\gamma_{\mathrm{Choi}}=1$ across $A|B$.  The Hamiltonian error from this replacement is controlled at the ERI level by
\begin{equation}
    \Delta V
    =
    \sum_{\ell\in S}
    \lambda_\ell
    \left[
      g^{(\ell)}\otimes g^{(\ell)}
      -
      \Pi_{A|B}(g^{(\ell)})\otimes \Pi_{A|B}(g^{(\ell)})
    \right],
    \label{eq:block_projection_eri_error}
\end{equation}
where $S$ is the set of projected fragments and the tensor product denotes the
outer product in the reshaped ERI supermatrix.  This gives a direct way to
trade distribution overhead against ERI error, and in small systems the same
approximation can be calibrated by the change in correlation energy.

As we discuss in \cref{sec:gauge}, any degenerate eigenspaces of $g^{(\ell)}$, including the nullspace associated with zero or truncated
$\varepsilon_j^{(\ell)}$, lead to gauge degrees of freedom in the eigendecomposition. Near-degeneracies make this choice unstable under
small changes in geometry, basis, or truncation threshold. A nearly
block-diagonal matrix with almost equal diagonal entries can have delocalized
eigenvectors even when its cross-block entries are tiny.  Since
\cref{eq:gamma_choi} depends on the singular values of subblocks of the chosen
Gaussian rotation, a raw numerical eigensolver can overstate the distribution
cost of a fragment that is physically close to block local.  The gauge and
nullspace optimizations in the main text are designed to remove this
coordinate artifact when possible.  Block projection is a stronger approximation, which specifically truncates against the partition directly, enforcing local structure of fragments.

Finally, block-local diagonalization of $g^{(\ell)}$ should not be confused with
removing all nonlocal dynamics.  If a projected fragment has
$D_\ell=D_{\ell,A}+D_{\ell,B}$, then
\begin{equation}
    D_\ell^2
    =
    D_{\ell,A}^2 + D_{\ell,B}^2
    + 2D_{\ell,A}D_{\ell,B}.
\end{equation}
The Gaussian change of basis can be local while the diagonal squared fragment
still contains the cross-boundary interaction, under the block-diagonal projection approximation.  

\section{Localized Fragments gap and interface sweep}
\label{app:localized-fragments-gap-interface-sweep}

\begin{figure}[!ht]
    \centering
    \includegraphics[width=\textwidth]{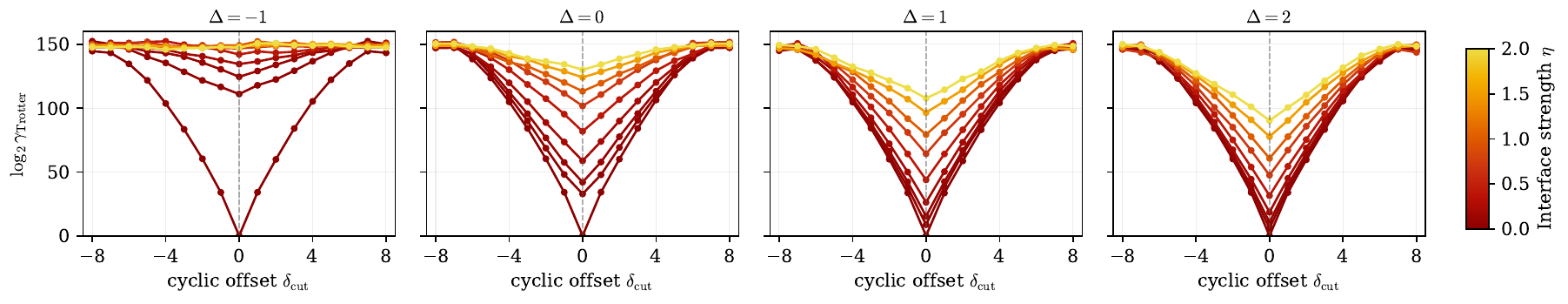}
    \caption{\textbf{Localized Fragments numerical stability and interface sweep.}
    We plot the interfragment Gaussian contribution
    $\log_2\gamma_{\mathrm{Trotter}}(A_j)$ for cyclic mode partitions of the Localized Fragments ensemble.
    Each panel fixes the block spectral separation $\Delta$, defined such that $\Delta = \lambda_\mathrm{max}(A_\ell) - \lambda_\mathrm{min}(B_\ell)$ for each fragment. Each colored curve fixes the
    interface strength $\eta$. For $\Delta>0$, the local block spectra are separated and the
    planted partition remains inexpensive until the interface is large enough
    to hybridize modes across the cut. At $\Delta=0$, the bands touch and the
    planted cost rises more quickly with $\eta$. At $\Delta=-1$, the bandwidths
    overlap completely, so even a small nonzero interface rapidly degrades the planted-cut
    advantage as the eigenvectors' locality becomes numerically unstable to perturbation.}
    \label{fig:localized-fragments-gap-interface-sweep}
\end{figure}

We now discuss the Localized Fragments ensembles used in \cref{sec:model_localized_fragments} in more detail.
The goal is to isolate how the interfragment Gaussian cost
depends on two physical parameters: the cross-block interface strength $\eta$
and in particular the spectral separation $\Delta$ between the two local blocks.
For each fragment $\ell=1,\ldots,L$, with $N=30$, $L=8$, and $N_A=N/2=15$, we
construct a real symmetric single-particle matrix
\begin{equation}
g^{(\ell)}(\eta)=
\begin{pmatrix}
A_\ell & \eta X_\ell\\
\eta X_\ell^\dagger & B_\ell
\end{pmatrix}.
\end{equation}
For each $\ell$, we draw independent Haar-random orthogonal matrices
$Q_{\ell,A}$ and $Q_{\ell,B}$ to serve as the eigenbases of blocks $A_\ell$ and $B_\ell$.
The off-diagonal block $X_\ell$ is drawn by sampling an element $G_\ell$ of the Gaussian orthogonal ensemble (GOE) and then normalizing by
spectral norm,
\begin{equation}
X_\ell = G_\ell / \|G_\ell\|_2,
\end{equation}
so that $\eta$ directly controls the interface scale.
The spectra of the two diagonal blocks have fixed bandwidth equal to 1 and are generated in the following way 
Starting from an evenly spaced grid on $[0,1]$, we add Gaussian noise with variance $\sigma^2=0.1$, sort the resulting values, 
and rescale them back to the interval $[0,1]$. This gives two ordered spectra
$t^{A}_{\ell,i}$ and $t^{B}_{\ell,i}$. We then set
\begin{equation}
\lambda_i(A_\ell)=-\Delta/2-t^{A}_{\ell,i},
\qquad
\lambda_i(B_\ell)=\Delta/2+t^{B}_{\ell,i},
\end{equation}
and define
\begin{equation}
A_\ell=Q_{\ell,A}\operatorname{diag}(\lambda(A_\ell))Q_{\ell,A}^\dagger,
\qquad
B_\ell=Q_{\ell,B}\operatorname{diag}(\lambda(B_\ell))Q_{\ell,B}^\dagger.
\end{equation}
The same base random draw is reused across the plotted $\eta$ and
$\Delta$ values, so changes in the curves reflect the parameter change rather
than a change in the sample.

In \cref{fig:localized-fragments-gap-interface-sweep}, we demonstrate the relationship between the parameters $(\eta,\Delta)$ and the sensitivity of distribution cost to the optimal partition. We sweep through values of the parameters, using a
block-aware column ordering for each $R_\ell$: after diagonalization, the columns are
sorted by their total weight on the planted left block. This increases the numerical stability for comparing neighboring fragment bases.
It removes some artificial cost from numerical noise in the eigenvalue ordering, especially for small $\eta$ and overlapping spectra in the local blocks $A_\ell$ and $B_\ell$.
For each cyclically shifted partition of $N_A$ consecutive modes, the plotted value is the
interfragment Gaussian contribution
\begin{equation}
\log_2\gamma_{\mathrm{Trotter}} =
\sum_{\ell=1}^{L-1}
\log_2\gamma_{\mathrm{Choi}}\!\left(R_{\ell+1}^\dagger R_\ell\right).
\end{equation}

%=============================================================================
\section{Collective dephasing in random diagonal Hamiltonians}
\label{app:collective_dephasing}
%=============================================================================

We prove \cref{thm:slow-growth-MBL} and extend it to all R\'enyi orders.
Consider one diagonal Coulomb layer,
\begin{equation}
    \hat C_\ell(t_\ell)
    :=e^{-it_\ell \hat D_\ell^2},
    \qquad
    \hat D_\ell=\sum_{m=1}^{N}\varepsilon_m^{(\ell)}\hat n_m .
    \label{eq:app_rro_layer}
\end{equation}
Here, $t_\ell$ denotes the coefficient of $\hat D_\ell^2$ in the exponent;
for the convention in \cref{eq:trotter}, $t_\ell=\tau\lambda_\ell/2$.
Let $A|B$ be a mode partition, with $N_A=|A|$, $N_B=|B|$,
$d_A=2^{N_A}$, and $d_B=2^{N_B}$. We suppress the fragment label $\ell$
until the end of the section.

Write $\hat D=\hat D_A+\hat D_B$, where
\begin{equation}
    \hat D_A=\sum_{a\in A}\varepsilon_a\hat n_a,
    \qquad
    \hat D_B=\sum_{b\in B}\varepsilon_b\hat n_b.
\end{equation}
Since $[\hat D_A,\hat D_B]=0$,
\begin{equation}
    e^{-it\hat D^2}
    =
    \left(e^{-it\hat D_A^2}\otimes e^{-it\hat D_B^2}\right)
    e^{-i2t\hat D_A\hat D_B}.
    \label{eq:app_rro_boundary}
\end{equation}
The first factor is local across $A|B$ and does not affect the operator
Schmidt coefficients. It is therefore sufficient to study the boundary
unitary $e^{-i2t\hat D_A\hat D_B}$.

For occupation strings $x\in\{0,1\}^{N_A}$ and
$y\in\{0,1\}^{N_B}$, define the centered subsystem charges
\begin{equation}
    X_x=\sum_{a\in A}\varepsilon_a\left(x_a-\frac12\right),
    \qquad
    Y_y=\sum_{b\in B}\varepsilon_b\left(y_b-\frac12\right).
    \label{eq:app_rro_charges}
\end{equation}
We group the input and output copies of $A$ into
$\mathcal A=\mathcal H_{A,\mathrm{out}}\otimes
\mathcal H_{A,\mathrm{in}}$, and similarly define $\mathcal B$.
For a diagonal unitary, the Choi state is supported on the diagonal Choi
subspaces spanned by
$|\widetilde x\rangle_{\mathcal A}
=|x\rangle_{A,\mathrm{out}}|x\rangle_{A,\mathrm{in}}$ and
$|\widetilde y\rangle_{\mathcal B}
=|y\rangle_{B,\mathrm{out}}|y\rangle_{B,\mathrm{in}}$.
Up to local and global phases, its normalized Choi state is therefore
\begin{equation}
    |\Psi(t)\rangle
    =\frac{1}{\sqrt{d_A d_B}}
    \sum_{x,y}e^{-i2tX_xY_y}
    |\widetilde x\rangle_{\mathcal A}
    |\widetilde y\rangle_{\mathcal B}.
    \label{eq:app_rro_choi}
\end{equation}
Tracing out $\mathcal B$ gives the exact reduced Choi state
\begin{align}
    \rho_{\mathcal A}(t)
    &=\frac{1}{d_A}\sum_{x,x'}
      \chi_B\!\left(2t(X_x-X_{x'})\right)
      |\widetilde x\rangle\langle\widetilde x'|,
      \label{eq:app_rro_reduced}\\
    \chi_B(s)
    &:=\frac{1}{d_B}\sum_y e^{-isY_y}
      =\prod_{b\in B}\cos\!\left(\frac{s\varepsilon_b}{2}\right).
      \label{eq:app_rro_characteristic}
\end{align}
The exact charge variances are
\begin{equation}
    v_A^2:=\frac{1}{d_A}\sum_xX_x^2
      =\frac14\sum_{a\in A}\varepsilon_a^2,
    \qquad
    v_B^2:=\frac14\sum_{b\in B}\varepsilon_b^2.
    \label{eq:app_rro_variances}
\end{equation}
It is convenient to introduce the dimensionless collective-dephasing
strength
\begin{equation}
    g:=4|t|v_Av_B.
    \label{eq:app_rro_g}
\end{equation}
We use natural logarithms and
$S_\alpha(\rho)=(1-\alpha)^{-1}\log\Tr(\rho^\alpha)$ for
$\alpha>0$, $\alpha\neq1$, with the usual continuous limits at
$\alpha=1$ and $\alpha=\infty$. Below,
$S_\alpha\equiv S_\alpha^{\mathrm{Choi}}(\hat C_\ell(t_\ell);A)$
denotes the operator R\'enyi entropy across $A|B$.

\begin{proposition}[\textbf{R\'enyi entropies of a rank-one diagonal layer}]
\label{prop:rro_all_renyi}
For any finite set of coefficients with $v_Av_B>0$, the short-time
operator R\'enyi entropies of \cref{eq:app_rro_layer} obey
\begin{equation}
S_\alpha=
\begin{cases}
\displaystyle
\frac{g^{2\alpha}}{(1-\alpha)4^\alpha}+o(g^{2\alpha}),
&0<\alpha<1,\\[0.8em]
\displaystyle
\frac{g^2}{4}\left[1-\log\!\left(\frac{g^2}{4}\right)\right]
+o\!\left(g^2\log\frac1g\right),
&\alpha=1,\\[0.8em]
\displaystyle
\frac{\alpha}{4(\alpha-1)}g^2+o(g^2),
&1<\alpha<\infty,\\[0.8em]
\displaystyle
\frac{g^2}{4}+o(g^2),
&\alpha=\infty.
\end{cases}
\label{eq:app_rro_short_all_alpha}
\end{equation}
In particular, $S_{1/2}=g+o(g)$, whereas
$S_2=g^2/2+o(g^2)$.

In the Gaussian-charge approximation, in which the centered charges
$X$ and $Y$ are Gaussian with variances $v_A^2$ and $v_B^2$, the reduced
Choi spectrum is geometric,
\begin{equation}
    p_n=(1-q)q^n,
    \qquad
    q=\frac{\nu-1}{\nu+1},
    \qquad
    \nu=\sqrt{1+g^2},
    \qquad n=0,1,\ldots .
    \label{eq:app_rro_geometric}
\end{equation}
Consequently, for $0<\alpha<\infty$, $\alpha\neq1$,
\begin{equation}
    S_\alpha(g)
    =\frac{1}{\alpha-1}
      \log\!\left[
      \left(\frac{\nu+1}{2}\right)^\alpha
      -\left(\frac{\nu-1}{2}\right)^\alpha
      \right].
    \label{eq:app_rro_all_alpha}
\end{equation}
The continuous limits are
\begin{align}
    S_1(g)
    &=\frac{\nu+1}{2}\log\frac{\nu+1}{2}
      -\frac{\nu-1}{2}\log\frac{\nu-1}{2},
      \label{eq:app_rro_vn}\\
    S_\infty(g)
    &=\log\frac{\nu+1}{2}.
    \label{eq:app_rro_min}
\end{align}
Two useful special cases are
\begin{equation}
    S_{1/2}(g)=\operatorname{arsinh}(g),
    \qquad
    S_2(g)=\frac12\log(1+g^2).
    \label{eq:app_rro_half_second}
\end{equation}
For every fixed $\alpha>0$, before finite-size saturation,
\begin{equation}
    S_\alpha(g)=\log g+O_\alpha(1),
    \qquad g\gg1.
    \label{eq:app_rro_log_growth}
\end{equation}
For every finite system, the Choi lower bound satisfies the exact relation
$\gamma_{\mathrm{Choi}}=2e^{S_{1/2}}-1$. In the Gaussian-charge
approximation this becomes
\begin{equation}
    \gamma_{\mathrm{Choi}}^{(\mathrm G)}
    =2\left(g+\sqrt{1+g^2}\right)-1.
    \label{eq:app_rro_gamma}
\end{equation}
Consequently, $\gamma_{\mathrm{Choi}}-1=2g+o(g)$ for every finite
system at short times, while
$\gamma_{\mathrm{Choi}}^{(\mathrm G)}=4g+O(1)$ for $g\gg1$ in the
pre-saturation regime.
\end{proposition}

\begin{proof}
We first prove the short-time statement without a disorder assumption.
Let
\begin{equation}
    |0_A\rangle=\frac{1}{\sqrt{d_A}}\sum_x|\widetilde x\rangle,
    \qquad
    |0_B\rangle=\frac{1}{\sqrt{d_B}}\sum_y|\widetilde y\rangle,
\end{equation}
and define diagonal charge operators
$\hat X|\widetilde x\rangle=X_x|\widetilde x\rangle$ and
$\hat Y|\widetilde y\rangle=Y_y|\widetilde y\rangle$.
Centering implies $\langle0_A|\hat X|0_A\rangle=
\langle0_B|\hat Y|0_B\rangle=0$. Hence
\begin{equation}
    |1_A\rangle:=\frac{\hat X|0_A\rangle}{v_A},
    \qquad
    |1_B\rangle:=\frac{\hat Y|0_B\rangle}{v_B}
\end{equation}
are normalized and orthogonal to $|0_A\rangle$ and $|0_B\rangle$.
Expanding \cref{eq:app_rro_choi},
\begin{equation}
    |\Psi(t)\rangle
    =|0_A0_B\rangle
     -i2t v_Av_B|1_A1_B\rangle+O(t^2),
    \label{eq:app_rro_state_expansion}
\end{equation}
where the remainder is in Hilbert-space norm. The first two terms have
orthogonal left and right Schmidt vectors. Therefore the two leading
Schmidt coefficients are
\begin{equation}
    s_0=1+O(t^2),
    \qquad
    s_1=2|t|v_Av_B+O(t^2)=\frac{g}{2}+O(g^2),
\end{equation}
and all remaining Schmidt coefficients are $O(t^2)$. The reduced Choi
eigenvalues $p_j=s_j^2$ consequently satisfy
\begin{equation}
    p_0=1-\frac{g^2}{4}+o(g^2),
    \qquad
    p_1=\frac{g^2}{4}+o(g^2),
    \qquad
    \sum_{j\ge2}p_j=O(g^4).
    \label{eq:app_rro_small_spectrum}
\end{equation}
Substitution into the definition of $S_\alpha$ gives
\cref{eq:app_rro_short_all_alpha}. The distinction between
$S_{1/2}$ and $S_2$ follows directly from
$\sqrt{p_1}=O(|t|)$ but $p_1=O(t^2)$.

We next diagonalize the Gaussian-charge model. Replacing the empirical
charge distributions by
\begin{equation}
    P_A(X)=\frac{e^{-X^2/(2v_A^2)}}{\sqrt{2\pi v_A^2}},
    \qquad
    P_B(Y)=\frac{e^{-Y^2/(2v_B^2)}}{\sqrt{2\pi v_B^2}},
\end{equation}
the normalized Choi wavefunction is
$\Psi_t(X,Y)=\sqrt{P_A(X)P_B(Y)}e^{-i2tXY}$.
Tracing out $Y$ gives
\begin{equation}
    \rho_A(X,X')
    =\frac{1}{\sqrt{2\pi v_A^2}}
      \exp\!\left[
      -\frac{X^2+X'^2}{4v_A^2}
      -2t^2v_B^2(X-X')^2
      \right].
    \label{eq:app_rro_gaussian_kernel}
\end{equation}
With $u=X/(\sqrt2v_A)$ and $u'=X'/(\sqrt2v_A)$, the unitarily
rescaled kernel is
\begin{equation}
    \widetilde\rho_A(u,u')
    =\frac{1}{\sqrt\pi}
      \exp\!\left[
      -\frac{2+g^2}{4}(u^2+u'^2)
      +\frac{g^2}{2}uu'
      \right].
    \label{eq:app_rro_mehler_kernel}
\end{equation}
Let $\phi_n^{(\nu)}$ denote the normalized Hermite functions of width
$\nu$. Mehler's formula, with
$\nu=\sqrt{1+g^2}$ and $q=(\nu-1)/(\nu+1)$, gives
\begin{equation}
    \widetilde\rho_A(u,u')
    =(1-q)\sum_{n=0}^{\infty}q^n
      \phi_n^{(\nu)}(u)\phi_n^{(\nu)}(u').
    \label{eq:app_rro_mehler_decomposition}
\end{equation}
This proves \cref{eq:app_rro_geometric}. Summing the geometric series,
\begin{equation}
    \Tr\rho_A^\alpha
    =\frac{(1-q)^\alpha}{1-q^\alpha}
    =\frac{1}{
      \left(\frac{\nu+1}{2}\right)^\alpha
      -\left(\frac{\nu-1}{2}\right)^\alpha},
\end{equation}
which gives \cref{eq:app_rro_all_alpha} and its continuous limits.
The identities in \cref{eq:app_rro_half_second} follow by setting
$\alpha=1/2$ and $\alpha=2$.

At large $g$, $\nu=g+O(g^{-1})$. For fixed $\alpha\neq1,\infty$,
\begin{equation}
    S_\alpha(g)
    =\log\frac{g}{2}+\frac{\log\alpha}{\alpha-1}+O(g^{-1}),
    \label{eq:app_rro_large_alpha}
\end{equation}
while $S_1(g)=\log(g/2)+1+O(g^{-1})$ and
$S_\infty(g)=\log(g/2)+O(g^{-1})$. This proves
\cref{eq:app_rro_log_growth}. Finally, the short-time statement for $\gamma_{\mathrm{Choi}}$ follows
from $S_{1/2}=g+o(g)$, while \cref{eq:app_rro_gamma} follows from
$e^{\operatorname{arsinh}(g)}=g+\sqrt{1+g^2}$.
\end{proof}

The Gaussian-charge approximation is supported by a weighted central-limit
condition. Indeed, under the uniform sum over occupation strings,
$X=(1/2)\sum_{a\in A}\varepsilon_a r_a$, where the $r_a=\pm1$ are
independent and equiprobable. If
\begin{equation}
    \frac{\sum_{a\in A}\varepsilon_a^4}
    {\left(\sum_{a\in A}\varepsilon_a^2\right)^2}\longrightarrow0,
    \qquad
    \frac{\sum_{b\in B}\varepsilon_b^4}
    {\left(\sum_{b\in B}\varepsilon_b^2\right)^2}\longrightarrow0,
    \label{eq:app_rro_lindeberg}
\end{equation}
then $X/v_A$ and $Y/v_B$ converge in distribution to standard normal
variables. This follows by expanding the characteristic function,
$\prod_a\cos(s\varepsilon_a/(2v_A))=e^{-s^2/2+o(1)}$, and similarly
for $B$. The condition holds with high probability for independent
Gaussian coefficients and for vectors sampled uniformly from
$S^{N-1}$.

Restoring the fragment label, for
$\varepsilon_m^{(\ell)}\sim\mathcal N(0,\sigma_\ell^2)$,
\begin{equation}
    \omega_\ell
    =4|t_\ell|v_{A,\ell}v_{B,\ell}
    \simeq
    |t_\ell|\sigma_\ell^2\sqrt{N_AN_B}.
    \label{eq:app_rro_g_typical}
\end{equation}
For the Trotter convention in \cref{eq:trotter},
$t_\ell=\tau\lambda_\ell/2$, and hence
\begin{equation}
    \omega_\ell
    \simeq
    \frac{|\tau\lambda_\ell|}{2}\sigma_\ell^2\sqrt{N_AN_B}.
    \label{eq:app_rro_g_trotter}
\end{equation}
For a balanced cut, $\omega_\ell\simeq
|\tau\lambda_\ell|\sigma_\ell^2N/4$. If
$\boldsymbol\varepsilon^{(\ell)}$ is normalized to unit Euclidean norm,
then $\sigma_\ell^2\simeq1/N$, and the pre-saturation spectrum depends
on $\tau\lambda_\ell$ but not on $N$ to leading order.

For completeness, the order-zero R\'enyi entropy behaves differently.
For a finite system, it is the logarithm of the rank of the coefficient
matrix
\begin{equation}
    C_{xy}(t)=\frac{1}{\sqrt{d_Ad_B}}e^{-i2tX_xY_y},
\end{equation}
so $S_0\leq\min(N_A,N_B)\log2$. If all $X_x$ and all $Y_y$ are
pairwise distinct, an $r\times r$ minor, with
$r=\min(d_A,d_B)$, has the small-$t$ expansion
\begin{equation}
\det[e^{zX_iY_j}]_{i,j=1}^{r}
=\frac{z^{r(r-1)/2}}{\prod_{k=0}^{r-1}k!}
 \prod_{i<j}(X_j-X_i)\prod_{i<j}(Y_j-Y_i)
 +O\!\left(z^{r(r-1)/2+1}\right),
\end{equation}
where $z=-i2t$. The leading coefficient is nonzero, so the matrix has
maximal rank for sufficiently small $t\neq0$, and for all $t$ outside a
discrete exceptional set. Thus, for continuous disorder,
$S_0(t)=\min(N_A,N_B)\log2$ almost surely for generic nonzero $t$,
whereas $S_0(0)=0$. This discontinuity is why the logarithmic-growth
statement in \cref{eq:app_rro_log_growth} is made for fixed
$\alpha>0$.

%=============================================================================
\section{Operator R\'enyi entropies of Gaussian unitaries}
\label{app:operator_entanglement_Gaussian_covariance}
%=============================================================================

Let us first remind ourselves of the Schmidt decomposition for pure bipartite states, \(|\psi\rangle\in\mathcal{H}_{AB}\). Any such state can be written as
\begin{align}
|\psi\rangle=\sum_{j=1}^{\chi}\sqrt{\lambda_j}\,|j_A\rangle\otimes |j_B\rangle,
\end{align}
where \(\chi\le \min\{d_A,d_B\}\) is the Schmidt rank and \(\{|j_A\rangle\},\{|j_B\rangle\}\) are orthonormal bases in each subsystem, respectively. The Schmidt rank being greater than one is necessary and sufficient for a pure quantum state to be entangled; and more generally, one can define entanglement monotones by computing scalar entropic functions of the Schmidt values \(\{\lambda_j\}\), such as the von Neumann entropy.

In a similar spirit, the operator Schmidt decomposition of a unitary \(U\) acting on a bipartite quantum system quantifies the minimal number of product operators, such as Pauli strings, that are needed to represent \(U\). Namely,
\begin{align}
U=\sum_{j=1}^{\mathcal{X}}\sqrt{\lambda_j}\,V_j\otimes W_j,
\end{align}
where \(\mathcal{X}\le \min\{d_A^2,d_B^2\}\) is the operator Schmidt rank of the unitary and \(\{V_j\},\{W_j\}\) represent orthonormal bases of operators with respect to the Hilbert-Schmidt inner product. A unitary is a product unitary if and only if \(\mathcal{X}=1\), otherwise it is an operator-space entangled unitary. Similar to the pure state case, one can compute the entropy of the operator Schmidt coefficients \(\{\lambda_j\}\) to quantify operator-space entanglement of the unitary. It is worth emphasizing that operator-space entanglement of unitaries quantifies a number of operational characteristics such as the cost of distributed quantum simulation of \(U\), the cost of classically simulating \(U\) via a matrix product operator (MPO) representation, etc. In fact, operator entanglement has long been used to quantify the entanglement of \textit{mixed states}, under the name of the \textit{realignment criterion}, and has a rich history in quantum information theory; see e.g., \cite{Gittsovich2008} and references therein. For our discussion, the operator entanglement quantifies the distribution complexity of simulating \(U\) locally.

Generically, computing operator entanglement of a unitary is equivalent to performing a singular value decomposition of its realigned matrix. This entails a classical complexity \(O(d^3)\), where \(d=2^N\times 2^N\), namely exponential in the number of qubits. As a result, computing operator entanglement is generically intractable for many-body quantum dynamics. However, for special classes of unitaries, such as Cliffords, matchgates, and linear optical unitaries, one can often compute this more efficiently. This relies on using a compact representation of the unitary itself, which is intimately related to the fact that these unitaries form non-universal subgroups of the full unitary group. We will now show that for Gaussian unitaries that appear in the Trotter circuits of double-factorized electronic structure Hamiltonians, one can compute their operator entanglement with a covariance matrix that scales linearly with the system size, as opposed to exponential. This reduction relies on the standard covariance matrix formalism for free fermionic systems; see e.g., \cite{peschel2003rdm} for a review.

For consistency with the preceding section, we formulate the Gaussian
result directly in terms of the normalized Choi state and its operator
R\'enyi entropies. Let $\hat U$ be an even fermionic Gaussian unitary on
$N$ modes. In a Majorana basis $\gamma_1,\ldots,\gamma_{2N}$, write
\begin{equation}
    \hat U^\dagger\gamma_j\hat U
    =\sum_{k=1}^{2N}O_{jk}\gamma_k,
    \qquad O\in\mathrm{SO}(2N).
    \label{eq:app_gaussian_majorana_action}
\end{equation}
For a mode partition $A|B$, $O[A,A]$ denotes the
$2N_A\times2N_A$ block containing both Majoranas of every mode in $A$. For a bipartition \(A|B\), the operator entanglement of \(U\) can be obtained without constructing the \(2^N\times 2^N\) matrix of \(U\). One only needs the \(AA'\) block of \(R\). Namely, we can perform the realignment directly on the covariance matrix without ever actually computing a SVD of an exponentially large matrix.

\begin{proposition}[\textbf{All-order operator entanglement of a Gaussian unitary}]
\label{prop:gaussian_all_renyi}
Let $\nu_1,\ldots,\nu_{2N_A}$ be the singular values of $O[A,A]$.
The eigenvalues of the reduced normalized Choi state are
\begin{equation}
    p_{\mathbf n}
    =\prod_{j=1}^{2N_A}
      \frac{1+(-1)^{n_j}\nu_j}{2},
    \qquad
    \mathbf n\in\{0,1\}^{2N_A}.
    \label{eq:app_gaussian_spectrum}
\end{equation}
Therefore, for $0<\alpha<\infty$, $\alpha\neq1$,
\begin{equation}
    S_\alpha^{\mathrm{Choi}}(\hat U;A)
    =\frac{1}{1-\alpha}
      \sum_{j=1}^{2N_A}
      \log\!\left[
      \left(\frac{1+\nu_j}{2}\right)^\alpha
      +\left(\frac{1-\nu_j}{2}\right)^\alpha
      \right].
    \label{eq:app_gaussian_all_alpha}
\end{equation}
The remaining R\'enyi orders are
\begin{align}
    S_1^{\mathrm{Choi}}(\hat U;A)
    &=\sum_{j=1}^{2N_A}h\!\left(\frac{1+\nu_j}{2}\right),
    \\
    S_\infty^{\mathrm{Choi}}(\hat U;A)
    &=-\sum_{j=1}^{2N_A}\log\!\left(\frac{1+\nu_j}{2}\right),
    \\
    S_0^{\mathrm{Choi}}(\hat U;A)
    &=r_A\log2,
    \qquad
    r_A:=\#\{j:\nu_j<1\},
    \label{eq:app_gaussian_limits}
\end{align}
where $h(p)=-p\log p-(1-p)\log(1-p)$. For a particle-preserving Gaussian unitary $\hat R$ induced by the real
orbital rotation $R\in\mathrm{SO}(N)$, each singular value $\sigma_k$ of
$R[A,A]$ occurs twice among the $\nu_j$. Hence
\begin{equation}
    S_\alpha^{\mathrm{Choi}}(\hat R;A)
    =\frac{2}{1-\alpha}
      \sum_{k=1}^{N_A}
      \log\!\left[
      \left(\frac{1+\sigma_k}{2}\right)^\alpha
      +\left(\frac{1-\sigma_k}{2}\right)^\alpha
      \right].
    \label{eq:app_gaussian_particle_all_alpha}
\end{equation}
In particular,
\begin{align}
    S_{1/2}^{\mathrm{Choi}}(\hat R;A)
    &=2\sum_{k=1}^{N_A}
      \log\!\left(1+\sqrt{1-\sigma_k^2}\right),
      \label{eq:app_gaussian_half}\\
    S_1^{\mathrm{Choi}}(\hat R;A)
    &=2\sum_{k=1}^{N_A}
      h\!\left(\frac{1+\sigma_k}{2}\right),
      \label{eq:app_gaussian_vn}\\
    S_2^{\mathrm{Choi}}(\hat R;A)
    &=-2\sum_{k=1}^{N_A}
      \log\!\left(\frac{1+\sigma_k^2}{2}\right),
      \label{eq:app_gaussian_second}
\end{align}
The half-R\'enyi formula gives
\begin{equation}
    \gamma_{\mathrm{Choi}}(\hat R;A)
    =2e^{S_{1/2}^{\mathrm{Choi}}}-1
    =2\prod_{k=1}^{N_A}
      \left(1+\sqrt{1-\sigma_k^2}\right)^2-1,
\end{equation}
which recovers \cref{thm:gaussian_choi_bound}.
\end{proposition}

\begin{proof}
Regard $\hat U$ as a normalized vector $|\hat U)$ in Hilbert--Schmidt
operator space, with $(X|Y)=2^{-N}\Tr(X^\dagger Y)$. Let $\Pi$ be the
operator-parity map, and define left and right super-Majoranas by
\begin{equation}
    \ell_j|X)=|\gamma_jX),
    \qquad
    r_j|X)=i\Pi|X\gamma_j).
\end{equation}
On a Majorana monomial $X$, the parity map acts as
$\Pi|X)=(-1)^{|X|}|X)$, where $|X|$ is its fermion parity. The
"super-Majoranas" (Majorana superoperators) satisfy the Majorana anticommutation relations. With
$\eta=(\ell_1,\ldots,\ell_{2N},r_1,\ldots,r_{2N})$, define
\begin{equation}
    \Omega_{\mu\nu}
    :=\frac{i}{2}(\hat U|[\eta_\mu,\eta_\nu]|\hat U).
\end{equation}
The covariance matrix of the Gaussian operator vector $|\hat U)$ is
\begin{equation}
    \Omega
    =\begin{pmatrix}0&O\\-O^\dagger&0\end{pmatrix}.
    \label{eq:app_gaussian_operator_covariance}
\end{equation}
Indeed, the left-left and right-right blocks vanish, while the cross
block follows from
$2^{-N}\Tr(\hat U^\dagger\gamma_j\hat U\gamma_k)=O_{jk}$.

The operator Schmidt decomposition across $A|B$ is the ordinary Schmidt
decomposition of $|\hat U)$ after grouping the output and input copies
of $A$ into $\mathcal A$, and similarly for $B$. At the covariance-matrix
level this realignment only permutes
$(\ell_A,\ell_B,r_A,r_B)$ into
$(\ell_A,r_A,\ell_B,r_B)$. The restricted covariance matrix on
$\mathcal A$ is therefore
\begin{equation}
    \Omega_{\mathcal A}
    =\begin{pmatrix}
      0&O[A,A]\\
      -O[A,A]^\dagger&0
      \end{pmatrix}.
    \label{eq:app_gaussian_reduced_covariance}
\end{equation}
An SVD of $O[A,A]$ transforms this matrix into the direct sum
\begin{equation}
    \Omega_{\mathcal A}
    \sim
    \bigoplus_{j=1}^{2N_A}
    \begin{pmatrix}0&\nu_j\\-\nu_j&0\end{pmatrix}.
\end{equation}
A Gaussian density operator with one covariance block
$\left(\begin{smallmatrix}0&\nu\\-\nu&0\end{smallmatrix}\right)$ has
eigenvalues $(1\pm\nu)/2$. The reduced Choi spectrum therefore
factorizes as in \cref{eq:app_gaussian_spectrum}. Summing
$p_{\mathbf n}^\alpha$ over the independent binary indices gives
\cref{eq:app_gaussian_all_alpha}.

For a real particle-preserving orbital rotation, the Majorana
transformation is, up to a permutation of Majoranas, $O_R=R\oplus R$.
Thus the singular values of $O_R[A,A]$ are
$\sigma_1,\sigma_1,\ldots,\sigma_{N_A},\sigma_{N_A}$, which gives
\cref{eq:app_gaussian_particle_all_alpha} and its special cases.
\end{proof}

\PRLrefsep
\end{document}